\PassOptionsToPackage{unicode}{hyperref}
\PassOptionsToPackage{hyphens}{url}
\PassOptionsToPackage{dvipsnames,svgnames,x11names}{xcolor}
\documentclass[
  12pt]{article}

\usepackage{amsmath,amssymb,amsthm, mathtools}
\usepackage{iftex}
\ifPDFTeX
  \usepackage[T1]{fontenc}
  \usepackage[utf8]{inputenc}
  \usepackage{textcomp} 
\else 
  \usepackage{unicode-math}
  \defaultfontfeatures{Scale=MatchLowercase}
  \defaultfontfeatures[\rmfamily]{Ligatures=TeX,Scale=1}
\fi
\usepackage{lmodern}
\ifPDFTeX\else  
\fi
\IfFileExists{upquote.sty}{\usepackage{upquote}}{}
\IfFileExists{microtype.sty}{
  \usepackage[]{microtype}
  \UseMicrotypeSet[protrusion]{basicmath} 
}{}
\makeatletter
\@ifundefined{KOMAClassName}{
  \IfFileExists{parskip.sty}{%
    \usepackage{parskip}
  }{
    \setlength{\parindent}{0pt}
    \setlength{\parskip}{6pt plus 2pt minus 1pt}}
}{
  \KOMAoptions{parskip=half}}
\makeatother
\usepackage{xcolor}
\makeatletter
\ifx\paragraph\undefined\else
  \let\oldparagraph\paragraph
  \renewcommand{\paragraph}{
    \@ifstar
      \xxxParagraphStar
      \xxxParagraphNoStar
  }
  \newcommand{\xxxParagraphStar}[1]{\oldparagraph*{#1}\mbox{}}
  \newcommand{\xxxParagraphNoStar}[1]{\oldparagraph{#1}\mbox{}}
\fi
\ifx\subparagraph\undefined\else
  \let\oldsubparagraph\subparagraph
  \renewcommand{\subparagraph}{
    \@ifstar
      \xxxSubParagraphStar
      \xxxSubParagraphNoStar
  }
  \newcommand{\xxxSubParagraphStar}[1]{\oldsubparagraph*{#1}\mbox{}}
  \newcommand{\xxxSubParagraphNoStar}[1]{\oldsubparagraph{#1}\mbox{}}
\fi
\makeatother

\usepackage{longtable,booktabs,array}
\usepackage{calc} 
\usepackage{etoolbox}
\makeatletter
\patchcmd\longtable{\par}{\if@noskipsec\mbox{}\fi\par}{}{}
\makeatother
\IfFileExists{footnotehyper.sty}{\usepackage{footnotehyper}}{\usepackage{footnote}}
\makesavenoteenv{longtable}
\usepackage{graphicx}
\makeatletter
\def\maxwidth{\ifdim\Gin@nat@width>\linewidth\linewidth\else\Gin@nat@width\fi}
\def\maxheight{\ifdim\Gin@nat@height>\textheight\textheight\else\Gin@nat@height\fi}
\makeatother
\setkeys{Gin}{width=\maxwidth,height=\maxheight,keepaspectratio}
\makeatletter
\def\fps@figure{htbp}
\makeatother

\makeatletter
\@ifpackageloaded{caption}{}{\usepackage{caption}}
\AtBeginDocument{%
\ifdefined\contentsname
  \renewcommand*\contentsname{Table of contents}
\else
  \newcommand\contentsname{Table of contents}
\fi
\ifdefined\listfigurename
  \renewcommand*\listfigurename{List of Figures}
\else
  \newcommand\listfigurename{List of Figures}
\fi
\ifdefined\listtablename
  \renewcommand*\listtablename{List of Tables}
\else
  \newcommand\listtablename{List of Tables}
\fi
\ifdefined\figurename
  \renewcommand*\figurename{Figure}
\else
  \newcommand\figurename{Figure}
\fi
\ifdefined\tablename
  \renewcommand*\tablename{Table}
\else
  \newcommand\tablename{Table}
\fi
}
\@ifpackageloaded{float}{}{\usepackage{float}}
\floatstyle{ruled}
\@ifundefined{c@chapter}{\newfloat{codelisting}{h}{lop}}{\newfloat{codelisting}{h}{lop}[chapter]}
\floatname{codelisting}{Listing}

\makeatother
\makeatletter
\@ifpackageloaded{caption}{}{\usepackage{caption}}
\@ifpackageloaded{subcaption}{}{\usepackage{subcaption}}
\makeatother

\ifLuaTeX
  \usepackage{selnolig}  
\fi
\usepackage[]{natbib}
\usepackage{bookmark}

\IfFileExists{xurl.sty}{\usepackage{xurl}}{} 
\hypersetup{
  pdftitle={Title},
  pdfauthor={Author 1; Author 2},
  pdfkeywords={3 to 6 keywords, that do not appear in the title},
  colorlinks=true,
  linkcolor={blue},
  filecolor={Maroon},
  citecolor={Blue},
  urlcolor={Blue},
  pdfcreator={LaTeX via pandoc}}

\newcommand{\anon}{1}

\newtheorem{theorem}{Theorem}
\newtheorem{proposition}[theorem]{Proposition}
\newtheorem{example}{Example}%
\newtheorem{remark}{Remark}%
\newtheorem{definition}{Definition}
\newtheorem{lemma}{Lemma}

\usepackage{appendix}

\begin{document}

\def\spacingset#1{\renewcommand{\baselinestretch}%
{#1}\small\normalsize} \spacingset{1}


\if1\anon
{
  \title{\bf Compressed Bayesian Tensor Regression}
  \author{Roberto Casarin\\
    Department of Economics, Ca' Foscari University of Venice\\
    and \\
    Radu Craiu \\
    Department of Statistical Science, University of Toronto\\
    and\\
    Qing Wang \\
    Department of Economics, Ca' Foscari University of Venice
    }
  \maketitle
} \fi

\if0\anon
{
  \bigskip
  \bigskip
  \bigskip
  \begin{center}
    {\LARGE\bf Compressed Bayesian Tensor Regression}
\end{center}
  \medskip
} \fi

\bigskip
\begin{abstract}
To address the common problem of high dimensionality in tensor regressions, we introduce a generalized tensor random projection method that embeds high-dimensional tensor-valued covariates into low-dimensional subspaces with little loss of information about the response. The method is flexible, allowing for tensor-wise, mode-wise, or combined random projections as special cases. A Bayesian inference framework is provided, featuring a hierarchical prior distribution and a low-rank parameter representation. Strong theoretical support is provided for the concentration properties of random projections and for the posterior consistency of Bayesian inference. An efficient Gibbs sampler is developed to perform inference on the compressed data. To mitigate the sensitivity introduced by random projections, Bayesian model averaging is employed, with normalizing constants estimated using reverse logistic regression. An extensive simulation study is conducted to examine the effects of different tuning parameters. Simulations indicate, and real-data applications confirm, that compressed Bayesian tensor regressions can achieve better out-of-sample predictions while significantly reducing computational costs compared to standard Bayesian tensor regressions.
\end{abstract}

\noindent%
\vfill

\newpage
\spacingset{1.8} 

\section{Introduction}
Dimensionality reduction has been a key area of interest in learning from high-dimensional data. Traditional dimensionality reduction techniques, e.g., principal component analysis (PCA), linear discriminant analysis, factor analysis, and sufficient dimensionality reduction, suffer from severe computational restrictions that increase exponentially with the dimensions of the data \citep[e.g., see][for a comparison between PCA and random projection]{dasgupta2013experiments}. 

In this paper, we consider random projection techniques that use randomly generated matrices to embed high-dimensional data points into a lower-dimensional space. Under fairly general assumptions, random projection preserves pairwise distances within a certain tolerance, as proved in the celebrated Johnson-Lindenstrauss (JL) lemma \citep{beals_extensions_1984}. Random projection has been successfully applied in statistics to reduce computational costs or to improve the efficiency of a standard method or model when applied to large datasets. For instance, \cite{indyk1998approximate}, \cite{ailon2009fast}, \cite{datar2004locality} used it for the efficient approximation of the nearest-neighbor search; \cite{chakraborty_efficient_2023}, \cite{li_tec_2021},  \cite{cannings2017random} applied it to high-dimensional classification;
\cite{guhaniyogi_bayesian_2015}, \cite{farahmand2017random}, \cite{koop_bayesian_2019} introduced random projection into inference for large dynamic regression models, and \cite{gailliot2024data}, \cite{ guhaniyogi2016compressed} applied random projection in Bayesian non-parametric regression. 

In this paper, we focus on tensor regression models, which have recently become popular across many fields for inference and statistical learning on multi-dimensional data \citep{zhou2013tensor,zhou2014regularized,guhaniyogi2017bayesian, billio_bayesian_2022, Billio02012024,luo2025bayesian,CASARIN2025105427}. We consider Bayesian scalar--on-tensor linear regressions, where dimensionality reduction is essential to reduce the number of parameters to estimate \citep{guhaniyogi2017bayesian,guhaniyogi2020bayesian}. In this sense, tensor decompositions have been used to extract factors from the covariate tensor or to parameterize the coefficient tensor in a hierarchical prior setting. However, when the number of covariates is so large that optimal factor extraction is infeasible, random projection offers a viable, easy-to-implement solution that provides strong theoretical guarantees for preserving the explanatory power of the covariates.

Different random projection strategies for tensors have been studied in the literature. \cite{rakhshan_tensorized_2020} proposed two types of tensorized random projections to map a mode-$d$ tensor to a $q$-dimensional vector using low-rank random projection tensors constructed by Tensor Train \citep{oseledets2011tensor} or canonical polyadic (CP) representations such that each entry in $\mathbb{R}^q$  is computed from the inner product of a distinct random projection tensor and the tensor predictor. \cite{shi_higher-order_2019} proposed a higher-order count sketch (HCS) that reduces the dimension of the original tensor while still preserving the higher-order data structure. 
Similarly, \cite{li_tec_2021} proposed a random projection of a tensor by exploiting its CP representation, where the random projection is performed by randomly projecting each margin from the CP decomposition to a lower dimension.

Random projection results in loss of information, and the theoretical efforts to demonstrate the limited severity of such loss rely on the  JL lemma which asserts that any set of $n$ points in the $d$-dimensional Euclidean space can be embedded into the $k$-dimensional Euclidean space such that all pairwise distances are preserved within an arbitrarily small factor, $\epsilon>0$, for $k=\mathcal{O}(\epsilon^{-2}\log n)$.  
Central to the JL embedding is a $k\times d$ random projection matrix $\boldsymbol{\Phi}$. The original recipe requires $\boldsymbol{\Phi}$ to meet three properties, namely, spherical symmetry, orthogonality, and normality \citep{ailon2009fast}. 
Several variants of JL embeddings have been developed to simplify and sharpen the lemma. \cite{indyk1998approximate} showed that the JL guarantee can still be obtained without enforcing orthogonality and normality.  \cite{achlioptas_database-friendly_2003} not only dropped the spherical symmetry condition, but also proposed a sparse way to construct the random projection matrix. Each entry is independently drawn from a discrete distribution with atoms $-\sqrt{\psi}$, $0$, and $\sqrt{\psi}$ with probability $1/2\psi$, $1-1/\psi$, and $1/2\psi$ where $\psi=1$ or $3$. To encourage sparsity in the random projection matrix and speed up computation, \cite{li_very_2006} used $\psi\gg 3$ (e.g., $\psi=\sqrt{D}$, where $D$ is the number of features, or covariates). \cite{matouvsek2008variants} considered a version of the JL lemma with independent sub-Gaussian projection entries. 

This paper brings several contributions to the existing collection of methods for Bayesian tensor regression via random projections as it: i) extends the HCS method in \cite{shi_higher-order_2019} and the projection technique in \cite{li_tec_2021} to the case of tensor predictors; ii) provides concentration inequalities for the proposed projection by exploiting some properties of the Meijer G function in a significant departure from the existing literature \citep{mathai_h-function_2010,stojanac_products_2018}; iii) integrates the projection into a tensor regression framework; iv) demonstrates posterior consistency for the proposed compressed tensor regression; v) proposes a Monte Carlo sampling procedure for posterior approximation under different prior specifications. 
The paper is organized as follows. Section \ref{sec:model} introduces the compressed tensor regression model as well as the probabilistic bounds for the tensor random projection. Section \ref{sec:mcmc} presents the Gibbs sampler for sampling the posterior density. Section \ref{sec:asy} presents the theoretical properties of posterior consistency for the coefficient,  with all proofs included in the Supplements. Section \ref{sec:numerical} presents the simulation results and a data analysis. The Python code used in this section can be found at \url{https://github.com/qingwang13/CBTR_open.git}. Section \ref{sec:conclu} concludes and discusses future directions of research.

\section{A Compressed Bayesian Tensor Model}\label{sec:model}

\subsection{Preliminaries}\label{sec:prelim}
We introduce some basic notations for tensor and multi-linear algebra, which will be used to establish the Generalized Tensor Random Projection (GTRP). A more comprehensive and detailed introduction to tensors can be found in \cite{kolda_tensor_2009}.

A $N$-order real-valued tensor is a $N$-dimensional array $\mathcal{X}\in\mathbb{R}^{I_1\times \cdots \times I_N}$ with entries $\mathcal{X}_{i_1,\ldots,i_N}, i_n=1,\ldots,I_n$ and $n=1,\ldots,N$. The order is the number of tensor dimensions (also known as modes). We introduce some operators on real tensors. The first one is the $n$-mode product. Let $\mathcal{X}\in \mathbb{R}^{I_1\times \cdots\times I_n \times \cdots \times I_N}$ and $H\in \mathbb{R}^{J\times I_n}$, the $n$-mode product between $\mathcal{X}$ and $H$ is denoted by $\mathcal{Z}=\mathcal{X}\times_n H$ and yields the tensor $\mathcal{Z}\in\mathbb{R}^{I_1\times\cdots\times I_{n-1}\times J\times I_{n+1}\times\cdots\times I_N}$ with entries
\begin{align*}
    \mathcal{Z}_{i_1,\ldots, i_{n-1}, j,i_{n+1}\ldots, i_N} = \sum_{i_n=1}^{I_n}\mathcal{X}_{i_1,\ldots, i_{n},\ldots,i_N}H_{j,i_n}.
\end{align*}
The second operator is the contracted product, which generalizes the $n$-mode product. Let $\mathcal{X}\in \mathbb{R}^{J_1\times\cdots\times J_R\times I_{1}\times \cdots \times I_N}$ and $\mathcal{Y}\in\mathbb{R}^{K_1\times \cdots \times K_M\times I_1\times\cdots\times I_N}$. The contracted product between $\mathcal{X}$ and $\mathcal{Y}$ over the modes $(I_1,\ldots,I_N)$ is denoted by $\mathcal{Z=}\mathcal{X}\times_{\{R+1,\ldots,R+N\}}^{\{K+1,\ldots,K+N\}}\mathcal{Y}$, where the subscript and superscript index sets indicates the modes of $\mathcal{X}$ and $\mathcal{Y}$, respectively, that are contracted. The contracted product yields the tensor $\mathcal{Z}\in\mathbb{R}^{J_1\times\cdots\times J_R\times K_1\times\cdots\times K_M}$ with entries
\begin{align*}
    \mathcal{Z}_{j_1,\ldots, j_{R}, k_1,\ldots, k_M}=\sum_{i_1=1}^{I_1}\cdots\sum_{i_N=1}^{I_N}\mathcal{X}_{j_1,\ldots, j_{R},i_1,\ldots,i_N}\mathcal{Y}_{k_1,\ldots,k_M,i_1,\ldots,i_N}.
\end{align*}
The inner product of two tensors is a special case of the contracted product. The inner product of two tensors requires that both tensors have equal dimensions. Assuming $\mathcal{X,Y}\in\mathbb{R}^{I_1\times\cdots\times I_N}$ the inner product yields a scalar and is given by
\begin{align*}
\left<\mathcal{X,Y}\right>=\sum_{i_1=1}^{I_1}\cdots\sum_{i_N=1}^{I_N}\mathcal{X}_{i_1,\ldots,i_N}\mathcal{Y}_{i_1,\ldots,i_N}.
\end{align*}
\subsection{Tensor random projection}
A compressed Bayesian tensor regression (CBTR) model has the form
\begin{align}
    y_j = \mu + \left< \mathcal{B}, \texttt{GTRP}(\mathcal{X}_j) \right> + \sigma\varepsilon_j, \quad \varepsilon_j \overset{iid}{\sim} \mathcal{N}(0, 1), \label{eq. cbtr}
\end{align}
$j=1, \ldots, n$, where $\mu\in\mathbb{R}$ is the intercept, $\mathcal{B} \in \mathbb{R}^{q_1\times \cdots \times q_M}$ is the coefficient tensor, $\mathcal{X}_j \in \mathbb{R}^{p_1\times \cdots \times p_N}$ is the covariate tensor for the $j$th observation, and $\left<\cdot,\cdot\right>$ is the scalar product for tensors \citep{kolda_tensor_2009}. $\texttt{GTRP}(\mathcal{X}_j)$ denotes the \texttt{Generalized Tensor Random Projection (GTRP)}  operator applied to $\mathcal{X}_j$ defined as 
\begin{align}
    \texttt{GTRP}(\mathcal{X}) \coloneqq \mathcal{X}\times_1 H_1 \times_2 \ldots \times_R H_R \times_{\{R+1,\ldots,N\}}^{\{M-R+1,\ldots,M-R+N-R\}} \mathcal{H},
    \label{eq: gtrp}
\end{align}
where $\mathcal{X}\in \mathbb{R}^{p_1\times \cdots \times p_N}$, $\times_n$ and $\times_{\{\}}^{\{\}}$ denote the $n$-mode and the contracted products defined in Section \ref{sec:prelim}. $H_m \in\mathbb{R}^{q_m\times p_m}$, $m=1, \ldots, R$ and  $\mathcal{H} \in \mathbb{R}^{q_{R+1} \times \cdots \times q_{M} \times p_{R+1}\times \cdots \times p_N}$ 
are the random projection matrices and $(M+N-2R)$-mode random projection tensor, respectively, with $R< M \leq N$. Without loss of generality, we assumed mode-wise projection for the first $R$ modes, since the mode ordering can be chosen by the practitioner.
The \texttt{GTRP} proposed in Eq. \eqref{eq: gtrp} generalizes in two aspects the existing random projections for tensors. First, it extends the projection for 3-mode tensors to tensors with a general number of modes $N$. Secondly, the projection reduces the dimension of the covariate space, allowing fewer covariates within each mode and fewer modes. We define two distinct types of random projection that are used to construct our \texttt{GTRP}. The first type combines covariates of a given mode, while preserving the elements in the other modes. For that given mode, the proposed approach is similar to the classical techniques used in regression models, where new linear combinations of covariates are created to reduce collinearity.
\begin{definition}\label{def1}
    A random projection \texttt{GTRP-MW} is called mode-wise when $\texttt{GTRP-MW}(\mathcal{X}) \coloneqq \mathcal{X}\times_m H_m$ where $\mathcal{X}\in\mathbb{R}^{p_1\times \cdots \times p_N}$ and $H_m\in\mathbb{R}^{q_m\times p_m}$. The mode-wise projection reduces the size of mode $m$ of $\mathcal X$ from $p_m$ to $q_m$, yielding $\texttt{GTRP-MW}(\mathcal{X}) \in \mathbb{R}^{p_1\times\cdots\times q_m\times\cdots\times p_N}$.
\end{definition}
The second type constructs random linear combinations of entries across multiple modes simultaneously through a tensor contraction.
\begin{definition}\label{def2}
    A random projection \texttt{GTRP-TW} is called tensor-wise when $\texttt{GTRP-TW}(\mathcal{X}) \coloneqq\mathcal{X}\times_{\{R+1,\ldots,N\}}^{\{M+1,\ldots,M+N-R\}} \mathcal{H}$ where  $\mathcal{X}\in\mathbb{R}^{p_1\times \cdots\times p_{R}\times p_{R+1}\times\cdots \times p_N}$ and $\mathcal{H}\in\mathbb{R}^{q_1\times\cdots\times q_M\times p_{R+1}\times \cdots \times p_N}$. The tensor-wise projection contracts the modes $R+1,\ldots,N$ of $\mathcal X$ with the last $N-R$ modes of $\mathcal H$, yielding $\texttt{GTRP-TW}(\mathcal{X}) \in \mathbb{R}^{p_1\times\cdots\times p_R\times q_1\times\cdots\times q_M}$. 
\end{definition}

It is apparent that \texttt{GTRP-MW}$(\mathcal{X})$ effectively changes the size of the mode $m$ from $p_m$ to $q_m$, while still keeping the $N$-mode structure of $\mathcal{X}$, whereas \texttt{GTRP-TW}$(\mathcal{X})$ can be used to change the number of modes or the sizes of the modes, or both.  To gain an intuition of the \texttt{GTRP} in \eqref{eq: gtrp}, we consider some  special cases that can serve as reference:
\begin{itemize}
    \item [(a)] If $R=0$, $M=1$, \texttt{GTRP} corresponds to the random projection from $N$th-order tensor to $q_1$ dimensional vector: $\mathbb{R}^{p_1\times\ldots\times p_N} \rightarrow \mathbb{R}^{q_1}$. This setting does not exploit the original multiple-mode data structure and it is equivalent to the random projection in  \cite{achlioptas_database-friendly_2003} with $d = p_1\times \cdots \times p_N$ and $k=q_1$ applied to the vectorized tensor.
    \item[(b)] If $R=0$, $M\geq 1$, only \texttt{GTRP-TW}$(\mathcal{X})_{i_1,\ldots, i_M} = \left<\mathcal{X}, \mathcal{H}_{i_1,\ldots,i_M,:}\right>$ is carried out, which returns an $M$-mode tensor. If $M=N$, the number of modes will be preserved, while only the dimensions along each mode will be reduced. If $M<N$, then not only the dimensions of the tensor will be reduced, but the number of modes will also be reduced from $N$ to $M$.
    \item[(c)] If $R>0$, $N=M=R+1$, only \texttt{GTRP-MW}$(\mathcal{X})$ is carried out, where the dimension along each mode is reduced from $p_m$ to $q_m$, but the number of modes is preserved.
    \item[(d)] If $R \geq 1, M \geq R+1$, the \texttt{GTRP} involves both mode-wise random projection for the first $R$ modes and tensor-wise random projection for the $(R+1)$th to $N$th modes. Mode reduction can be achieved by choosing $M<N$.
\end{itemize}

To illustrate the effect of mode preservation in our general \texttt{GTRP}, the following 2-mode covariate example, with one of the projection matrices being the identity, provides some insights.
\begin{example}
    Considering a mode-wise random projection for a $3\times 2$ matrix $\mathcal{X}$, $f(\mathcal{X}) = \mathcal{X}\times_1 H_1 \times_2 H_2$,  where $H_1$ is a $3 \times 3$ identity matrix, $H_2$ is a $1 \times 2$ random row vector, this will map $\mathcal{X}$ into a $3\times 1$ vector $\boldsymbol{x}$ with the entries, 
    \begin{align*}
        \boldsymbol{x}_{i_1,i_2} = \sum_{j_1=1}^3\sum_{j_2=1}^2 \mathcal{X}_{j_1,j_2}H_{1,i_1,j_1}H_{2,i_2,j_2}=
        \sum_{j_2=1}^2 \mathcal{X}_{i_1,j_2}H_{2,i_2,j_2}.
    \end{align*}
    Since the random projection matrix $H_1$ is the identity matrix, consistent with the definition of \texttt{GTRP-MW}, the random projection will only be performed in the second mode, returning a vector where the component $i_1$-th is a linear combination of the elements of the row $i_1$-th of $\mathcal{X}$.
\end{example}

As shown in the above illustrations, the value of $R$ controls the extent of mode-wise random projection. The choice of using only mode-wise random projection, tensor-wise random projection, or a combination of the two should be evaluated on the basis of specific application requirements, as discussed in the numerical illustration section. In addition, a trade-off between model performance and computational cost may be considered. When dealing with very high--dimensional data with a large number of modes, mode reduction can be performed by selecting $M<N$ to achieve computational feasibility. In contrast, when preserving structural information is deemed necessary, the number of modes can remain unchanged by setting $M=N$, while reducing only the dimension along each mode.

Alternative random projections can be used. For instance, CP, TT, and Kronecker Product (KP) \citep{rakhshan_tensorized_2020,feng2024deep} decompositions can be applied with a given rank $D$ to generate low-rank random projection tensors. 

\begin{example}
    Considering random projections using the CP and TT methods in \cite{rakhshan_tensorized_2020} to map a $p_1\times p_2$ matrix $\mathcal{X}$ into a vector $\boldsymbol{x}$ as follows:
    \begin{equation}
    \texttt{CPRP}(\mathcal{X})_i = \left<\sum_{d=1}^D A^{1}_{i,:,d} \circ A^{2}_{i,:,d}, \mathcal{X}\right>, \quad \texttt{TTRP}(\mathcal{X})_i = \left<\mathcal{G}^1_i \times \mathcal{G}^2_i, \mathcal{X}\right>,
    \end{equation}
    where $A^n_i \in \mathbb{R}^{p_n\times D}$, $n=1,2$ and $\mathcal{G}^1_i \in \mathbb{R}^{1\times p_1\times D} \text{ and } \mathcal{G}^2_i \in \mathbb{R}^{D\times p_2\times 1}$, $i=1,...,q_1$. 
\end{example}


Note that our mode-wise random projection can be thought of as constructing CP random projection tensors with rank $1$ for each embedded entry. More importantly, CP, TT, and KP random projections map an order-$N$ tensor to a vector that collapses all structural information; however, our methods still preserve the tensor structure, which can be valuable for practical applications.

A wide variety of distributions can be used to construct random projection matrices or tensors, provided that the entries are iid with mean zero and finite fourth moment \citep{mukhopadhyay_targeted_2020}. A simple way to generate projections is to assume that the elements of $H_m$ and $\mathcal{H}$ are i.i.d. from a standard normal distribution. \cite{dasgupta_elementary_2003} gives a concise proof of the JL lemma under the assumption of standard Gaussian entries. Nevertheless, the dense projection matrix used in classical random projection is not well-suited for high-dimensional problems. Thus, sparse \citep{achlioptas_database-friendly_2003} and very sparse random projections \citep{kane2014sparser} have been proposed. In more applied literature, the Rademacher distribution is used in \cite{rakhshan2021rademacher}, to encourage sparsity in the constructed random projection matrices/tensors. In this paper, we follow \cite{achlioptas_database-friendly_2003} and \cite{li_very_2006} and assume that the entries are independent random variables having the following discrete distribution
\begin{equation}
r \in \{-\sqrt{\psi}, 0, +\sqrt{\psi}\}, \quad 
\mathbb{P}(r = \pm\sqrt{\psi}) = \frac{1}{2\psi}, \quad 
\mathbb{P}(r = 0) = 1 - \frac{1}{\psi}.
\label{eq:rpd}
\end{equation}
\subsection{Model properties}
In our model, the random projection $\texttt{GTRP}(\mathcal{X})$ projects the covariate tensor $\mathcal{X}_j\in \mathbb{R}^{p_1 \times \cdots \times p_N}$ onto a lower-dimensional subspace, that is, $\texttt{GTRP}(\mathcal{X}_j) \in \mathbb{R}^{ q_1 \times \cdots \times q_M}$, $j=1,\ldots,n$. The following results show that, when projecting, the distances between points in the original sample spaces are preserved by random projection under some suitable conditions. In the following, we define the constants $p(N)=\prod_{m=1}^{N}p_m$, $q(M)=\prod_{m=1}^{M}q_m$ and $c(N,M)=p(N)/q(M)$.

When $R = 0$ and $M=1$, then $\texttt{GTRP}(\mathcal{X}_j)$ randomly projects all tensor entries into a vector space and the following JL concentration inequality holds uniformly in both the number of elements in each mode and in the number of modes.

\begin{proposition}[A JL inequality for tensor-wise random projection]\label{cor: jl}
Let $\mathbb{X}$ be an arbitrary set of $n$ order $N$ tensors in $\mathbb{R}^{p_1 \times \cdots \times p_N}$. Define $\texttt{GTRP-TW}(\mathcal{X})=\mathcal{X}\times_{\{1,\ldots,N\}}^{\{1,\ldots,N\}}\mathcal{H}$ with $\mathcal{H}$ an $N+1$ order random tensor in $\mathbb{R}^{p_1 \times \cdots \times p_N\times q_1}$ with entries from the distribution in \eqref{eq:rpd}, and the multilinear mapping $f(\mathcal{X})=\frac{1}{\sqrt{q(M)}}\texttt{GTRP-TW}(\mathcal{X})$ from $\mathbb{R}^{p_1 \times \cdots \times p_N}$ to $ \mathbb{R}^{q_1}$. Given $\epsilon, \beta >0$, and a positive integer $q_1\geq q_0$ where $q_0 = (4+2\beta)(\epsilon^2/2-\epsilon^3/3)^{-1}\log n$, $f$ satisfies with high probability and for all tensors $\mathcal{U}, \mathcal{V} \in \mathbb{X}$:
\begin{equation*}  
        (1-\epsilon)\lVert\mathcal{U}-\mathcal{V}\rVert^2 \leq \lVert f(\mathcal{U}) - f(\mathcal{V})\rVert^2 \leq (1+\epsilon)\lVert\mathcal{U}-\mathcal{V}\rVert^2.
\end{equation*}
\end{proposition}

Similarly, a concentration inequality can be proved when projecting mode-wise, that is, $R=M-1, M=N$. The concentration bound is uniform in the number of elements in each mode, but not in the number of modes.

\begin{theorem}[JL inequality for mode-wise random projection] \label{thm: jl}
    Let $\mathbb{X}$ be an arbitrary set of $n$ order $N$ tensors in $\mathbb{R}^{p_1 \times \cdots \times p_N}$.  Let $\epsilon, \beta > 0$ and set
    \begin{align*}
        q_0 = \frac{4+2\beta}{\frac{\epsilon^2}{3^N-1}-\frac{(3^{N+1}-2)\epsilon^3}{3(3^N-1)^3}} \log n.
    \end{align*}
Assume a sequence of positive integers $q_j$ $j=1,\ldots,N$ satisfy $q(M)=q(N) \geq q_0$ 
with probability at least $1 - n^{-\beta}$.
    
 Define 
 $\texttt{GTRP}(\mathcal{X})=\mathcal{X}\times_1 H_1\times_2 \cdots \times_{N} H_N$, 
 where the entries of $H_m \in \mathbb{R}^{p_m \times q_m}$ for $m = 1, \ldots, N$ are independently distributed following the distribution given in \eqref{eq:rpd}, and the multilinear mapping $f(\mathcal{X})=\frac{1}{\sqrt{q(M)}}\texttt{GTRP}(\mathcal{X})$ from $\mathbb{R}^{p_1 \times \cdots \times p_N}$ to $ \mathbb{R}^{q_1 \times \cdots \times q_N}$. Then for all $\mathcal{U}, \mathcal{V} \in \mathbb{X}$, $f$ satisfies
    \begin{align*}
        (1-\epsilon)\lVert \mathcal{U} - \mathcal{V}\rVert^2 \leq \lVert f(\mathcal{U}) - f(\mathcal{V}) \rVert^2 \leq (1+\epsilon)\lVert \mathcal{U} - \mathcal{V}\rVert^2.
    \end{align*}
\end{theorem}

Theorem \ref{thm: jl} extends classical JL inequalities from vectors to multi-mode tensors that are projected along each mode. It also provides a theoretical foundation for the use of structured random projection in scalable Bayesian tensor regression. Note that setting $N=1$ in Theorem \ref{thm: jl} yields the JL inequality from Proposition \ref{cor: jl}.

\begin{figure}[]
    \centering
    \vspace*{3pt}
    \includegraphics[width=.45\linewidth]{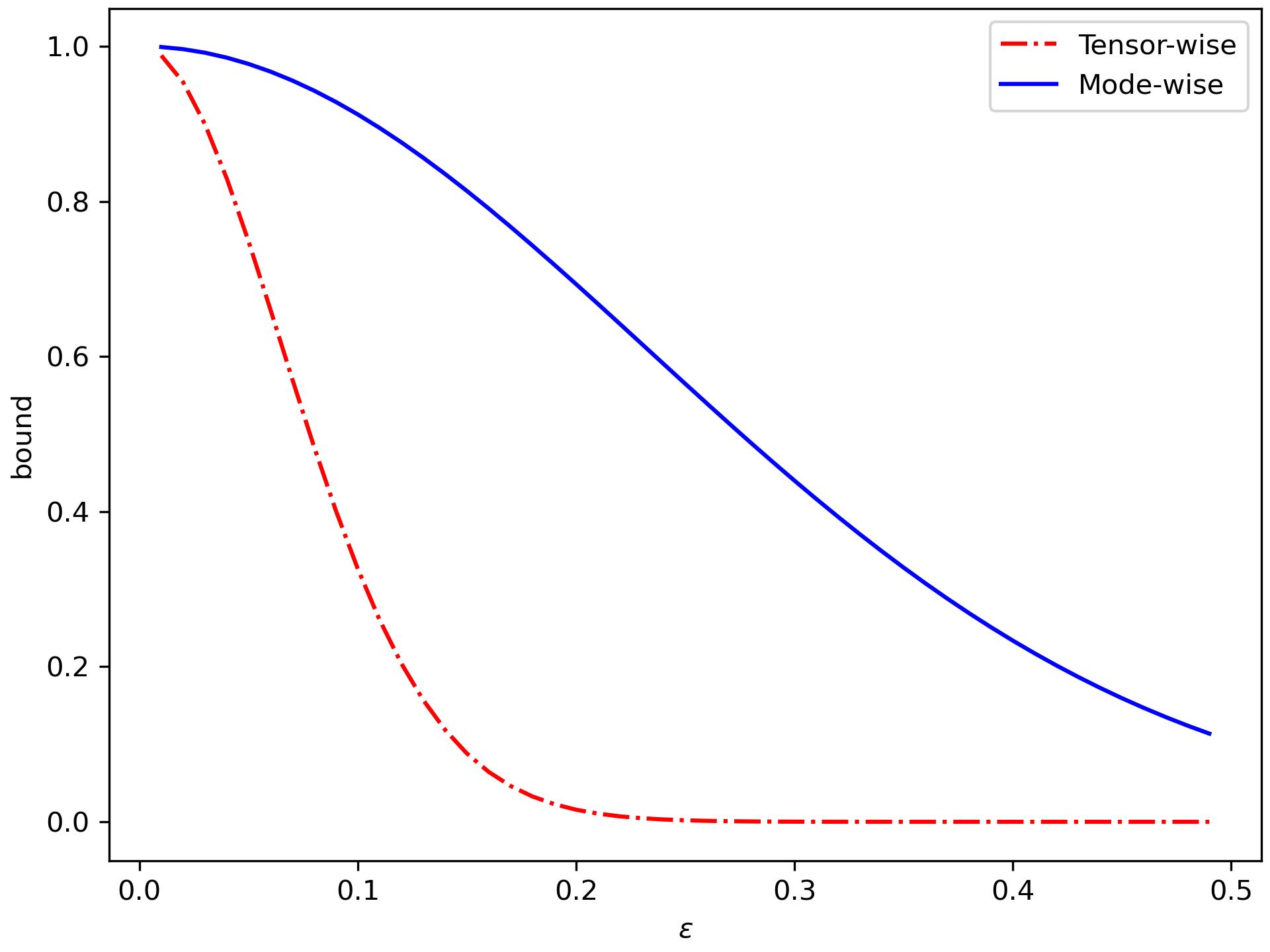}
    \caption{The plot shows the two bounds obtained by tensor-wise random projection according to Proposition \ref{cor: jl} (red curve) and mode-wise random projection (blue curve) according to Theorem \ref{thm: jl}. We considered a $3$-mode $\mathbb{R}^{20\times 60 \times 50}$ tensor (i.e., $N=3$, $p_1=20$, $p_2=60$ and $p_3=50$) projected into a $\mathbb{R}^{480}$ vector and a $\mathbb{R}^{4\times 12 \times 10}$ tensor (i.e., $M=N=3$, $q_1=4$, $q_2=12$ and $q_3=10$) with the mode-wise and tensor-wise projection, respectively, assuming $n=10^4$ data points, and a concentration rate $\beta=0.2$.}
    \label{fig: bounds_compa}
\end{figure}

To get JL-embedding, we need that for each of the ${n \choose 2}$ pairs of $\mathcal{U}, \mathcal{V} \in \mathbb{X}$, the squared norm of $(\mathcal{U}-\mathcal{V})$ is maintained within a factor of $1\pm \epsilon$. If we can show that for some $\beta >0$ and any fixed tensor $\mathcal{A} \in \mathbb{R}^{p_1 \times \cdots \times p_N}$,
\begin{align*}
    \Pr[(1-\epsilon)\lVert\mathcal{A}\rVert^2 \leq \lVert f(\mathcal{A})\rVert^2 \leq (1+\epsilon)\lVert\mathcal{A}\rVert^2] \geq 1-\frac{2}{n^{2+\beta}},
\end{align*}
then the probability of not getting a JL-embedding is bounded by ${n \choose 2} \times \frac{2}{n^{2+\beta}} < \frac{1}{n^{\beta}}$.

A comparison of the bounds obtained from tensor- and mode-wise random projections is shown in Figure \ref{fig: bounds_compa}. The results show the trade-off between maintaining the original data structure, such as some of the tensor modes, and the dimensions of the random subspace. In the case where all modes are preserved, the dimensionality reduction (blue line) is less effective than the case where the original structure vanishes completely (red line). Two main advantages of preserving the covariates' structure are interpretability and reduced computational cost.

The bounds presented above are optimal, as they have been derived following a Chernoff--Cramér procedure. Alternative concentration bounds for our projections can be derived to provide some guarantees on the preservation of distance. For example, based on a general hyper-contractivity result of the Hanson--Wright type for polynomials of Gaussian and Rademacher variables \citep{hanson1971bound,rakhshan_tensorized_2020, kane2014sparser} and bounds on the moments up to the fourth-order, one can show the following bounds.
\begin{theorem}[Alternative bounds using hyper-contractivity] \label{thm: alt_bnd}
    The JL-embedding can be achieved for \texttt{GTRP} with mode-wise random projection if $q(N) \geq q_0$, such that
    \begin{eqnarray}
        q_0 > C \varepsilon^{-2}3^N\left(2+\beta\right)^{2N}\log^{2N}n,
    \end{eqnarray}
    with $C$ an absolute constant.  
\end{theorem}

\begin{remark}
For the CP and TT projections, the following bounds on embedding dimensions have been obtained in \cite{rakhshan_tensorized_2020}
\begin{eqnarray}
&& q_0 > C' \varepsilon^{-2}3^{N-1}\left(1+\frac{2}{R}\right)^N\log^{2N}\left(\frac{n^{2+\beta}}{2}\right)\\
&& q_0 > C'' \varepsilon^{-2}\left(1+\frac{2}{R}\right)^N\log^{2N}\left(\frac{n^{2+\beta}}{2}\right),
\end{eqnarray}        
where $R$ denotes the rank of the random projection tensor, and  $C'$ and $C''$ are absolute constants.
\end{remark}
The bounds given above are exponential and apply to general projection tensors even when the entries are not normally distributed. While they provide a Chernoff--like estimate, the bounds are not optimal in the Chernoff--Cramér sense.  We note that they depend on constants that are not easy to compute. The bounds in this section provide a theoretical basis for the methodological developments proposed in this paper. The bounds in Figure \ref{fig: bounds_compa} are derived under general assumptions about the covariate tensor $\mathcal{X}$ and thus are conservative for those cases where $\mathcal{X}$ exhibits a more restrictive structure, such as high sparsity, or sparsity aligned with several coordinates.  The difference in performance between tensor- and mode-wise projections, as suggested by Figure \ref{fig: bounds_compa}, has been confirmed by the numerical experiments in Section \ref{sec:numerical} and depends on the sparsity pattern in the covariate tensor. 

\subsection{Prior distributions}
We consider two alternative prior specifications. In the first one, we assume independent Gaussian and inverse gamma prior distributions. $
    \mathcal{B}\sim \mathcal{TN}_{p_1,\ldots,p_N}(\mathbf{0},\Sigma_1,\ldots,\Sigma_N), \mu\sim\mathcal{N}(0,\sigma^2_{\mu}), \sigma^2\sim\mathcal{IG}(a,b).$

In the second specification, we assume a hierarchical prior structure which builds on, as in \cite{guhaniyogi2017bayesian}, a Parallel Factor (PARAFAC) representation of $\mathcal{B}$ for further dimensionality reduction on tensor coefficients $\mathcal{B} = \sum_{d=1}^D \boldsymbol{\gamma}^{(d)}_1\circ\dots\circ \boldsymbol{\gamma}^{(d)}_N,$
where $\circ$ denotes the \textit{external product} of vectors, and $\boldsymbol{\gamma}^{(d)}_m$ are the margins from PARAFAC decomposition of tensor coefficient $\mathcal{B}$.
At first level, we assume that the margins from the PARAFAC decomposition are independent and follow multivariate normal distributions with  zero  mean vector and scales given by the product of the scalars $\tau$, $\zeta^{(d)}$, and the diagonal matrix $W^{(d)}_{m} = \text{diag}(w^{(d)}_{m,1},\ldots,w^{(d)}_{m,j_m},\ldots,$ $w^{(d)}_{m,q_{m}})$, i.e. 
\begin{align}
    \boldsymbol{\gamma}_{m }^{(d)}\;\sim\;\mathcal{N}_{q_{m}}(\boldsymbol{0},\tau\zeta^{(d)}W_{m }^{(d)}), \quad  m=1,\ldots,M, \quad d=1,\ldots,D.\label{eq.gamma}
\end{align}
This random-scale specification allows for shrinkage at different levels. 

To complete the hierarchical prior, at the second level, we modify the priors from \cite{guhaniyogi2017bayesian} and assume the following prior distributions for the scales.
\begin{align}
&\tau   \sim  \mathcal{IG}(a_\tau , b_\tau),\quad
w_{m,j_{m} }^{(d)}\sim \mathcal{E}xp((\lambda_{m }^{(d)})^2/2)\notag \\
&\lambda_{m }^{(d)}\sim \mathcal{G}a(a_\lambda ,b_\lambda),\quad (\zeta^{(1)},\ldots,\zeta^{(D)})\sim \mathcal{D}ir(\alpha,\ldots,\alpha),\label{eq.zeta.prior}
\end{align}
$m=1,\ldots, M$, $d=1,\ldots, D$ where $\mathcal{IG}(a,b), \mathcal{G}a(a,b)$, $\mathcal{E}xp(\lambda)$ and $\mathcal{D}ir(\nu_1,\ldots,\nu_D)$ denote the Inverse Gamma, Gamma, Exponential and Dirichlet distributions, respectively. The only difference compare to \cite{guhaniyogi2017bayesian} is assuming the prior distribution of global shrinkage parameter $\tau$ is an Inverse Gamma instead of Gamma, largely due to the fact that $\tau$ appears as a variance parameter in the Gaussian prior of $\boldsymbol{\gamma}_{m}^{(d)}$, it is natural to assume $\tau\overset{\text{i.i.d.}}{\sim} \mathcal{IG}(a,b)$ to get a more tractable full conditional distribution in the posterior approximation procedure.

\section{Posterior approximation} \label{sec:mcmc}
\subsection{Gibbs sampling}
The joint posterior distribution $f(\boldsymbol{\gamma}_{m}^{(d)},\zeta^{(d)}, \tau, \lambda_{m}^{(d)}, w_{m}^{(d)}, \sigma^2, \mu \mid \boldsymbol{y}, \texttt{GTRP}(\mathcal{X}))$ is not tractable, so it must be approximated using the Monte Carlo method. We achieve this using a custom-built  Gibbs sampler. Below, we describe the conditional sampling steps required by the algorithm's design. The derivation of the full conditionals can be found in Appendix \ref{app:Gibbs}.1. The sampler cycles between the following steps:

\begin{enumerate}
    \item Draw $\boldsymbol{\gamma}_{m}^{(d)}$ from a multivariate normal distribution $f(\boldsymbol{\gamma}_{m}^{(d)}\mid \boldsymbol{y}, \texttt{GTRP}(\mathcal{X}), \boldsymbol{\gamma}_{-m}, \tau, \boldsymbol{\zeta},$ $ \boldsymbol{w},\mu,\sigma^2)$ for $d\in\{1,\ldots,D\} $ and $ m\in\{1,\ldots,M\}$.
\end{enumerate}

Let us denote the Generalized Inverse Gaussian distributions with $\mathcal{GIG}$. The Gibbs updates for the remaining parameters and hyper--parameters are:

\begin{enumerate}
\setcounter{enumi}{1}
    \item Draw $\zeta^{(d)}$ from the $\mathcal{GIG}$ distribution $f(\zeta^{(d)} \mid  \boldsymbol{\gamma}_{}^{(d)}, \tau, \boldsymbol{w}^{(d)})$.
    \item Draw $\tau$ from the $\mathcal{IG}$ distribution $f(\tau \mid  \boldsymbol{\gamma},\boldsymbol{\zeta}, \boldsymbol{w})$.
    \item Draw $\lambda^{(d)}_{m}$ from a Gamma distribution $f(\lambda^{(d)}_{m} \mid \boldsymbol{\gamma}^{(d)}_{m}, \tau, \zeta^{(d)} )$.
    \item Draw $w^{(d)}_{m,j_m}$ from the $\mathcal{IG}$ distribution $f(w^{(d)}_{m,j_m} \mid \gamma^{(d)}_{m,j_m}, \lambda^{(d)}_{m}, \tau_, \zeta^{(d)})$.
    \item Draw $\sigma^2$ from the $\mathcal{IG}$ distribution $f(\sigma^2\vert \boldsymbol{y}, \texttt{GTRP}(\mathcal{X}), \mu,\boldsymbol{\gamma})$.
    \item Draw $\mu$ from the Gaussian distribution $f(\mu \mid \boldsymbol{y}, \texttt{GTRP}(\mathcal{X}), \boldsymbol{\gamma}, \sigma^2)$.
\end{enumerate}
The full conditional distributions of the Gibbs sampler for the hierarchical Normal-Inverse Gamma prior are given in Appendix \ref{app:Gibbs}.2. Both variants of the Gibbs algorithm involve conditional densities that are available in closed form.

\subsection{Model averaging}
Reliance on a single random projection is a risky approach, since one may not know the optimal type of projection or how far the projection matrix is from an optimal one. Moreover, it is straightforward to parallelise the computation and substantially reduce the time to obtain estimates or predictions from several projections. Many combination methods are available in the literature to aggregate predictions from different models. If the interest is in prediction, then Bayesian predictive stacking has been shown to be advantageous in the setting where the true model is not included in the model space \citep{wolpert1992stacked, yao2018using, yao2021bayesian}.  Since none of the compressed tensor regressions is expected to be the true data-generating model, \cite{gailliot2024data} adapted predictive stacking to data sketching and incorporated strategies to reduce computational load. When model inference is computationally expensive, and the interest is in emphasizing the models with higher marginal posterior probabilities, then Bayesian model averaging \citep[BMA,][]{raftery1997bayesian} is recommended \citep[see, for instance][who aggregated different random projections via BMA]{guhaniyogi_bayesian_2015, guhaniyogi2016compressed}. In this paper, we use Bayesian model averaging to combine different compressed tensor regressions. 
 
Specifically, we generate $L$ different random projections for each compressed tensor regression using entries randomly drawn from the distribution proposed in \eqref{eq:rpd}. Let $\mathcal{M}_{\ell}, \ell=1,\ldots,L,$ represent the model in \eqref{eq. cbtr} with $\texttt{GTRP}^{(\ell)}(\cdot)$ denoting the distinct random projection for $\mathcal{M}_{\ell}$. We further denote $f_{\ell}$ the predictive density for $\mathcal{M}_{\ell}$ and $\boldsymbol{\theta}^{(\ell)} =(\mu^{(\ell)}, \mathcal{B}^{(\ell)},{\sigma^2}^{(\ell)})$ its parameters, $\mathcal{D} = \{(y_j, \texttt{GTRP}(\mathcal{X}_j)), j=1,\ldots,n\}$ the observed data, and we are interested in the predictive density of $y_{n+j'}$ given $\mathcal{X}_{n+j'}$:
\begin{align}
    &f(y_{n+j'} | \texttt{GTRP}^{(\ell)}(\mathcal{X}_{n+j'}), \mathcal{D}) = \sum_{\ell=1}^{L} p_\ell(\mathcal{M}_{\ell} | \mathcal{D})f_\ell(y_{n+j'} | \texttt{GTRP}^{(\ell)}(\mathcal{X}_{n+j'}), \mathcal{D},\mathcal{M}_{\ell}) \label{pred-dist}\\
    &f_\ell(y_{n+j'} | \texttt{GTRP}^{(\ell)}(\mathcal{X}_{n+j'}), \mathcal{D},\mathcal{M}_{\ell})=\int f_\ell(y_{n+j'} | \texttt{GTRP}^{(\ell)}(\mathcal{X}_{n+j'}), \boldsymbol{\theta}^{(\ell)}, \mathcal{M}_{\ell})p_\ell(\boldsymbol{\theta}^{(\ell)} | \mathcal{M}_{\ell}, \mathcal{D})d\boldsymbol{\theta}^{(\ell)}
    \label{mod-pred_dist}
    \end{align}
for $j'=1,\ldots,m$ where $m$ is the size of the validation set. Since the normalizing constant $c_{\ell} = p_\ell(\mathcal{M}_{\ell}\mid \mathcal{D})$ of $p_\ell(\boldsymbol{\theta}^{(\ell)}\mid \mathcal{M}_{\ell}, \mathcal{D})$ is not available in closed form,  we approximate it using reverse logistic regression  \citep{geyer1994estimating}.

To approximate the predictive density in \eqref{pred-dist}, we first evaluate empirically the predictive distribution of $y_{n+j',s}^{(\ell)}$ produced by the $\ell$th random projection. In particular, at the $s$th MCMC step a random draw $y_{n+j',s}^{(\ell)}$ is generated from the posterior predictive distribution
\begin{align}
    y_{n+j',s}^{(\ell)} \mid \texttt{GTRP}^{(\ell)}(\mathcal{X}_{n+j'}), \boldsymbol{\theta}_s^{(\ell)} \sim \mathcal{N}\left(\mu_{s}^{(\ell)}+\left<\texttt{GTRP}^{(\ell)}(\mathcal{X}_{n+j'}), \mathcal{B}_{s}^{(\ell)}\right>, {\sigma_{s}^{2}}^{(\ell)}\right),
\end{align}
where $\boldsymbol{\theta}_{s}^{(\ell)}$, $s=1,\ldots,S$ denote the MCMC draws from the posterior distribution. 

We pool $y_{n+j',s}^{(\ell)}$ across $\ell$ and $s$ to obtain an empirical distribution which approximates the distribution of $y_{n+j'}$.  Let $\tilde{y}_{n+j'}^{(\ell)}$  denotes the approximated prediction given $\mathcal{D}$ and the $\ell$th projection $\texttt{GTRP}^{(\ell)}(\mathcal{X}_{n+j'})$, we approximate the posterior predictive mean with
$
\tilde{y}_{n+j'}=\sum_{\ell=1}^{L}w_{\ell}\tilde{y}_{n+j'}^{(\ell)},\quad \tilde{y}_{n+j'}^{(\ell)}=\frac{1}{S}\sum_{s=1}^{S}y_{n+j',s}^{(\ell)},
$
where $w_{\ell} = c_{\ell}/\sum_{k=1}^L c_{k}$,  for all $\ell=1,\ldots,L$.

To evaluate the quantiles
of the predictive distribution $f(y_{n+j'} \mid \texttt{GTRP}(\mathcal{X}_{n+j'},\mathcal{D})$
define $z_{n+j',s} = \sum_{\ell=1}^L u_{n+j',s}^{(\ell)}y_{n+j',s}^{(\ell)}$, where $(u_{n+j',s}^{(1)}, \ldots,u_{n+j',s}^{(L)}) \sim \text{Multinomial }(1, (w_1, \ldots,w_L))$. Because $
    P(z_{n+j',s} \leq t) =
    \sum_{\ell=1}^{L}P\left(y_{n+j',s}^{(\ell)} \leq t\right)w_{\ell}
$
 we have $f(t\mid \texttt{GTRP}(\mathcal{X}_{n+j'},{\mathcal{D}}) = \sum_{\ell=1}^L w_{\ell}f^{(\ell)}(t\mid \texttt{GTRP}(\mathcal{X}_{n+j'}), \mathcal{D},\mathcal{M}_{\ell})$. So the quantiles for the density $f$ in \eqref{pred-dist} can be evaluated from the sample quantiles of the $L$ predictive distributions defined in \eqref{mod-pred_dist}. 



\section{Posterior Consistency}\label{sec:asy}
Projection of the tensor predictor is justifiable from a computational point of view, but the statistical validity of the resulting inference must be theoretically defensible. We present in this section theoretical results that demonstrate that the predictions generated by the original and compressed tensor predictors and variables can be made arbitrarily close for particular choices of the projection matrix. 

\subsection{Notation and background}
To show the posterior consistency of the model predictions, we consider, without loss of generality, the modewise random projection of the $3$-mode tensors onto $3$-mode tensors with a reduced number of elements along each modes. Let $\mathcal{X}_j \in \mathbb{R}^{p_{1, n}\times p_{2, n}\times p_{3, n}}$ denote the 3-mode tensor predictor for observation $j = 1, \ldots, n$. We assume that there is a true tensor coefficient $\mathcal{B}_0 \in \mathbb{R}^{p_{1, n}\times p_{2, n}\times p_{3, n}}$. Denote by $\texttt{GTRP-M}(\mathcal{X}_j), \; \mathcal{B} \in \mathbb{R}^{q_{1, n}\times q_{2, n}\times q_{3, n}}$ the compressed tensor predictor and coefficient, respectively. Let $p_n = p_{1, n}\times p_{2, n}\times p_{3, n}$ and $q_n = q_{1, n}\times q_{2, n}\times q_{3, n}$ denote the number of predictors for a given sample size $n$ before and after compression, respectively. 
In addition, assume that all the covariates are bounded, which means $\lvert x_{jkl} \rvert < 1$ and $\lim_{n\to\infty} \sum_{j=1}^{p_{1,n}}\sum_{k=1}^{p_{2,n}}\sum_{l=1}^{p_{3,n}} \lvert b_{jkl, 0} \rvert < K$.

Let $f_0$ be the true posterior predictive density given the predictors $\mathcal{X}$, and $f$ be the predictive density given the coefficients $\mathcal{B}$ drawn from its posterior distribution and the predictors $\mathcal{X}$. Let $
\nu_{\mathcal{X}}(d\mathcal{X})$ be the probability measure for $\mathcal{X}$, and $\nu_y(dy)$ be the dominating measure for conditional densities $f$ and $f_0$. We assume that the true relationship between the response $y$ and the predictors $\mathcal{X}$ follows a parametric generalized linear model (GLM) of the form $f(y \mid \mathcal{X}, \mathcal{B}_0) = \exp \{a(h)y + b(h) + c(y)\}$, where $h = \left<\mathcal{X}, \mathcal{B}_0\right>$. In the case of a normal linear regression, with mean $h$ and variance $\sigma^2$, the density is obtained by choosing $a(h) = h/\sigma^2$, $b(h)=-h^2(2\sigma^2)^{-1}-1/2\ln (2\pi\sigma^2)$ and $c(y) = -y^2(2\sigma^2)^{-1}$.

The following measures of closeness are used to show posterior consistency. The Hellinger distance between $f$ and $f_0$ and the Kullback-Leibler divergence of $f$ from $f_0$ are $
    d(f,f_0) = \sqrt{\iint \left(\sqrt{f} - \sqrt{f_0}\right)^2\nu_{\mathcal{X}}(d\mathcal{X})\nu_y(dy)}$ and $ 
    d_{KL}(f,f_0) = \iint f_0 \ln \left(\frac{f_0}{f}\right)\nu_{\mathcal{X}}(d\mathcal{X})\nu_y(dy),$ respectively. In addition, we define $
    d_t(f,f_0) = t^{-1}\left(\iint f_0 \left(\frac{f_0}{f}\right)^t\nu_{\mathcal{X}}(d\mathcal{X})\nu_y(dy) - 1\right), \; \forall t>0. $

\subsection{Posterior results}
We present two important theoretical results on posterior consistency of CBTR using two different priors for the tensor coefficients $\mathcal{B}$: the Gaussian prior and the PARAFAC prior. Let \(\Tilde{\lambda}_n \) and $\underline{\lambda}_n$ be the largest and smallest eigenvalues of the covariance matrices of the Gaussian prior. The following theorems on consistency are proven by verifying that the sufficient conditions $a, b$ and $c$ in Theorem 4 of \cite{jiang_bayesian_2007}  are satisfied. The theoretical results derived in this section rely on the following assumptions for a sequence $\varepsilon_n$ satisfying $0<\varepsilon_n^2<1$ and $n\varepsilon_n^2 \to \infty$:

\medskip

\par \textbf{A.1}
        \(\frac{q_n\log(1/\varepsilon_n^2)}{n\varepsilon_n^2} \to 0, \quad \frac{\log(q_n)}{n\varepsilon_n^2} \to 0, \quad \frac{q_n\log G(\theta_n\sqrt{8\tilde{\lambda}_n n\varepsilon_n^2})}{n\varepsilon_n^2} \to 0 \label{eq: assump_1}
        \).      
\smallskip
        
\par\noindent Where $G(R) = 1 + R\sup_{\lvert h \rvert \leq R}\lvert a'(h)\rvert\sup_{\lvert h \rvert \leq R} \lvert \frac{b'(h)}{a'(h)}\rvert$, $\theta_n = \sqrt{q_np_n}$. Assumption A.1 imposes restrictions on the growth rate of the number of regressors, $q_n$, so that $q_n$ grows sublinearly with the total number of observations. Intuitively, this assumption prevents the projected model from being ``too'' complex.
\smallskip

\par \textbf{A.2}\textit{
        \(
            \tilde{\lambda}_n\leq Bq_n^v, \quad \underline{\lambda}_n \geq B_1\left(\log(q_n)\right)^{-1} 
        \)
        for some positive constants $B$, $B_1$, $v$.} 
\smallskip

\par\noindent Assumption \textbf{A.2} imposes some constraints on the prior covariance matrix of $\mathcal{B}$ by bounding the eigenvalues of the covariance matrix to ensure that the prior is well-defined and does not allow it to be too diffuse or too concentrated. 
\smallskip

\par \textbf{A.3}\textit{
        \(\frac{\log\left(\lVert \texttt{GTRP}(\mathcal{X}) \rVert\right)}{n\varepsilon_n^2} \to 0, \quad \lVert \texttt{GTRP}(\mathcal{X}) \rVert^2 > 8\frac{(K^2 + 1)}{B_1}\frac{\log(q_n)}{n\varepsilon_n^2},\quad 
        \forall \mathcal{X} = \mathcal{X}_1, \ldots, \mathcal{X}_n
        \).}
\smallskip

\par\noindent Assumption \textbf{A.3} ensures that the tensor random projection operation \texttt{GTRP}$(\cdot)$ does not excessively distort the norm of the tensor covariates $\mathcal{X}$, thus preserving the power of the covariates to explain the responses. This assumption is typically satisfied with high probability for carefully designed random projections as described in Proposition \ref{cor: jl} and Theorem \ref{thm: jl}. 
%
\smallskip
\par \textbf{A.4} There exists a pseudo-true compressed coefficient $\mathcal{B}_0^c$ such that $\lvert \left<\mathcal{X}_i, \mathcal{B}_0\right>-\left<\texttt{GTRP}(\mathcal{X}_i), \mathcal{B}_0^c\right>\rvert \leq o(q_{0,n}^{-1})$, where $q_{0,n}$ is the lower bound given by Proposition \ref{cor: jl} and Th.\ref{thm: jl}.

\smallskip

\par \textbf{A.5} \textit{$D(\log (\lVert\texttt{GTRP}(\mathcal{X}_i)\rVert)+\log C_0DM)\sum_{m=1}^M q_{m,n}<Mn\varepsilon_n^2  C$ for some positive constant $C$ and $C_0$.}

\smallskip

\par \textbf{A.6}  \textit{$\varepsilon_n^2 = n^{\delta}$ with $b-1<\delta<0$ where $\sum_{m=1}^M q_{m,n} = \mathcal{O}(n^b)$.}

\par\noindent Assumption \textbf{A.4} and \textbf{A.5} target PARAFAC priors on compressed tensor coefficients $\mathcal{B}$. Assumption \textbf{A.4} controls the complexity of the model by bounding the projection norm $\lVert\texttt{GTRP}(\mathcal{X}_i)\rVert$, the PARAFAC component $D$, and the number of coefficients $D\sum_{m=1}^Mq_{m,n}$. 
Assumption \textbf{A.5} mainly specifies how fast the posterior contracts, at a rate slower than $n^{-1}$, but still converging. It also controls the growth of the projected dimension: the total number of compressed parameters $q_{m,n}$ must grow sublinearly with $n$. Altogether, assumption \textbf{A.5} ensures that the predictive distribution does not overfit as $n$ grows.

\begin{theorem} \label{thm: pos_cons}
Assume that \textbf{A.1}, \textbf{A.2} and \textbf{A.3}  hold. Let $\mathcal{B} \sim \mathcal{TN}(\boldsymbol{0}, \boldsymbol{\Sigma}_{1}, \ldots, \boldsymbol{\Sigma}_{N})$ a priori and \(\Tilde{\lambda}_n \) and $\underline{\lambda}_n$ be the largest and smallest eigenvalues of $\boldsymbol{\Sigma}_{1}, \ldots, \boldsymbol{\Sigma}_{N}$. In addition, suppose all the covariates are bounded, i.e., $\lvert x_{jkl} \rvert < 1$ and $\lim_{n\to\infty} \sum_{j=1}^{p_{1,n}}\sum_{k=1}^{p_{2,n}}\sum_{l=1}^{p_{3,n}} \lvert b_{jkl, 0} \rvert < K$. 

For a sequence $\varepsilon_n$ satisfying $0<\varepsilon_n^2<1$ and $n\varepsilon_n^2 \to \infty$, then
    \begin{align}
        E_{f_0}\pi \left(d(f, f_0)>4\varepsilon_n \mid \{y_j, \mathcal{X}_j\}_{j=1}^n\right) \leq 4e^{-n\varepsilon_n^2/2},
    \end{align}
    where $\pi(\cdot \mid \{y_j, \mathcal{X}_j\}_{j=1}^n)$ is the posterior measure.
\end{theorem}

\begin{theorem} \label{thm: pos_parafac}
\sloppy
    Assume that \textbf{A.1}, \textbf{A.4} and \textbf{A.5} hold. Let $\boldsymbol{\gamma}_{m }^{(d)}\; \sim\;\mathcal{N}_{p_{m }}(\boldsymbol{0},\tau\zeta^{(d)}W_{m }^{(d)})$ a priori, and further assume that all covariates are bounded, i.e., $\lvert x_{jkl} \rvert < 1$ and $\lim_{n\to\infty} \sum_{j=1}^{p_{1,n}}\sum_{k=1}^{p_{2,n}}\sum_{l=1}^{p_{3,n}} \lvert b_{jkl, 0} \rvert < K$. For a sequence $\varepsilon_n$ satisfying $0<\varepsilon_n^2<1$ and $n\varepsilon_n^2 \to \infty$, 
    then
    \fussy
    \begin{align}
        E_{f_0}\pi \left(d(f, f_0)>4\varepsilon_n \mid \{y_j, \mathcal{X}_j\}_{j=1}^n\right) \leq 4e^{-n\varepsilon_n^2/2},
    \end{align}
    where $\pi(\cdot \mid \{y_j, \mathcal{X}_j\}_{j=1}^n)$ is the posterior measure.
       
\end{theorem}

In Theorems \ref{thm: pos_cons} and \ref{thm: pos_parafac}, we have assumed that the $\ell_1$ norm of the tensor coefficient $\mathcal{B}$ is bounded by a constant $K$. While this is in agreement with \cite{guhaniyogi_bayesian_2015} and \cite{ mukhopadhyay_targeted_2020}, it would be a natural extension to let $K$ grow slowly with $n$ (e.g., $K_n=o(n)$). In our current setting, this would require an admissible growth rate for $K_n: \frac{K_n^2 \log(q_n)}{n\varepsilon_n^2} \to 0$ which implies $K_n=o(\sqrt{n\varepsilon_n^2/\log(q_n)})$. Under the canonical high–dimensional choice $\varepsilon_n^2 \sim (q_n \log n)/n$, this reduces to $K_n = o\!\left(\sqrt{q_n \log n/\log(q_n)}\right)$, and further to $K_n = o(\sqrt{q_n})$ when $q_n$ grows polynomially in $n$ (e.g., $q_n = n^a$ for $0<a<1$). A complete characterization of the trade–off between the growth of $K_n$ and posterior contraction remains an open problem and presents a promising avenue for future research.

\section{Numerical Illustrations} \label{sec:numerical}
The posterior consistency analysis was established under the broader GLM framework for theoretical generality, while numerical illustrations are presented for Gaussian regression models.  Similar exposition strategies are common in the literature \cite{guhaniyogi_bayesian_2015,mukhopadhyay_targeted_2020} since under standard regularity conditions, non-Gaussian GLM models can be well approximated by Gaussian ones.
\subsection{Simulations}
We performed simulations under different settings for the type of random projection (tensor-wise and mode-wise), covariate tensor dimensions ($20 \times 20$ and $60 \times 60$ mode-2 tensors), and the number of observations (from $500$ to $2000$ at an interval of $500$). In addition, we investigated the sensitivity to compression rate, defined as $r=1/C(N,M)$, where we recall $C(N,M)=p(N)/q(M)$ with  $p(N)=\prod_{m=1}^{N}p_m$, and $q(M)=\prod_{m=1}^{M}q_m$, and different values of the sparsity coefficient $\psi$ used in generating projection matrices (tensors) and the PARAFAC decomposition rank.

The configurations of the tensor coefficient are presented in panel (a) of Figure \ref{fig: sim_types} and are labeled circle ({\bf CI}), cross ({\bf CR}), line ({\bf L}), and block ({\bf B}). The {\bf CI} and {\bf CR} configurations are symmetric along all modes and are sparse with different sparsity levels. The {\bf L} and {\bf B} configurations are asymmetric along at least one mode and represent scenarios in which projections that preserve that mode can improve results. The tensor covariates are drawn independently from the standard normal distribution. The efficiency of Gibbs sampling has been proved computationally on a tensor regression model without projection \citep{CASARIN2025105427}. See Appendix \ref{app:NumRes} for an illustrative example of MCMC output.

For each simulation setting, we performed $L=10$ independent random projections of the same type and combined the results using Bayesian model averaging. This required $2560$ simulations for a given $\psi$. We evaluated the performance of different models using posterior predictive checks. Several quantities are used to evaluate the model fitting. The distance of the actual data from their mean is defined as 
\begin{align}
    d_{j}=(y_j-\bar{y})^2,\qquad j=n+1,\ldots,n+m, \qquad \bar{y}=\frac{1}{m}\sum_{j=n+1}^{n+m}y_j. \label{eq: dist}
\end{align}
The root mean square error across the $L$ independent projections of the same type is defined as
\begin{equation}
    RMSE_{j,n}=\sqrt{\frac{1}{L}\sum_{\ell=1}^{L}(y_j-\tilde{y}_{j,n})^2},\, j=n+1,\ldots,n+m,
\end{equation}
where $\tilde{y}_{j,n}$ is the point prediction obtained for the $j$th out-of-sample item, based on a training sample of size $n$.

\subsubsection{Type of projection}

The top plots in Figure \ref{fig: rmse} show the RMSE (vertical axis) for the different baseline settings, where $20\times 20$ (panel a) and a $60\times 60$ (panel b) true tensor coefficients are used in generating $n=1,500$ i.i.d. samples from the tensor-regression model. In each plot, the RMSEs are reported for each projection method (horizontal axis) and configuration setting (different lines and symbols). The tensor-wise projection (first symbol in the four lines) underperformed the mode-wise projections in our four simulation settings.

The bottom plots in Figure \ref{fig: rmse} show the RMSE (vertical axis) for different training sample sizes (horizontal axis) for the simulation setting {\bf CR} with different types of random projections (different lines). Note that the larger dots in each line indicate the RMSEs reported in the blue line of panel (a). There is a clear downward-sloping trend as the training sample size increases across all random projection types, with mode-wise projections outperforming the tensor-wise. Among the mode-wise projections, the one preserving the second mode performs best (dotted line).

We investigate the features of the different projection methods by comparing the actual values $y_{n+j}$ in the test set with their predicted values $\tilde{y}_{n+j}^{(\ell)}$ (scatter plots in Figure \ref{fig: sim_types}).
Column 1 of panel (b) has been obtained using tensor-wise random projection (\texttt{GTRP-TW}) and Bayesian tensor regression on a training sample of $n=1,000$ observations and a test sample of $m=500$ observations. Compression rate $r=0.36$ and sparsity coefficient $\psi=3$ are used to generate the random projection tensors. Each plot reports the true (horizontal axis) and the predicted response variable (vertical axis). The estimation and prediction exercise has been performed using $L=10$ independent projections of the same type (different colored dots) and different data-generating settings (different rows). 

\begin{figure}[ht!]
    \centering
     \begin{tabular}{cc}
     
     \small (a) $20\times 20$ coefficient tensor, all settings & \small(b) $60\times 60$ coefficient tensor, all settings \\
    \includegraphics[trim={.5cm .5cm 0 .52cm}, clip, width=0.35\linewidth]{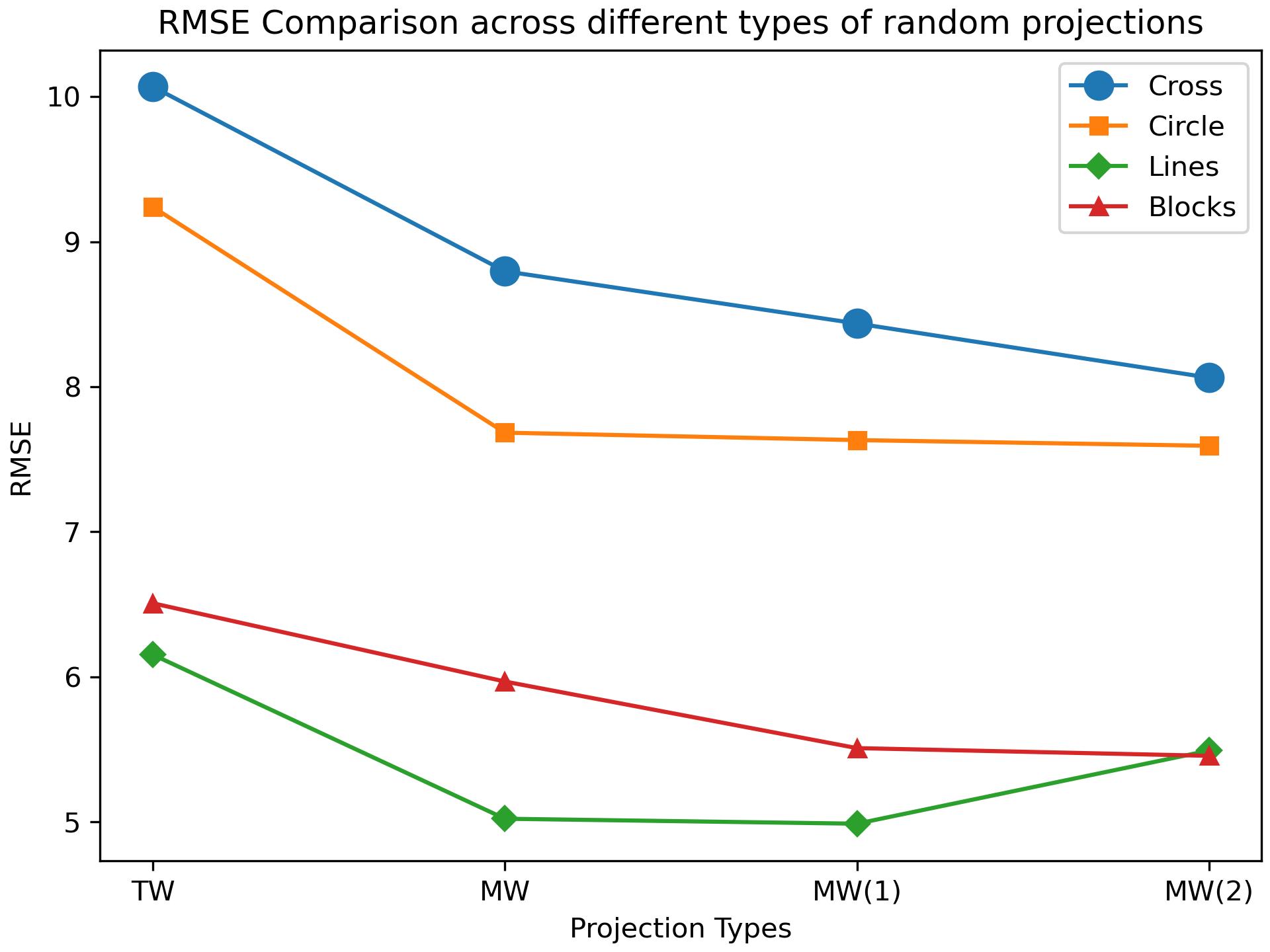}& \includegraphics[trim={.5cm .5cm 0 .52cm}, clip, width=.35\linewidth]{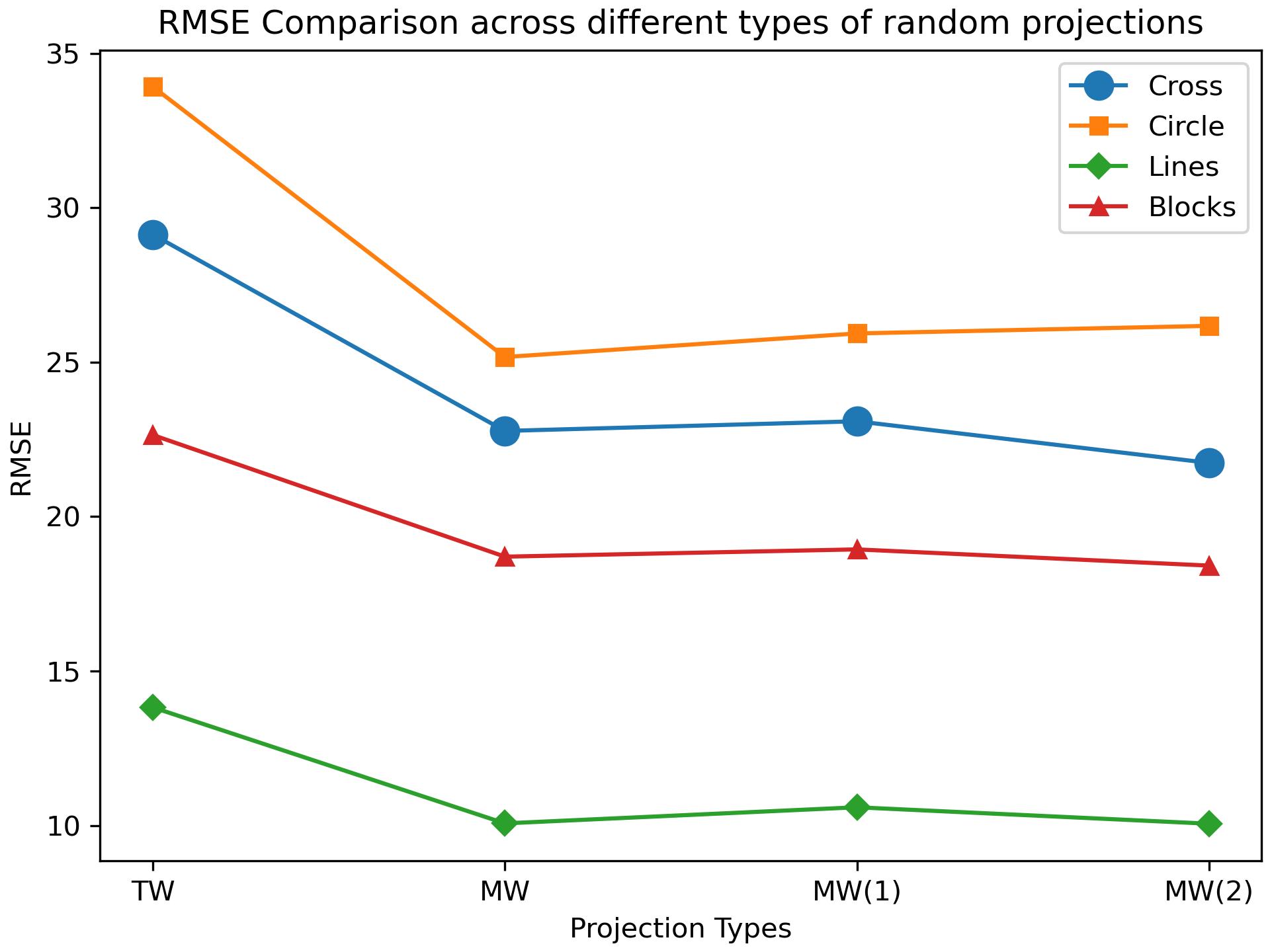}\\
     \small (c) $20\times 20$ coefficient tensor, {\bf CR} setting & \small(d) $60\times 60$ coefficient tensor, {\bf CR} setting \\
    \includegraphics[trim={0 .5cm 0 .52cm}, clip, width=0.35\linewidth]{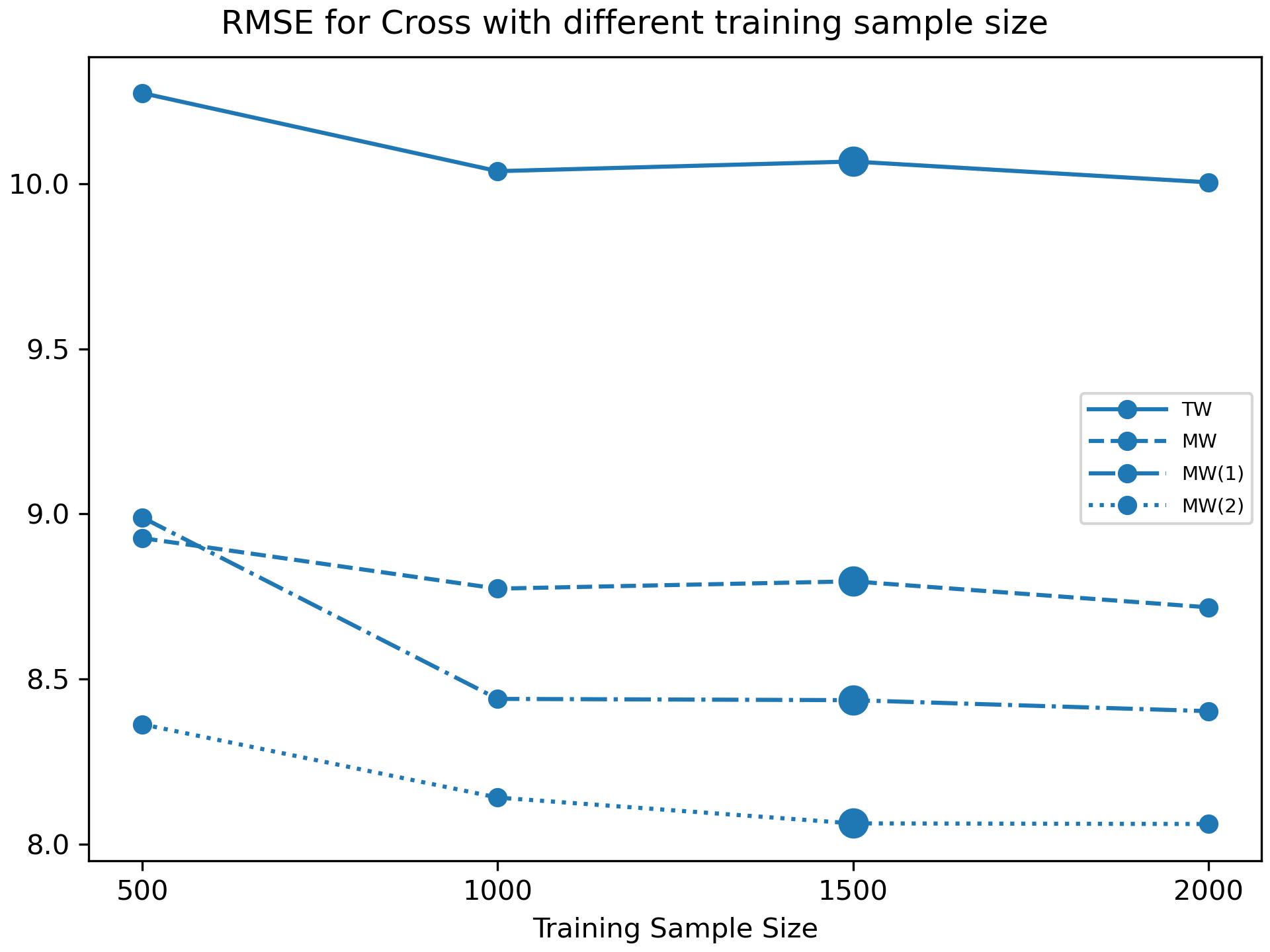}&\includegraphics[trim={0 .5cm 0 .52cm}, clip, width=.35\linewidth]{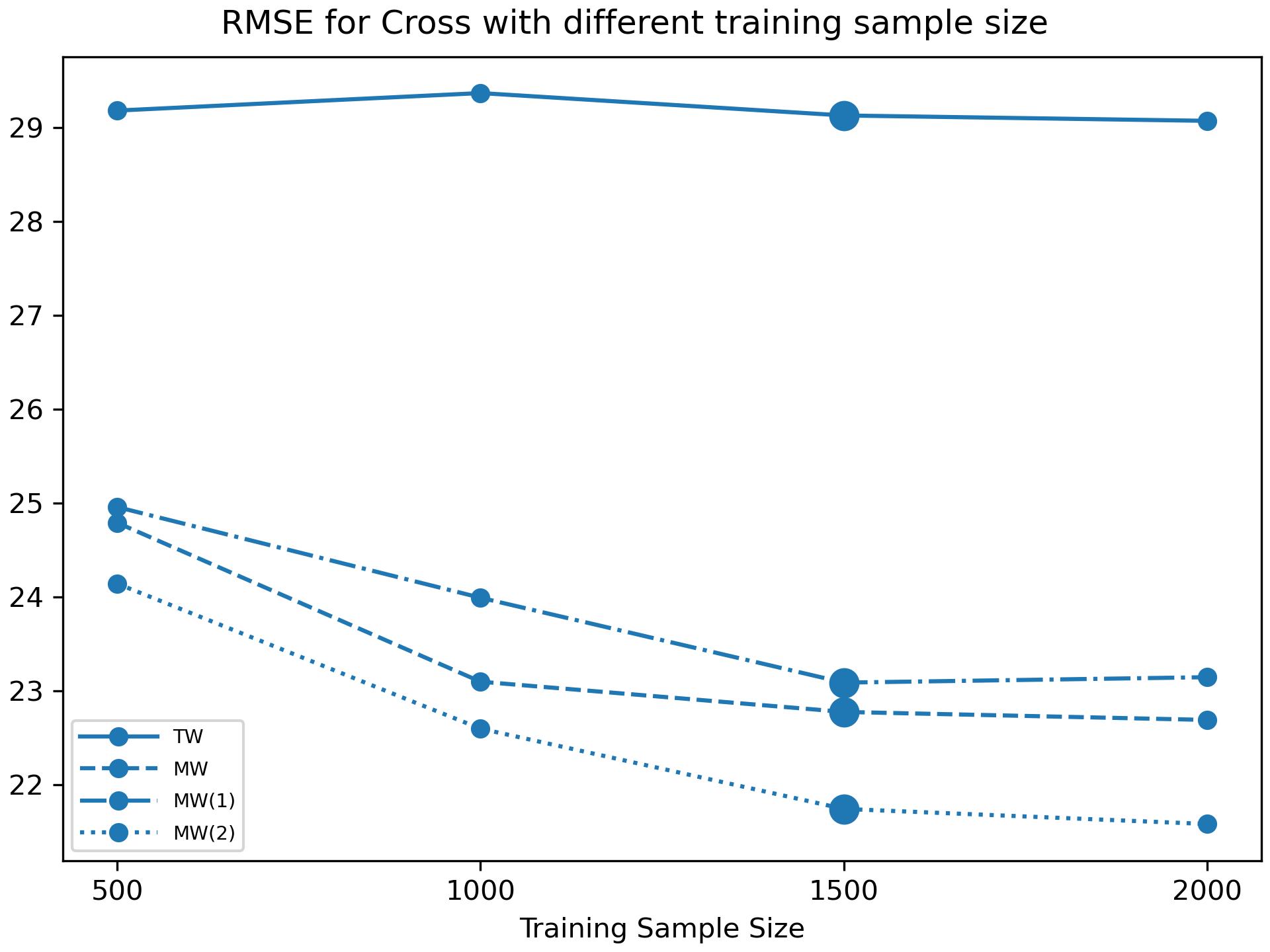} \\  \end{tabular}
    \caption{RMSE comparison across types of random projections (TW: tensor-wise, MW: mode-wise, MW(1): mode-wise preserving the first mode, and MW(2): mode-wise preserving the second mode), settings (blue: {\bf CR}, orange: {\bf CI}, green: {\bf L} and red: {\bf B}), and dimensions ((a):  $20 \times 20$ and (b): $60 \times 60$). The top panels show the RMSEs (vertical axis) obtained for a training sample of size 1500 for different projection types (horizontal axis) in different settings (colors and symbols).  The bottom panels show the RMSEs (vertical axis) for different training sample sizes (horizontal axis) and different projection types (line types). Each estimate is obtained via BMA over $L=10$ independent projection matrices of the same type and $500$ data points from the validation set. The larger dots in plots (c) and (d) indicate the RMSEs reported in the blue line of plots (a) and (b), respectively.}
    \label{fig: rmse}
\end{figure}

\begin{figure}[t]
    \centering
    \begin{tabular}{ccccc}
    \scriptsize (a) True Coefficient & \multicolumn{4}{c}{\scriptsize (b) Forecast Performance}\\
    &\scriptsize TW & \scriptsize MW & \scriptsize MW$(1)$ & \scriptsize MW $(2)$ \\
        \raisebox{2mm}{\includegraphics[width=.19\linewidth]{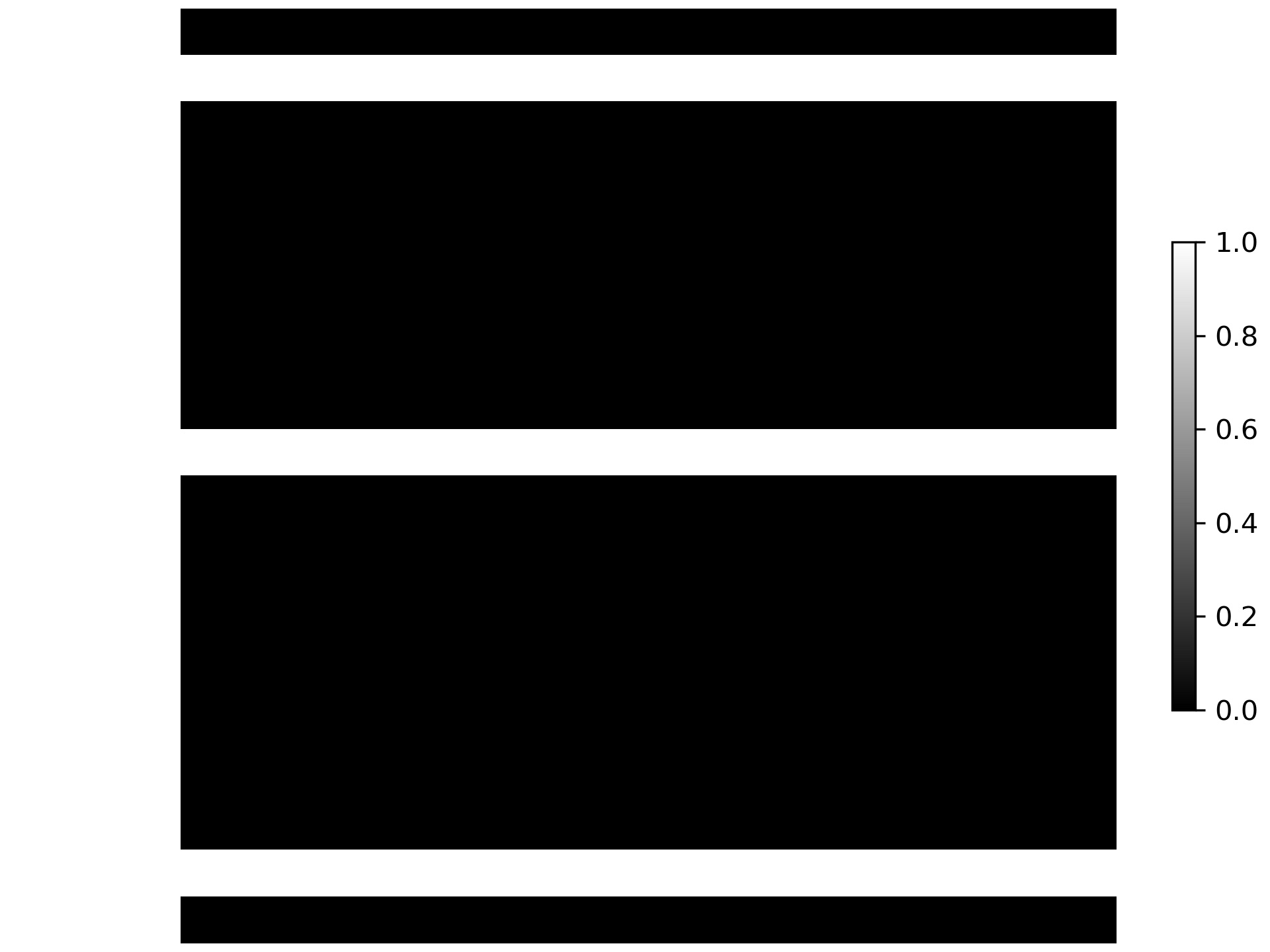}} & \includegraphics[width=.16\linewidth]{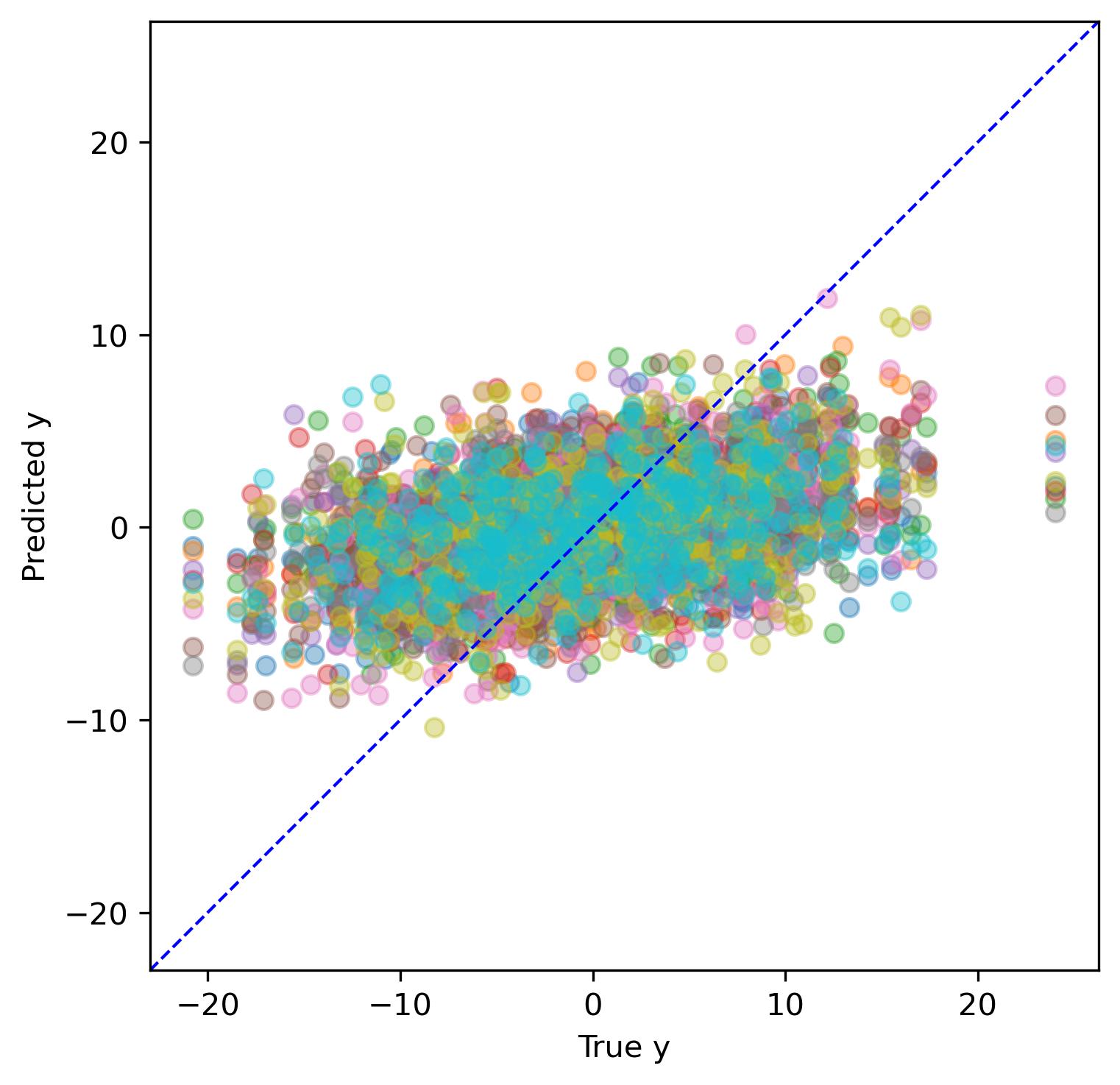} & \includegraphics[width=.16\linewidth]{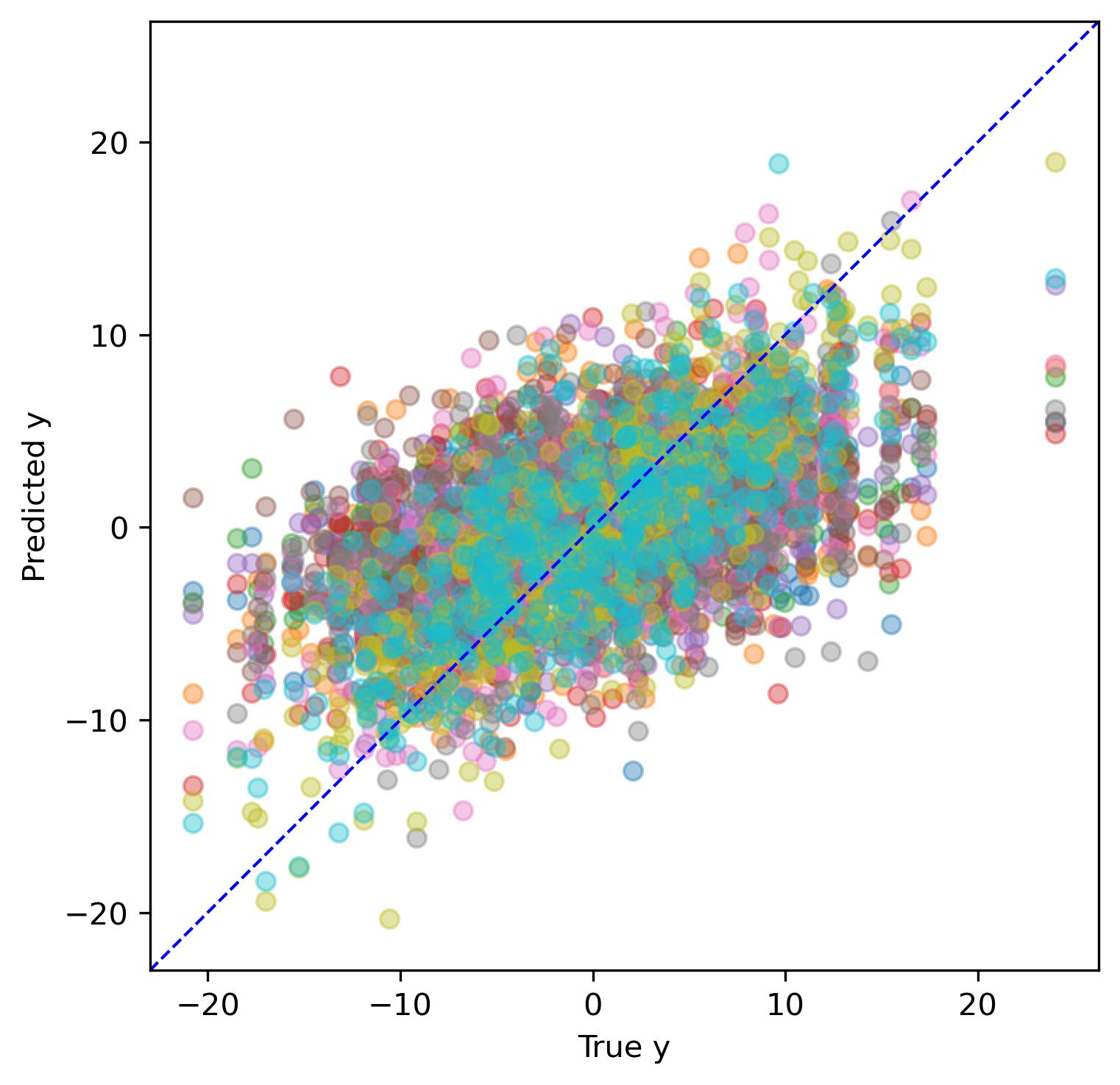} & \includegraphics[width=.16\linewidth]{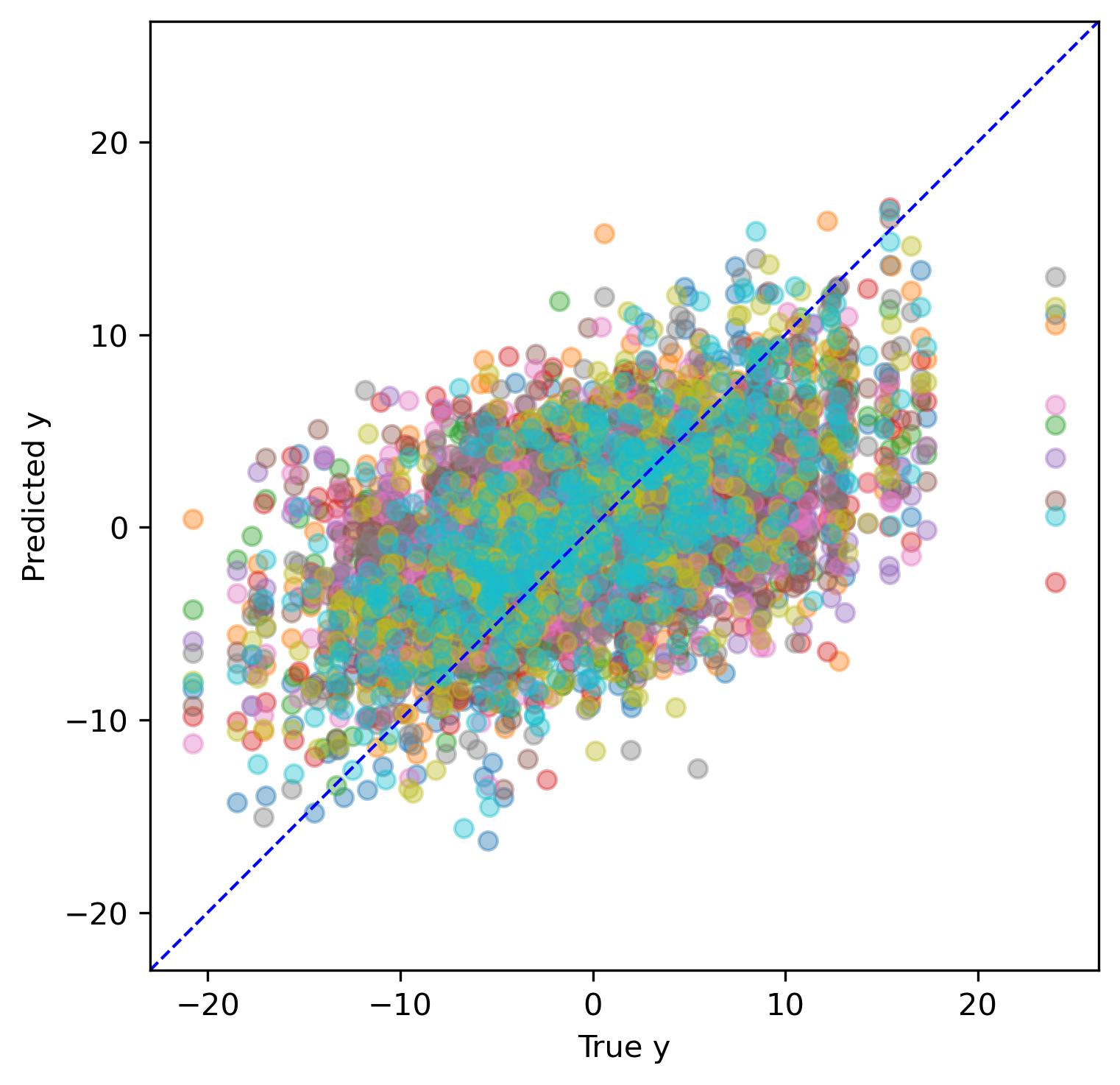} & \includegraphics[width=.16\linewidth]{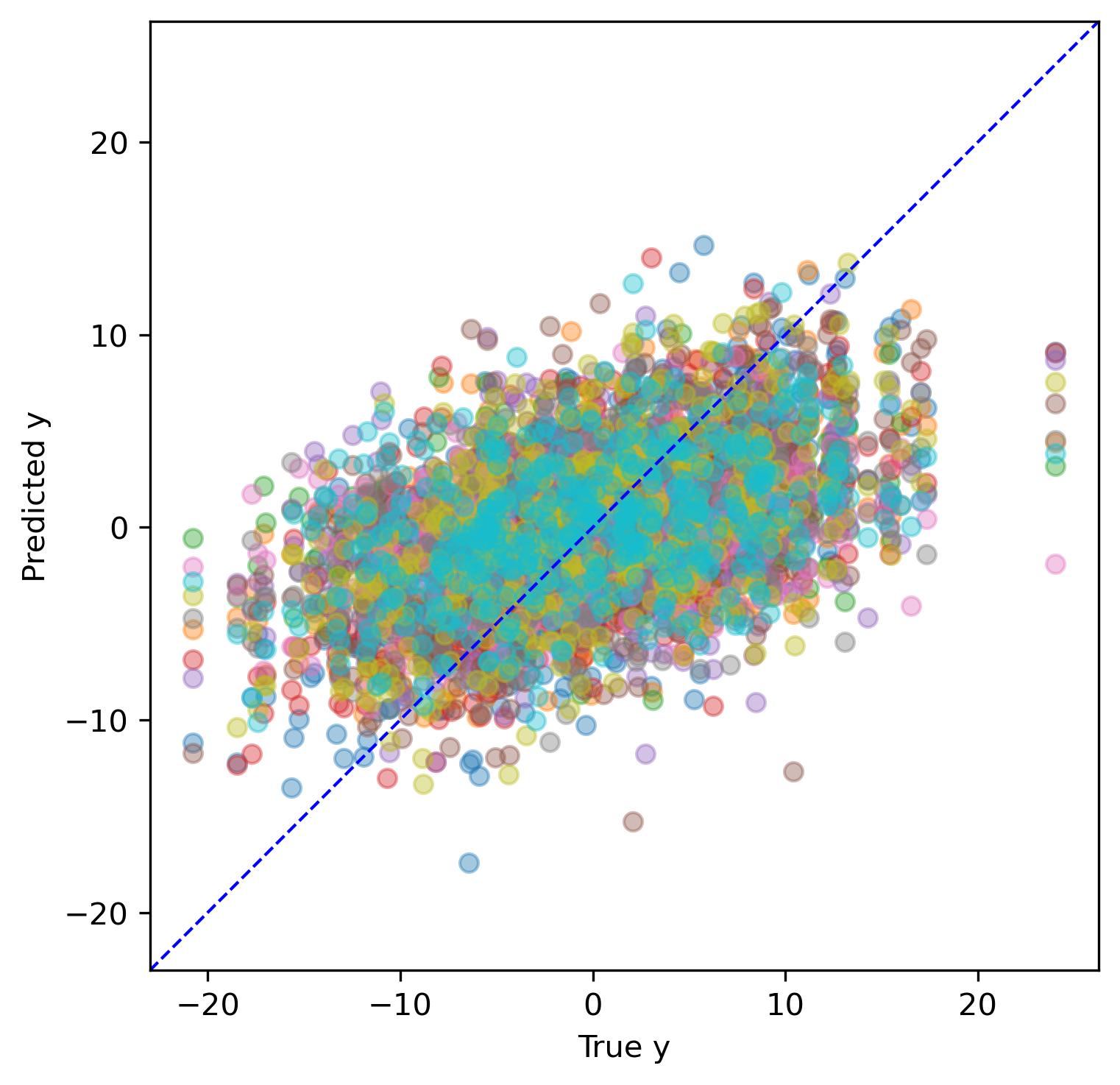}\\
        \raisebox{2mm}{\includegraphics[width=.19\linewidth]{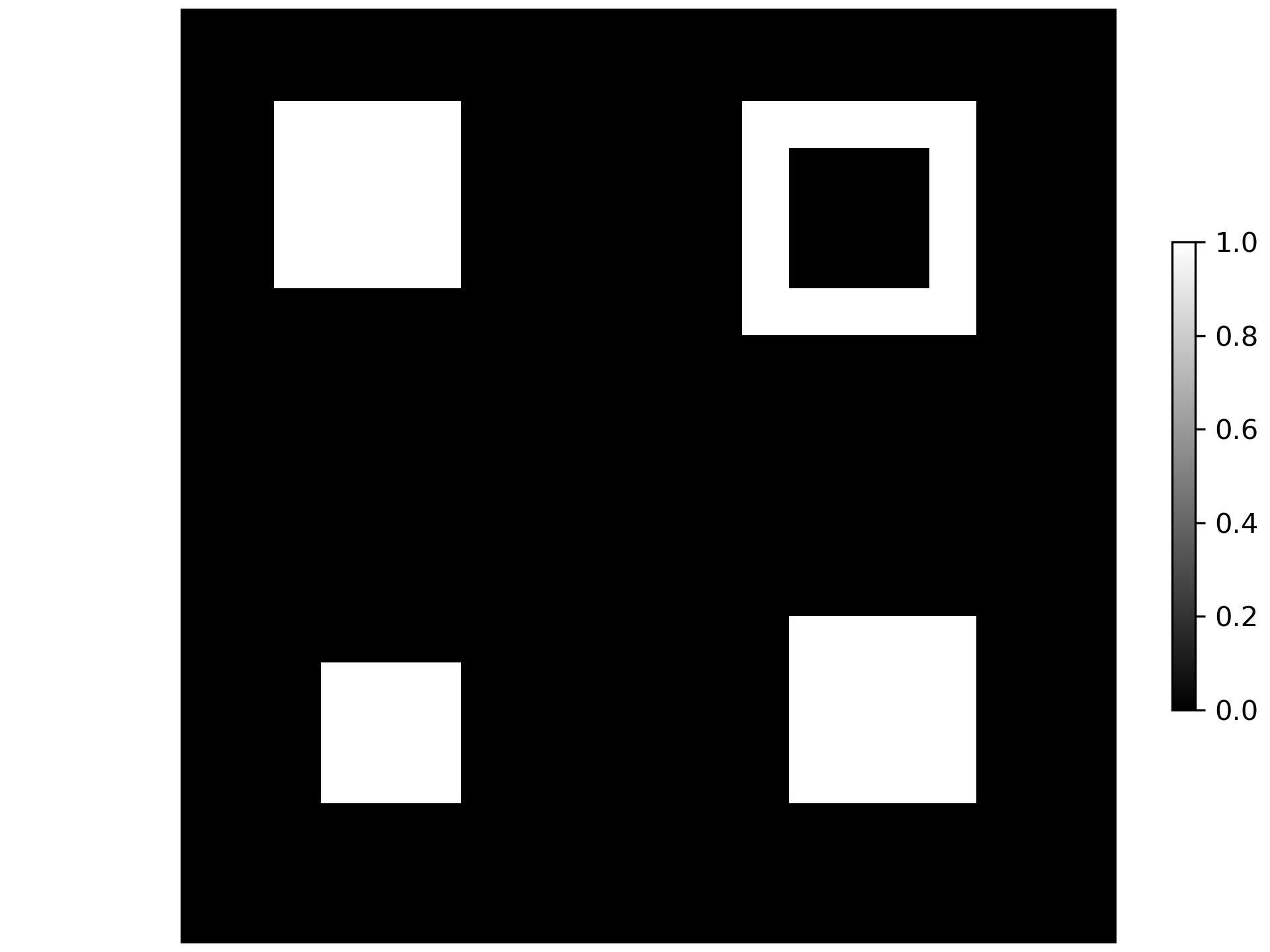}} & \includegraphics[width=.16\linewidth]{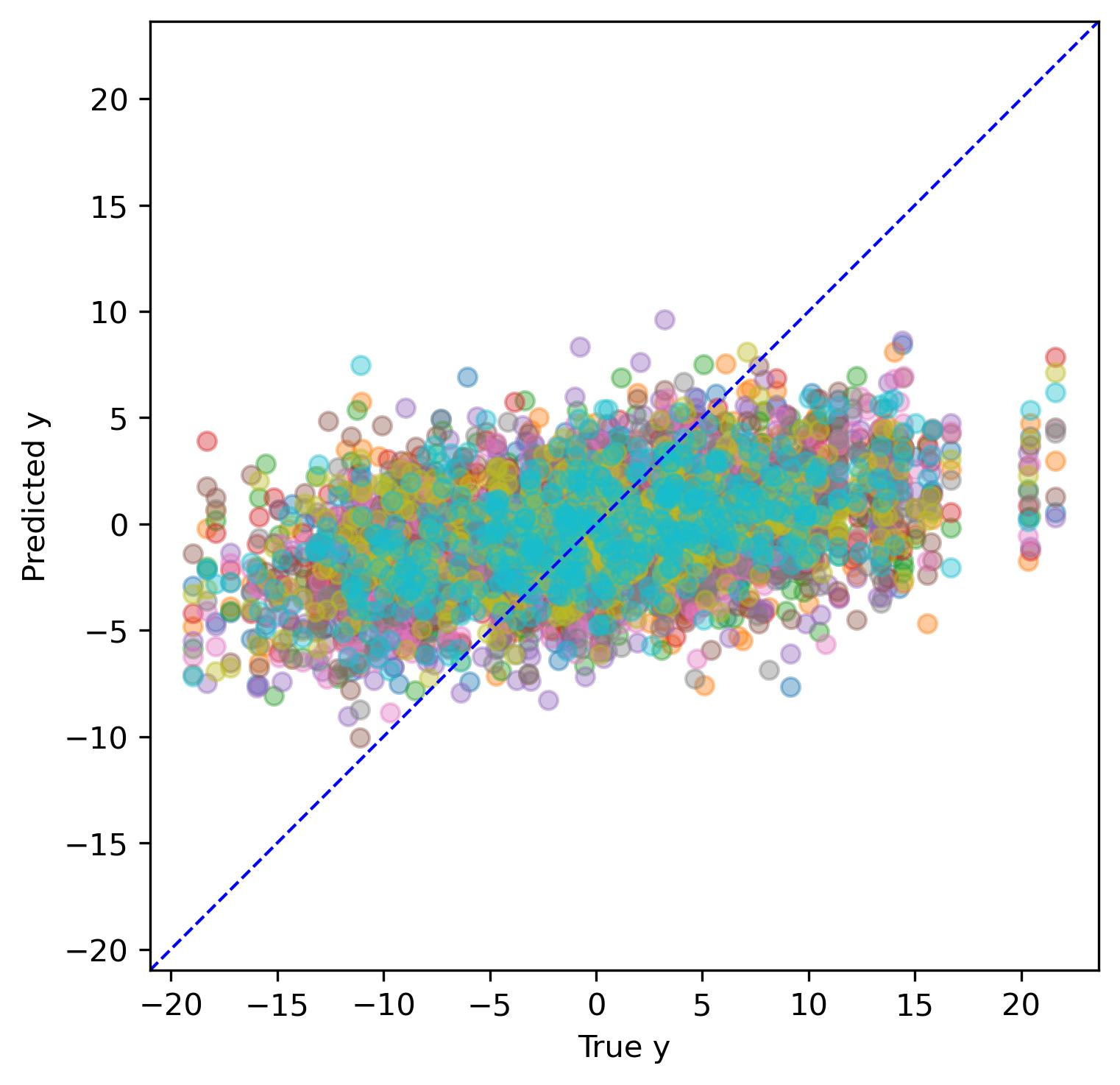} & \includegraphics[width=.16\linewidth]{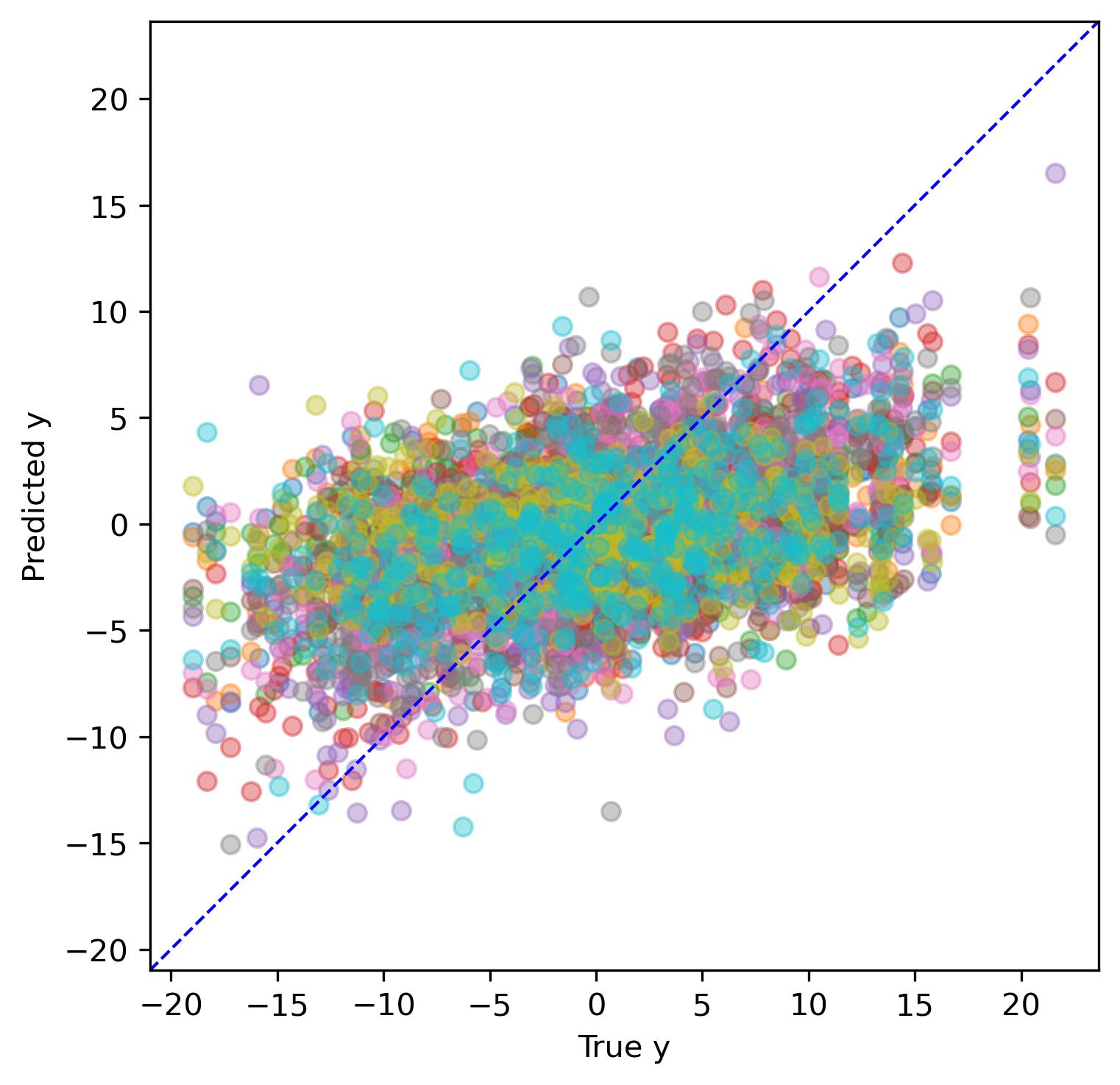} & \includegraphics[width=.16\linewidth]{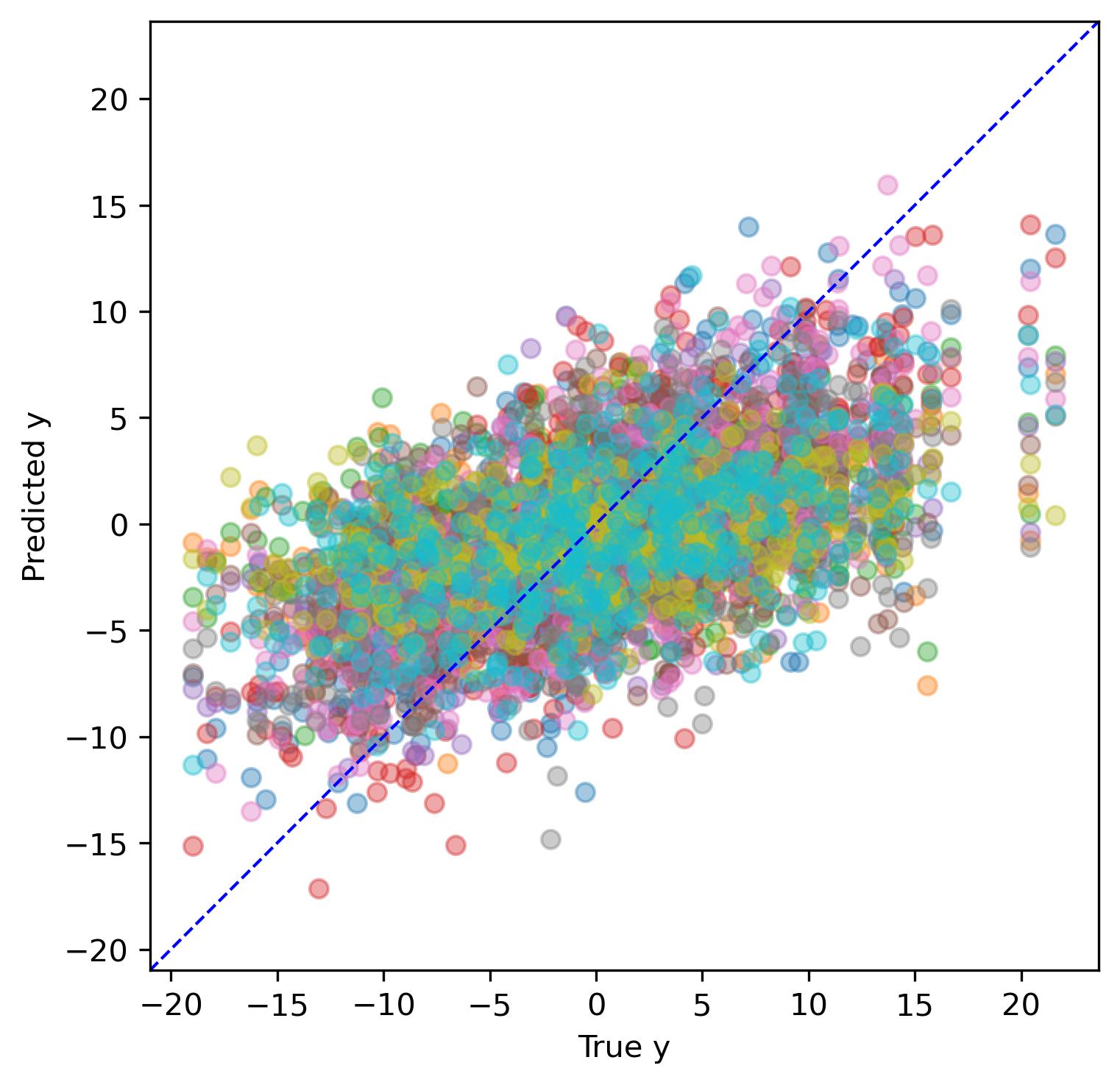} & \includegraphics[width=.16\linewidth]{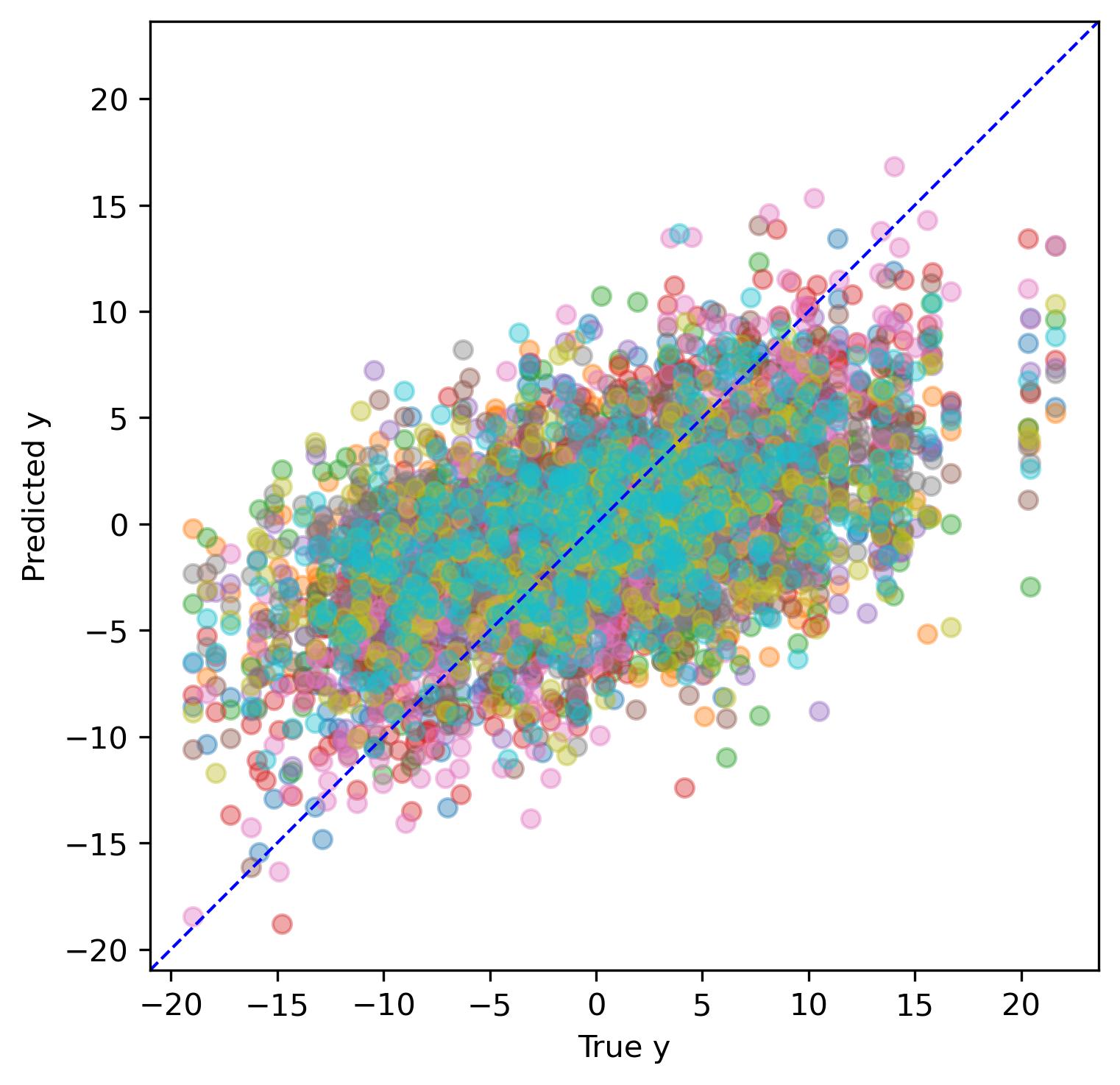}
    \end{tabular}
    \caption{True coefficient (panel a) and forecast (panel b). In each scatter plot: actual data (horizontal axis) against the predicted data (vertical axis) for different sparsity levels and structures (rows) and different types of random projections (columns), using $L=10$ independent projection matrices of the same random projection type (colors). In the experiments: training sample size $n=1000$, compression rate: $r=0.36$, sparsity parameter $\psi=3$.}
    \label{fig: sim_types}
\end{figure}

\begin{figure}[h!]
    \centering
    \begin{tabular}{ccc}
    (a) RMSE vs distance & (b) Variance within MSE & (c) Bias within MSE\\
        \includegraphics[width=.31\linewidth]{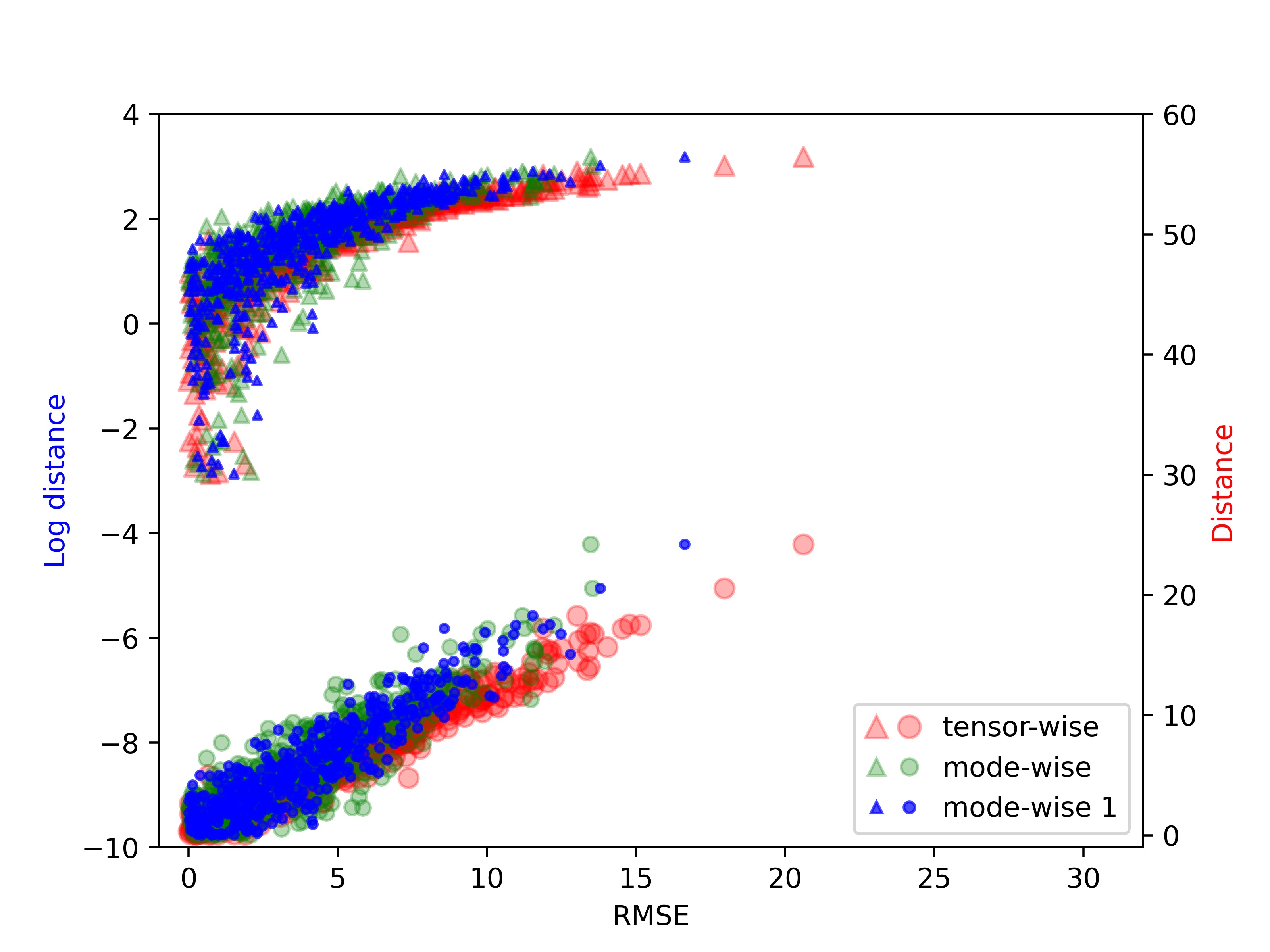}& \includegraphics[width=.28\linewidth]{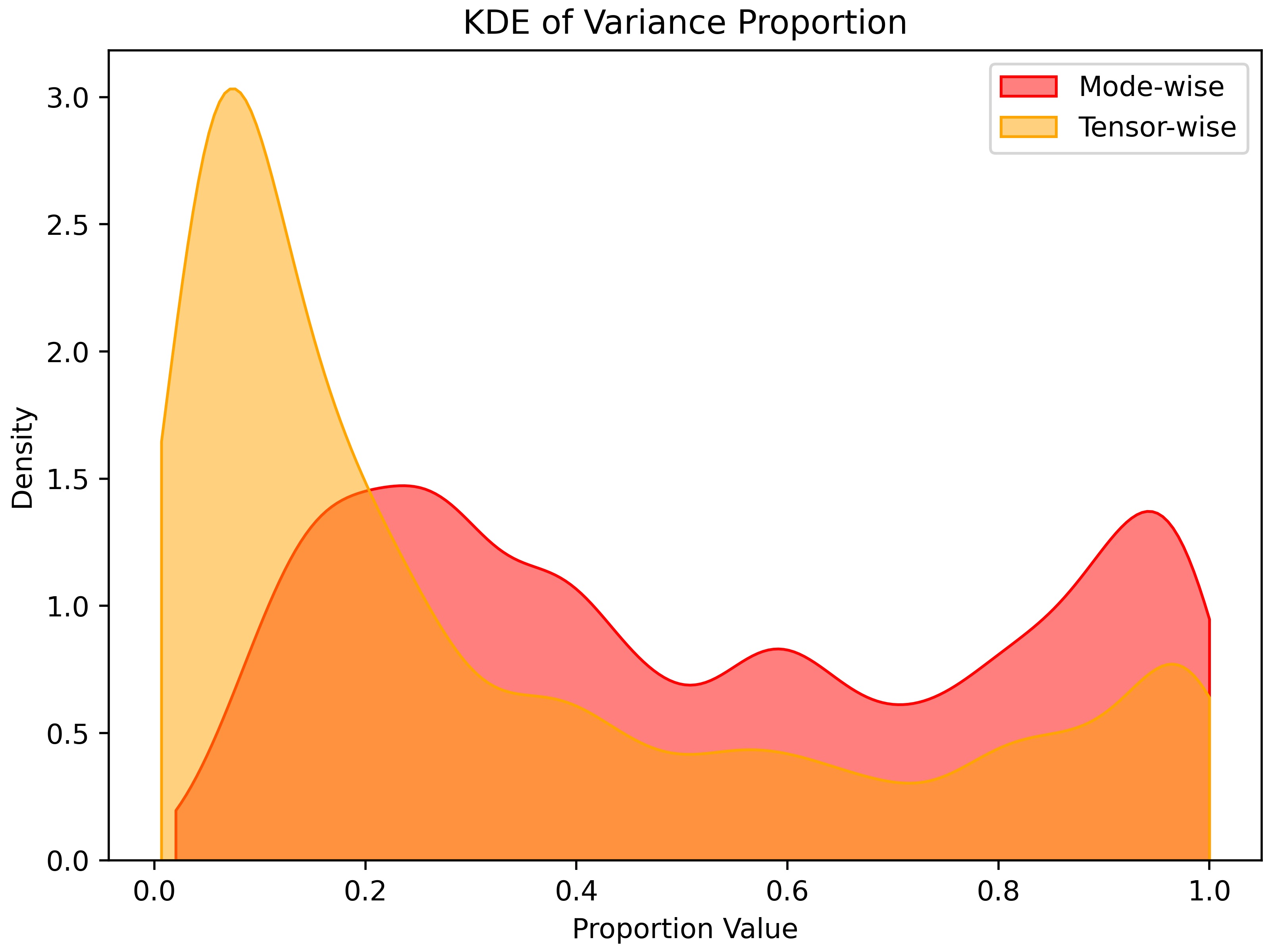} & \includegraphics[width=.28\linewidth]{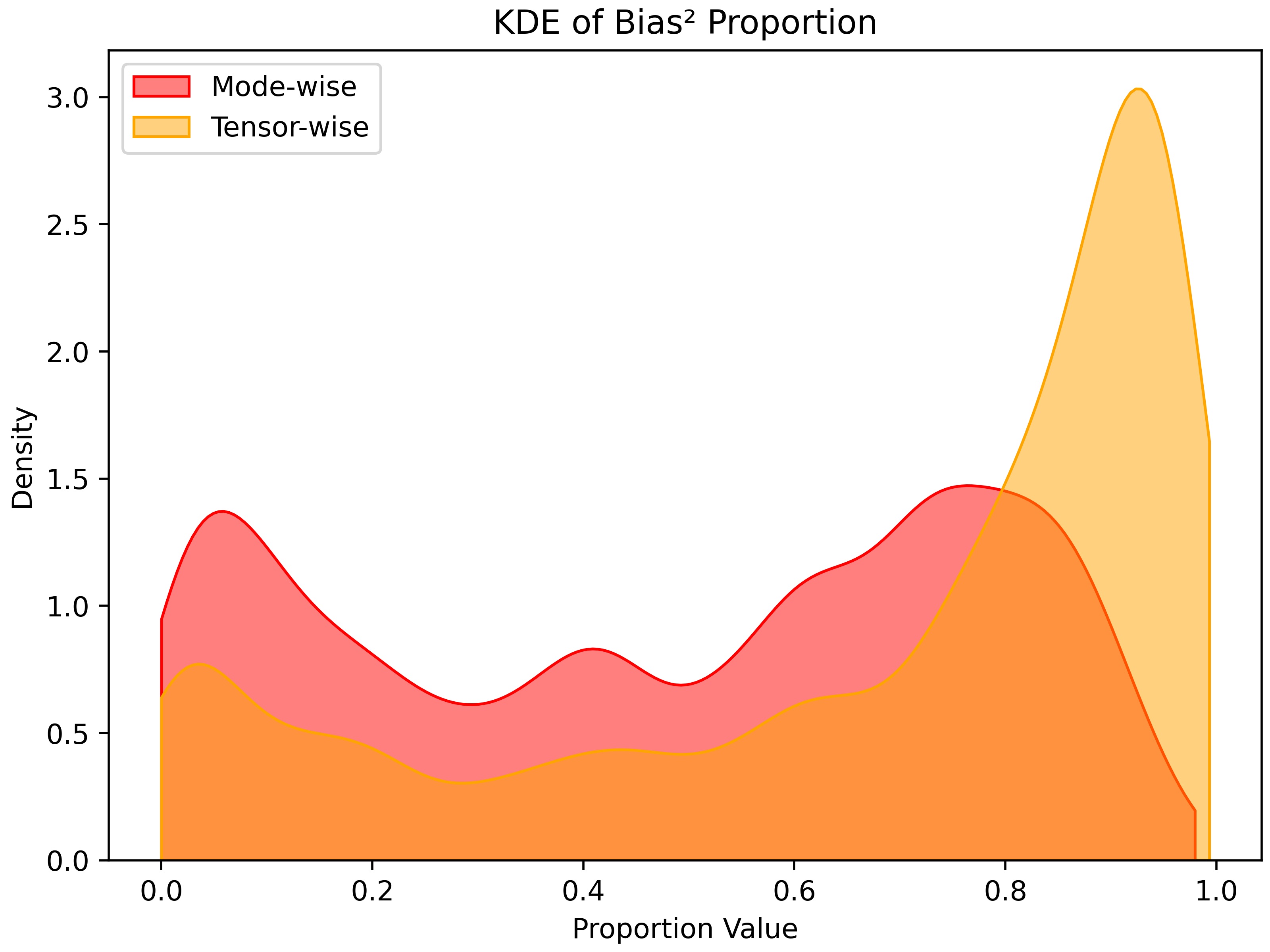}\\
        \includegraphics[width=.31\linewidth]{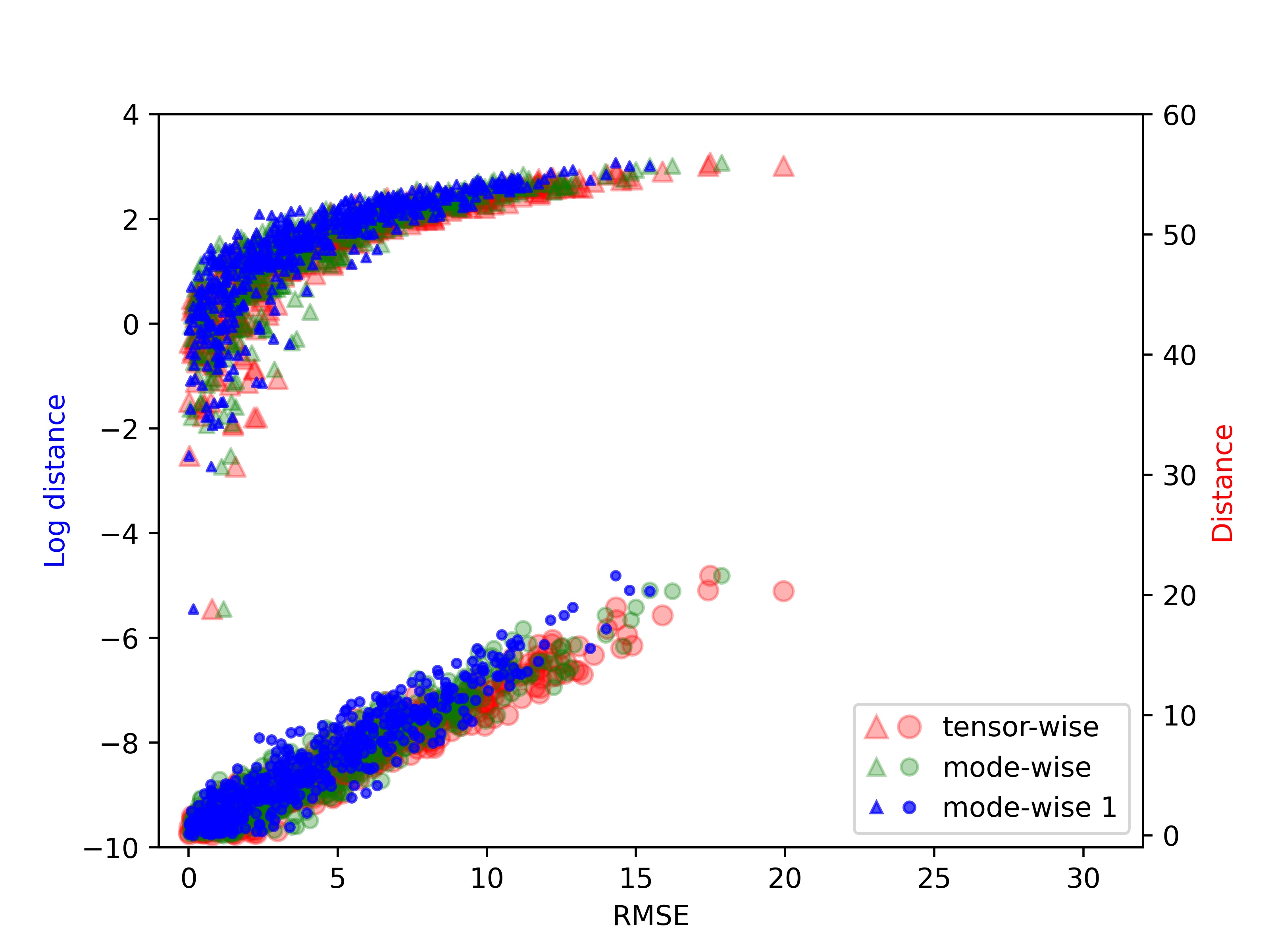}& \includegraphics[width=.28\linewidth]{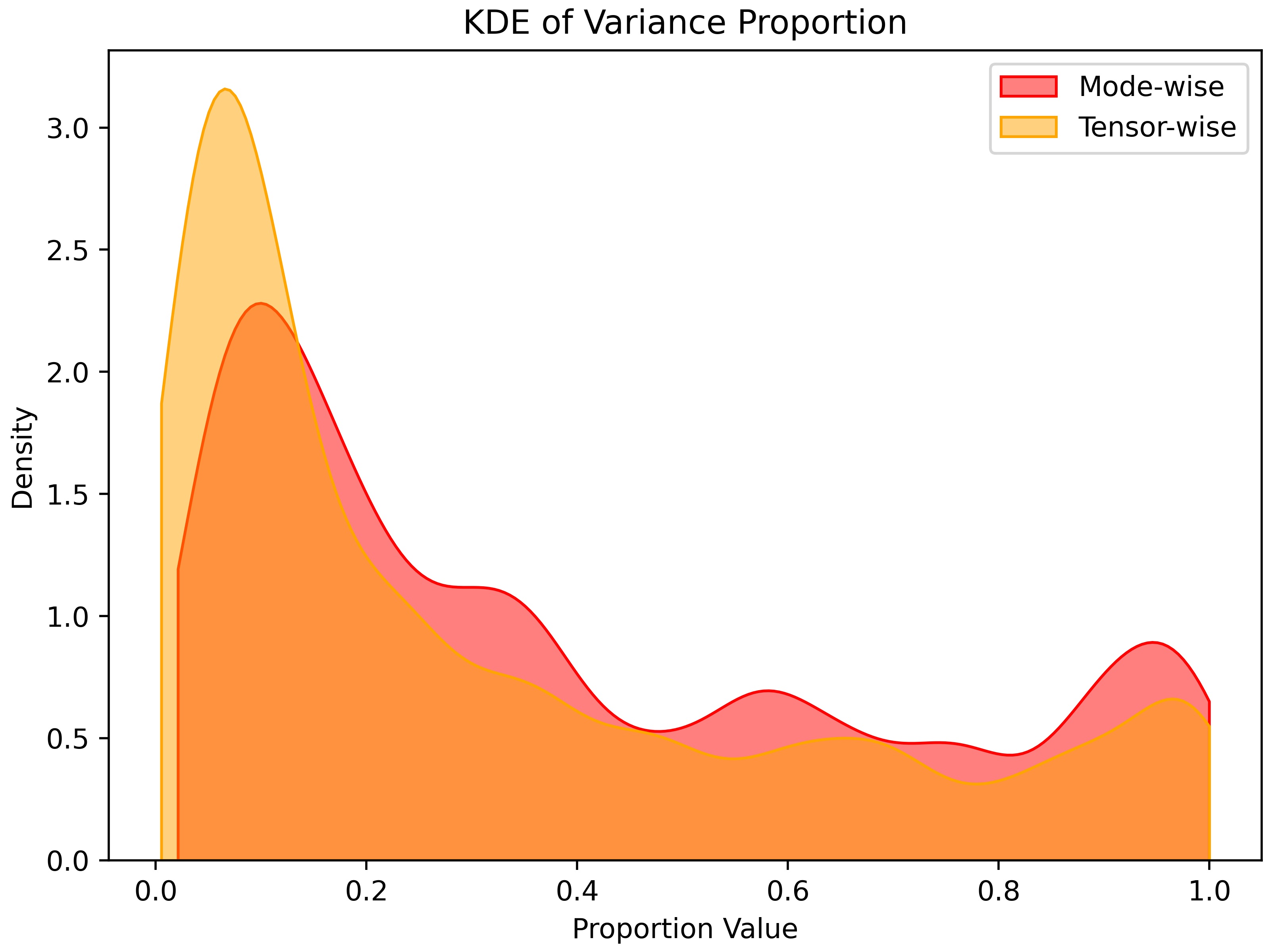} & \includegraphics[width=.28\linewidth]{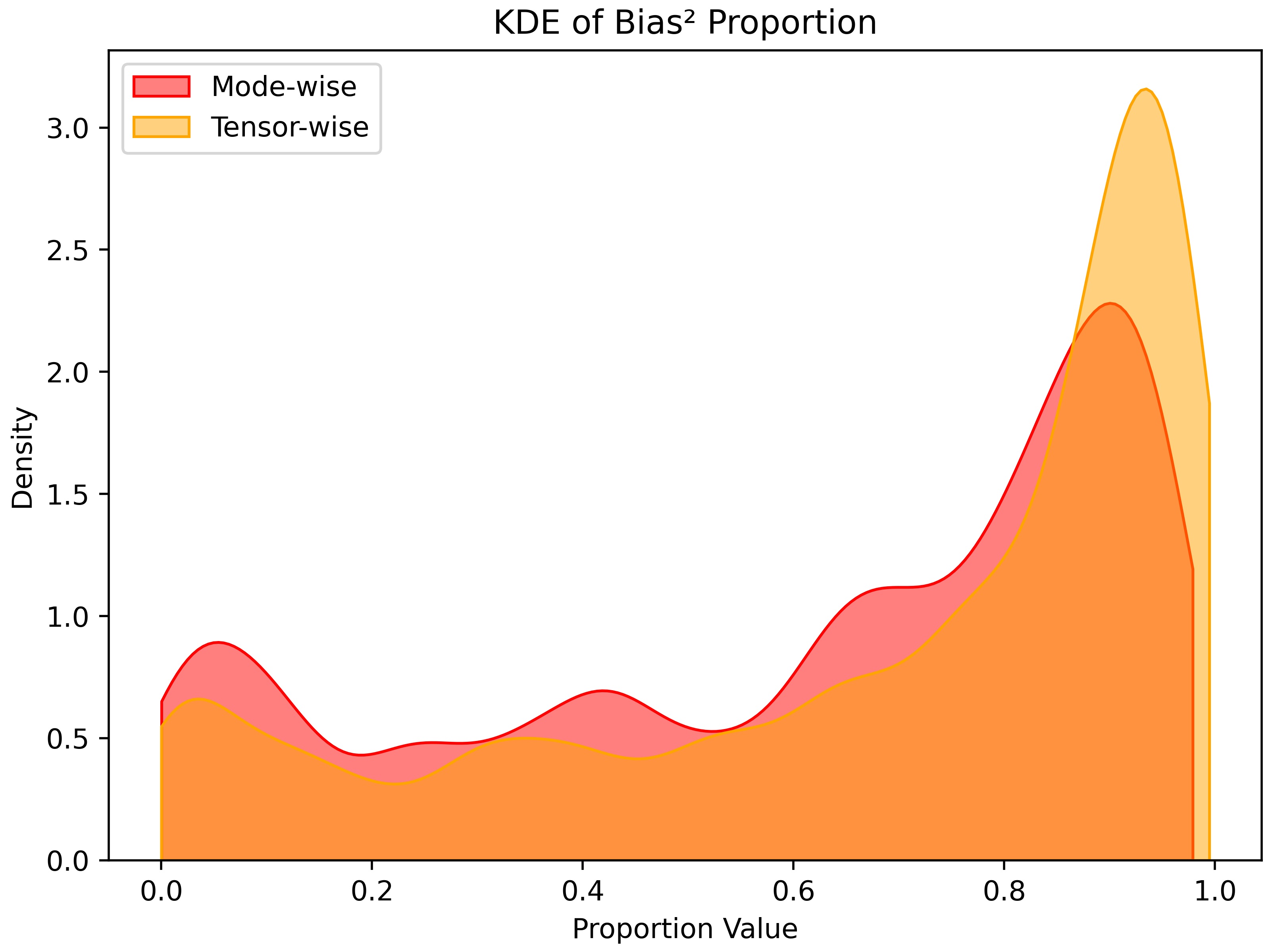}\\
    \end{tabular}
    \caption{Prediction errors. Panel (a) shows RMSE vs actual distance $d_j$ as defined in \eqref{eq: dist} (circle plotting symbols, right axis) and log-distance (triangle plotting symbols, left axis) between $m=500$ data points and their mean obtained from different types of random projections: TW (pink), MW (green), and MW(1) (blue). Panels (b) and (c) show the decomposition of MSE obtained from the $m=500$ test samples for two different types of random projections: TW (yellow) and MW (pink). Panel (b) shows the variance contribution to the MSE, and panel (c) shows the bias contribution to the MSE.}
    \label{fig: rmse1}
\end{figure}

The same prediction evaluation has also been carried out for random projection types: Mode-wise (\texttt{GTRP-MW}), Mode-wise preserving mode 1 (\texttt{GTRP-MW(1)}), and Mode-wise preserving mode 2 (\texttt{GTRP-MW(2)}) (columns from 2 to 4, respectively).
The plots in panel (b) show that \texttt{GTRP-TW} has difficulties in fitting the actual data (comparing the distance of the clouds from the 45$^{\circ}$ reference line). In contrast, \texttt{GTRP-MW}, \texttt{GTRP-MW(1)}, \texttt{GTRP-MW(2)} perform better for values of the actual data both close and far from the mean. 

To further investigate the relationship between the incurred errors and the relative distance of an observation from its distribution's mean,  
we produce graphical representations of  the relationship between the  distance $d_j$ defined in \eqref{eq: dist} and the forecasting error RMSE$_{j,n}$ $j=n+1,\ldots,n+m$. 

In the leftmost column of Figure \ref{fig: rmse1} we show scatter plots of distance (marked with circles), and scatter plots of distances on the log-scale (marked with triangles) versus RMSE. We use two different scales for distances because we are interested in regions where distances are small (and the log scale goes to $-\infty$) and in regions in the right tail where the log scale is more interpretable.  In every plot, the top cloud (triangle symbols) shows the tail behavior, while the bottom one (circle symbols) shows the relationship in the center of the distribution. Blue symbols are generally to the left of the other color symbols, suggesting that mode-wise random projection with mode-preserving yields smaller RMSE for the same distances. 

The right columns present the empirical distributions of the variance and bias proportions for the $m$ points in the test sample.  The forecasts for $m=500$ points in the test sample are obtained using a training sample of size $n=1000$. The decomposition of MSE shows that tensor-wise random projection yields smaller variances but higher bias across all four different simulation settings than mode-wise random projection.

\subsubsection{Sparsity and compression rates}
Parameter $\psi$ controls the sparsity level in the random projection tensor. When $\psi=1$, the entries of the random projection tensor are essentially drawn from $\{-1, 1\}$ with equal probabilities (a Rademacher distribution), a non-sparse projection tensor. As $\psi$ increases, the entries of the random projection tensor will be drawn from $\{-1, 0, 1\}$ with increasing probability that $0$ is drawn, and the projection tensor becomes sparser as $\psi$ increases. 

Figure \ref{fig:rmse_psi} reports the RMSE for simulation configurations of {\bf L} and {\bf B} using dense to sparse random projection tensors, i.e. $\psi \in \{2,3,4\}$ (other configurations can be found in the Supplementary Materials). Model averaging is performed across different random projections and different training sample sizes, and the BMA's performance is evaluated using the RMSE. Figure \ref{fig:rmse_psi} suggests that tensor-wise random projection is not as sensitive as mode-wise random projection for varying sparsity of the random projection matrices. Mode-wise random projections still outperform tensor-wise random projections. In most scenarios ({\bf CI}, {\bf CR}, and {\bf L}), mode-wise random projection has the lowest RMSE. A V-shape curve is observed for mode-wise random projection, suggesting that a moderate sparsity in the random projection process is preferred and helps preserve more information.

\begin{figure}
\vspace*{3pt}
    \centering
    \begin{tabular}{cc}
        \includegraphics[width=.4\linewidth]{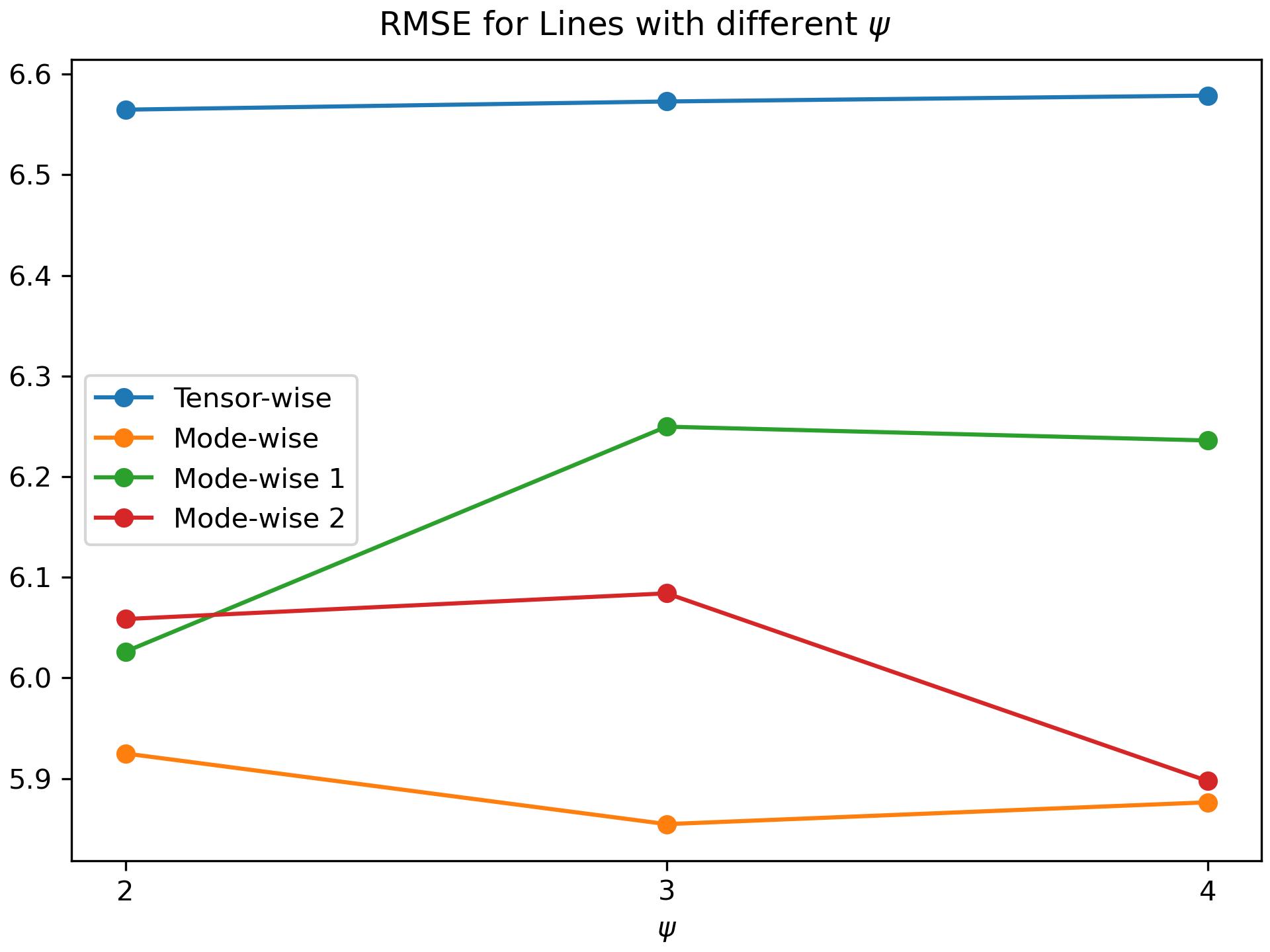}
         & \includegraphics[width=.4\linewidth]{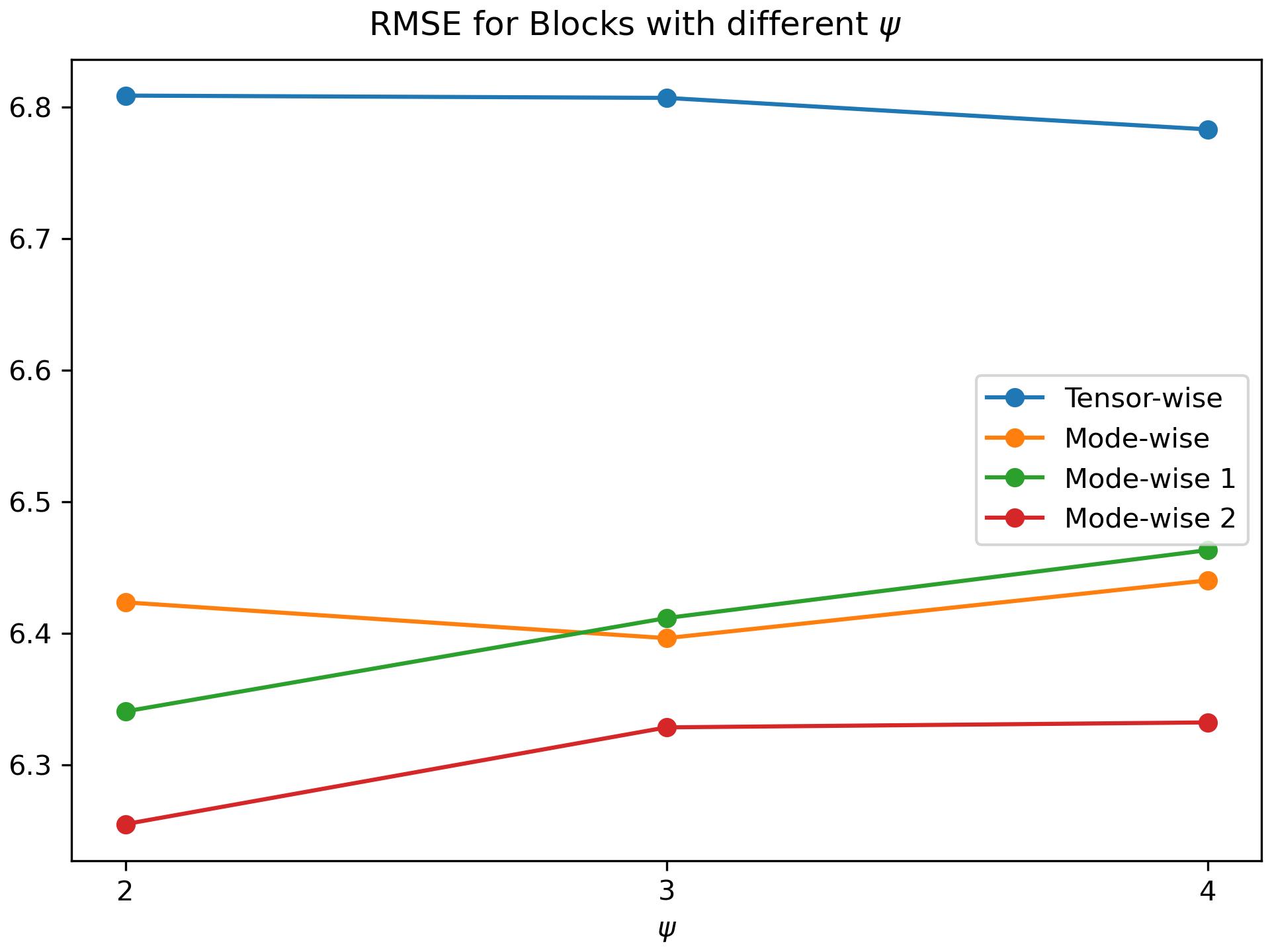}
    \end{tabular}
    \caption{Effects of using random projection matrices of different sparsity levels on prediction errors (RMSE) for {\bf L} (left) and {\bf B} simulation settings (right). In the two plots: the RMSE (vertical axis) obtained from $m=500$ test samples versus the sparsity levels ($\psi\in\{2,3,4\}$) (horizontal axis) for random projection types: TW (blue), MW (orange), MW(1) (green), and MW(2) (red).}
    \label{fig:rmse_psi}
\end{figure}

Figure \ref{fig: comp_rate} in the Supplementary materials shows the prediction performance (out-of-sample scatter plots and RMSE plots) for different compression rates and training sample sizes in the `Cross' setting. The random projection is performed with the first mode preserved. From the scatter plots, it's clear that as the compression rate increases, the regression line's slope increases, suggesting improved predictive performance. This is also shown in the RMSE plots in the right column of Figure \ref{fig: comp_rate}. 


The computational cost of CBTR with different compression rates ($r\in\{0.09, 0.16, 0.25, 0.36 \}$) is compared to that of Bayesian tensor regression using Gaussian priors, Lasso priors and PARAFAC priors. The computational time is reported for the simulation setting {\bf CR} with a tensor coefficient size of $60 \times 60$ and a training sample size of $n=2000$. $1000$ Gibbs iterations are used to sample the unknowns. Table \ref{tab: rmse_sim} reports the RMSE and computational time for all competing models. As the compression rate increases (i.e., more predictors are retained after compression), the computational time also increases. Nevertheless, the CBTR models remain substantially faster than the BTR models with Gaussian, LASSO and PARAFAC priors, with computational gains of approximately one order of magnitude. Among all methods, the BTR model with Gaussian priors achieves the best predictive performance in terms of RMSE, although this comes at a considerably higher computational cost. In contrast, the BTR model with LASSO and PARAFAC priors fail to provide competitive predictive performance relative to the other approaches. The CBTR models exhibit slightly worse predictive accuracy compared with the BTR model with Gaussian priors, but they achieve this performance at a substantially lower computational cost. To further evaluate the trade-off between predictive accuracy and computational efficiency, we consider the following efficiency measure: $\text{Efficiency Score}=1/(\text{RMSE}\times \text{Computational Time})$, which quantifies predictive performance per unit of computational cost. A higher score indicates a more favorable balance between accuracy and computational efficiency. The efficiency scores, reported in the last row of Table \ref{tab: rmse_sim}, show that the CBTR models substantially outperform the BTR models according to this criterion. A figure illustrating the relationship between computational cost and efficiency scores is provided in the Supplementary Material.

\begin{table}[ht]
\caption{Root Mean Square Errors for out-of-sample forecasting of Compressed Bayesian Tensor Regressions (CBTR) with different compression rates ($r\in\{0.09, 0.16, 0.25, 0.36 \}$) as well as Bayesian Tensor Regression (BTR) with LASSO, Gaussian and PARAFAC priors.}
\label{tab: rmse_sim}
\centering
\resizebox{\textwidth}{!}{%
\begin{tabular}{ccccc|ccc}

&  \multicolumn{4}{c|}{CBTR} & \multicolumn{3}{c}{BTR} \\

& $r=0.09$ & $r=0.16$& $r=0.25$ & $r=0.36$ & LASSO & Gaussian & PARAFAC\\
\hline
RMSE     &  $31.47$ & $32.37$ & $32.48$ & $32.33$  & $138.73$ & $\textbf{19.94}$  & $154.77$  \\

Time (mins) &  $\textbf{6.88}$ & $7.71$ & $8.64$ & $9.08$ & $1077$ & $1124$ & $15.57$ \\

Efficiency Score & $\textbf{0.28}$& $0.24$&$0.21$ &$0.20$ & $0.0004$& $0.003$ & $0.025$

\end{tabular}
}
\end{table}

\subsection{Guidelines for implementations}
The proposed methodology entails several choices, and predictive performances may depend strongly on the projection settings. Thus, here we summarize some guidelines suggested by the theory and numerical experiments.

Regarding the type of projection, in general, tensor--wise projection should be avoided, and mode-preserving projections should be preferred, since in the worst case, they return results comparable to the tensor--wise. The identification of the modes to preserve and the choice of the compression rates, when not clear from subject--matter knowledge, can be aided by an exploratory analysis based on mode-specific measures of sparsity. For instance, a screen-then-compress strategy, as proposed by \cite{mukhopadhyay_targeted_2020} or \cite{gailliot2024data}, can be adapted for this purpose.

Concerning the projection matrix sparsity, our numerical studies suggest that a moderate level, e.g., $\psi=3$, should serve as a baseline, and that a more conservative level, e.g., $\psi=2$, could be included in the set of projections if the computational cost is manageable.

Given the difficulties that can arise in choosing a projection setting in many practical applications, Bayesian model averaging and stacking should be used to incorporate model uncertainty into the prediction distribution.

\subsection{Empirical application}
We demonstrate the performance of compressed Bayesian tensor regression (CBTR) using a real-world application studying the effects of oil volatility on the return of stock markets (S\&P $500$). We apply our tensor regression framework to a large dataset with mixed-frequency variables, as used in \cite{CASARIN2025105427}. We regress the monthly log-returns of S\&P $500$ (SP) on covariates sampled at daily frequency with monthly lags ranging from one to four. The daily observations that we included are good oil volatility (GV), bad oil volatility (BV), US dollar index (ER), TED spread (IR), VIX index (VI), T-bill rate (TB), and bond spread (BD). Thus, the tensor predictors and coefficients are of size $(4, 7, 22)$, corresponding to the number of temporal lags, the number of regressors, and the number of daily observations per month. We use $350$ observations for training and $31$ for testing. A representation of the model is:
\begin{align}
    y_{t} &= \mu + \sum_{i_3=1}^{4}\left<B_{\tilde{I}(i_3)},
    A
    \right> + \sigma\varepsilon_{t},
    \label{eq:oil}
    \end{align}

where $\tilde{I}(i_3) = \{(i_1, i_2, i_3), i_h\in\{1,\ldots,p_h\}, \forall h\neq 3\}$ and $B_{\tilde{I}(i_3)}$ denotes the $i_{3}$th slice of tensor coefficients $B$ along the third mode and the $i$th line in the matrix $A$ has the form $X_{t-1/22-i_3+1},\ldots, X_{t-i_3}$ and $X=GV, BV, ER, IR, VI, TB, BD$ for, respectively, rows $1$ to $7$. The conditional mean of the model in \eqref{eq:oil} is given as the sum over slices corresponding to different temporal lags (third mode).

In Figure \ref{fig: pred_comp} of Appendix \ref{app:RealRes}, we compare the in-sample fittings as well as out-of-sample predictions of tensor regression without applying random projection and with different random projection methods (TW: tensor-wise without mode preservation, MW: mode-wise without mode preservation, MW$(1)$: mode-wise preserving first mode, MW$(1,2)$: mode-wise preserving first and second mode). As shown in the figure, the in-sample fittings of BTR and CBTR are relatively similar. This is also reflected in the RMSE reported in Table \ref{tab: rmse_macro}.
\begin{table}[h!]
\vspace*{2pt}
\caption{Root Mean Square (Forecasting) Errors for in-sample fitting (out-of-sample forecasting) of Bayesian Tensor Regression (BTR) and Compressed Bayesian Tensor Regressions (CBTR) with different random projection types.}
\label{tab: rmse_macro}
\makebox[\textwidth][c]{
\begin{tabular}{cc|cccccc}
              &  BTR   & \multicolumn{6}{c}{CBTR} \\
              &  & TW & MW & MW$(1)$ & MW$(1, 2)$ & MW$(1, 3)$ & MW$(2, 3)$\\
              \hline
In-sample     & $0.0338$ & $0.0355$ & $0.0346$ & $0.0356$ & $0.0333$  & $0.0323$ & $0.0329$    \\
Out-of-sample & $0.1148$ & $0.0676$ & $0.0623$ & $0.0723$ & $0.0383$ & $0.0600$ & $0.0508$    
\end{tabular}
}
\end{table}

What differentiates CBTR from Bayesian tensor regression (BTR) is out-of-sample forecasting performances, where every random projection method outperforms BTR. Among the different CBTRs, MW performs better than TW in terms of RMSE, consistently with the simulation results. Among MW models, those preserving modes (MW$(1)$ and MW$(1,2)$) perform better than those not preserving modes (MW). The performances of preserving $1$ and $2$ modes are very close, where preserving $2$ modes offers slightly better in-sample fitting but worse out-of-sample forecasting.

The empirical application demonstrates the validity of random projection for reducing data dimensionality while preserving important information for inference and forecasting. The fact that CBTR outperforms BTR in forecasting is encouraging. Moreover, we explore different random projection methods and find that CBTR-MW performs better than CBTR-TW in both simulation and empirical applications.


\section{Conclusion} \label{sec:conclu}
This paper introduces a Compressed Bayesian Tensor Regression (CBTR) framework that efficiently addresses the challenges of high-dimensional tensor covariates through a novel Generalized Tensor Random Projection (\texttt{GTRP}) strategy. The proposed method extends existing tensor projection approaches by allowing both mode-wise and tensor-wise projections, offering flexibility to preserve or reduce tensor modes and dimensions. Theoretical guarantees are provided by concentration inequalities and posterior consistency results, ensuring that inference and prediction remain valid after compression.

We design a Gibbs sampling algorithm tailored to hierarchical priors, including PARAFAC-based shrinkage priors, and introduce Bayesian model averaging to account for variability introduced by random projections. Our extensive simulation studies demonstrate that CBTR achieves substantial computational gains and improved prediction accuracy compared to standard Bayesian tensor regression, especially when the random projection preserves meaningful tensor structures. These findings are reinforced by an empirical application to financial data, where CBTR outperforms its uncompressed counterpart in out-of-sample forecasting.

Overall, our work establishes CBTR as a scalable, theoretically grounded alternative to conventional tensor regression methods, with potential applications across a wide range of domains involving structured, high-dimensional data. 
Current work can be expanded in several directions. Since a random projection may degrade predictive performance by compressing large sets of uninformative features, a pre-screening step can be applied. Discarding predictors with very low marginal association to the response prior to compression as in \cite{mukhopadhyay_targeted_2020} or  \cite{gailliot2024data} is promising. Moreover, Bayesian predictive stacking \citep{gailliot2024data} can be used as an alternative to BMA for aggregating inference results across different projections. Finally, alternative constructions of the random projection tensors  (e.g., Kronecker-based, tensor train-based, ect.) will be explored in a future communication.

\section*{Acknowledgments and Funding}
We thank David Ardia, Federico Bassetti, Rajarshi Guhaniyogi, Jim Griffin, Alessandra Luati, and Francesco Sanna Passino for their suggestions, which have improved the manuscript. RC was supported by the MUR - PRIN project `\textit{Discrete random structures for Bayesian learning and prediction}' under g.a. n. 2022CLTYP4, QW by the Next Generation EU - `\textit{GRINS - Growing Resilient, INclusive and Sustainable}' project (PE0000018), National Recovery and Resilience Plan (NRRP), and RVC by NSERC of Canada discovery grants RGPIN-2018-249547 and RGPIN-2024-04506.

\begin{appendices}
\renewcommand{\thesection}{A}
\renewcommand{\theequation}{A.\arabic{equation}}
\renewcommand{\thefigure}{A.\arabic{figure}}
\renewcommand{\thetable}{A.\arabic{table}}
\setcounter{table}{0}
\setcounter{figure}{0}
\setcounter{equation}{0}

\section{Proofs of the results} \label{app: proofs}
\subsection{Proof of Proposition \ref{cor: jl}} \label{prf: propjl}
When $R = 0$ and $M=1$ the projection writes as a scalar product between vector and a matrix, that is $\texttt{GTRP}(\mathcal{X}_{j})=\mathcal{X}_{j}\times_{1:N}\mathcal{H}_{1:N}=\sum_{j_1=1}^{p_1}\ldots\sum_{j_N=1}^{p_N}\mathcal{X}_{j,j_1,\ldots,j_N}\mathcal{H}_{j_1,\ldots,j_N,:}=\texttt{vec}(\mathcal{X}_{j})\texttt{mat}_{1:N}(\mathcal{H})$ where $\mathcal{H}$ is a $N+1$-mode projection tensor with iid entries. $\texttt{vec}(\cdot)$ is a vectorization operator and $\texttt{mat}_{1:N}(\cdot)$ is a matricisation operator stacking in one mode all elements from mode 1 to mode $N$ \citep[e.g., see ][Ch. 5]{hackbusch_tensor_2019}. The proof follows by setting $d=p_1\cdots p_N$ and $k=q_1$ in JL's Lemma of \cite{achlioptas_database-friendly_2003}.

\subsection{Proof of Theorem \ref{thm: jl}}
Before proving the theorem, we provide some preliminary results.


\begin{lemma}\label{lem1}
Let $\mathcal{T}=\boldsymbol{\tau}_1 \otimes \cdots \otimes \boldsymbol{\tau}_N$ be a $q_1\times \cdots \times q_N$ tensor with $\boldsymbol{\tau}_{m}\in\mathbb{R}^{q_m}, \boldsymbol{\tau}_{m,i_m}\sim\mathcal{N}(0,1/p_m)$ independent normal. Entries of $\mathcal{T}$ are $\mathcal{T}_{i_1,\ldots,i_N}=\boldsymbol{\tau}_{1,i_1}\cdots \boldsymbol{\tau}_{N,i_N}$. Let
\begin{equation}
\mathcal{Q}=\frac{1}{p(N)}\sum_{j_1=1}^{p_1}\cdots\sum_{j_N=1}^{p_N}H_{1,j_1,:}\otimes \cdots \otimes H_{N,j_N,:}    
\end{equation}
be the $q_1\times \cdots \times q_N$ tensor obtained by projecting the rescaled $p_1\times \cdots \times p_N$ unit tensor (a tensor of ones normalized to have unit Frobenius norm) with random matrices $H_m$ defined in Theorem \ref{thm: jl}, $H_{m,j_m,:}$ indicates the $j_m$th row of $H_m$. Let $\mathcal{Q}(\mathcal{A})$ be the random projection of any arbitrary unit tensor $\mathcal{A}$. The entries $Q_{i_1,\ldots,i_N}(\mathcal{A})$ of the tensor
$\mathcal{Q}(\mathcal{A})$ satisfy the following properties
\begin{itemize}
    \item[i.] $\mathbb{E}(\mathcal{Q}_{i_1,\ldots,i_N}(\mathcal{A})^{2k})\leq \mathbb{E}(\mathcal{Q}_{i_1,\ldots,i_N}^{2k})$ 
    \item[ii.] $\mathbb{E}(\mathcal{Q}_{i_1,\ldots,i_N}^{2k})\leq \mathbb{E}(\mathcal{T}_{i_1,\ldots,i_N}^{2k})$
\end{itemize}
\end{lemma}
\begin{proof}
Without loss of generality, we prove the results for the case $N=3$.
\begin{itemize}
\item[i.] This follows by the same argument as in the proof of Lemma 6.1 in \cite{achlioptas_database-friendly_2003}.
\item[ii.]
\begin{align*}
    &\mathbb{E}(\mathcal{Q}_{i_1,i_2,i_3}^{2k})=\mathbb{E}\left(\left(\frac{1}{p_1 p_2 p_3}\sum_{j_1=1}^{p_1}\sum_{j_2=1}^{p_2}\sum_{j_3=1}^{p_3}H_{1,j_1,i_1}H_{2,j_2,i_2}H_{3,j_3,i_3}\right)^{2k}\right)\\
    & =\mathbb{E}\left(\left(\frac{1}{p_1 p_2 p_3}\sum_{j_1=1}^{p_1}\sum_{j_2=1}^{p_2}H_{1,j_1,i_1}H_{2,j_2,i_2}\sum_{j_3=1}^{p_3}H_{3,j_3,i_3}\right)^{2k}\right)\\
    & =\mathbb{E}\left(\left(\frac{1}{p_1 p_2}\frac{H_{3,i_3}}{p_3}\sum_{j_1=1}^{p_1}\sum_{j_2=1}^{p_2}H_{1,j_1,i_1}H_{2,j_2,i_2}\right)^{2k}\right)\\
    & =\mathbb{E}\left(\left(\frac{H_{1,i_1}}{p_1}\frac{H_{2,i_2}}{p_2}\mathbf{h}_{3,i_3}\right)^{2k}\right)=\mathbb{E}\left(\left(\mathbf{h}_{1,i_1}\right)^{2k}\right)\mathbb{E}\left(\left(\mathbf{h}_{2,i_2}\right)^{2k}\right)\mathbb{E}\left(\left(\mathbf{h}_{3,i_3}\right)^{2k}\right)\\ 
    &\leq \mathbb{E}\left(\left(\boldsymbol{\tau}_{1,i_1}\right)^{2k}\right)\mathbb{E}\left(\left(\boldsymbol{\tau}_{2,i_2}\right)^{2k}\right)\mathbb{E}\left(\left(\boldsymbol{\tau}_{3,i_3}\right)^{2k}\right)=\mathbb{E}\left(\left(\boldsymbol{\tau}_{1,i_1}\boldsymbol{\tau}_{2,i_2}\boldsymbol{\tau}_{3,i_3}\right)^{2k}\right)\\&=\mathbb{E}\left(\left(\mathcal{T}_{i_1,i_2,i_3}\right)^{2k}\right)
\end{align*}
where $H_{m,i_m}=\sum_{j_m=1}^{p_m}H_{m,j_m,i_m}$ with $\mathbb{E}(H_{m,i_m})=0, \mathbb{V}(H_{m,i_m})=p_m$, $\mathbf{h}_{m,i_m}=H_{m,i_m}/p_m$ with $\mathbb{E}(\mathbf{h}_{m,i_m})=0, \mathbb{V}(\mathbf{h}_{m,i_m})=1/p_m$, the inequality  follows using the same argument as in \cite[Lemma 6.2]{achlioptas_database-friendly_2003}.
\end{itemize}
\end{proof}

\begin{lemma}\label{lem3}
Let $x_j \overset{ind}{\sim} \mathcal{G}a(\alpha, \beta_j)$ with pdf
\begin{align*}
    f(x) = \frac{\beta_j^{\alpha}}{\Gamma(\alpha)}x^{\alpha-1}e^{-\beta_j x}, \; x>0
\end{align*}
$\mathbb{E}\left(e^{hx_1\cdots x_N}\right)=\left(\frac{1}{\Gamma(\alpha)}\right)^N G^{1,N}_{N,1}\left(-\frac{h}{\beta_1 \cdots \beta_N}\left| \begin{array}{c}1-\alpha, \ldots, 1-\alpha\\0\end{array}\right.\right)$, where $G^{m,n}_{p,q}(\cdot|\begin{array}{c}a_1,\ldots,a_p\\b_1,\ldots,b_q\end{array})$ is the Meijer G-function given in \citep[][ Def. 1.5]{mathai_h-function_2010}.
\end{lemma}

\begin{proof}
Let $H^{m,n}_{p,q}\left(\cdot\left| \begin{array}{c}  (a_1,A_1),\ldots,(a_p,A_p)\\(b_1,B_1),\ldots,(b_q,B_q)\end{array}\right.\right)$ be the Fox H-function given in \citep[][ Def. 1.1]{mathai_h-function_2010} and define $z=x_2\cdots x_N$. Since $\exp\{hxz\}=H^{1,0}_{0,1}\left(-hzx \left| \begin{array}{c} - \\(0, 1)\end{array}\right. \right)$ \citep[][ Eq. 1.39]{mathai_h-function_2010}, then by the law of iterated expectation
\begin{align}
    & \mathbb{E}\left(\mathbb{E}\left(e^{hx_1z} \mid z\right)\right)  = \mathbb{E} \left(\frac{\beta_1^{\alpha}}{\Gamma(\alpha)}\int^{\infty}_0 e^{-\beta_1 x} x^{\alpha - 1} e^{hxz}dx \right) \notag \\
    & = \mathbb{E} \left(\frac{\beta_1^{\alpha}}{\Gamma(\alpha)}\int^{\infty}_0 e^{-\beta_1 x} x^{\alpha - 1} H^{1,0}_{0,1}\left(-hzx \left| \begin{array}{c} - \\(0, 1)\end{array}\right. \right)dx \right) \label{eq: e2h}
    \end{align}
which is the Laplace transform of $x^{\alpha-1} H^{1,0}_{0,1}\left(-hzx \left| \begin{array}{c} - \\(0, 1)\end{array}\right. \right)$. From Eq. 2.19 in \cite{mathai_h-function_2010}, with $\varrho=\alpha$, $a=-hz$ and $s=\beta$, Eq. \ref{eq: e2h} becomes    
   \begin{align}
    & \mathbb{E} \left(\frac{\beta_1^{\alpha}}{\Gamma(\alpha)} \beta_1^{-\alpha} H^{1,1}_{1,1}\left(-\frac{hz}{\beta_1} \left| \begin{array}{c} (1-\alpha, 1)\\(0, 1)\end{array}\right. \right) \right)\label{eq:laplace}=...\\
    & = \mathbb{E} \left(\left(\frac{1}{\Gamma(\alpha)}\right)^{N-1} H^{1,N-1}_{N-1,1}\left(-\frac{h}{\beta_1\cdots\beta_{N-1}}\left|\begin{array}{c}(1-\alpha, 1), \ldots, (1-\alpha, 1) \\(0,1) \end{array}\right.\right)\right)\notag\\
    & = \frac{1}{\Gamma(\alpha)^N}\beta_N^{\alpha}\int_0^{\infty} e^{-\beta_N x}x^{\alpha -1}H^{1,N-1}_{N-1,1}\left(-\frac{h}{\beta_1\cdots\beta_{N-1}}\left|\begin{array}{c} (1-\alpha, 1), \ldots, (1-\alpha, 1) \\(0,1)\end{array}\right.\right)dx \notag \\
    & = \Gamma(\alpha)^{-N}H^{1,N}_{N,1}\left(-\frac{h}{\beta_1\cdots\beta_{N} }\left|\begin{array}{c}(1-\alpha, 1), \ldots, (1-\alpha, 1) \\(0,1)\end{array}\right.\right)\notag\\
    & = \Gamma(\alpha)^{-N}G^{1,N}_{N,1}\left(-\frac{h}{\beta_1\cdots\beta_{N} }\left|\begin{array}{c} 1-\alpha, \ldots, 1-\alpha \\0\end{array}\right.\right)\notag\\
    & = \Gamma(\alpha)^{-N}G^{N,1}_{1,N}\left(-\frac{\beta_1\cdots\beta_{N}}{h} \left|\begin{array}{c} 1\\1-\alpha, \ldots, 1-\alpha \end{array}\right.\right)  \notag  
\end{align}
where the before last equality follows from the definition of Meijer G-function $G^{m,n}_{p,q}(\cdot|\begin{array}{c}a_1,\ldots,a_p\\b_1,\ldots,b_q\end{array})$ given in \citep[][ Def. 1.5]{mathai_h-function_2010}, and the last equality from Eq. (1.58) in \cite{mathai_h-function_2010}.
\end{proof}

\begin{lemma}\label{lem2}
\begin{equation}
\mathbb{E}\left(\exp\{h \mathcal{Q}_{1,\ldots,1}(\mathcal{A})^2\}\right) \leq  \frac{1}{\pi^{N/2}}G_{1,N}^{N,1}\left(\frac{1}{p(N) 2^N h}\left| \begin{array}{c}1\\1/2,\ldots,1/2\end{array}\right.\right) 
\end{equation}
\end{lemma}

\begin{proof}
By  Monotone Convergence Theorem
\begin{align}
    &\mathbb{E}\left(\exp\{h \mathcal{Q}_{1,\ldots,1}(\mathcal{A})^2\}\right) = \sum_{k=0}^{\infty}\frac{h^{k}}{k!} \mathbb{E}\left(\mathcal{Q}_{1,\ldots,1}(\mathcal{A})^{2k}\right)\leq \sum_{k=0}^{\infty}\frac{h^{k}}{k!} \mathbb{E}\left(\mathcal{T}_{1,\ldots,1}^{2k}\right)\notag\\&=\mathbb{E}\left(\exp\{h \mathcal{T}_{1,\ldots,1}^2\}\right)=\frac{1}{\pi^{N/2}}G_{1,N}^{N,1}\left(-\frac{p(N)}{2^N h }\left| \begin{array}{c}1\\1/2,\ldots,1/2\end{array}\right.\right) \label{eq: mct}
\end{align}
where the inequality follows from Lemma \ref{lem1} and the last equality from Lemma \ref{lem3}, where we set $\alpha=1/2$ and $\beta_j=p_j/2$ in the Meijer G-function, and from $\Gamma(1/2)=\sqrt{\pi}$.
\end{proof}

\subsubsection{Proof of \autoref{thm: jl}}
 To facilitate the proof in this section, we will work with random projection matrices scaled by $1/\sqrt{p_m}$, denoting $\tilde{H}_m=H_m/\sqrt{p_m}$. As a result, the $(i_1, \ldots, i_N)$-th element of $f(\mathcal{U})$ write as:
\begin{equation}
f(\mathcal{U})_{i_1,\ldots,i_N}=\sqrt{C(N,M)}\sum_{j_1=1}^{p_1}\cdots\sum_{j_N=1}^{p_N}\mathcal{X}_{t,j_1,\ldots,j_N}\tilde{H}_{1,j_1,i_1}\cdots \tilde{H}_{N,j_N,i_N} \notag
\end{equation}
where $C(N,M)=p(N)/q(M)$, $p(N)=\prod_{m=1}^{N}p_m$, and $q(M)=\prod_{m=1}^{M}q_m$. We denote with $||f(\mathcal{U})||$ the Frobenius' norm of $f(\mathcal{U})$ and prove that 
\begin{equation}
||\mathcal{U}-\mathcal{V}||^2(1-\varepsilon)\leq ||f(\mathcal{U})-f(\mathcal{V})||^2\leq ||\mathcal{U}-\mathcal{V}||^2(1+\varepsilon)    \notag
\end{equation}
with probability at least $1-\kappa_n$ for any pair $\mathcal{U},\mathcal{V}\in\mathbb{R}^{p_1\times \ldots\times p_N}$.

Since the map satisfies $f(\mathcal{U})-f(\mathcal{V})=f(\mathcal{U}-\mathcal{V})$ the statement becomes
\begin{equation}
||\mathcal{A}||^2(1-\varepsilon) \leq ||f(\mathcal{A})||^2\leq ||\mathcal{A}||^2(1+\varepsilon)    \label{eq: fofa}
\end{equation}
with probability at least $1-\kappa_n$. 
Without loss of generality, it is sufficient to prove that \eqref{eq: fofa} holds for arbitrary unit tensor ($||\mathcal{A}||=1$): \begin{equation}
(1-\varepsilon) \leq ||f(\mathcal{A})||^2\leq (1+\varepsilon)    \notag
\end{equation}
Define $S(\mathcal{A})= ||f(\mathcal{A})||^2/ C(N,M)$ and $\mathcal{Q}(\mathcal{A})$ as the tensor with elements 
\begin{align}
\mathcal{Q}_{i_1,\ldots,i_N}(\mathcal{A})&=\sum_{j_1=1}^{p_1}\cdots\sum_{j_N=1}^{p_N}\mathcal{A}_{j_1,\ldots,j_N}\tilde{H}_{1,j_1,i_1}\cdots \tilde{H}_{N,j_N,i_N}
\end{align}

Then $S(\mathcal{A})=\sum_{i_1=1}^{q_1}\cdots\sum_{i_N=1}^{q_N} \mathcal{Q}_{i_1,\ldots,i_N}(\mathcal{A})^2$. By Markov's inequality, it follows
\begin{eqnarray}
&&P\left(\left\{||f(\mathcal{A})||^2> (1+\varepsilon)\right\}\right)=P\left(\left\{\exp\{h  S(\mathcal{A})\}>\exp\left\{\frac{h}{C(N,M)} (1+\varepsilon)\right\}\right\}\right) \nonumber\\
&&\leq  
\mathbb{E}\left(\exp\{h  S(\mathcal{A})\}\right)\exp\left\{-\frac{h}{C(N,M)} (1+\varepsilon)\right\}\nonumber\\
&&\leq  
\left(\mathbb{E}\left(\exp\{h \mathcal{Q}_{1,\ldots,1}(\mathcal{A})^2\}\right)\right)^{q(N)}\exp\left\{-\frac{h}{C(N,M)} (1+\varepsilon)\right\}\nonumber\\
&&\leq \left(f(h)\exp\left\{-\frac{h}{p(N)} (1+\varepsilon)\right\}\right)^{q(N)}\label{eq: fofh}
\end{eqnarray}

The inequality in \eqref{eq: fofh} follows from Lemma \ref{lem2}, where we defined
\begin{align}
    f(h) = \frac{1}{\pi^{N/2}}G_{1,N}^{N,1}\left(-\frac{p(N)}{ 2^N h}\left| \begin{array}{c}1\\1/2,\ldots,1/2\end{array}\right.\right)
\end{align}
 One can obtain the optimal exponential bound for the upper tail by optimizing in $h$. However, the first-order condition is intractable due to the presence of the Meijer G-function. But a ``good enough'' solution of $h$ can be obtained using power-log expansion for the Meijer G-function as in \cite{stojanac_products_2018}. 

The first order condition of \eqref{eq: fofh} with respect to $h$ after simplification is
\begin{align}
    \frac{1}{h} G^{1,N}_{N,1}\left(\frac{2^N h}{p(N)}\left| \begin{array}{c}1/2,\ldots,1/2\\1\end{array}\right.\right) + \frac{1+\epsilon}{p(N)}G^{1,N}_{N,1}\left(\frac{2^N h}{p(N)} \left| \begin{array}{c}1/2,\ldots,1/2\\0\end{array}\right.\right) = 0 \label{eq: foch}
\end{align}

Applying the lowest order power-log series expansion for the above Meijer G-function
\begin{align}
    G_{1,N}^{N,1}\left(\frac{2^N h}{p(N)}\left| \begin{array}{c}1/2,\ldots,1/2\\x\end{array}\right.\right) \approx \left(\frac{2^N h}{p(N)}\right)^{\frac{1}{2}} \bar{H}^x_{0, N-1}\left[\log ( \frac{2^N h}{p(N)})\right]^{N-1}
\end{align}
where 
\begin{align}
    \bar{H}^0_{0, N-1} = -\frac{1}{(N-1)!}\Gamma (\frac{1}{2}), \qquad \bar{H}^1_{0, N-1} = \frac{1}{2} \bar{H}^0_{0, N-1}. \notag
\end{align}
Equation \eqref{eq: foch} approximates as follows
\begin{align}
    &\frac{1}{h} \left(\frac{2^N h}{p(N)}\right)^{\frac{1}{2}}\frac{1}{2} \bar{H}^0_{0, N-1}\left[\log (\frac{2^N h}{p(N)})\right]^{N-1} + \notag \\ &\quad \frac{1+\epsilon}{p(N)}\left(\frac{2^N h}{p(N)}\right)^{\frac{1}{2}}  \bar{H}^0_{0, N-1}\left[\log ( \frac{2^N h}{p(N)})\right]^{N-1} = 0 \\
    &\left(\frac{2^N h}{p(N)}\right)^{\frac{1}{2}} \bar{H}^0_{0, N-1}\left[\log (\frac{2^N h}{p(N)})\right]^{N-1} \left(\frac{1}{2h} + \frac{1+\epsilon}{p(N)}\right) = 0
\end{align}
Since $h>0$, the only solution is $h = p(N)/2^N$, and it follows that
\begin{align}
    &P\left(\left\{||f(\mathcal{A})||^2> (1+\varepsilon)\right\}\right) \notag\\
    &\leq \left(\frac{1}{\pi^{N/2}}G_{1,N}^{N,1}\left(1\left| \begin{array}{c}1\\1/2,\ldots,1/2\end{array}\right.\right) \exp \left\{-\frac{1}{2^N }(1+\epsilon)\right\}\right)^{q(N)}\\
    & = \exp \left\{q(N)\left(-\frac{N}{2}\ln{\pi} + \ln G_{1,N}^{N,1}\left(1\left| \begin{array}{c}1\\1/2,\ldots,1/2\end{array}\right.\right) - \frac{1+\epsilon}{2^{N}}\right)\right\}\notag
\end{align}
Given that the Meijer G-function is fully specified, we can evaluate its value and the above bound can be approximated as
\begin{align*}
    P\left(\left\{||f(\mathcal{A})||^2> (1+\varepsilon)\right\}\right) \leq \exp \left\{q_1 q_2 q_3 \left(-\frac{7}{100} -\frac{1+\epsilon}{8}\right)\right\}    
\end{align*}
For the lower tail exponential bound, consider
\begin{align*}
    & P\left(\left\{||f(\mathcal{A})||^2 < (1-\varepsilon)\right\}\right)=P\left(\left\{\exp\{-h  S(\mathcal{A})\right\}>\exp\left\{-\frac{h}{C(N,M)} (1-\varepsilon)\right\}\right) \\
    & \leq \mathbb{E}\left(\exp\{-h  S(\mathcal{A})\}\right)\exp\left\{\frac{h}{C(N,M)} (1-\varepsilon)\right\}\\
    &\leq  
\left(\mathbb{E}\left(\exp\{-h \mathcal{Q}_{1,\ldots,1}(\mathcal{A})^2\}\right)\right)^{q(N)}\exp\left\{\frac{h}{C(N,M)} (1-\varepsilon)\right\} 
\end{align*}
By expanding $\exp\{-h \mathcal{Q}_{1,\ldots,1}(\mathcal{A})^2\}$ we have
\begin{align}
    &P\left(\left\{||f(\mathcal{A})||^2 < (1-\varepsilon)\right\}\right)\notag \\
    &\leq \left(1 - h\mathbb{E}\left(\mathcal{Q}_{1,\ldots,1}(\mathcal{A})^2\right) + \frac{h^2}{2}\mathbb{E}\left(\mathcal{Q}_{1,\ldots,1}(\mathcal{A})^4\right) \right)^{q(N)}\exp\left\{\frac{h}{C(N,M)} (1-\varepsilon)\right\}\notag\\
    &\leq\left(1 - \frac{h}{p(N)} + \frac{3^Nh^2}{2(p(N))^2} \right)^{q(N)}\exp\left\{\frac{h}{C(N,M)} (1-\varepsilon)\right\} \label{eq: lower bound}
\end{align}
\bigskip

\begin{figure}[t]
\vspace*{2pt}
    \centering
    \begin{tabular}{c}
    \includegraphics[width=.5\linewidth]{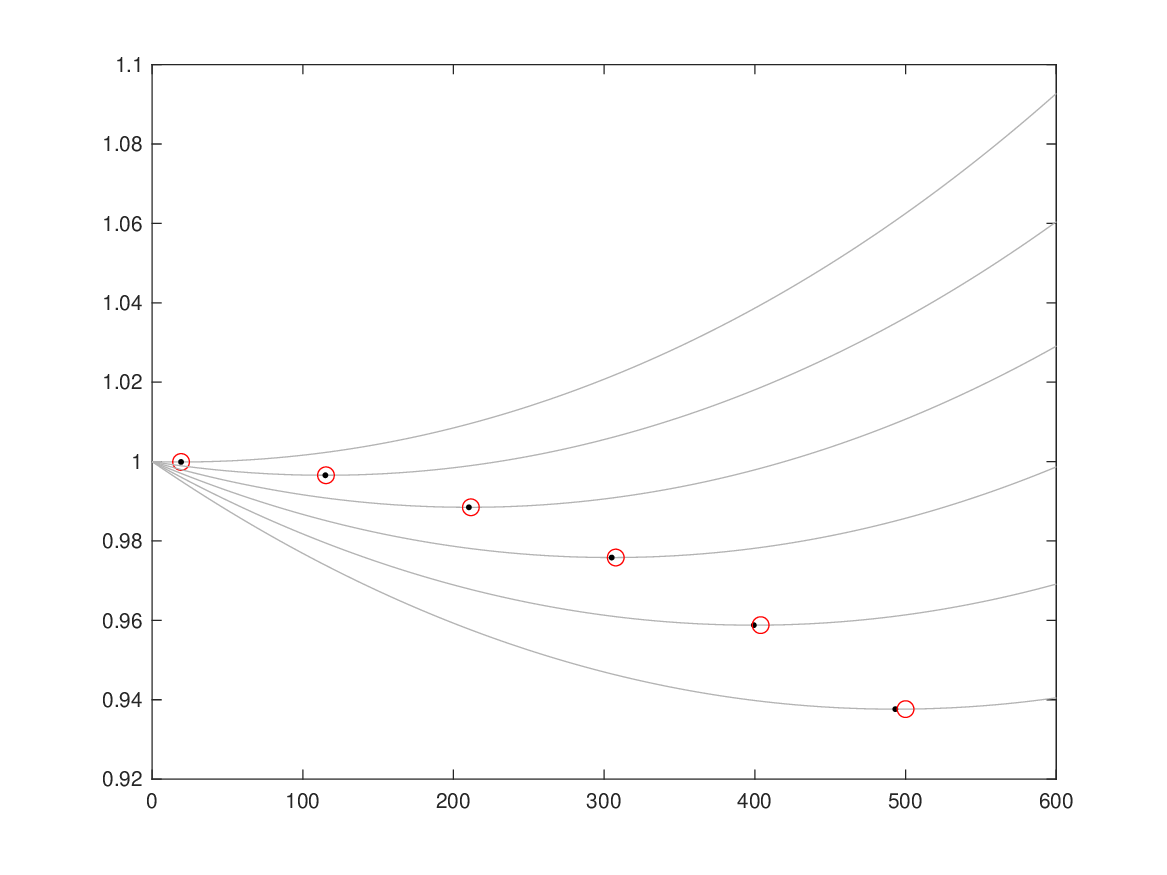} 
    \end{tabular}
    \caption{\small The plot of the lower bound as a function of $h$ for different values of $\varepsilon = 0.01, 0.06, 0.11, 0.16, 0.21, 0.26$. The optimal values of $h$ that minimize the bound are shown in the black dots, the approximated values of $h$ are shown as red circles.}
    \label{fig: bnd}
\end{figure}
To optimize the bound, solving the first order condition of \eqref{eq: lower bound} with respect to $h$, this gives $h = \sqrt{\frac{2(p(N))^2\varepsilon}{3^N(1-\varepsilon)}+\left(\frac{(p(N))(3^N-1+\varepsilon)}{3^N(1-\varepsilon)}\right)^2} - \frac{(p(N))(3^N-1+\varepsilon)}{3^N(1-\varepsilon)}$. Numerical studies (\autoref{fig: bnd}) suggest  $h^*=\frac{p(N)}{3^N-1}\varepsilon$ is a good approximation. Nevertheless, the absolute approximation error $\lvert h-h^*\rvert
\approx
\frac{p(N)(3^N-2)}{(3^N-1)^3}\varepsilon^2
$ confirms $h^*$ is a good approximation when $\varepsilon = O(p(N)^{-(1+\xi)})$, with $\xi>0$.  Substituting this value of $h$, we get \eqref{eq: lb_foc}, series expansion gives \eqref{eq: lb_expan}.
\begin{align}
    &P\left(\left\{||f(\mathcal{A})||^2 < (1-\varepsilon)\right\}\right) \notag\\
    &< \exp \left\{(q(N))\ln \left(1 - \frac{\epsilon}{3^N-1} + \frac{3^N\epsilon^2}{2(3^N-1)^2}\right)+\frac{q(N)}{3^N-1}\epsilon(1-\epsilon)\right\} \label{eq: lb_foc}\\
    &\approx \exp\left\{-q(N)\left(\frac{\epsilon^2}{2(3^N-1)}-\frac{(3^{N+1}-2)\epsilon^3}{6(3^N-1)^3}\right)\right\}\label{eq: lb_expan}
\end{align}

To get JL-embedding, we need $2\times \exp\left\{-q(N)\left(\frac{\epsilon^2}{2(3^N-1)}-\frac{(3^{N+1}-2)\epsilon^3}{6(3^N-1)^3}\right)\right\} \leq \frac{2}{n^{2+\beta}}$, thus $q(N) \geq \frac{4+2\beta}{\frac{\epsilon^2}{3^N-1}-\frac{(3^{N+1}-2)\epsilon^3}{3(3^N-1)^3}} \log n$.

\subsection{Proof of \autoref{thm: alt_bnd}}
Note that the $(i_1, i_2,\ldots, i_N)$-th entry from our mode-wise random projection can be written equivalently as the inner product of the tensor $\mathcal{X}$ and a rank $1$ tensor constructed by the outer product of the corresponding columns of matrices $H_{n,:,i_n}$:
\begin{align*}
f(\mathcal{X})_{i_1,\ldots,i_N} = \frac{1}{\sqrt{q(N)}}\left<H_{1,:,i_1} \circ H_{2,:,i_2} \circ \cdots \circ H_{N,:,i_N}, \mathcal{X}\right>=\frac{1}{\sqrt{q(N)}}u_{i_1,\ldots,i_N}\\
\end{align*}
To find the bound on the embedding dimensions, we follow the similar arguments from \cite{rakhshan_tensorized_2020} to first bound the variance of the Frobenius norm of $f(\mathcal{X})$ and then applying Hypercontractivity Concentration Inequality \citep{schudy2012concentration} to bound the embedding dimension.

\begin{align*}
\mathbb{V}\left(||f(\mathcal{X}||_F^2\right) = \mathbb{E}||f(\mathcal{X})||_F^4-\left(\mathbb{E}||f(\mathcal{X})||_F^2\right)^2
\end{align*}

Due to expected isometry, it can be shown that $\mathbb{E}||f(\mathcal{X})||_F^2 = ||\mathcal{X}||_F^2=1$, and
\begin{align*}
    \mathbb{E}||\mathcal{U}||_F^4 = \sum_{i_1=1}^{q_1}\cdots\sum_{i_N=1}^{q_N}\mathbb{E} u_{i_1,\ldots,i_N}^4 + \sum_{i_1\ldots i_N \neq i_1'\ldots i_N'}\mathbb{E}(u_{i_1,\ldots,i_N}^2u_{i_1',\ldots,i_N'}^2)
\end{align*}

Since $u_{i_1,\ldots,i_N}^2$ and $u_{i_1',\ldots,i_N'}^2$ are independent, the second term on the right hand side amounts to $q(N)(q(N)-1)||\mathcal{X}||_F^4=q(N)(q(N)-1)$. Using the same argument in \cite{rakhshan_tensorized_2020} we can bound $\mathbb{E}u_{i_1,\ldots,i_N}^4$,

\begin{align*}
    \mathbb{E}u_{i_1,\ldots,i_N}^4 =\mathbb{E}\left<H_{1,:,i_1} \circ H_{2,:,i_2} \circ \cdots \circ H_{N,:,i_N}, \mathcal{X}\right>^4\leq 3^N ||f(\mathcal{X})||_F^4= 3^N
\end{align*}

Then,
\begin{align*}
    \mathbb{V}\left(||f(\mathcal{X}||_F^2\right) &= \mathbb{V}\left(||\frac{1}{\sqrt{q(N)}}\mathcal{U}||_F^2\right) = \frac{1}{q(N)^2}\left(\mathbb{E}||\mathcal{U}||_F^4-\left(\mathbb{E}||\mathcal{U}||_F^2\right)^2\right)\\
    &\leq  \frac{1}{q(N)^2} \left[q(N)3^N + q(N)(q(N)-1)\right] -1 = \frac{3^N-1}{q(N)}
\end{align*}

By Hypercontractivity Concentration Inequality, for some positive constants $C$ and $K$ we have,
\begin{eqnarray*}
    \mathbb{P}\left(\left|||f(\mathcal{X})||_F^2-||\mathcal{X}||_F^2\right| \geq \varepsilon ||\mathcal{X}||_F^2 \right) &\leq& C\exp \left[-\left(\frac{\varepsilon^2}{K \mathbb{V}(||f(\mathcal{X})||_F^2)}\right)^{\frac{1}{2N}}\right]\\
    &\leq& C\exp \left[-\frac{(\sqrt{q(N)}\varepsilon)^{\frac{1}{N}}}{(K 3^N)^{\frac{1}{2N}}}\right]
\end{eqnarray*}

\subsection{Proof of \autoref{thm: pos_cons}}
Let $\mathcal{P}_n$ denote a sequence of sets of probability densities, $N(\varepsilon_n, \mathcal{P}_n)$ the minimum number of Hellinger balls of radius $\varepsilon_n$ needed to cover $\mathcal{P}_n$. Define the following conditions:
\begin{enumerate}
    \item[a)] $\log N(\varepsilon_n, \mathcal{P}_n) \leq n\varepsilon_n^2$ for all large $n$
    \item[b)] $\pi(\mathcal{P}_n^c) \leq e^{-2n\varepsilon_n^2}$ for all large $n$
    \item[c)] $\pi \left[f: d_t(f, f_0) < \frac{\varepsilon_n^2}{4}\right] \geq e^{-n\varepsilon_n^2/4}$ for all large $n$.
\end{enumerate}
\begin{proposition} \label{prop: pos_cons}
    If $n\varepsilon_n^2 \to \infty$, then under conditions $a, b, c$ (for some $t > 0$), we have
    \begin{align}
        E_{f_0}\pi \left[d(f, f_0) > 4\varepsilon_n \mid \left(y_i, \mathcal{X}_i\right)_{i=1}^{n}\right] \leq 4e^{-n\varepsilon_n^2 \min (1/2, t/4)} \notag
    \end{align}
\end{proposition}

Proposition \ref{prop: pos_cons} has been proved in \cite{jiang_bayesian_2007}. We prove Theorem \ref{thm: pos_cons} by showing conditions $a, b$ and $c$ hold in our case for some positive $t$.
\vspace{1mm}

\begin{proposition} \label{prop: a2}
    Assume $\mathcal{B} \sim \mathcal{TN}\left(\boldsymbol{0}, \boldsymbol{\Sigma}_1,\ldots, \boldsymbol{\Sigma}_N\right)$, where $\mathcal{TN}$ denotes the Tensor Normal distribution and $\boldsymbol{\Sigma}_n$ is the covariance matrix for mode $n$. Then
    \begin{align*}                      P\left(\left\lvert\left<f(\mathcal{X}), \mathcal{B}\right> - \left<\mathcal{X}, \mathcal{B}_0\right>\right\rvert< \Delta\right) > P(X-Y\geq 2),
    \end{align*}
    where $X\sim Poi\left(\frac{\Delta_1}{2}\right), Y\sim Poi(\frac{\lambda}{2})$ with $\Delta_1 = \frac{\Delta^2}{\text{Var}\left(\left<f(\mathcal{X}), \mathcal{B}\right>\right)}, \lambda = \frac{\left<\mathcal{X}, \mathcal{B}_0\right>^2}{\text{Var}\left(\left<f(\mathcal{X}), \mathcal{B}\right>\right)}, \text{Var}\left(\left<f(\mathcal{X}), \mathcal{B}\right>\right) = \text{vec}(f(\mathcal{X}))'(\boldsymbol{\Sigma}_1 \otimes \cdots \otimes \boldsymbol{\Sigma}_N) \text{vec}(f(\mathcal{X}))$.
\end{proposition}

\begin{proof}
    Note that $\left<f(\mathcal{X}), \mathcal{B}\right> \sim \mathcal{N}(0, \text{vec}(f(\mathcal{X}))'(\boldsymbol{\Sigma}_1 \otimes \cdots \otimes \boldsymbol{\Sigma}_N) \text{vec}(f(\mathcal{X}))$. This implies 
    \begin{align*}
        \frac{\left\lvert\left<f(\mathcal{X}), \mathcal{B}\right> - \left<\mathcal{X}, \mathcal{B}_0\right>\right\rvert^2}{\text{Var}\left(\left<f(\mathcal{X}), \mathcal{B}\right>\right)} \sim \chi^2_1(\lambda),
    \end{align*}
    where $\chi^2_1(\lambda)$ is the noncentral chi-squared distribution with degrees of freedom $1$ and noncentral parameter $|\lambda|\sqrt{\text{Var}(\left<f(\mathcal{X}), \mathcal{B}\right>)}=|\mathbb{E}(\left<f(\mathcal{X}), \mathcal{B}\right>-\left<\mathcal{X}, \mathcal{B}_0\right>)|=\left<\mathcal{X}, \mathcal{B}_0\right>$. It is known that a noncentral chi-squared distribution can also be written a gamma mixture with Poisson weights
    \begin{align}    P\left(\left\lvert\left<f(\mathcal{X}), \mathcal{B}\right> - \left<\mathcal{X}, \mathcal{B}_0\right>\right\rvert< \Delta\right) &= P\left(\frac{\left\lvert\left<f(\mathcal{X}), \mathcal{B}\right> - \left<\mathcal{X}, \mathcal{B}_0\right>\right\rvert^2}{\text{Var}\left(\left<f(\mathcal{X}), \mathcal{B}\right>\right)} < \Delta_1\right)\notag \\
        & = \sum_{i=0}^{\infty}\frac{e^{-\frac{\lambda}{2}}(\frac{\lambda}{2})^{i}}{i!}P(Z_{1+2i}<\Delta_1),\label{eq: chi-squ}
    \end{align}
    where $\Delta_1 = \Delta^2/\text{Var}(\left<f(\mathcal{X}), \mathcal{B}\right>)$ and $Z_{1+2i} \sim \chi^2_{1+2i}$. Note that $P(Z_{1+2i}<\Delta_1)>P(Z_{2+2i}<\Delta_1) = P(G<\Delta_1)$ where $G\sim \mathcal{G}a(1+i, \frac{1}{2})$. From Proposition A.2 in \cite{guhaniyogi_bayesian_2015}, we obtain $P\left(\left\lvert\left<f(\mathcal{X}), \mathcal{B}\right> - \left<\mathcal{X}, \mathcal{B}_0\right>\right\rvert< \Delta\right) > P(X-Y\geq 2)$.
\end{proof}

\textit{Proof of Theorem \ref{thm: pos_cons}}. We will check the three conditions with $t=1$. Let $b_n = \sqrt{8\Tilde{\lambda}_nn\varepsilon_n^2}$.
\vspace{1mm}

\textit{Condition a.} Let $\mathcal{P}_n$ be the set of all densities that can be represented by the $\mathcal{B}$ with entries $\lvert b_{jkl}\rvert < b_n$, $j = 1, \ldots, q_{1,n}, k = 1, \ldots, q_{2,n}, l = 1, \ldots, q_{3,n}$. Let's consider the $l^{\infty}$ balls of the form $(a_{jkl}-\delta, a_{jkl}+\delta)$ for each entry of the tensor coefficients with the center of each ball inside $\mathcal{P}_n$, the external covering number of $\mathcal{P}_n$ is bounded above by $\left(\frac{b_n}{\delta}+1\right)^{q_n}$ where $q_n=q_{1,n}q_{2,n}q_{3,n}$.

Let $f_u$ be any density in $\mathcal{P}_n$, $\exists \mathcal{B}$ s.t. $u = \left<\texttt{GTRP}(\mathcal{X}_i), \mathcal{B}\right>$, $\lvert b_{jkl}\rvert \leq b_n$ and 
\begin{align*}
    f_u(y) = \exp \left\{ya(u)+b(u)+c(y) \right\}
\end{align*}
    
Let $b_{jkl} \in (c_{jkl}-\delta, c_{jkl}+\delta)$, s.t. $\lvert b_{jkl}-c_{jkl}\rvert\leq \delta$ and $\lvert c_{jkl}\rvert \leq b_n$. Let $v = \left<\texttt{GTRP}(\mathcal{X}_i), \mathcal{C}\right>$, and 
\begin{align*}
    f_v(y) = \exp \left\{ya(v)+b(v)+c(y)\right\}
\end{align*}

We then find the number of Hellinger balls that required to cover $\mathcal{P}_n$ by using the fact $d(f, f_0) \leq \left(d_{KL} (f, f_0)\right)^{1/2}$, where 
\begin{align}
    d_{KL}(f_u,f_v) 
    & = \iint f_v\log \left(\frac{f_v}{f_u}\right)\nu_y(dy)\nu_{\mathcal{X}}(d\mathcal{X}) \notag \\
    & = \iint \left[y\left(a(v) - a(u)\right) + \left(b(v) - b(u)\right)\right]f_v\nu_y(dy)\nu_{\mathcal{X}}(d\mathcal{X}) \notag \\
    & = \int \left[ \left(a(v) - a(u)\right)E\left[y\mid \mathcal{X}\right] + \left(b(v)-b(u)\right)\right]\nu_{\mathcal{X}}(d\mathcal{X}) \notag \\
    & = \int (v-u)\left[ a'(u_v)\left(-\frac{b'(v)}{a'(v)}\right)+b'(u_v)\right]\nu_{\mathcal{X}}(d\mathcal{X}). \label{eq: d_0}
\end{align}
The last two steps are achieved by first integrating with respect to $y$ and then applying the mean value theorem, where $u_v$ is the intermediate point between $u$ and $v$. By Cauchy-Schwartz inequality, the condition $|b_{jkl}-c_{jkl}|<\delta$ and the assumption $|x_{jkl}|<1$ we have,
\begin{align*}
    \lvert u - v \rvert 
    & = \lvert \left<\texttt{GTRP}(\mathcal{X}_i), \mathcal{B}\right> - \left<\texttt{GTRP}(\mathcal{X}_i), \mathcal{C}\right>\rvert = \lvert \left<\texttt{GTRP}(\mathcal{X}_i), \mathcal{B - C}\right>\rvert \\
    & \leq \lVert\texttt{GTRP}(\mathcal{X}_i) \rVert \lVert\mathcal{B - C} \rVert \leq \lVert \mathcal{X}_i\rVert \sqrt{q_n}\delta \leq \sqrt{p_nq_n}\delta = \theta_n\delta,
\end{align*}
where we defined $\theta_n = \sqrt{q_np_n}$. Since $\lvert u\rvert = \lvert \left<\texttt{GTRP}(\mathcal{X}_i), \mathcal{B}\right> \rvert \leq \lVert \texttt{GTRP}(\mathcal{X}_i) \rVert \lVert \mathcal{B}\rVert \leq \sqrt{q_np_n}b_n \leq b_n\theta_n$, similarly, $\lvert v\rvert \leq b_n\theta_n$, thus $\lvert u_v\rvert \leq b_n\theta_n$. Combining the results and \eqref{eq: d_0}, we have,
\begin{align*}
    d(f_u, f_v) &\leq \sqrt{d_{KL}(f_u, f_v)} \\
    &\leq \sqrt{\int |v-u|\left| a'(u_v)\left(-\frac{b'(v)}{a'(v)}\right)+b'(u_v)\right|\nu_{\mathcal{X}}(d\mathcal{X})}\\
    &\leq \sqrt{2\theta_n \delta \sup_{\lvert h\rvert \leq b_n\theta_n} \lvert a'(h)\rvert \sup_{\lvert h\rvert \leq b_n\theta_n} \left\lvert \frac{b'(h)}{a'(h)}\right\rvert \int \nu_{\mathcal{X}}(d\mathcal{X})}.
\end{align*}

Let $\delta = \varepsilon_n^2/(2\theta_n\sup_{\lvert h\rvert \leq b_n\theta_n} \lvert a'(h)\rvert \sup_{\lvert h\rvert \leq b_n\theta_n} \left\lvert \frac{b'(h)}{a'(h)}\right\rvert)$, one gets $d(f_u, f_v) \leq \varepsilon_n$. The entropy of $\mathcal{P}_n$ is therefore bounded from above by 
\begin{align*}
    \left( 1+\frac{2b_n\theta_n}{\varepsilon_n^2}\sup_{\lvert h\rvert \leq b_n\theta_n} \lvert a'(h)\rvert \sup_{\lvert h\rvert \leq b_n\theta_n} \left\lvert \frac{b'(h)}{a'(h)}\right\rvert\right)^{q_n} & = \left(1-\frac{1}{\varepsilon_n^2}+\frac{D(b_n\theta_n)}{\varepsilon_n^2}\right)^{q_n} \\
    & \leq \left(\frac{D(b_n\theta_n)}{\varepsilon_n^2}\right)^{q_n}
\end{align*}
where we defined $G(R) = 1 + R\sup_{\lvert h \rvert \leq R}\lvert a'(h)\rvert \sup_{\lvert h \rvert \leq R} \lvert \frac{b'(h)}{a'(h)}\rvert$ and the inequality follows from the assumption $\varepsilon_n^2<1$. Thus the Hellinger covering number satisfied $N(\varepsilon_n, \mathcal{P}_n) \leq \left(\frac{D(b_n\theta_n)}{\varepsilon_n^2}\right)^{q_n}$, implying $\log N(\varepsilon_n, \mathcal{P}_n) \leq q_n (\log D(b_n\theta_n)+\log(1/ \varepsilon_n^2))$. Using the assumptions in \textbf{A.1}, $\frac{q_n\log(1/\varepsilon_n^2)}{n\varepsilon_n^2} \to 0 \text{ and } \frac{q_n\log G(\theta_n\sqrt{8\Tilde{\lambda}_n n\varepsilon_n^2})}{n\varepsilon_n^2} \to 0$, \textit{condition a} follows.
\vspace{5mm}

\textit{Condition b.} 

By union bound inequality, it follows:

\begin{align*}
    \pi (\mathcal{P}_n^c) = \pi \left(\cup_{j=1}^{q_{1,n}}\cup_{k=1}^{q_{2,n}}\cup_{l=1}^{q_{3,n}} \lvert b_{jkl}\rvert > b_n\right) \leq \sum_{j=1}^{q_{1,n}}\sum_{k=1}^{q_{2,n}}\sum_{l=1}^{q_{3,n}} \pi (\lvert b_{jkl}\rvert > b_n)
\end{align*}
Since $b_{jkl} \sim \mathcal{N}(0, \sigma_{jkl}^2), \frac{1}{\sigma_{jkl}^2} > \frac{1}{\Tilde{\lambda}_n}$. By Mill's ratio $\pi(\lvert \frac{b_{jkl}}{\sqrt{\Tilde{\lambda}_n}}\rvert > \frac{b_n}{\sqrt{\Tilde{\lambda}_n}})<2\frac{\exp \{-b_n^2/2\Tilde{\lambda}_n\}}{\sqrt{2\pi b_n^2/\Tilde{\lambda}_n}}$, the above quantity is bounded above by $2q_n\frac{\exp \{-b_n^2/2\Tilde{\lambda}_n\}}{\sqrt{2\pi b_n^2/\Tilde{\lambda}_n}} = 2q_n\frac{\exp \{-4n\varepsilon_n^2\}}{4\sqrt{\pi n\varepsilon_n^2}} \leq \exp \{-2n\varepsilon_n^2\}$ for sufficiently large $n$, since $\log(q_n)/(n\varepsilon_n^2)\to 0$ from the assumptions \textit{i)} of \autoref{thm: pos_cons} and $n\varepsilon_n^2\to\infty$. \textit{Condition b} follows.
\vspace{5mm}

\textit{Condition c.} We verify condition $c$ for $t=1$. From Proposition \ref{prop: a2}, we have,
\begin{align*}
    P\left(\left\lvert\left<\texttt{GTRP}(\mathcal{X}), \mathcal{B}\right> - \left<\mathcal{X}, \mathcal{B}_0\right>\right\rvert< \Delta\right) > P(X-Y\geq 2).
\end{align*} 

Since $X\sim Poi\left(\frac{\Delta_1}{2}\right), Y\sim Poi(\frac{\lambda}{2})$, $X-Y$ follows Skellam distribution with PMF
\begin{align}
    P(X-Y=k) = \exp\{-(\lambda + \Delta_1)\}\left(\frac{\Delta_1}{\lambda}\right)I_k\left(2\sqrt{\lambda \Delta_1}\right)\notag 
\end{align}

Plug in $\lambda$, $\Delta_1$ and $k=2$ we have,
\begin{align}
    P(X-Y=2) = \exp \left\{-\frac{\Delta^2 + \left<\mathcal{X}, \mathcal{B}_0\right>^2}{\text{Var}(\left<\texttt{GTRP}(\mathcal{X}), \mathcal{B}\right>)}\right\}\left(\frac{\Delta^2}{\left<\mathcal{X}, \mathcal{B}_0\right>^2}\right)I_2\left(2\frac{\Delta \lvert\left<\mathcal{X}, \mathcal{B}_0\right>\rvert}{\text{Var}(\left<\texttt{GTRP}(\mathcal{X}), \mathcal{B}\right>)}\right)
\end{align}
using the fact that for $z>0, I_k(z)>2^{k}z^{k}\Gamma(k+1)$ \citep{joshi1991some}, we have
\begin{align*}
    &P(X-Y\geq 2) > P(X-Y = 2) \\
    & > \exp \left\{-\frac{\Delta^2 + \left<\mathcal{X}, \mathcal{B}_0\right>^2}{\text{Var}(\left<\texttt{GTRP}(\mathcal{X}), \mathcal{B}\right>)}\right\}\left(\frac{\Delta}{\left<\mathcal{X}, \mathcal{B}_0\right>}\right)^{2}2^2\left(2\frac{\Delta \left<\mathcal{X}, \mathcal{B}_0\right>}{\text{Var}(\left<\texttt{GTRP}(\mathcal{X}), \mathcal{B}\right>)}\right)^2\Gamma (3)\\
    & > \exp \left\{-\frac{\Delta^2 + \left<\mathcal{X}, \mathcal{B}_0\right>^2}{\text{Var}(\left<\texttt{GTRP}(\mathcal{X}), \mathcal{B}\right>)}\right\} \frac{2^5 \Delta^4}{\text{Var}(\left<\texttt{GTRP}(\mathcal{X}), \mathcal{B}\right>)^2}\\
    & > \exp \left\{-\frac{\Delta^2 + \left<\mathcal{X}, \mathcal{B}_0\right>^2}{\underline{\lambda}\lVert\texttt{GTRP}(\mathcal{X})\rVert_F^2}\right\}\frac{2^5 \Delta^4}{\Tilde{\lambda}^2\lVert\texttt{GTRP}(\mathcal{X})\rVert_F^4}> \exp \left\{-\frac{n\varepsilon_n^2}{4}\right\} 
\end{align*}
where the last inequality follows from 
\begin{align*}
    &\exp \left\{-\frac{\Delta^2 + \left<\mathcal{X}, \mathcal{B}_0\right>^2}{\underline{\lambda}\lVert\texttt{GTRP}(\mathcal{X})\rVert_F^2}\right\}>\exp \left\{-\frac{n\varepsilon_n^2}{8}\frac{8}{n\varepsilon_n^2}\frac{\Delta^2 + K^2}{\underline{\lambda}\lVert\texttt{GTRP}(\mathcal{X})\rVert_F^2}\right\}\\
    &>\exp \left\{-\frac{n\varepsilon_n^2}{8}\frac{8}{n\varepsilon_n^2}\frac{\log(q_n)(1 + K^2)}{B_1\lVert\texttt{GTRP}(\mathcal{X})\rVert_F^2}\right\}> \exp \left\{-\frac{n\varepsilon_n^2}{8}\right\} 
\end{align*}    
choosing $\Delta=\varepsilon_n^2/(4\eta)$ and assuming $\underline{\lambda}>B_1/\log(q_n)$ as in \textbf{A.2} and $\lVert\texttt{GTRP}(\mathcal{X})\rVert_F^2>8(1 + K^2)\log(q_n)/(n\varepsilon_n^2B_1)$ as in \textbf{A.3}, and from
\begin{align*}
&\frac{2^5 \Delta^4}{\Tilde{\lambda}^2\lVert\texttt{GTRP}(\mathcal{X})\rVert_F^4}=\exp\left\{-\frac{n\varepsilon_n^2}{8}\left(\frac{8\log(\Tilde{\lambda}^2)-8\log(2^5 \Delta^4)}{n\varepsilon_n^2}+\frac{8\log(\lVert\texttt{GTRP}(\mathcal{X})\rVert_F^4)}{n\varepsilon_n^2}\right)\right\}\\
&\exp\left\{-\frac{n\varepsilon_n^2}{8}\left(8\frac{2\log(B)+2v\log(q_n)-\log(2^5 \Delta^4)}{n\varepsilon_n^2}+\frac{8\log(\lVert\texttt{GTRP}(\mathcal{X})\rVert_F^4)}{n\varepsilon_n^2}\right)\right\}> \exp \left\{-\frac{n\varepsilon_n^2}{8}\right\} 
\end{align*}    
due to assumptions $\log(q_n)/(n\varepsilon^2_n)\to 0$ in \textbf{A.1}, $\tilde{\lambda}_n\leq Bq_n^v$ in \textit{ii)} and $\log\left(\lVert \texttt{GTRP}(\mathcal{X}) \rVert\right)/(n\varepsilon_n^2) \to 0$ in \textbf{A.3} as $n\to\infty$. We conclude that for all large $n$
\begin{align*}
    P\left(\lvert\left<\texttt{GTRP}(\mathcal{X}_i), \mathcal{B}\right> - \left<\mathcal{X}_i, \mathcal{B}_0\right>\rvert < \frac{\varepsilon_n^2}{4\eta}\right) > \exp \left\{-\frac{n\varepsilon_n^2}{4}\right\}.
\end{align*}

For $\mathcal{X} = \mathcal{X}_1, \ldots, \mathcal{X}_n$, let $\mathcal{S} = \left\{\mathcal{B}: \lvert\left<\texttt{GTRP}(\mathcal{X}_i), \mathcal{B}\right> - \left<\mathcal{X}_i, \mathcal{B}_0\right>\rvert < \frac{\varepsilon_n^2}{4\eta}\right\}$. For $t=1$,
\begin{align*}%
    d_{t=1} &= \iint f_0\left(\frac{f_0}{f}-1\right)\nu_y(dy)\nu_{\mathcal{X}}(d\mathcal{X})\\
    &= \int E_{y\mid\mathcal{X}}\left[\frac{f_0}{f}(Y)-1\right]\nu_{\mathcal{X}}(d\mathcal{X}) = E_{\mathcal{X}}\left[g(u^*)\left(\left<\texttt{GTRP}(\mathcal{X}_i), \mathcal{B}\right> - \left<\mathcal{X}_i, \mathcal{B}_0\right>\right)\right]
\end{align*}
where the last steps are achieved by first integrating out $y$ and applying mean value theorem. $g$ is a continuous derivative function and $u^*$ is an intermediate point between $\left<\texttt{GTRP}(\mathcal{X}_i), \mathcal{B}\right>$ and $\left<\mathcal{X}_i, \mathcal{B}_0\right>$. Since $\lvert \left<\mathcal{X}_i, \mathcal{B}_0\right> \rvert < \sum_n \lvert b_{jkl, 0} \rvert < K$, we can bound $u^*$ by the following,
\begin{align*}
    \lvert u^*\rvert < \lvert \left<\texttt{GTRP}(\mathcal{X}_i), \mathcal{B}\right> - \left<\mathcal{X}_i, \mathcal{B}_0\right>\rvert + \lvert \left<\mathcal{X}_i, \mathcal{B}_0\right>\rvert < \frac{\varepsilon_n^2}{4\eta} + K
\end{align*}

Choosing $\eta$ such that $\lvert g(u^*)\rvert < \eta$ in the interval $\left[-(K+1), (K+1)\right]$ for all large $n$, this implies $d_t(f, f_0) < \frac{\varepsilon_n^2}{4}$ is a subset of $\mathcal{S}$, hence confirming condition $c$.

\subsection{Proof of Theorem \ref{thm: pos_parafac}}
We show that the three conditions are also satisfied with PARAFAC priors. Following the prior imposed on the margins from the PARAFAC decomposition from \eqref{eq.gamma}, we have $\boldsymbol{\gamma}_{m }^{(d)} \sim\mathcal{N}_{p_{m }}(\boldsymbol{0},\tau\zeta^{(d)}W_{m }^{(d)})$.

\textit{Condition a} is easily verified with the same spirits as in the proof of Thm \ref{thm: pos_cons}.

\textit{Condition b.}

By PARAFAC decomposition, we have:
\begin{align*}
    \pi (\lvert b_{jkl}\rvert \leq b_n) & = \pi \left(\lvert\sum_{d=1}^D \gamma_{1,j}^{(d)}\gamma_{2,k}^{(d)}\gamma_{3,l}^{(d)} \rvert \leq b_n\right) \geq \pi \left(\sum_{d=1}^D\lvert \gamma_{1,j}^{(d)}\gamma_{2,k}^{(d)}\gamma_{3,l}^{(d)} \rvert \leq b_n\right) \\
    & \geq \pi \left(\lvert \gamma_{1,j}^{(d)}\gamma_{2,k}^{(d)}\gamma_{3,l}^{(d)} \rvert \leq \frac{b_n}{D}\right)  \geq \pi \left(\lvert \gamma_{m,j_m}^{(d)} \rvert \leq \left(\frac{b_n}{D}\right)^{1/M}\right)
\end{align*}

The inequalities follow the inclusion of events: let $A_d =\gamma_{1,j}^{(d)}\gamma_{2,k}^{(d)}\gamma_{3,l}^{(d)}$, then $\left\{|A_d|\le \frac{b_n}{D}\text{ for } d=1 \dots,D \right\}\subseteq\left\{\sum_{d=1}^D |A_d|\le b_n\right\} \subseteq \left\{\left|\sum_{d=1}^D A_d\right|\le b_n\right\}$, now consider $\lvert A_d\rvert=\prod_{m=1}^{M}\lvert\gamma_{m,j_m}^{(d)}\rvert$, $\left\{|\gamma_{m,j_m}^{(d)}|\le\left(\frac{b_n}{D}\right)^{1/M}\right\}\subseteq\left\{|A_d|\le \frac{b_n}{D} \right\}$.

Therefore, $\pi (\lvert b_{jkl}\rvert > b_n) \leq \pi \left(\lvert \gamma_{m,j_m}^{(d)} \rvert > \left(\frac{b_n}{D}\right)^{1/M}\right)$. By Mill's ratio $\pi(\lvert \frac{\gamma_{m,j_m}^{(d)}}{\sqrt{\Tilde{\lambda}_n}}\rvert > \frac{\left(\frac{b_n}{D}\right)^{1/M}}{\sqrt{\Tilde{\lambda}_n}})<2\frac{\exp \{-\left(\frac{b_n}{D}\right)^{2/M}/2\Tilde{\lambda}_n\}}{\sqrt{2\pi \left(\frac{b_n}{D}\right)^{2/M}/\Tilde{\lambda}_n}}$. Let $b_n = D(8\tilde{\lambda}_n n\varepsilon_n^2)^{M/2}$, the results follow from the same arguments used in proof of Thm \ref{thm: pos_cons} condition $b$.

\textit{Condition c.}

We are interested in a lower bound for
\begin{align}
    P\left(\lvert\left<\texttt{GTRP}(\mathcal{X}_i), \mathcal{B}\right> - \left<\mathcal{X}_i, \mathcal{B}_0\right>\rvert < \Delta_n\right).
\end{align}
Following Assumption A5, the projection error satisfies $\lvert\left<\texttt{GTRP}(\mathcal{X}_i), \mathcal{B}_0^c\right> - \left<\mathcal{X}_i, \mathcal{B}_0\right>\rvert = o(q_{0,n}^{-1})=\varepsilon_{n}$ a.s. w.r.t. $f_{0}$, where $q_{0,n}$ is the lower bound provided by the concentration inequalities in Prop. \ref{cor: jl} and Thm. \ref{thm: jl},  and $\lvert\left<\texttt{GTRP}(\mathcal{X}_i), \mathcal{B}\right> - \left<\mathcal{X}_i, \mathcal{B}_0\right>\rvert \leq \lvert\left<\texttt{GTRP}(\mathcal{X}_i), \mathcal{B}\right> - \left<\texttt{GTRP}(\mathcal{X}_i), \mathcal{B}_0^c\right>\rvert + \lvert\left<\texttt{GTRP}(\mathcal{X}_i), \mathcal{B}_0^c\right> - \left<\mathcal{X}_i, \mathcal{B}_0\right>\rvert=\lvert\left<\texttt{GTRP}(\mathcal{X}_i), \mathcal{B}\right>- \left<\mathcal{X}_i, \mathcal{B}_0\right>\rvert+\varepsilon_n$, then
\begin{align*}
    & P\left(\lvert\left<\texttt{GTRP}(\mathcal{X}_i), \mathcal{B}\right> - \left<\mathcal{X}_i, \mathcal{B}_0\right>\rvert < \Delta_n\right) \\
    \geq &P\left(\lvert\left<\texttt{GTRP}(\mathcal{X}_i), \mathcal{B}\right> - \left<\texttt{GTRP}(\mathcal{X}_i), \mathcal{B}_0^c\right>\rvert < \Delta_n-\varepsilon_n\right)\\
    = & P\left(\lvert\left<\texttt{GTRP}(\mathcal{X}_i), \mathcal{B} - \mathcal{B}_0^c\right>\rvert < \frac{\Delta_n}{2}\right)\\
    \geq &P\left(\lVert\texttt{GTRP}(\mathcal{X}_i)\rVert \lVert\mathcal{B} - \mathcal{B}_0^c\rVert< \frac{\Delta_n}{2}\right)\\
    =& P\left( \lVert\mathcal{B} - \mathcal{B}_0^c\rVert< \frac{\Delta_n}{2\lVert\texttt{GTRP}(\mathcal{X}_i)\rVert}\right)
\end{align*}

the second line follows from the triangular inequality, the third line from choosing $\varepsilon_n=\Delta_n/2$, and the fourth line from the Cauchy-Schwartz inequality $\lvert\left<\texttt{GTRP}(\mathcal{X}_i), \mathcal{B} - \mathcal{B}_0^c\right>\rvert \leq \lVert\texttt{GTRP}(\mathcal{X}_i)\rVert \lVert\mathcal{B} - \mathcal{B}_0^c\rVert$. Let $\omega_n = \frac{\Delta_n}{2\lVert\texttt{GTRP}(\mathcal{X}_i)\rVert}$. Assume $\mathcal{B}_0^c$ has rank at most $D$: $\mathcal{B}_0^c = \sum_{d=1}^D\boldsymbol{\beta}_1^{(d)}\circ\cdots\circ\boldsymbol{\beta}_M^{(d)}$ and $\mathcal{B} = \sum_{d=1}^D\boldsymbol{\gamma}_1^{(d)}\circ\cdots\circ\boldsymbol{\gamma}_M^{(d)}$, 
\begin{align}
    &P\left(\lVert\mathcal{B} - \mathcal{B}_0^c\rVert < \omega_n\right)\notag\\
    =& P \left(\lVert\sum_{d=1}^D\left(\boldsymbol{\gamma}_1^{(d)}\circ\cdots\circ\boldsymbol{\gamma}_M^{(d)}-\boldsymbol{\beta}_1^{(d)}\circ\cdots\circ\boldsymbol{\beta}_M^{(d)}\right)\rVert \leq \omega_n\right)\label{eq:nor_sum}\\
    \geq & P \left(\sum_{d=1}^D\lVert\boldsymbol{\gamma}_1^{(d)}\circ\cdots\circ\boldsymbol{\gamma}_M^{(d)} - \boldsymbol{\beta}_1^{(d)}\circ\cdots\circ\boldsymbol{\beta}_M^{(d)}\rVert \leq \omega_n\right) \label{eq:sum_nor}
    \end{align}
    where the inequality from \eqref{eq:nor_sum} to \eqref{eq:sum_nor} follows triangular inequality. By Lemma 7 of \cite{guhaniyogi2017bayesian}, the differences between two rank-one tensors can be decomposed into a finite sum of outer products involving the factor perturbation:
    \begin{align}
    \left\{\|\boldsymbol{\gamma}_m^{(d)}-\boldsymbol{\beta}_m^{(d)}\|\le \kappa_n,\; m=1,\ldots,M,\; d=1,\ldots,D\right\}\subseteq\left\{\|\mathcal{B}-\mathcal{B}_0^c\|<\omega_n\right\} \label{eq:fac_purt}
    \end{align}
    where $\kappa_n$ is chosen to be small enough such that the inclusion holds.
Therefore, 
\begin{align}
P\left(\|\mathcal{B}-\mathcal{B}_0^c\|<\omega_n\right) &\ge P\left(\bigcap_{d=1}^D\bigcap_{m=1}^M\left\{\|\boldsymbol{\gamma}_m^{(d)}-\boldsymbol{\beta}_m^{(d)}\|\le \kappa_n\right\}\right)\notag\\
&=\prod_{d=1}^D\prod_{m=1}^MP\left(\|\boldsymbol{\gamma}_m^{(d)}-\boldsymbol{\beta}_m^{(d)}\|\le \kappa_n\right)\notag
\end{align}
thus, it is sufficient to bound $P \left(\lVert\boldsymbol{\gamma}_m^{(d)}-\boldsymbol{\beta}_m^{(d)}\rVert \le \kappa_n \right)$.
\begin{align*}
    &P\left(\lVert\boldsymbol{\gamma}_m^{(d)}-\boldsymbol{\beta}_m^{(d)}\rVert \leq \kappa_n \vert \tau, \zeta^{(d)}, w^{(d)}_{m,j_m} \right) \\
    \geq & \prod_{j_m=1}^{q_{m,n}} P\left(\lvert\gamma^{(d)}_{m,j_m}-\beta^{(d)}_{m,j_m}\rvert \leq \frac{\kappa_n}{\sqrt{q_{m,n}}} \vert \tau, \zeta^{(d)}, w^{(d)}_{m,j_m}\right)\\
    \geq &  \prod_{j_m=1}^{q_{m,n}} \left(\frac{2\kappa_n}{\sqrt{q_{m,n}\tau\zeta^{(d)}w^{(d)}_{m,j_m}}} \exp \left\{-\frac{\lvert\beta^{(d)}_{m,j_m}\rvert^2+\kappa_n^2/q_{m,n}}{\tau\zeta^{(d)}w^{(d)}_{m,j_m}}\right\}\right)
\end{align*}

where the last step follows from the fact that $\int_a^be^{-x^2/2}dx\geq e^{-(a^2+b^2)/2}(b-a)$. Let $\varphi(\kappa_n) = \prod_{j_m=1}^{q_{m,n}} \left(\frac{2\kappa_n}{\sqrt{q_{m,n}\tau\zeta^{(d)}w^{(d)}_{m,j_m}}} \exp \left\{-\frac{\lvert\beta^{(d)}_{m,j_m}\rvert^2+\kappa_n^2/q_{m,n}}{\tau\zeta^{(d)}w^{(d)}_{m,j_m}}\right\}\right)$. We want to show $-\log \varphi(\kappa_n) < \frac{n\varepsilon_n^2}{a}$.

Note that
\begin{align*}
    &P\left(\lVert\boldsymbol{\gamma}_m^{(d)}-\boldsymbol{\beta}_m^{(d)}\rVert \leq \kappa_n \vert \tau, \zeta^{(d)} \right) =\mathbb{E}\left[P\left(\lVert\boldsymbol{\gamma}_m^{(d)}-\boldsymbol{\beta}_m^{(d)}\rVert \leq \kappa_n \vert \tau, \zeta^{(d)}, w^{(d)}_{m,j_m} \right)\right]\\
    \geq & \left(\frac{2\kappa_n}{\sqrt{q_{m,n}\tau\zeta^{(d)}}}\right)^{q_{m,n}}\prod_{j_m=1}^{q_{m,n}} \mathbb{E}\left[\left(\frac{1}{\sqrt{w^{(d)}_{m,j_m}}} \exp \left\{-\frac{\lvert\beta^{(d)}_{m,j_m}\rvert^2+\kappa_n^2/q_{m,n}}{\tau\zeta^{(d)}w^{(d)}_{m,j_m}}\right\}\right)\right]\\
    =& \left(\frac{2\kappa_n {\lambda^{(d)}_m}^2}{2\sqrt{q_{m,n}\tau\zeta^{(d)}}}\right)^{q_{m,n}}\prod_{j_m=1}^{q_{m,n}} \int\left(\frac{1}{\sqrt{w^{(d)}_{m,j_m}}} \exp \left\{-\frac{\lvert\beta^{(d)}_{m,j_m}\rvert^2+\kappa_n^2/q_{m,n}}{\tau\zeta^{(d)}w^{(d)}_{m,j_m}}-\frac{{\lambda^{(d)}_m}^2 w^{(d)}_{m,j_m}}{2}\right\}\right)dw^{(d)}_{m,j_m}\\
    =& \left(\frac{\kappa_n {\lambda^{(d)}_m}^2}{\sqrt{q_{m,n}\tau\zeta^{(d)}}}\right)^{q_{m,n}}\exp \left\{-\lambda^{(d)}_m\sqrt{\frac{2(\lvert\beta^{(d)}_{m,j_m}\rvert^2+\kappa_n^2/q_{m,n}}{\tau\zeta^{(d)}}}\right\}
\end{align*}
Following similar reasoning as in \cite{guhaniyogi2017bayesian} we move on to integrate out $\lambda_{m}^{(d)}, \tau$ and $\zeta^{(d)}$, and we end up with the following expression
\begin{align*}
    &P\left(\lVert\boldsymbol{\gamma}_m^{(d)}-\boldsymbol{\beta}_m^{(d)}\rVert \leq \kappa_n, d=1,\ldots,D, m=1,\ldots,M \right)\\
    \geq & \frac{\lambda_{2}^{\lambda_1}\Gamma(Da)}{\Gamma(\lambda_1)\Gamma(a)^D}\prod_{m=1}^M\prod_{d=1}^D\left[\left(\frac{\kappa_n}{\sqrt{q_{m,n}}b_{\lambda,d}}\right)^{q_{m,n}}\frac{\Gamma(q_{m,n}+a_{\lambda,d}\frac{M}{2})}{\Gamma(a_{\lambda,d})}\right]\\
    &\prod_{m=1}^M\prod_{d=1}^D\frac{1}{\left(\frac{\sqrt{2q_{m,n}}\kappa_n}{b_{\lambda,d}}+1\right)^{q_{m,n}+a_{\lambda,d}}}\frac{\exp\{-\lambda_2\}}{(\lambda_1+\sum_{d=1}^{D}a_{\lambda,d}\frac{M}{2})}\frac{\prod_{d=1}^D[\Gamma(a+a_{\lambda,d}\frac{M}{2})]}{\Gamma(Da+\frac{M}{2}\sum_{d=1}^Da_{\lambda,d})}
\end{align*}

Let $C_1 = \frac{\lambda_{2}^{\lambda_1}\Gamma(Da)}{\Gamma(\lambda_1)\Gamma(a)^D}\frac{\exp\{-\lambda_2\}}{(\lambda_1+\sum_{d=1}^{D}a_{\lambda,d}\frac{M}{2})}\frac{\prod_{d=1}^D[\Gamma(a+a_{\lambda,d}\frac{M}{2})]}{\Gamma(Da+\frac{M}{2}\sum_{d=1}^Da_{\lambda,d})}$, then we have
\begin{align*}
    &-\log P\left(\lVert\boldsymbol{\gamma}_m^{(d)}-\boldsymbol{\beta}_m^{(d)}\rVert \leq \kappa_n\right) \leq -\log C_1 \\
    +&\sum_{m=1}^M\sum_{d=1}^D\left(q_{m,n}\left[-\log \kappa_n +\frac{1}{2}\log q_{m,n} +\log b_{\lambda,d}\right]-\log \Gamma(q_{m,n}+a_{\lambda,d})+\log \Gamma(a_{\lambda,d})\right)\\
    +&\sum_{m=1}^M\sum_{d=1}^D(q_{m,n}+a_{\lambda,d})\log\left(\frac{2\sqrt{q_{m,n}}\kappa_n}{b_{\lambda,d}}+1\right)\\
    =&\frac{n\varepsilon^2_n}{4}\left(-\frac{4\log C_1}{n\varepsilon^2_n}+\sum_{m=1}^M\sum_{d=1}^D\left(4\left[-\frac{q_{m,n}\log \kappa_n}{n\varepsilon^2_n}+\frac{q_{m,n}\log q_{m,n}}{2n\varepsilon^2_n}+\frac{q_{m,n}\log b_{\lambda,d}}{n\varepsilon^2_n}\right]\right.\right.\\&\left.\left.-\frac{4 \log \Gamma(q_{m,n}+a_{\lambda,d})}{n\varepsilon^2_n} +\frac{4 \log \Gamma(a_{\lambda,d})}{n\varepsilon^2_n}\right)+\sum_{m=1}^M\sum_{d=1}^D\frac{4(q_{m,n}+a_{\lambda,d})}{n\varepsilon^2_n}\log\left(\frac{2\sqrt{q_{m,n}}\kappa_n}{b_{\lambda,d}}+1\right)\right)
\end{align*}

Notice that $-\frac{\log C_1}{n\varepsilon_n^2}\rightarrow0$ as $n\varepsilon_n^2\rightarrow \infty$. By choosing $\kappa_n = \frac{\Delta_n}{C_0DM\lVert\texttt{GTRP}(\mathcal{X}_i)\rVert}$, where $C_0$ is some positive constant, we have that
    \begin{align*}
    &-\sum_{m=1}^M\sum_{d=1}^D\frac{q_{m,n}\log \kappa_n}{n\varepsilon_n^2}\\
    =&-\frac{D\left[\log \Delta_n-\log (\lVert\texttt{GTRP}(\mathcal{X}_i)\rVert)-\log D\right]}{M}\frac{\sum_{m=1}^M q_{m,n}}{n\varepsilon_n^2} \\
    = &-\frac{D}{M}\frac{\log \Delta_n\sum_{m=1}^M q_{m,n}}{n\varepsilon_n^2} + \frac{D\log (\lVert\texttt{GTRP}(\mathcal{X}_i)\rVert)+D\log C_0DM}{M}\frac{\sum_{m=1}^M q_{m,n}}{n\varepsilon_n^2}\\
    > & -\frac{D}{M}\frac{\log \Delta_n\sum_{m=1}^M q_{m,n}}{n\varepsilon_n^2} + C 
    \end{align*}
    by assumption \textbf{A.5}. From assumption \textbf{A.6} it follows that
    \begin{align*}
        \frac{\log \Delta_n\sum_{m=1}^M q_{m,n}}{n\varepsilon_n^2} \rightarrow 0.
    \end{align*}
    and \(\sum_{m=1}^M q_{m,n}\log q_{m,n} / n\varepsilon_n^2 \rightarrow 0\), which implies
 \(\sum_{m=1}^M q_{m,n} / n\varepsilon_n^2 \rightarrow 0,\quad \sum_{m=1}^M \log q_{m,n} / n\varepsilon_n^2 \rightarrow 0\) and \(\frac{4(q_{m,n}+a_{\lambda,d})}{n\varepsilon^2_n}\log\left(\frac{2\sqrt{q_{m,n}}\kappa_n}{b_{\lambda,d}}+1\right)\rightarrow 0\). By the Stirling approximation of the Gamma function and from assumption \textbf{A.2}, it follows  \(\frac{4 \log \Gamma(q_{m,n}+a_{\lambda,d})}{n\varepsilon^2_n}\rightarrow 0\). Thus we can claim that \(-\log P\left(\lVert\boldsymbol{\gamma}_m^{(d)}-\boldsymbol{\beta}_m^{(d)}\rVert \leq \kappa_n\right)\leq \frac{n\varepsilon_n^2}{4}\), thus \(P\left(\lVert\boldsymbol{\gamma}_m^{(d)}-\boldsymbol{\beta}_m^{(d)}\rVert \leq \kappa_n\right) \geq \exp \left\{-\frac{n\varepsilon_n^2}{4}\right\}\), which implies \[ P\left(\lvert\left<\texttt{GTRP}(\mathcal{X}_i), \mathcal{B}\right> - \left<\mathcal{X}_i, \mathcal{B}_0\right>\rvert < \Delta\right) \geq \exp \left\{-\frac{n\varepsilon_n^2}{4}\right\}.\]
 The result follows from the same arguments used in the proof of Proposition \ref{prop: pos_cons}.

\renewcommand{\thesection}{B}
\renewcommand{\theequation}{B.\arabic{equation}}
\renewcommand{\thefigure}{B.\arabic{figure}}
\renewcommand{\thetable}{B.\arabic{table}}
\setcounter{table}{0}
\setcounter{figure}{0}
\setcounter{equation}{0}

\section{Full conditional distributions} \label{app:Gibbs}
\subsection{PARAFAC priors}
Given the PARAFAC priors, the posterior of the unknowns of the model is given by
\begin{align}
p(\boldsymbol{\gamma}_{m}^{(d)}, \sigma^2,\mu, w_{m,j_{m}}^{(d)},\lambda_{m}^{(d)},\tau,\zeta^{(d)} \mid \mathbf{y}, \mathcal{X})
\end{align}
We adopt the MCMC procedure based on the Gibbs sampling algorithm to sample the unknowns from 3 blocks to reduce autocorrelation.
\subsubsection{Block 1: Sampling \texorpdfstring{$\zeta^{(d)}$}{zeta^(d)} and \texorpdfstring{$\tau$}{tau} from \texorpdfstring{$p(\zeta^{(d)},\tau\mid \boldsymbol{\gamma},\boldsymbol{w})$}{p(zeta^(d), tau | gamma, w)}}

\begin{align}
    p(\zeta^{(d)} \mid \boldsymbol{\gamma}, \tau, \boldsymbol{w}) &\propto p(\boldsymbol{\gamma}\mid \boldsymbol{\zeta}, \tau, \boldsymbol{w})p(\boldsymbol{\zeta})\notag\\
    & \propto \prod_{d=1}^D\prod_{m=1}^M{\zeta^{(d)}}^{-\frac{p_m}{2}}\exp \left\{-\frac{1}{2}{\boldsymbol{\gamma}_m^{(d)}}^T\frac{{W_m^{(d)}}^{-1}}{\tau \zeta^{(d)}}\boldsymbol{\gamma}_m^{(d)}\right\}\prod_{d=1}^D {\zeta^{(d)}}^{\alpha-1}\notag\\
    &= \prod_{d=1}^D {\zeta^{(d)}}^{-\sum_{m=1}^Mp_m / 2 + \alpha-1} \exp\left\{-\frac{1}{2\tau\zeta^{(d)}}\sum_{m=1}^M{\boldsymbol{\gamma}_m^{(d)}}^T{W_m^{(d)}}^{-1}\boldsymbol{\gamma}_m^{(d)}\right\}\notag\\
    &\sim \mathcal{G}i\mathcal{G} \left(\alpha-\frac{\sum_{m=1}^Mp_m}{2},0,\frac{\sum_{m=1}^M{\boldsymbol{\gamma}_m^{(d)}}^T{W_m^{(d)}}^{-1}\boldsymbol{\gamma}_m^{(d)}}{\tau}\right)\notag\\
    & \sim \mathcal{IG}\left(\frac{\sum_{m=1}^Mp_m}{2} - \alpha, \frac{\sum_{m=1}^M{\boldsymbol{\gamma}_m^{(d)}}^T{W_m^{(d)}}^{-1}\boldsymbol{\gamma}_m^{(d)}}{2\tau}\right) \notag
\end{align}

\begin{align}
    p(\tau \mid \boldsymbol{\gamma}, \boldsymbol{\zeta}, \boldsymbol{w}) &\propto p(\boldsymbol{\gamma}\mid \boldsymbol{\zeta}, \tau,  \boldsymbol{w})p(\tau) \notag \\
    &\propto \prod_{d=1}^D\prod_{m=1}^M \tau^{-\frac{p_m}{2}}\exp \left\{-\frac{1}{2}{\boldsymbol{\gamma}_m^{(d)}}^T\frac{{W_m^{(d)}}^{-1}}{\tau \zeta^{(d)}}\boldsymbol{\gamma}_m^{(d)}\right\}\tau^{a_{\tau}-1}\exp \left\{-b_{\tau}\tau\right\}\notag\\
    & = \tau^{a_{\tau}-\frac{D\sum_{m=1}^Mp_m}{2}-1}\exp\left\{-\frac{1}{2\tau}\sum_{d=1}^D\frac{\sum_{m=1}^M {\boldsymbol{\gamma}_m^{(d)}}^T{W_m^{(d)}}^{-1}\boldsymbol{\gamma}_m^{(d)}}{\zeta^{(d)}}-b_{\tau}\tau\right\}\notag\\
    &\sim \mathcal{G}i\mathcal{G} \left(a_{\tau}-\frac{D\sum_{m=1}^Mp_m}{2}, 2b_{\tau}, \sum_{d=1}^D\frac{\sum_{m=1}^M {\boldsymbol{\gamma}_m^{(d)}}^T{W_m^{(d)}}^{-1}\boldsymbol{\gamma}_m^{(d)}}{\zeta^{(d)}}\right)\notag
\end{align}

\subsubsection{Block 2: Sampling \texorpdfstring{$\lambda_{m}^{(d)}$}{lambda} and \texorpdfstring{$w_{m,j_{m}}^{(d)}$}{w} from \texorpdfstring{$p(\lambda_{m}^{(d)}, w_{m,j_{m}}^{(d)}|\gamma_{m,j_{m}}^{(d)}, \tau, \zeta^{(d)})$}{p(lambda, w | gamma, tau, zeta)}}

Notice that by the construction of the prior distributions, $\gamma_{m,j_{m}}^{(d)}$ follows a double exponential distribution given $\lambda_{m}^{(d)}$, $\tau$, $\zeta^{(d)}$, that is $\gamma_{m,j_{m}}^{(d)} \sim \mathcal{D}\mathcal{E }\left(0,\sqrt{\tau\zeta^{(d)}}/\lambda_{m}^{(d)}\right)$. The full conditional of $\lambda_{m}^{(d)}$ can be written as
\begin{align}
&p\left(\lambda_{m}^{(d)}\mid\gamma_{m,j_{m}}^{(d)}, \tau, \zeta^{(d)}\right)\propto \pi(\lambda_{m}^{(d)}) p\left(\gamma_{m,j_{m}}^{(d)}\mid\lambda_{m}^{(d)},\tau, \zeta^{(d)}\right)\notag\\
&\propto\left(\tau\zeta^{(d)}\right)^{-\frac{p_{m}}{2}}\left(\lambda_{m}^{(d)}\right)^{a_{\lambda}+p_{m}-1}\exp\left\{-\left(\frac{\sum_{j_{m}=1}^{p_{m}}\left|\gamma_{m,j_{m}}^{(d)}\right|}{\sqrt{\tau\zeta^{(d)}}}+b_{\lambda}\right)\lambda_{m}^{(d)}\right\}\notag\\
& \propto \mathcal{G}a \left(a_{\lambda}+p_{m}, \sum_{j_{m}=1}^{p_{m}}\left|\gamma_{m,j_{m}}^{(d)}\right|/\sqrt{\tau\zeta^{(d)}}+b_{\lambda}\right)\notag
\end{align}

The full conditional for $w_{m,j_{m}}^{(d)}$ is
\begin{align}
&p\left(w_{m,j_{m}}^{(d)}\mid\gamma_{m,j_{m}}^{(d)},\lambda_{m}^{(d)}, \tau, \zeta^{(d)}\right) \propto \pi\left(w_{m,j_{m}}^{(d)}\right)p\left(\gamma_{m,j_{m}}^{(d)}\mid\lambda_{m}^{(d)}, \tau, \zeta^{(d)}, w_{m,j_{m}}^{(d)}\right)\notag\\
&\propto {w_{m,j_{m}}^{(d)}}^{\frac{1}{2}-1}\exp \left\{-\frac{1}{2}\left({\lambda_{m}^{(d)}}^2w_{m,j_{m}}^{(d)}+\frac{{\gamma_{m,j_{m}}^{(d)}}^2}{\tau\zeta^{(d)}w_{m,j_{m}}^{(d)}}\right)\right\}\notag\\
&\propto \mathcal{G}i\mathcal{G}\left(1/2, {\lambda_{m}^{(d)}}^2, {\gamma_{m,j_{m}}^{(d)}}^2/\tau\zeta^{(d)}\right)\notag
\end{align}

\subsubsection{Block 3: Sampling \texorpdfstring{$\boldsymbol{\gamma}_{m}^{(d)}, \mu, \sigma^2$}{gamma, mu, sigma^2}}

\begin{align}
&p(\boldsymbol{\gamma}_m^{(d)}\mid\mathbf{y}, \mathcal{X},\tau, \boldsymbol{\zeta},\boldsymbol{w},\mu,\sigma^2) \propto p(\mathbf{y}\mid \boldsymbol{\gamma}_m^{(d)},\mathcal{X},\tau, \boldsymbol{\zeta},\boldsymbol{w},\mu,\sigma^2) p(\boldsymbol{\gamma}_m^{(d)})\notag\\
    & \propto \prod_{t=1}^T \exp\left\{-\frac{1}{2}\frac{\left(y_t - \mu -\left<\mathcal{B},\mathcal{X}_t\right>\right)^2}{\sigma^2}\right\}\exp\left\{-\frac{1}{2\tau\zeta^{(d)}}{\boldsymbol{\gamma}_m^{(d)}}^T{W_m^{(d)}}^{-1}\boldsymbol{\gamma}_m^{(d)}\right\}\notag
\end{align}
notice that
\begin{align}
    \left<\mathcal{B},\mathcal{X}_t\right> & = \left<\mathcal{B}^{(d)},\mathcal{X}_t\right> + \sum_{d' \neq d}^D \left<\mathcal{B}^{(d')},\mathcal{X}_t\right> \notag\\
    &= {\boldsymbol{\gamma}_m^{(d)}}^T \left(\mathcal{X}_t \times_1 \boldsymbol{\gamma}_1^{(d)} \cdots \times_{m-1} \boldsymbol{\gamma}_{m-1}^{(d)} \times_{m+1} \boldsymbol{\gamma}_{m+1}^{(d)} \cdots \times_M \boldsymbol{\gamma}_M^{(d)}\right) + \sum_{d' \neq d}^D \left<\mathcal{B}^{(d')},\mathcal{X}_t\right> \notag\\
    & = {\boldsymbol{\gamma}_m^{(d)}}^T \psi_{mt}^{(d)} + R_{t}^{(d)} \notag
\end{align}
where 
\begin{align}
    \psi_{mt}^{(d)} & = \mathcal{X}_t \times_1 \boldsymbol{\gamma}_1^{(d)} \cdots \times_{m-1} \boldsymbol{\gamma}_{m-1}^{(d)} \times_{m+1} \boldsymbol{\gamma}_{m+1}^{(d)} \cdots \times_M \boldsymbol{\gamma}_M^{(d)}\notag\\
    R_{t}^{(d)} & = \sum_{d' \neq d}^D \left<\mathcal{B}^{(d')},\mathcal{X}_t\right>. \notag
\end{align}
The  quadratic term in the likelihood becomes 
\begin{align}
    &\left(y_t - \mu -\left<\mathcal{B}, \mathcal{X}_t\right>\right)^2\notag \\
    & = \left(y_t - \mu - R_t^{(d)}\right)^2 -2\left(y_t - \mu - R_t^{(d)}\right){\boldsymbol{\gamma}_m^{(d)}}^T \psi_{mt}^{(d)}+{\boldsymbol{\gamma}_m^{(d)}}^T \psi_{mt}^{(d)} {\psi_{mt}^{(d)}}^T\boldsymbol{\gamma}_m^{(d)}\notag\\
    & = (\tilde{y}_{t}^{(d)})^{2}-2\tilde{y}_t^{(d)}{\boldsymbol{\gamma}_m^{(d)}}^T \psi_{mt}^{(d)} +{\boldsymbol{\gamma}_m^{(d)}}^T \psi_{mt}^{(d)} {\psi_{mt}^{(d)}}^T\boldsymbol{\gamma}_m^{(d)} \notag
\end{align}
where $\tilde{y}_t^{(d)} = y_t - \mu - R_t^{(d)}$.

Then we have the full conditional for $\boldsymbol{\gamma}_m^{(d)}$
\begin{align}
    &p(\boldsymbol{\gamma}_m^{(d)}\mid\mathbf{y}, \mathcal{X},\tau, \boldsymbol{\zeta},\boldsymbol{w},\mu,\sigma^2)\notag\\
    &\propto \exp\left\{-\frac{1}{2\sigma^2}\left[{\boldsymbol{\gamma}_m^{(d)}}^T\sum_{t=1}^T\psi_{mt}^{(d)} {\psi_{mt}^{(d)}}^T\boldsymbol{\gamma}_m^{(d)}-2{\boldsymbol{\gamma}_m^{(d)}}^T\sum_{t=1}^T\tilde{y}_t^{(d)} \psi_{mt}^{(d)} \right]-\frac{1}{2\tau\zeta^{(d)}}{\boldsymbol{\gamma}_m^{(d)}}^T{W_m^{(d)}}^{-1}\boldsymbol{\gamma}_m^{(d)}\right\} \notag\\
    & \propto \exp\left\{-\frac{1}{2}\left[{\boldsymbol{\gamma}_m^{(d)}}^T\left(\frac{\sum_{t=1}^T\psi_{mt}^{(d)} {\psi_{mt}^{(d)}}^T}{\sigma^2}+\frac{{W_m^{(d)}}^{-1}}{\tau\zeta^{(d)}}\right)\boldsymbol{\gamma}_m^{(d)}-2{\boldsymbol{\gamma}_m^{(d)}}^T\frac{\sum_{t=1}^T\tilde{y}_t^{(d)} \psi_{mt}^{(d)}}{\sigma^2}\right]\right\}\notag \\
    & \sim \mathcal{MN}_{p_m}(\boldsymbol{\mu}^*, \Sigma^*) \notag
\end{align}
where 
\begin{align}
    \Sigma^* & = \left(\frac{\sum_{t=1}^T\psi_{mt}^{(d)} {\psi_{mt}^{(d)}}^T}{\sigma^2}+\frac{{W_m^{(d)}}^{-1}}{\tau\zeta^{(d)}}\right)^{-1}\notag\\
     \boldsymbol{\mu}^* & = \left(\frac{\sum_{t=1}^T\psi_{mt}^{(d)} {\psi_{mt}^{(d)}}^T}{\sigma^2}+\frac{{W_m^{(d)}}^{-1}}{\tau\zeta^{(d)}}\right)^{-1}\frac{\sum_{t=1}^T\tilde{y}_t^{(d)} \psi_{mt}^{(d)}}{\sigma^2}\notag
\end{align}

The full conditional of $\sigma^2$ can be written as:
\begin{align}
    p\left(\sigma^2 \mid \mathbf{y},\mathcal{X}, \mu, \boldsymbol{\gamma}\right) &\propto p\left(\mathbf{y} \mid \mathcal{X},\mu,\boldsymbol{\gamma},\sigma^2\right)p\left(\sigma^2\right) \notag \\
    & \propto \left(\sigma^2\right)^{-\left(a_{\sigma}+\frac{T}{2}\right)-1}\exp \left\{-\frac{1}{\sigma^2}\left(\frac{1}{2}\sum_{t=1}^T\left(y_t - \left<\mathcal{B}, \mathcal{X}_t\right>-\mu\right)^2 +b_{\sigma}\right)\right\},\notag
\end{align}

which is the kernel of the IG distribution $\mathcal{IG} \left(a_{\sigma}^*, b_{\sigma}^*\right)$, where $a_{\sigma}^* = a_{\sigma}+\frac{T}{2}$ and $b_{\sigma}^*  = \frac{1}{2}\sum_{t=1}^T\left(y_t - \left<\mathcal{B}, \mathcal{X}_t\right> - \mu\right)^2 +b_{\sigma}$. Finally, let $\mu^* = \sum_{t=1}^T\left(y_t - \left<\mathcal{B}, \mathcal{X}_t\right>\right) {\sigma_{\mu}^*}^2 / \sigma^2$ and 
    ${\sigma_{\mu}^*}^2  = \left(T/\sigma^2+1/\sigma_{\mu}^2\right)^{-1}$, the full conditional of $\mu$ is:
\begin{align}
    &p\left(\mu \mid \mathbf{y},\mathcal{X}, \boldsymbol{\gamma},\sigma^2\right) \\&\propto p\left(\mathbf{y} \mid  \mathcal{X}, \mu, \boldsymbol{\gamma},\sigma^2 \right) \pi \left(\mu\right)
     \propto \exp \left\{-\frac{1}{2\sigma^2}\left[T\mu^2 -2\mu\sum_{t=1}^T\left(y_t - \left<\mathcal{B}, \mathcal{X}_t\right>\right)\right] - \frac{1}{2}\frac{\mu^2}{\sigma_{\mu}^2}\right\} \notag\\
    &= \exp \left\{-\frac{1}{2}\left[\left(\frac{T}{\sigma^2}+\frac{1}{\sigma_{\mu}^2}\right)\mu^2 - 2\mu\frac{\sum_{t=1}^T\left(y_t - \left<\mathcal{B}, \mathcal{X}_t\right>\right)}{\sigma^2}\right]\right\}
     \propto \mathcal{N}\left(\mu^*, {\sigma_{\mu}^*}^2\right). \notag
\end{align}

\subsection{Gaussian priors}
Given the Gaussian prior for the tensor coefficients specified in Theorem \ref{thm: pos_cons}, we further more assume that $\sigma^2 \sim \mathcal{IG}(a_{\sigma}, b_{\sigma})$ and $\mu \sim \mathcal{N}(0, \sigma_{\mu}^2)$, w.o.l.g we assume the tensor coefficient is a mode-$2$ tensor, then we have the following full conditionals for the tensor coefficients, $\sigma^2$ and $\mu$:
\begin{align*}
    &p\left(\mathcal{B}_{\text{vec}} \mid \boldsymbol{y}, X, \boldsymbol{\mu}, \sigma^2\right) \propto p\left(\boldsymbol{y} \mid \mathcal{B}_{\text{vec}}, X\right)p\left(\mathcal{B}_{\text{vec}}\right)\\
    & \propto \exp\left\{-\frac{1}{2\sigma^2}\left(\boldsymbol{y}-\boldsymbol{\mu}-X\mathcal{B}_{\text{vec}}\right)^{\top}\left(\boldsymbol{y}-\boldsymbol{\mu}-X\mathcal{B}_{\text{vec}}\right)\right\}\exp\left\{-\frac{1}{2}\mathcal{B}_{\text{vec}}^{\top}\left(\Sigma_1\otimes\Sigma_2\right)^{-1}\mathcal{B}_{\text{vec}}\right\}\\
    &\sim \mathcal{MN}\left(\mu_{\mathcal{B}_{\text{vec}}}, \Sigma_{\mathcal{B}_{\text{vec}}}\right)
\end{align*}
where $\mathcal{B}_{\text{vec}}$ is the vectorized tensor coefficient $\mathcal{B}$ and $X$ is the matrix obtained stacking vertically vectorized covariate tensors $\text{vec}(\mathcal{X}_t)^{\top}$, $t=1,\ldots,T$, $\boldsymbol{\mu}=\mu\boldsymbol{\iota}_{T}$, $\Sigma_{\mathcal{B}_{\text{vec}}} = (X^{\top}X/\sigma^2 + \left(\Sigma_1\otimes\Sigma_2\right)^{-1})^{-1}$ and $\mu_{\mathcal{B}_{\text{vec}}} = \Sigma_{\mathcal{B}_{\text{vec}}} X^{\top}(\boldsymbol{y}-\boldsymbol{\mu})/\sigma^2$.

\begin{align*}
    &p\left(\sigma^2\mid \boldsymbol{y}, X, \mathcal{B}_{\text{vec}}, \boldsymbol{\mu}\right) \propto p\left(\boldsymbol{y} \mid X, \boldsymbol{\mu}, \mathcal{B}_{\text{vec}}, \sigma^2\right) p(\sigma^2)\\
    &\propto (\sigma^2)^{\frac{T}{2}}\exp\left\{-\frac{1}{2\sigma^2}\left(\boldsymbol{y}-\boldsymbol{\mu}-X\mathcal{B}_{\text{vec}}\right)^{\top}\left(\boldsymbol{y}-\boldsymbol{\mu}-X\mathcal{B}_{\text{vec}}\right)\right\}(\sigma^2)^{-a_{\sigma}-1}\exp\left\{-\frac{b_{\sigma}}{\sigma^2}\right\}\\
    &\sim \mathcal{IG}\left(a_{\sigma}^*, b_{\sigma}^*\right)
\end{align*}
where $a_{\sigma}^*=a_{\sigma} + T/2$ and $b_{\sigma}^*=b_{\sigma}+(\boldsymbol{y}-\boldsymbol{\mu})^{\top}(\boldsymbol{y}-\boldsymbol{\mu})/2$.

\begin{align*}
    &p(\mu \mid \boldsymbol{y}, X,  \mathcal{B}_{\text{vec}}, \sigma^2) \propto p\left(\boldsymbol{y}\mid \mu, X,  \mathcal{B}_{\text{vec}}, \sigma^2\right)p(\mu)\\
    &\propto \prod_{t=1}^T\exp\left\{-\frac{1}{2\sigma^2}(y_t-\mu-\mathcal{X}_t^{\top}\mathcal{B}_{\text{vec}})^2\right\}\exp\left\{-\frac{\mu^2}{2\sigma_{\mu}^2}\right\}\\
    &\sim \mathcal{N}\left(\mu^*, {\sigma_{\mu}^*}^2\right)
\end{align*}
where $\mu^*=(T/\sigma^2+1/\sigma_{\mu}^2)^{-1}\sum_{t=1}^T(y_t-\mathcal{X}_t^{\top}\mathcal{B}_{\text{vec}})/\sigma^2$ and ${\sigma_{\mu}^*}^2=(T/\sigma^2+1/\sigma_{\mu}^2)^{-1}$.

\renewcommand{\thesection}{C}
\renewcommand{\theequation}{C.\arabic{equation}}
\renewcommand{\thefigure}{C.\arabic{figure}}
\renewcommand{\thetable}{C.\arabic{table}}
\setcounter{table}{0}
\setcounter{figure}{0}
\setcounter{equation}{0}

\section{Further numerical results}\label{app:NumRes}
In this section, we provide further illustration of the effectiveness of the Bayesian compressed tensor regression model proposed in Section \ref{sec:model}.

\subsection{Sample size}
We consider three different simulation settings. In each setting, a different $20\times 20$ true tensor coefficient is used to generate the $n=1,500$ i.i.d. samples. The tensor covariates, which are also $20 \times 20$, are drawn i.i.d. from the standard normal distribution. The simulated results are presented in Fig. \ref{fig: sim_parafac}.

\begin{figure}[h!]
\vspace*{2pt}
    \centering
    \begin{tabular}{ccc}
        \includegraphics[width=.31\linewidth]{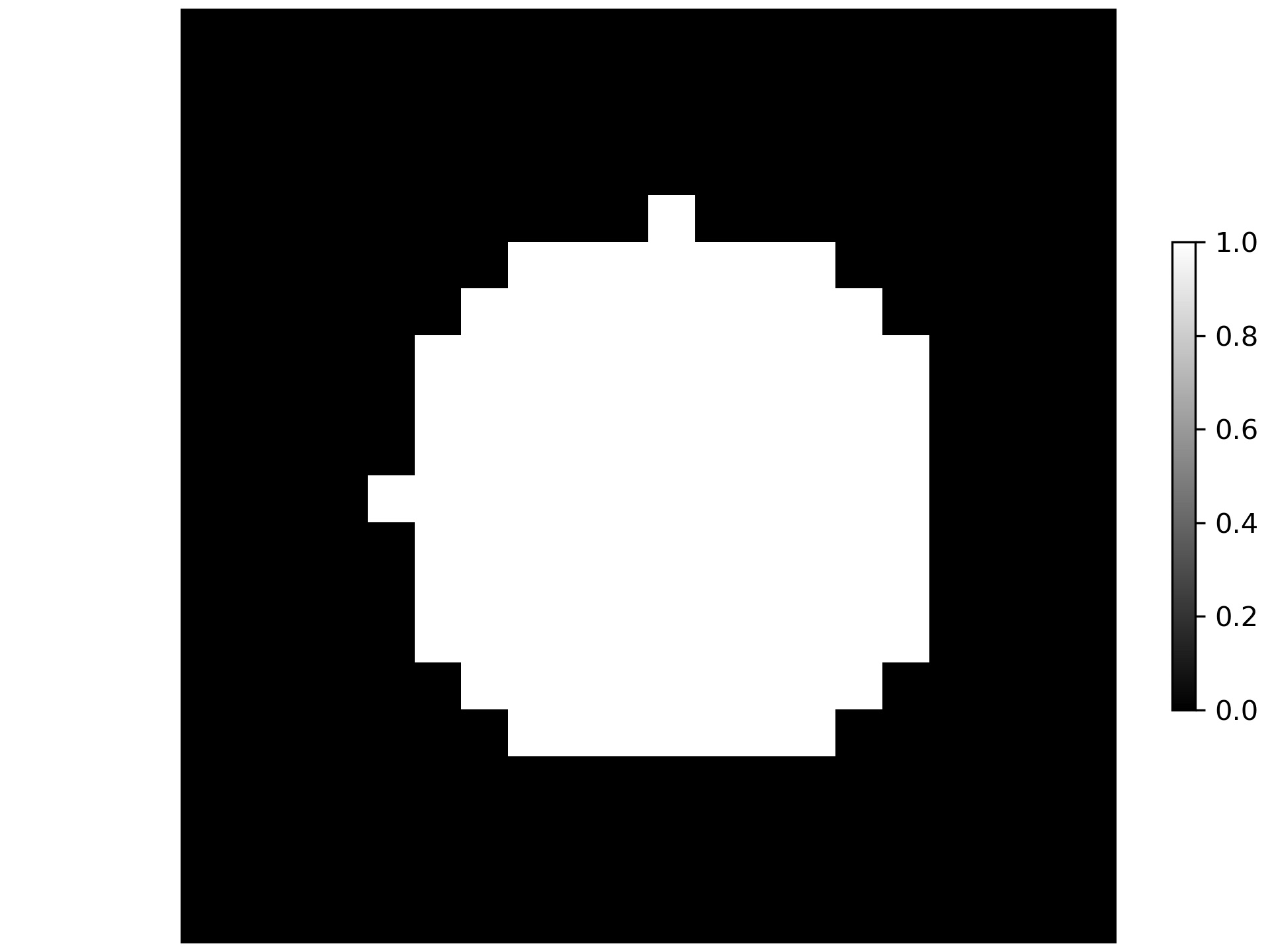} & \includegraphics[width=.31\linewidth]{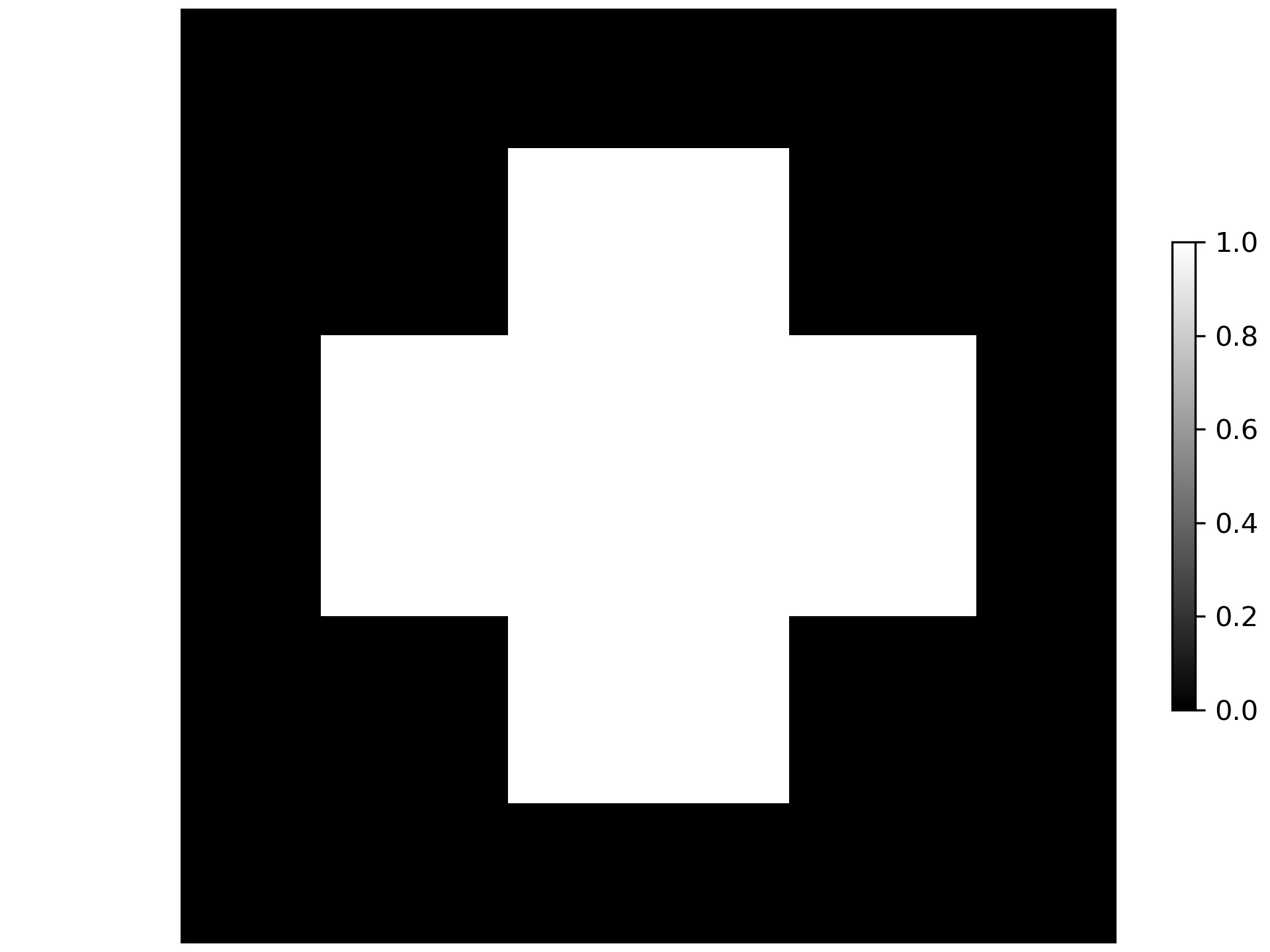} & \includegraphics[width=.31\linewidth]{imgs/coefficients_setting3.jpeg} \\
        \includegraphics[width=.31\linewidth]{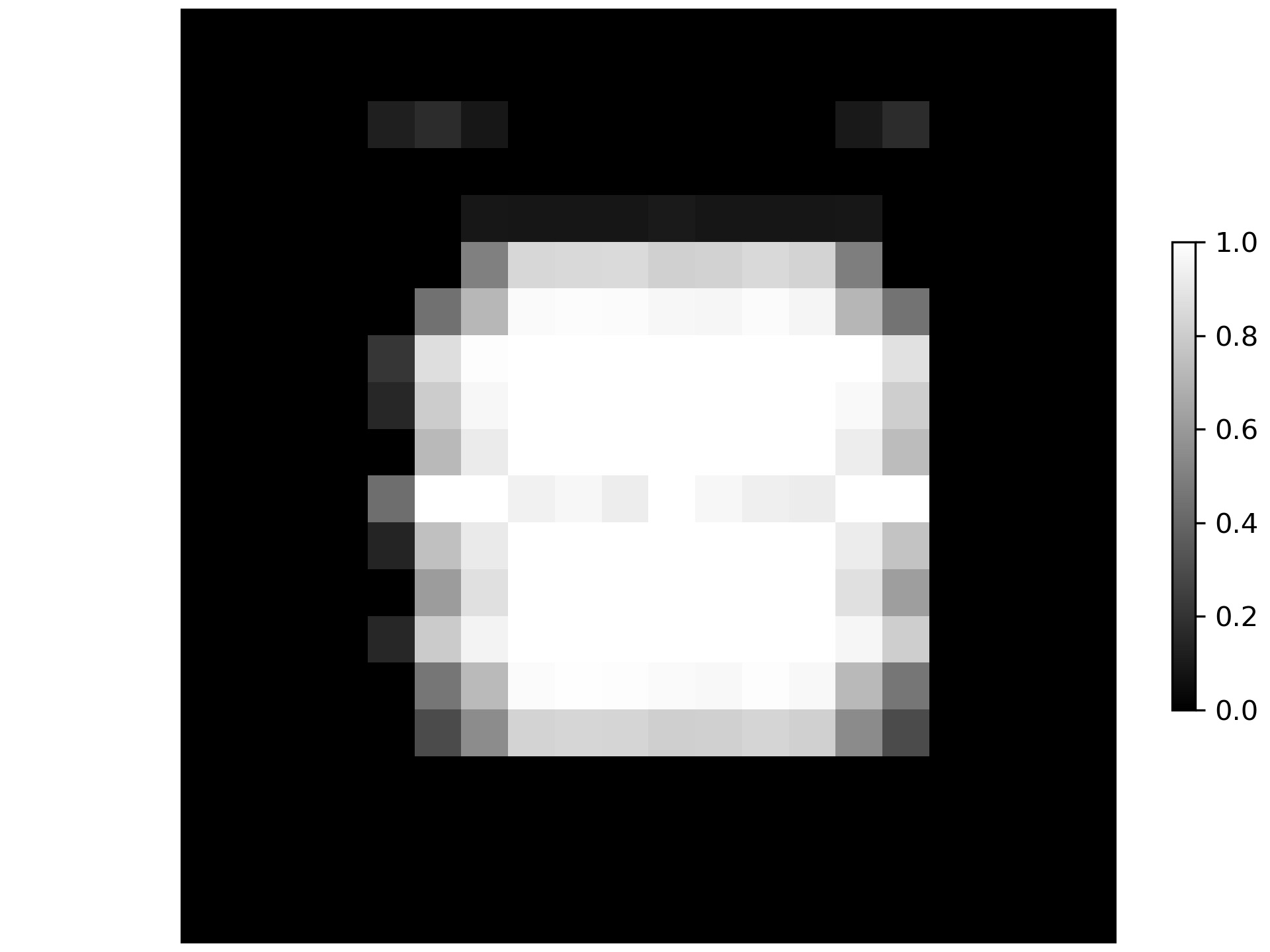} & \includegraphics[width=.31\linewidth]{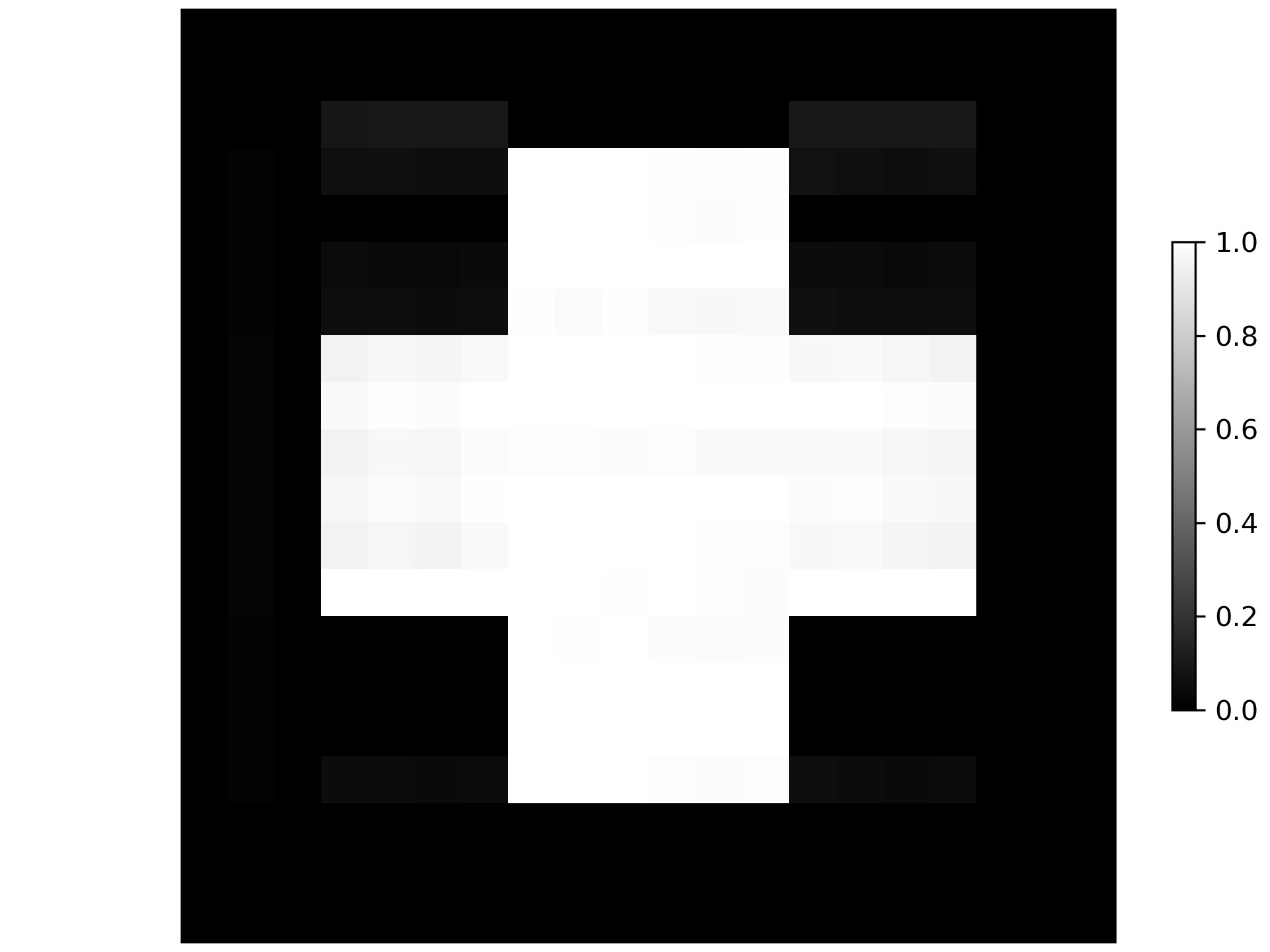} & \includegraphics[width=.31\linewidth]{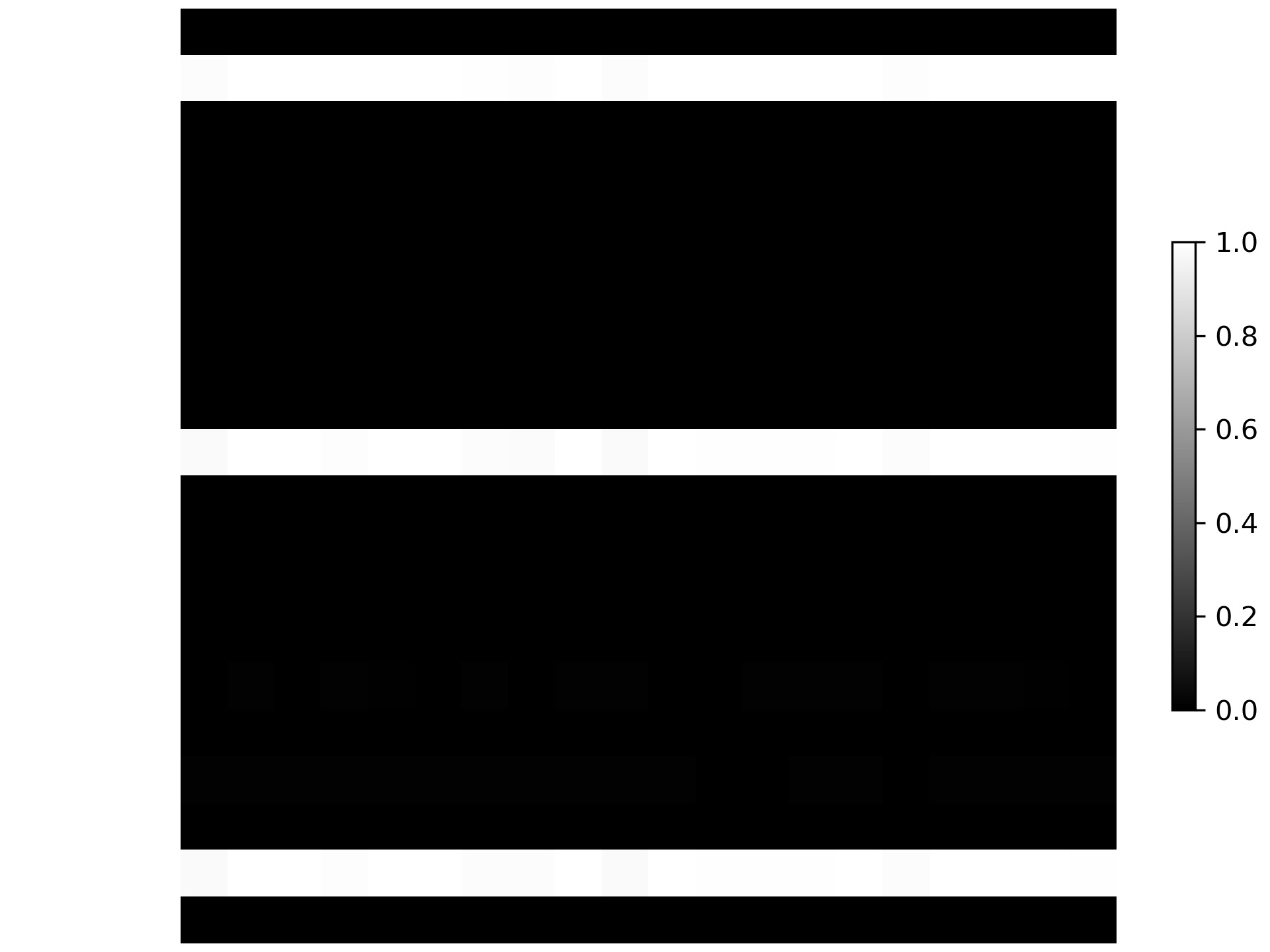} \\
        \includegraphics[width=.31\linewidth]{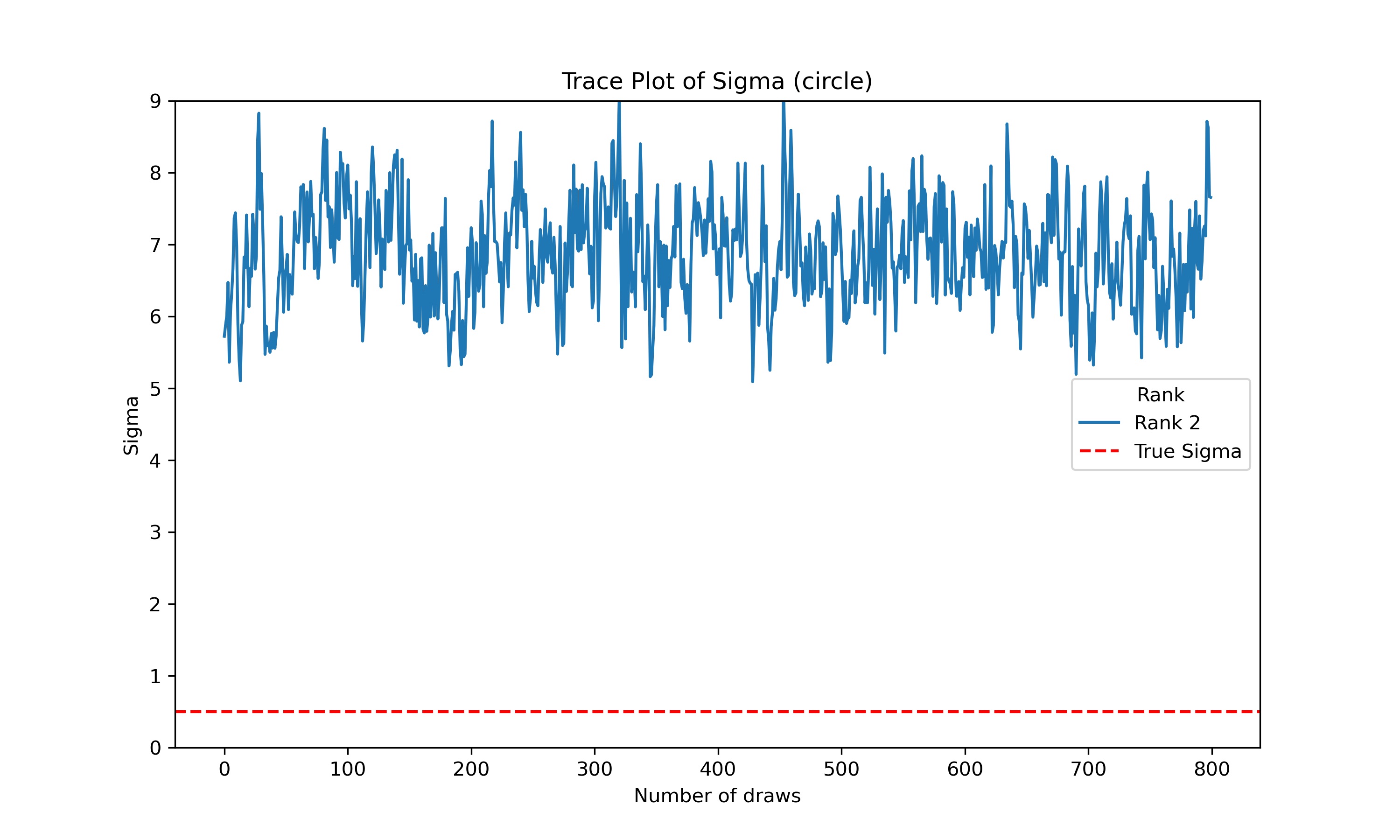} & \includegraphics[width=.31\linewidth]{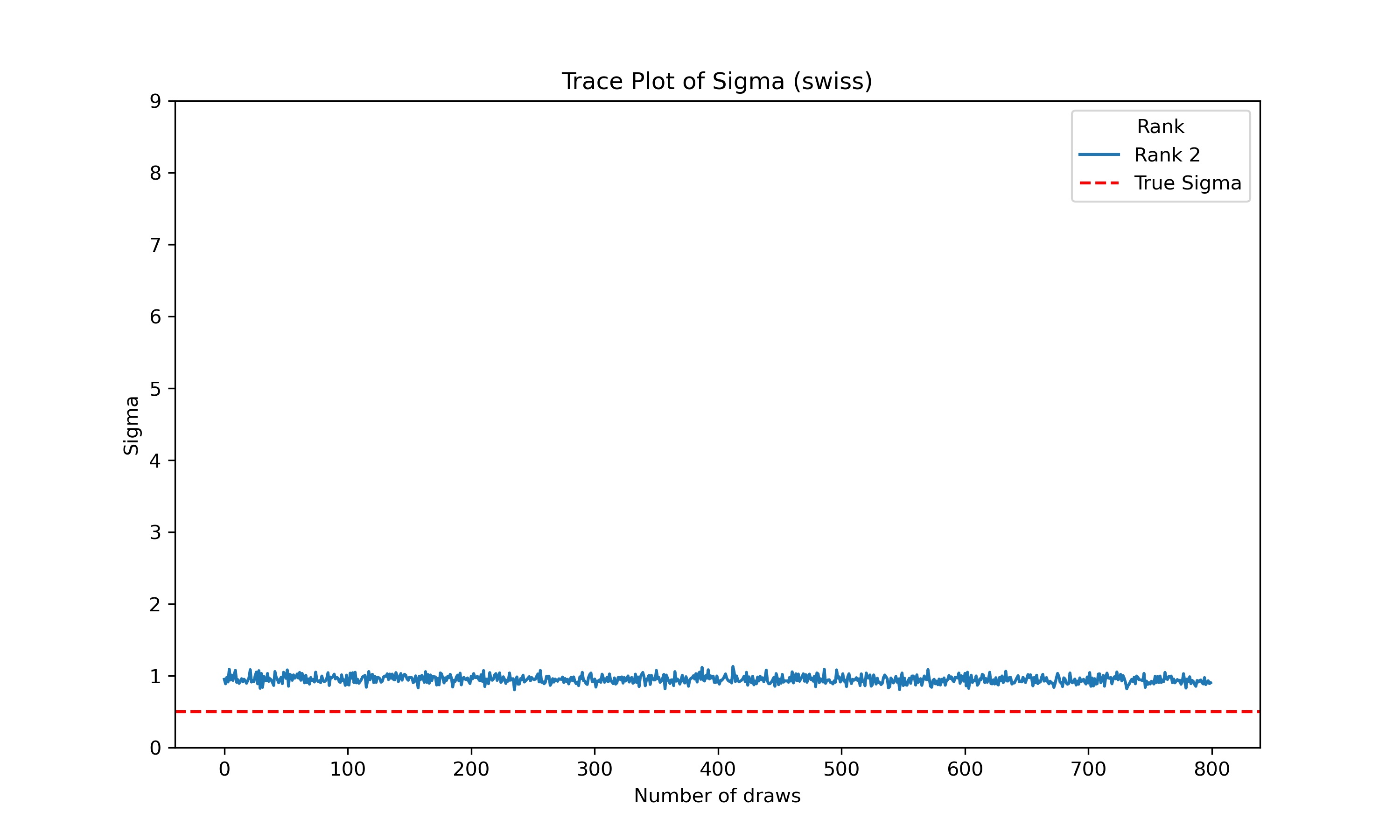} & \includegraphics[width=.31\linewidth]{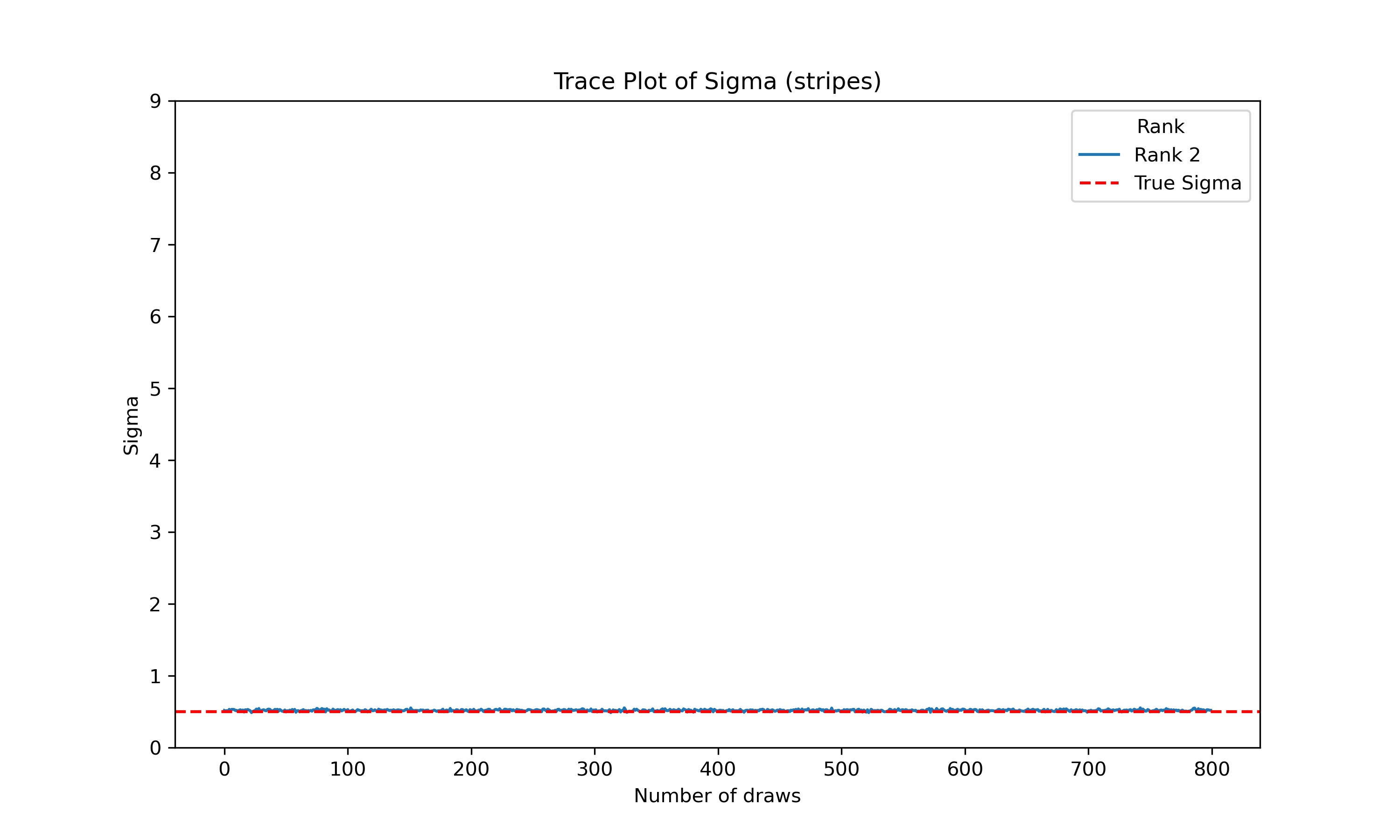} \\
        \includegraphics[width=.31\linewidth]{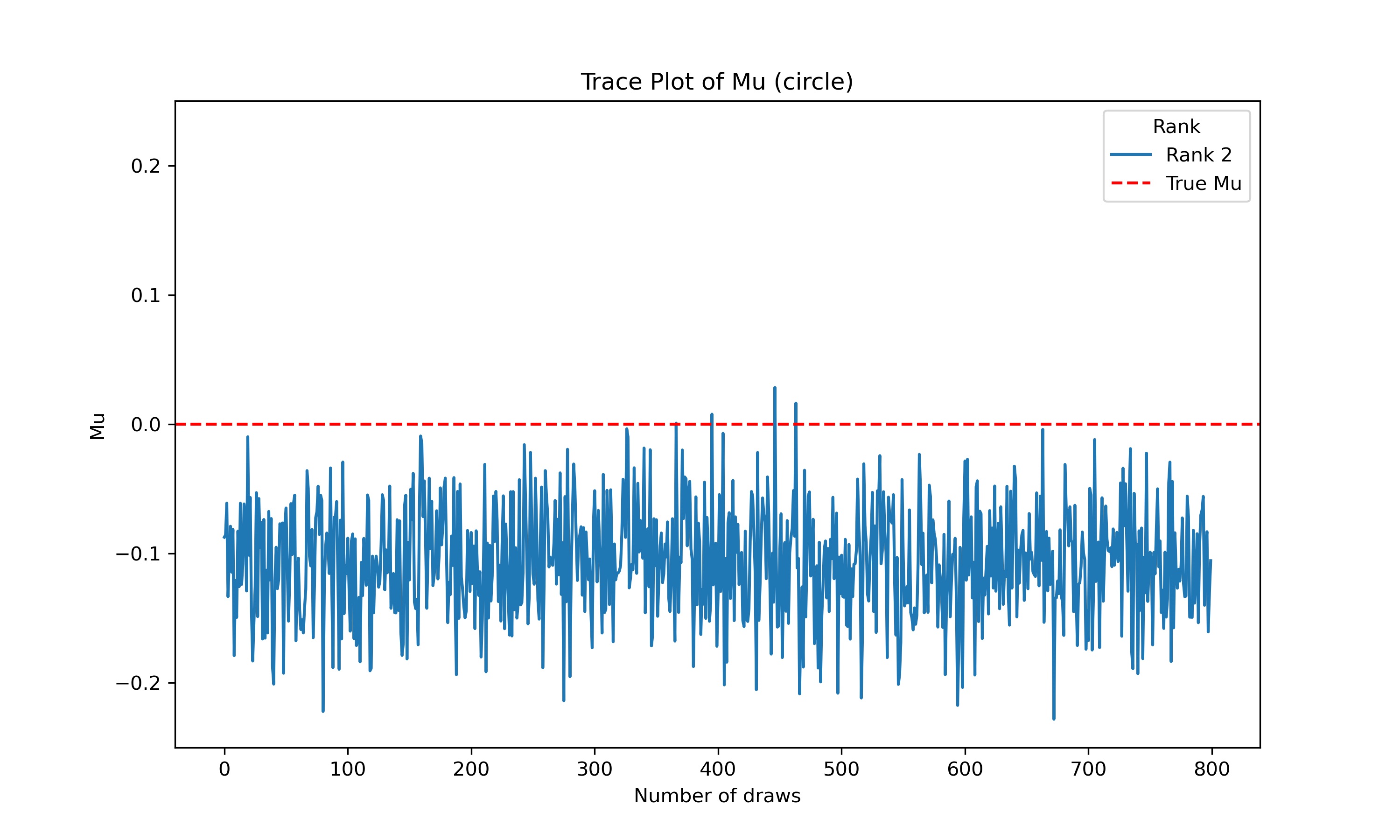} & \includegraphics[width=.31\linewidth]{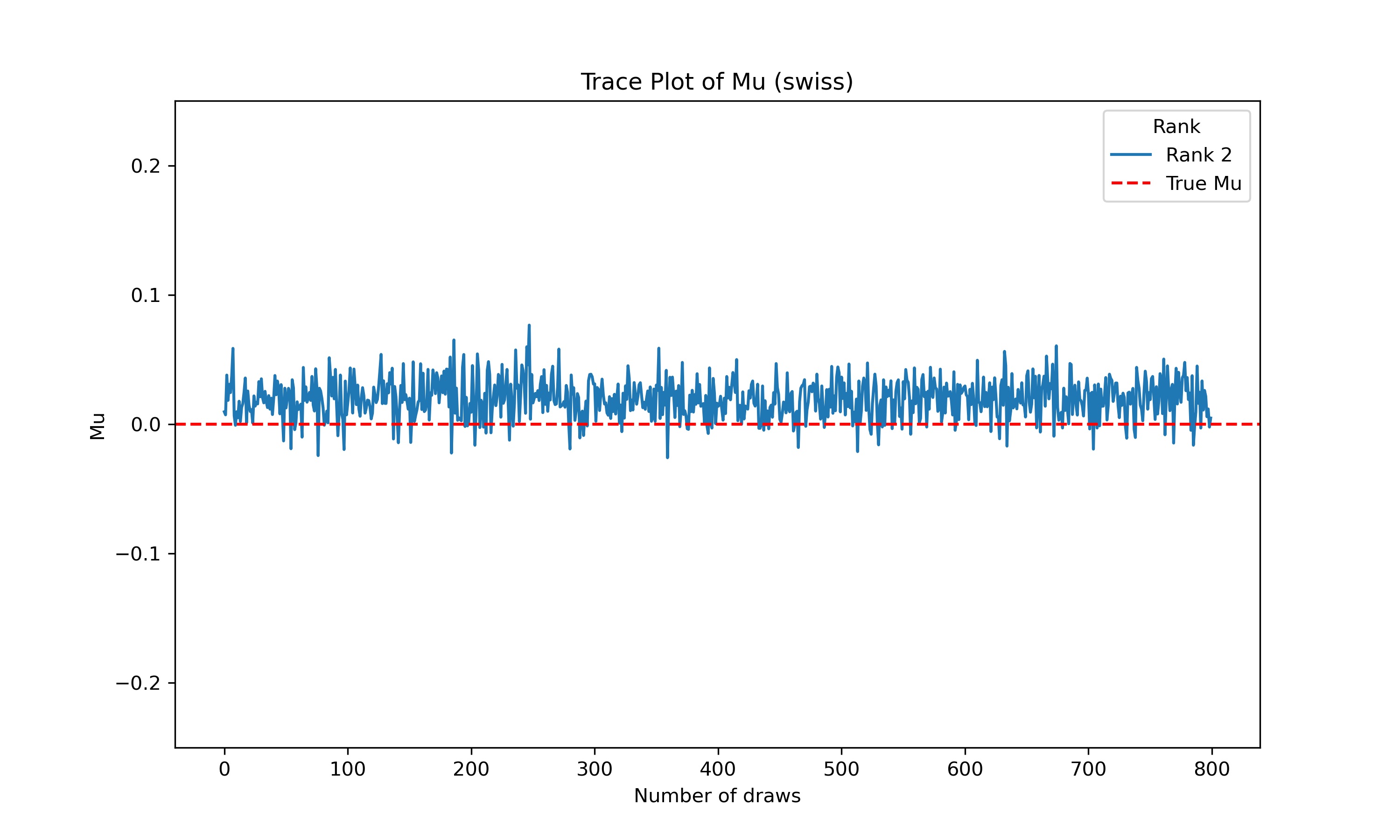} & \includegraphics[width=.31\linewidth]{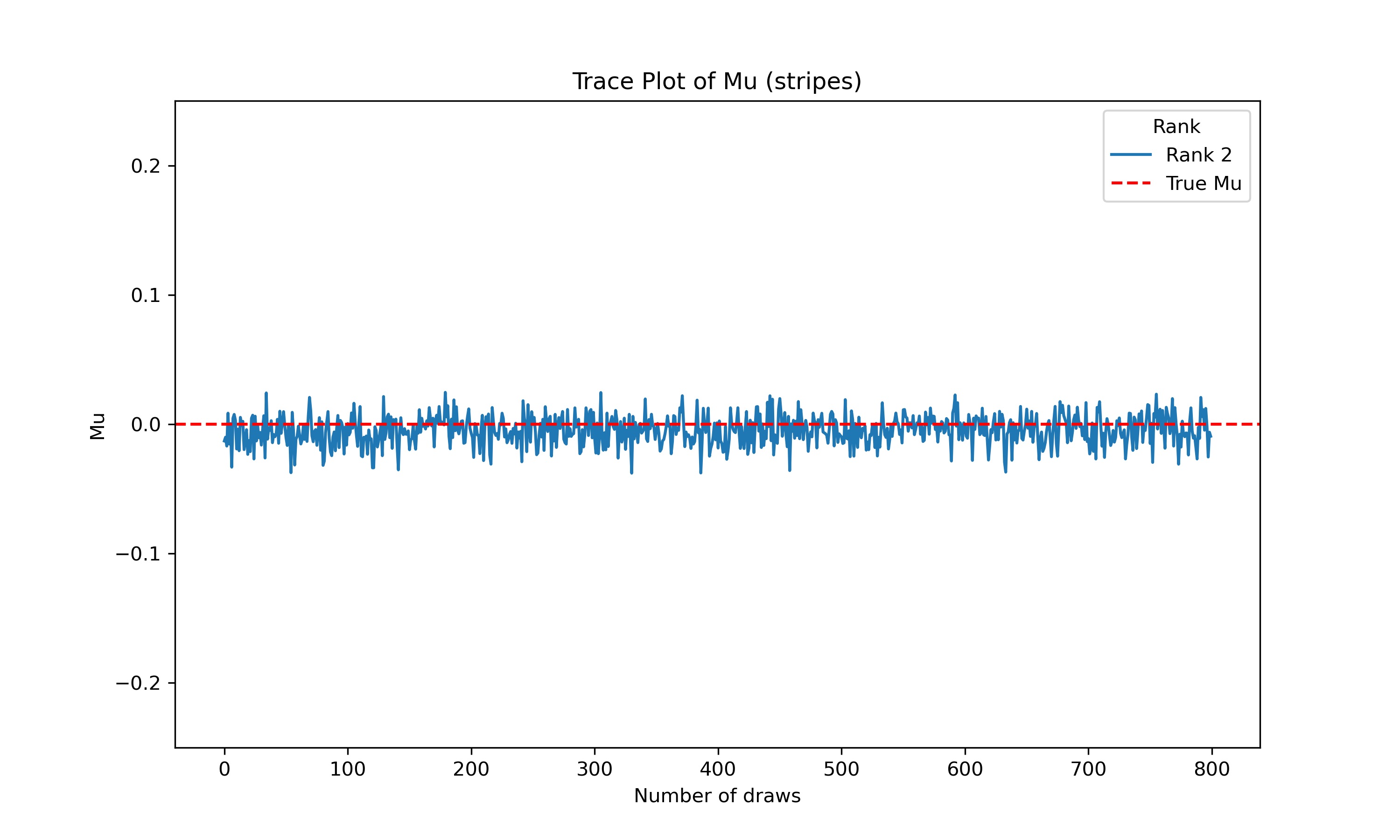} \\
        \includegraphics[width=.31\linewidth]{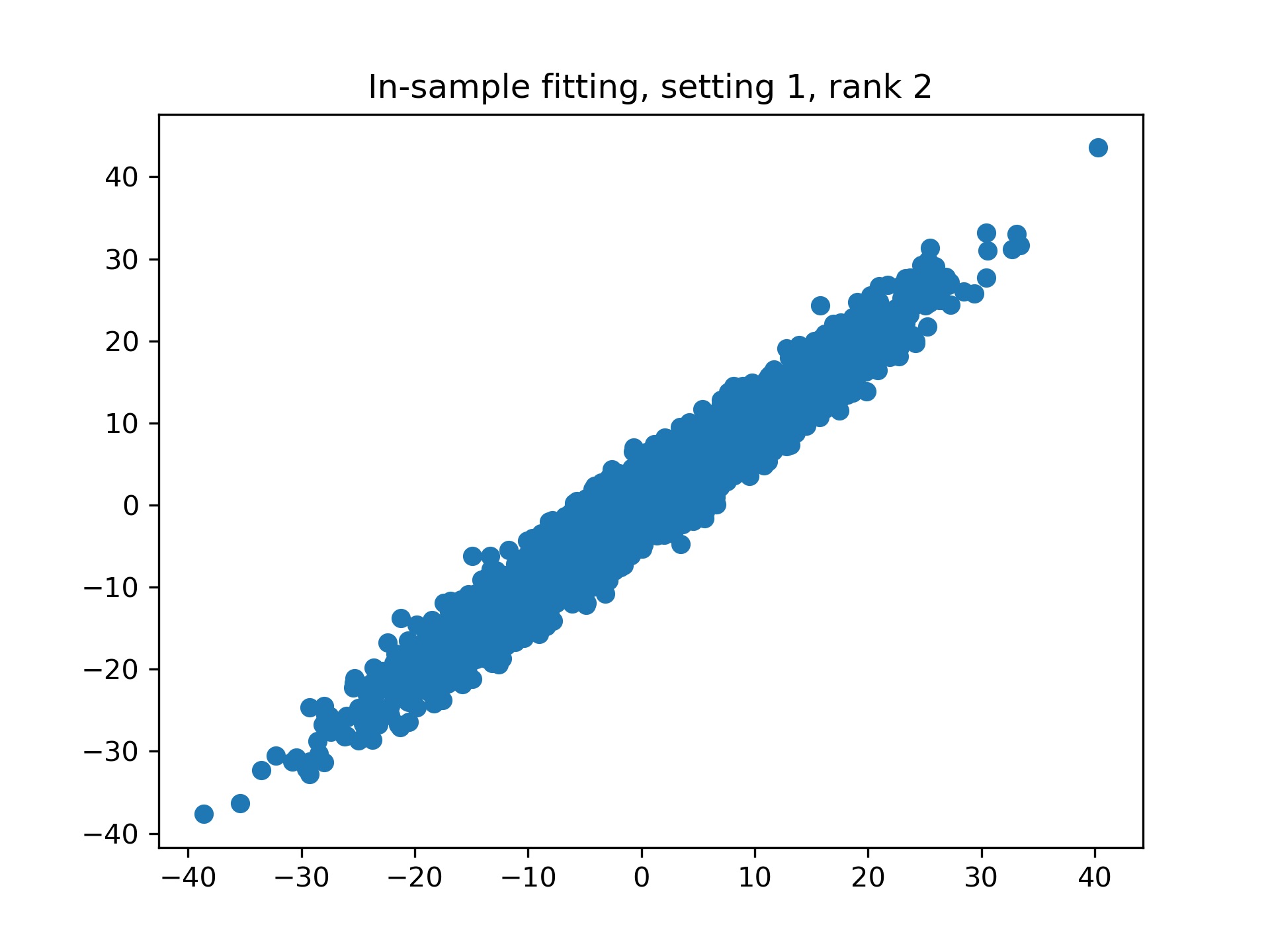} & \includegraphics[width=.31\linewidth]{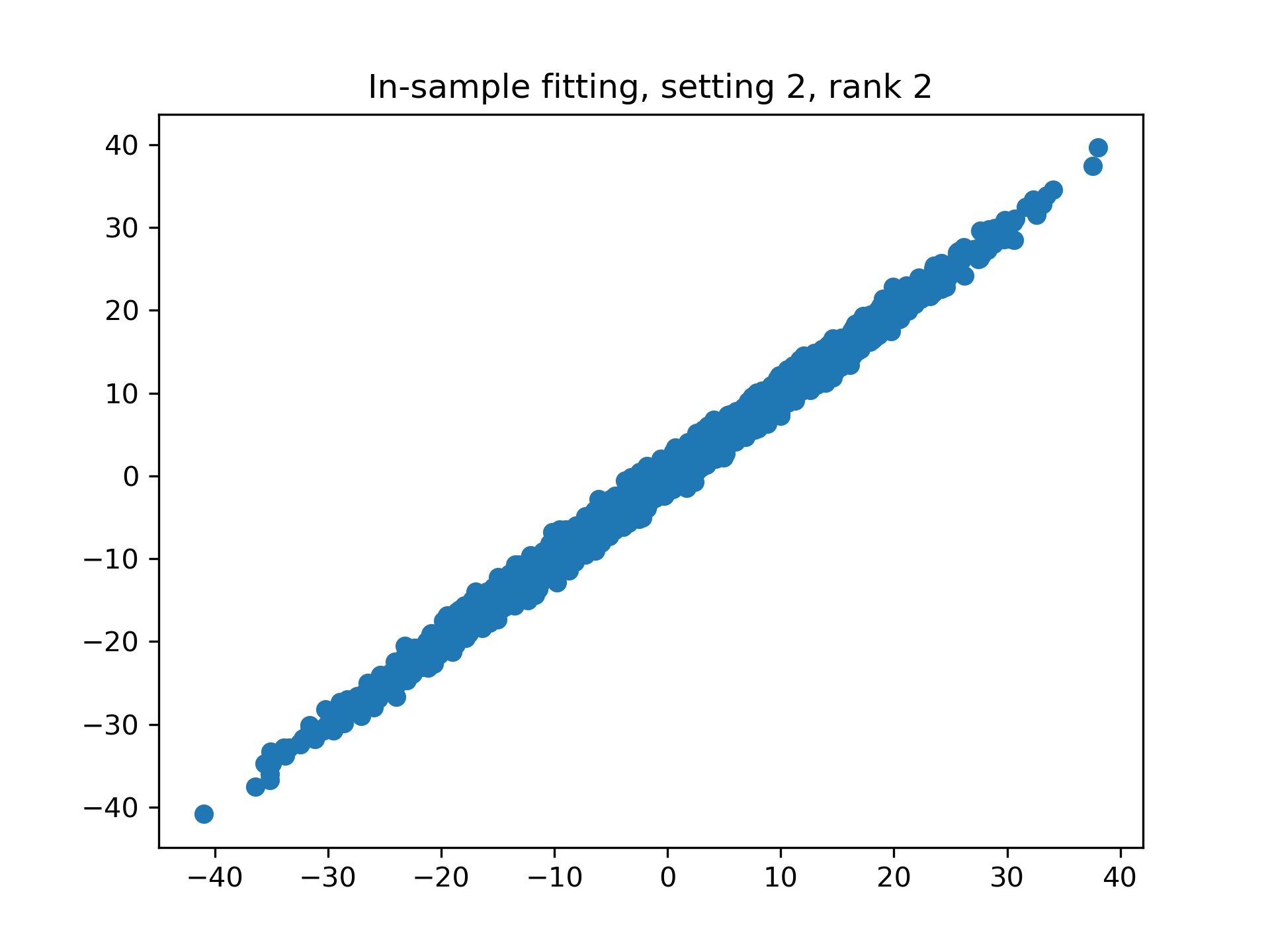} & \includegraphics[width=.31\linewidth]{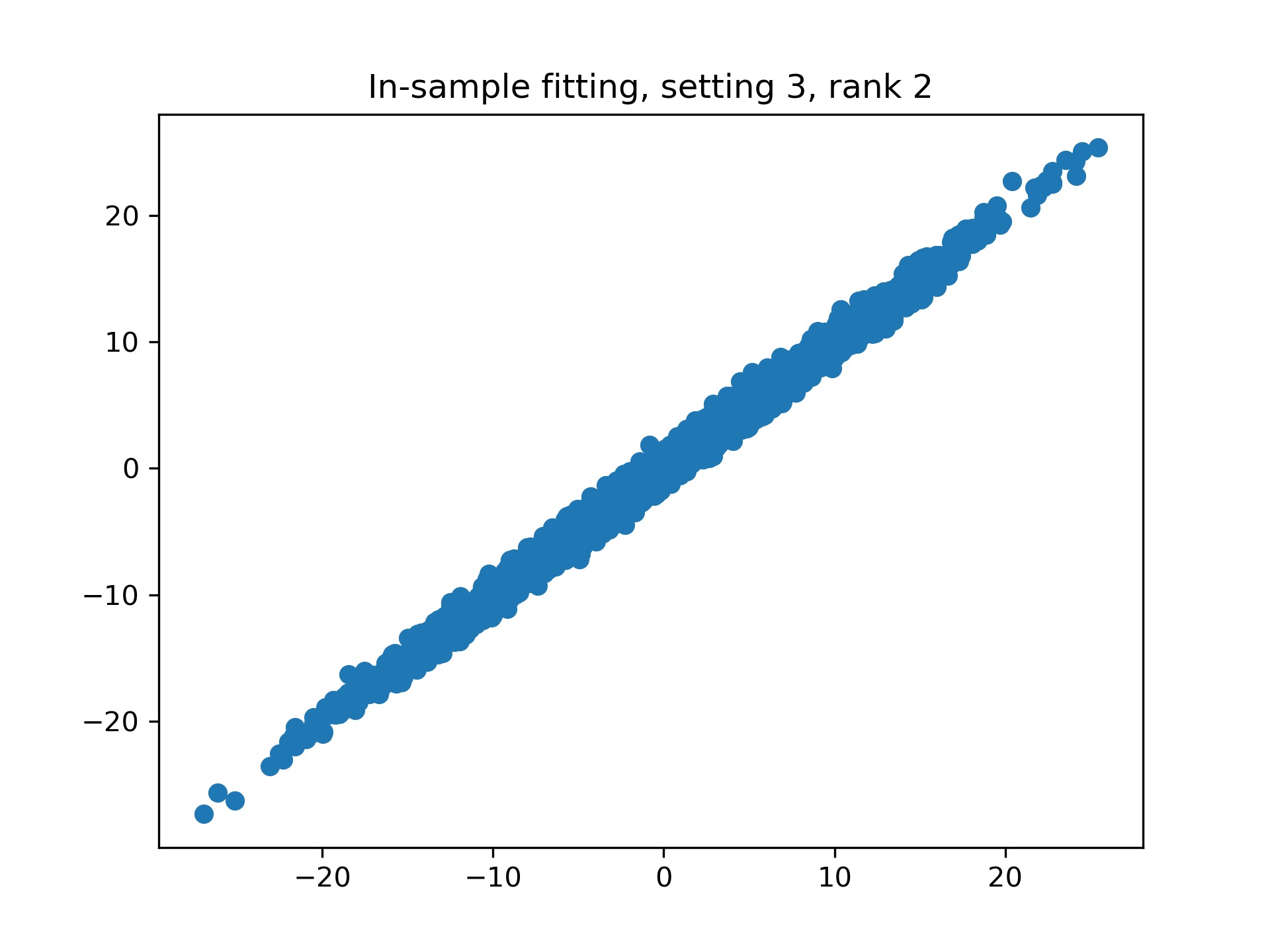}
    \end{tabular}
    \caption{Simulation results for Bayesian tensor regression. First row: true coefficients. Second row: estimated coefficients. Third and fourth row: trace plots of $\sigma^2$ and $\mu$, true values are the red dashed lines. Fifth row: scatter plots for in-sample fitting, true values (horizontal axis) versus fitted values (vertical axis).}
    \label{fig: sim_parafac}
\end{figure}

Parameter estimation is based on the first 1,000 observations, and out-of-sample forecasts are generated for the remaining 500 samples. The following hyper-parameter setting is considered: $D=5, \alpha=D^{-2}, a_{\tau}=3, b_{\tau}=100, a_{\lambda}=20, b_{\lambda}=2, a_{\sigma}=3, b_{\sigma}=1, \sigma^2_{\mu}=1$. We ran the Gibbs sampler for $1,000$ iterations and removed $200$ burn-in samples.

\subsection{Non-structured coefficients}
To explore the effects of random projection on tensor coefficients without underlying structure, unlike the settings in previous simulations, we carry out further simulations with the true coefficients, where the entries are i.i.d. drawn from $\{0, 1\}$ at sparsity levels of $75\%, 50\%$, and $25\%$.

Fig. \ref{fig: non-stru} shows the scatter plots of predicted data against the actual data across different random projection methods for coefficients with different sparsity levels using compression rate $=0.36$, training sample size $=1000$, and $\psi=3$. When the sparsity level of the true coefficients is moderate ($25\%$ and $50\%$), mode-wise random projection and mode-wise random projection with mode preservation still outperform the tensor-wise random projection as in the case of coefficients with some underlying structures. However, when the true coefficients become highly sparse ($75\%$), the performance of the different random projections becomes very close, with tensor-wise random projection slightly outperforming mode-wise random projection. This can also be seen in Fig. \ref{fig: rmse_no_stru}, which shows the RMSE across different random projection methods for different true coefficients.

\begin{figure}
    \centering
    \includegraphics[width=0.5\linewidth]{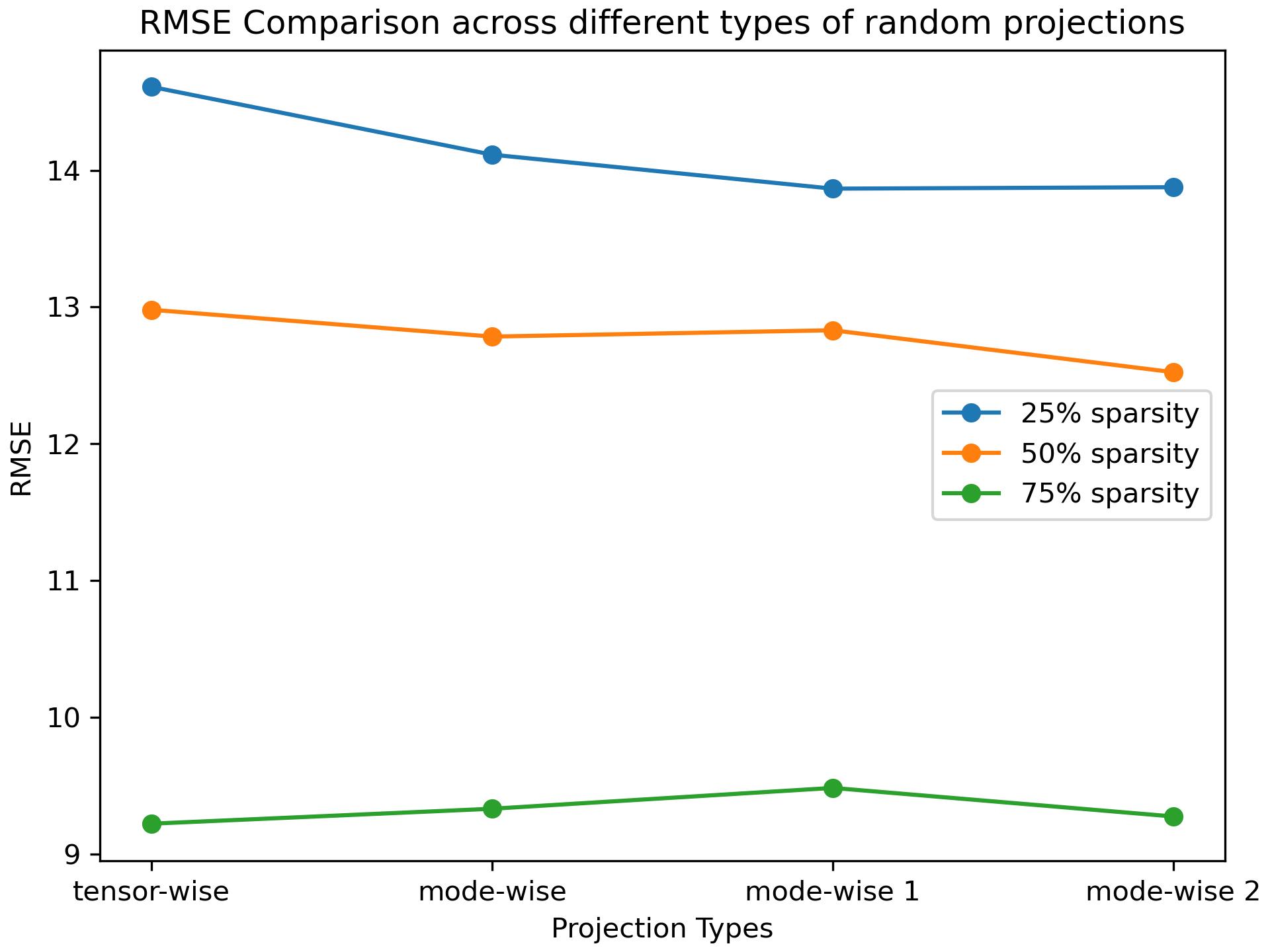}
    \caption{RMSE (vertical axis) comparison across different random projection methods (horizontal axis) and coefficients with different sparsity levels shown as lines in different colors (25\%: blue, 50\%: yellow, 75\%: green).}
    \label{fig: rmse_no_stru}
\end{figure}

\begin{figure}[]
    \centering
    \begin{tabular}{ccccc}
    \scriptsize (a) True Coefficient & \multicolumn{4}{c}{\scriptsize (b) Model Fitting}\\
    &\scriptsize Tensor-wise & \scriptsize Mode-wise & \scriptsize Mode-wise $1$ & \scriptsize Mode-wise $2$ \\
        \raisebox{2mm}{\includegraphics[width=.16\linewidth]{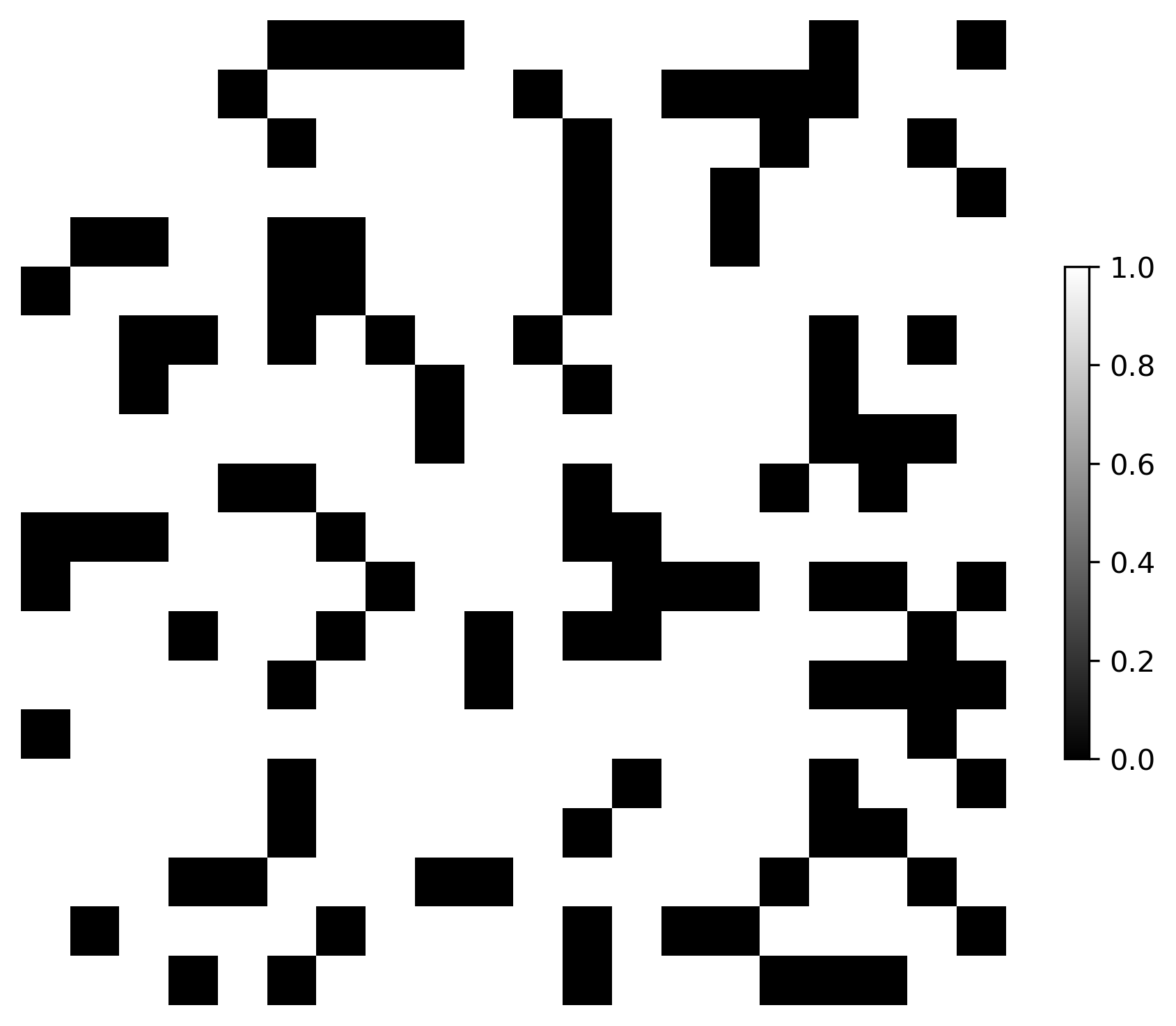}} & \includegraphics[width=.16\linewidth]{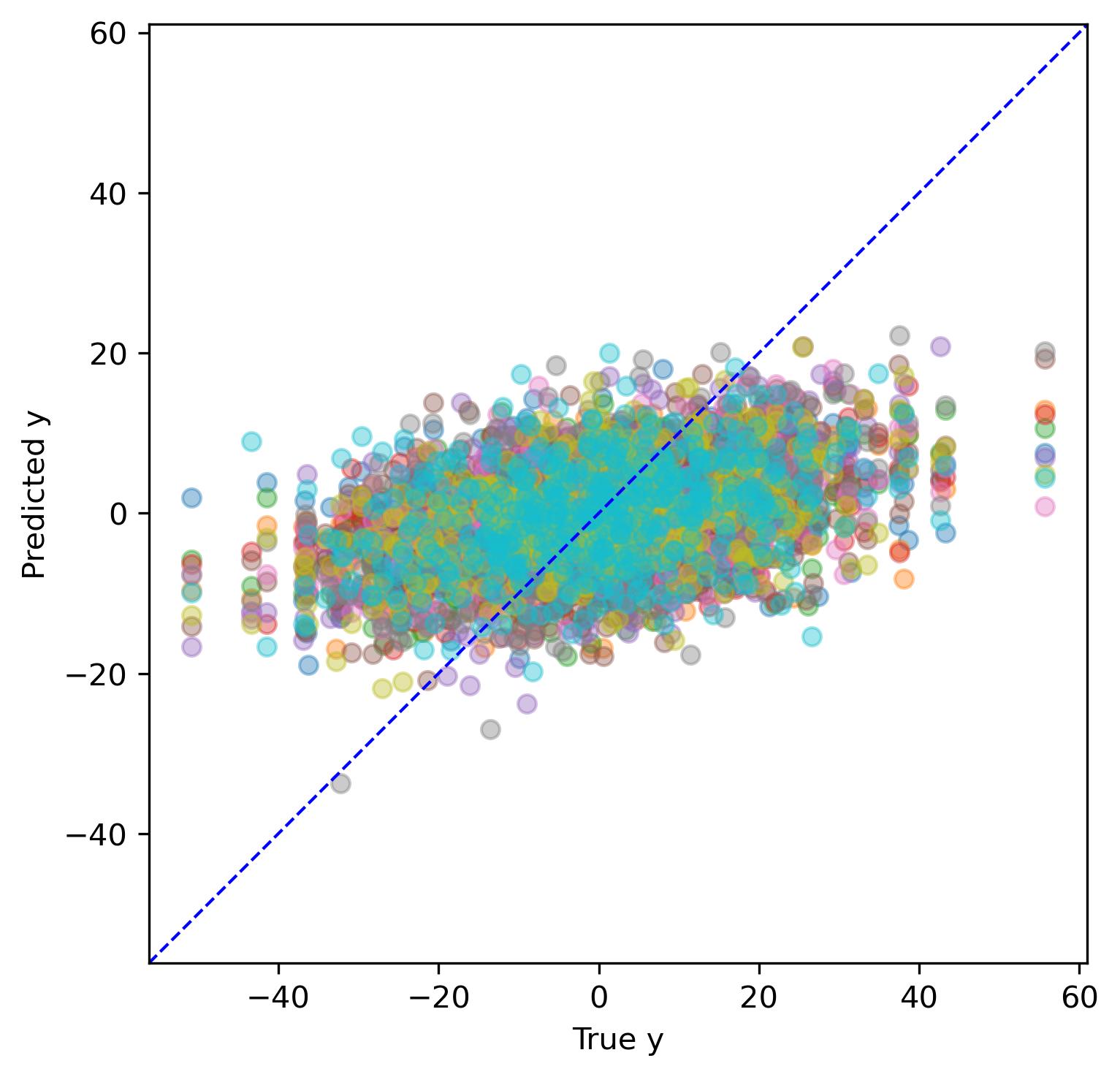} & \includegraphics[width=.16\linewidth]{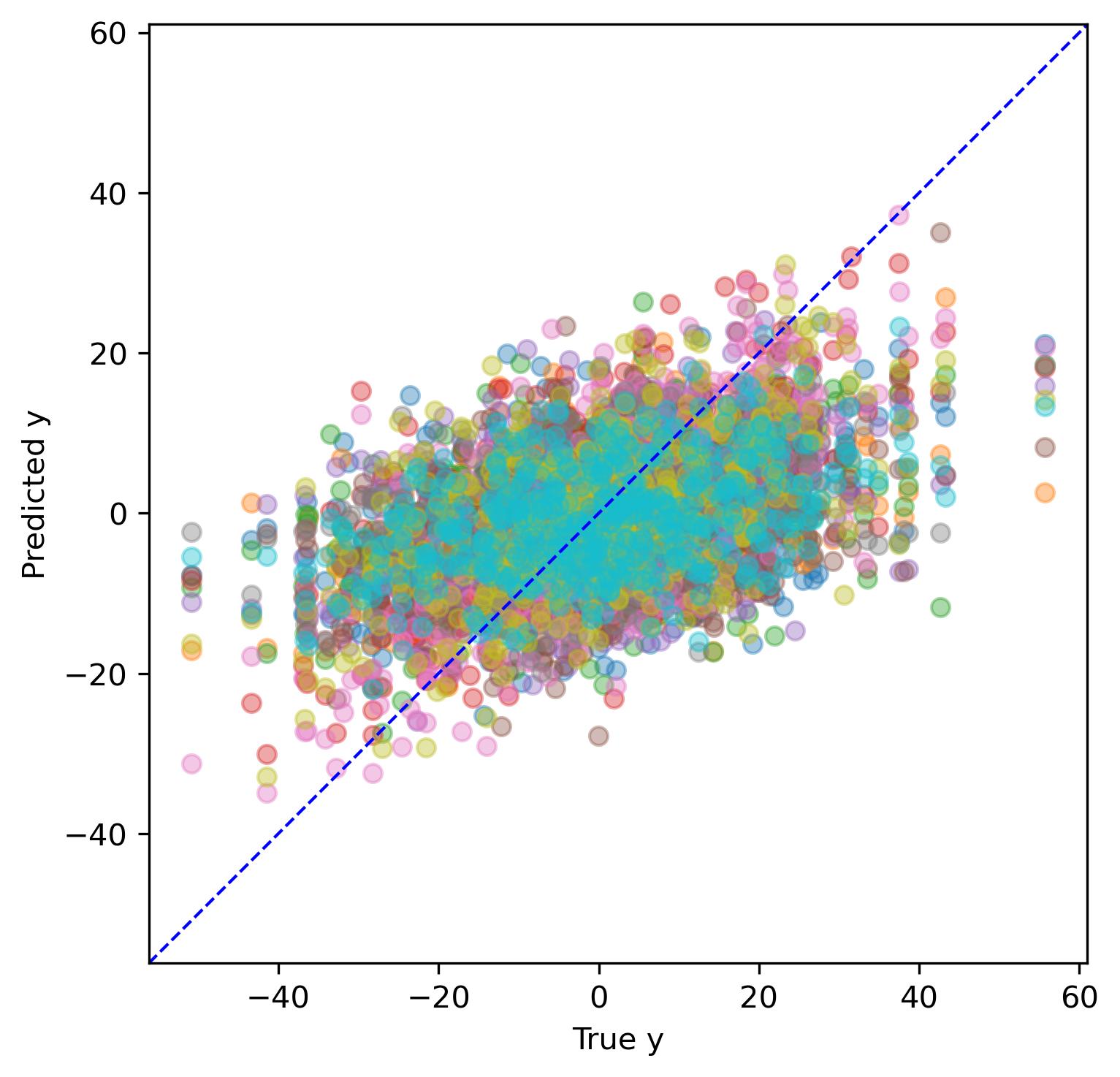} & \includegraphics[width=.16\linewidth]{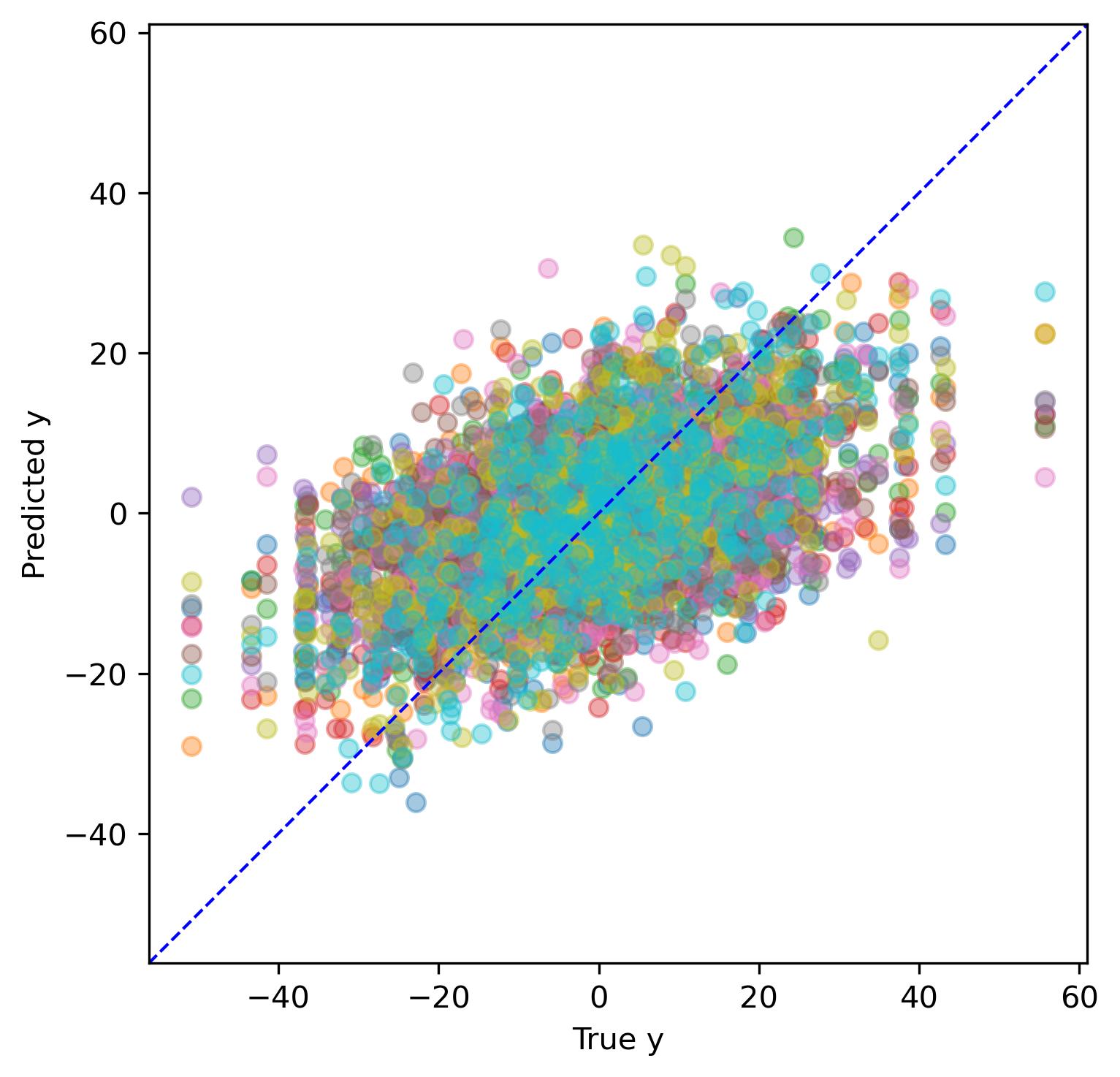} & \includegraphics[width=.16\linewidth]{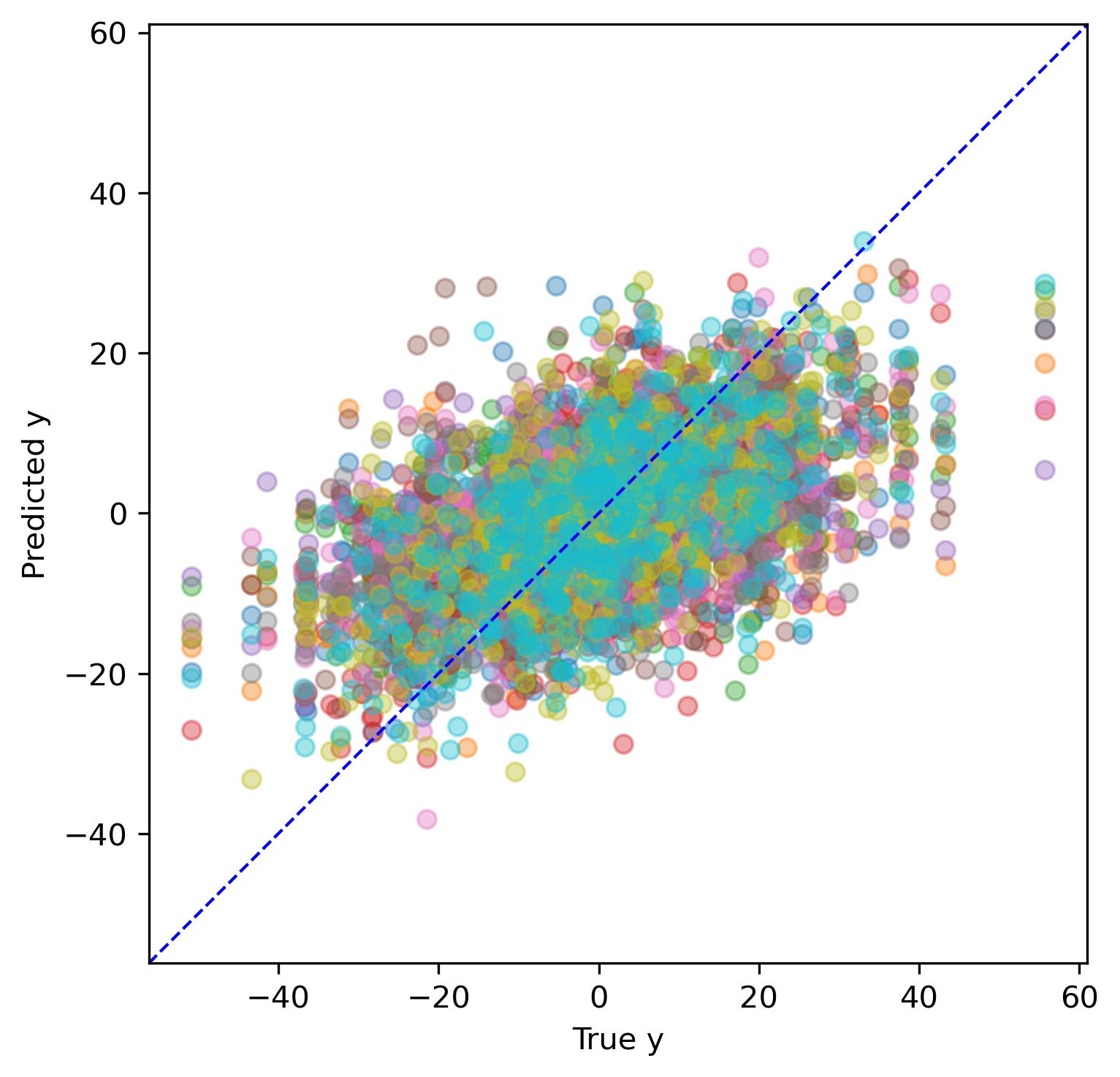}\\
        \raisebox{2mm}{\includegraphics[width=.16\linewidth]{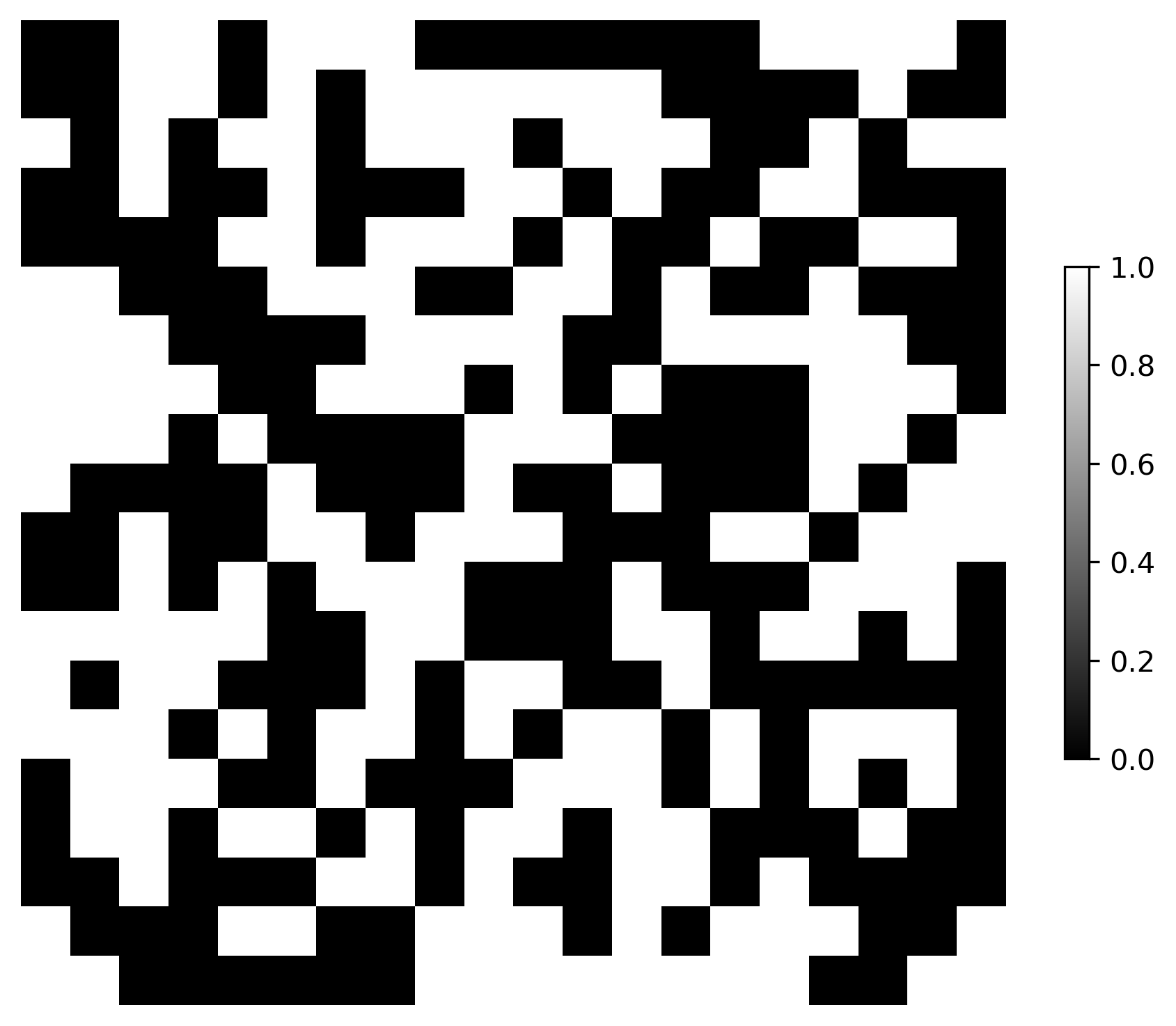}} & \includegraphics[width=.16\linewidth]{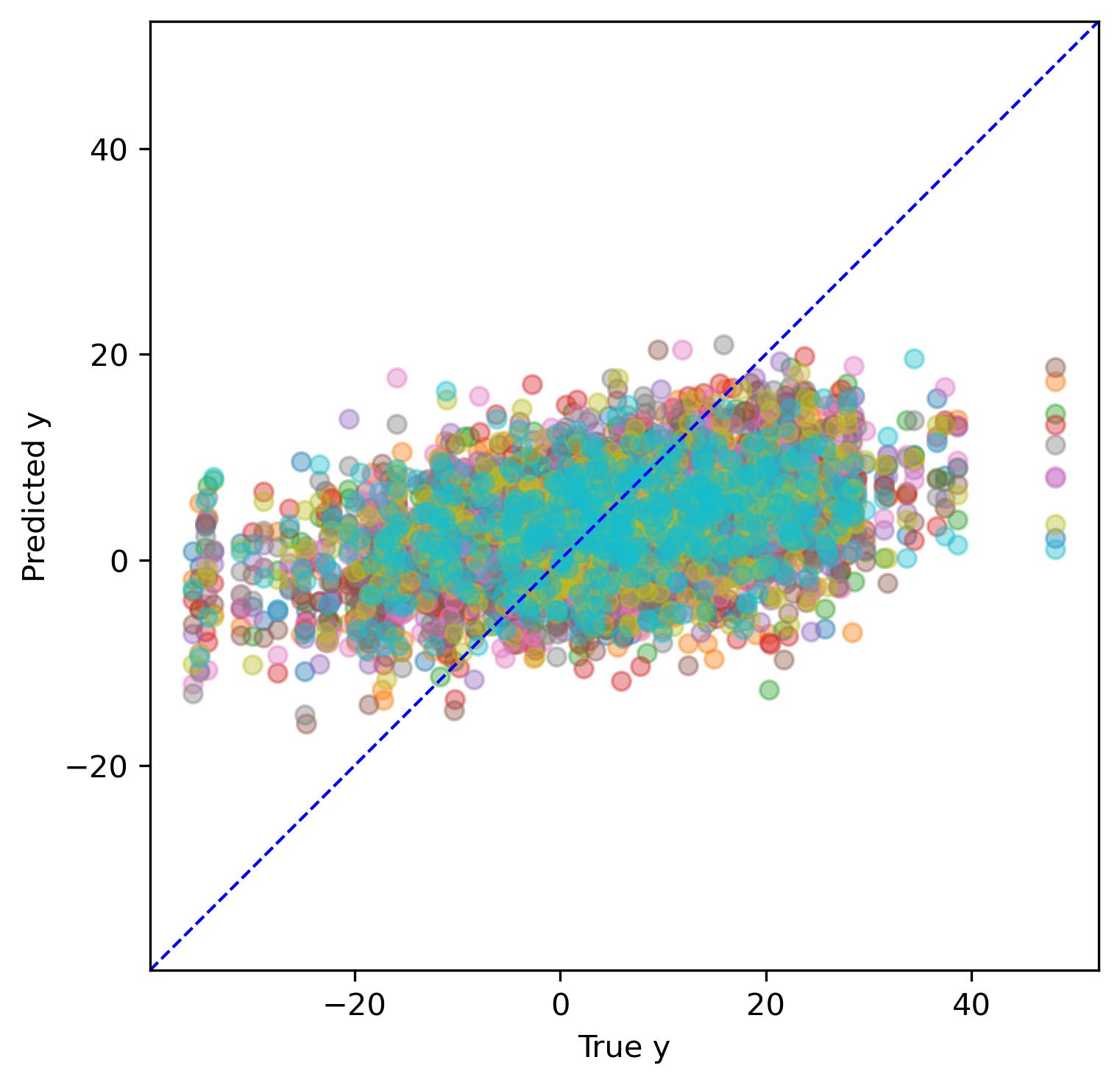} & \includegraphics[width=.16\linewidth]{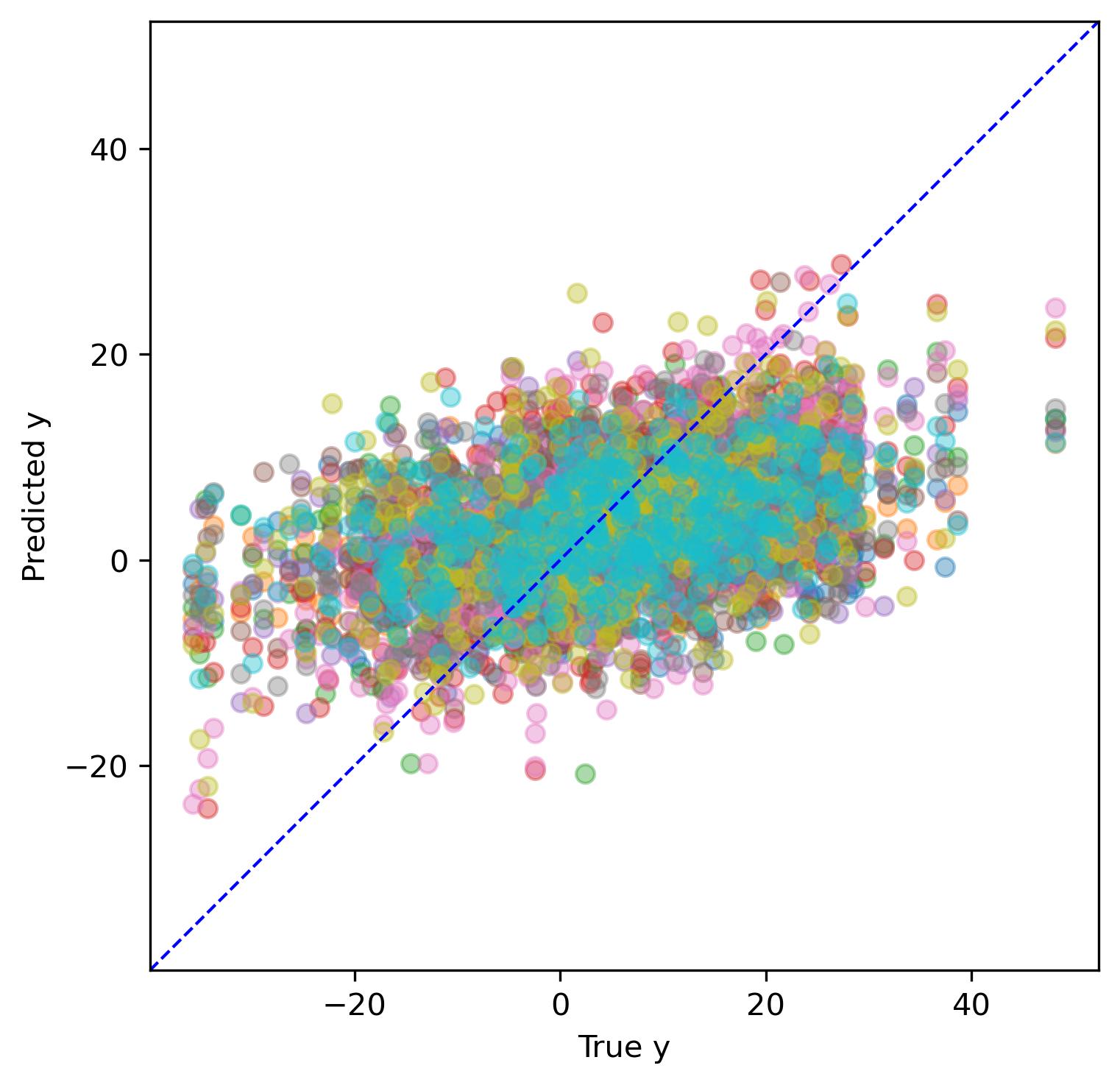} & \includegraphics[width=.16\linewidth]{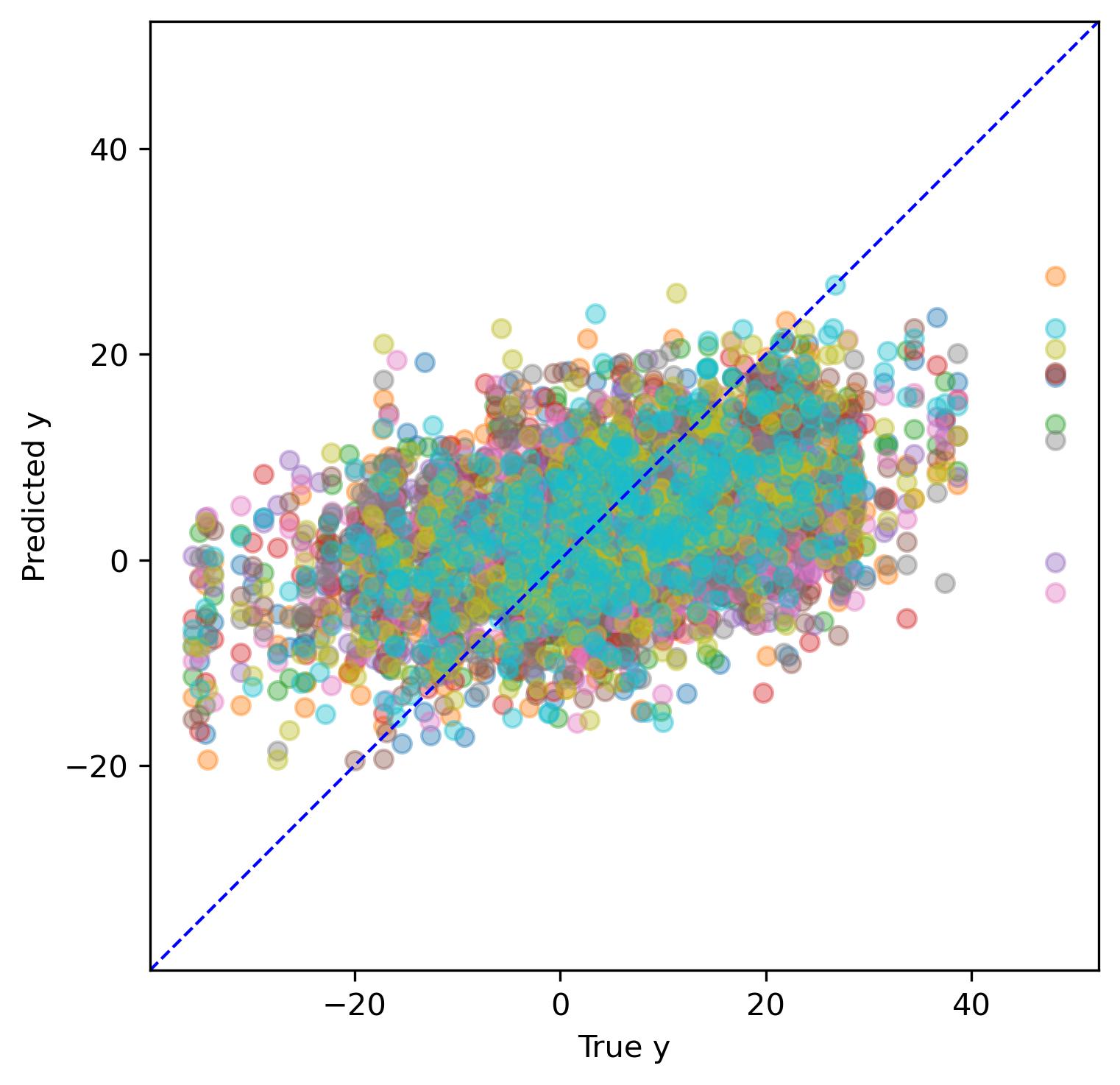} & \includegraphics[width=.16\linewidth]{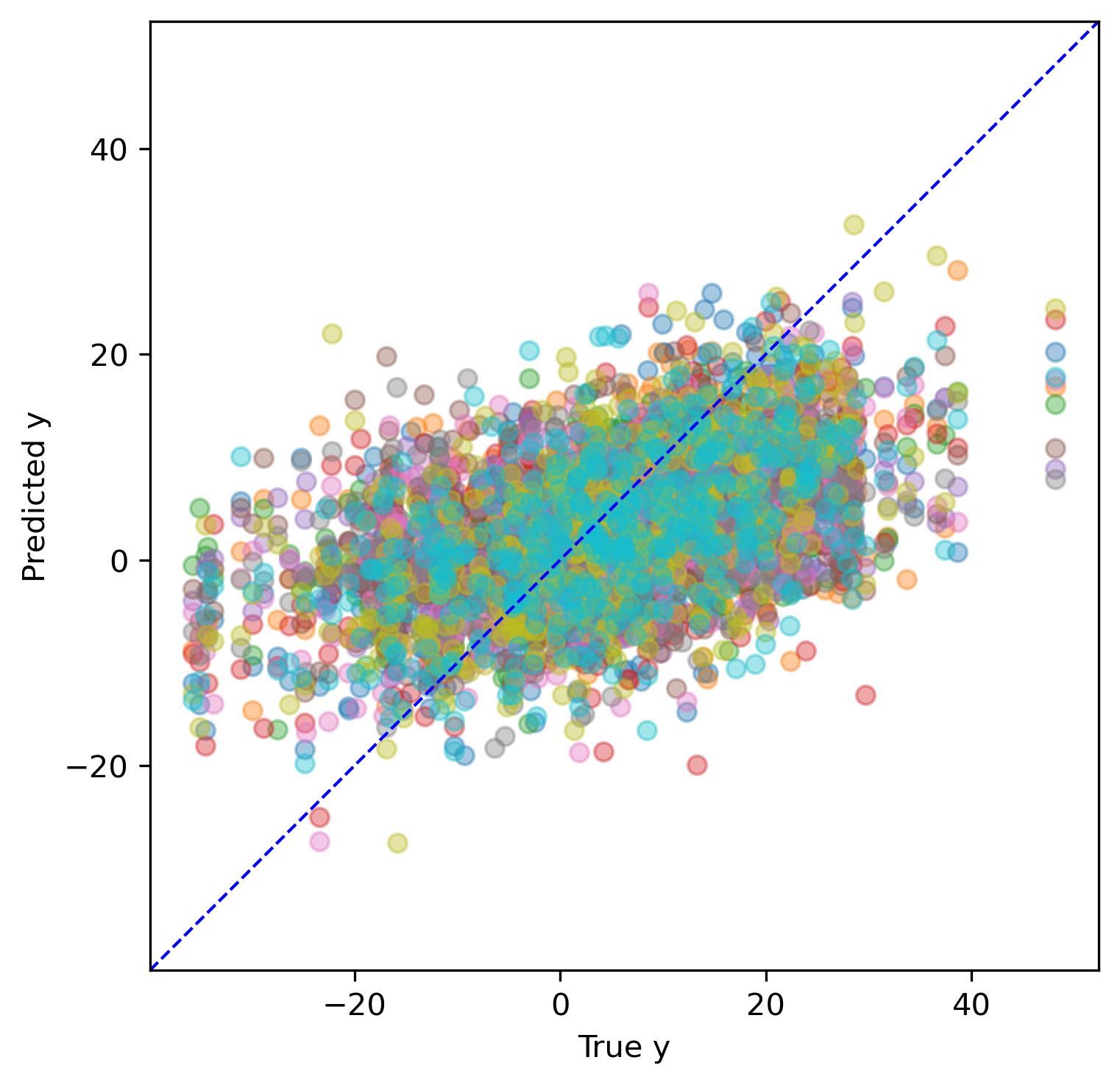} \\
        \raisebox{2mm}{\includegraphics[width=.16\linewidth]{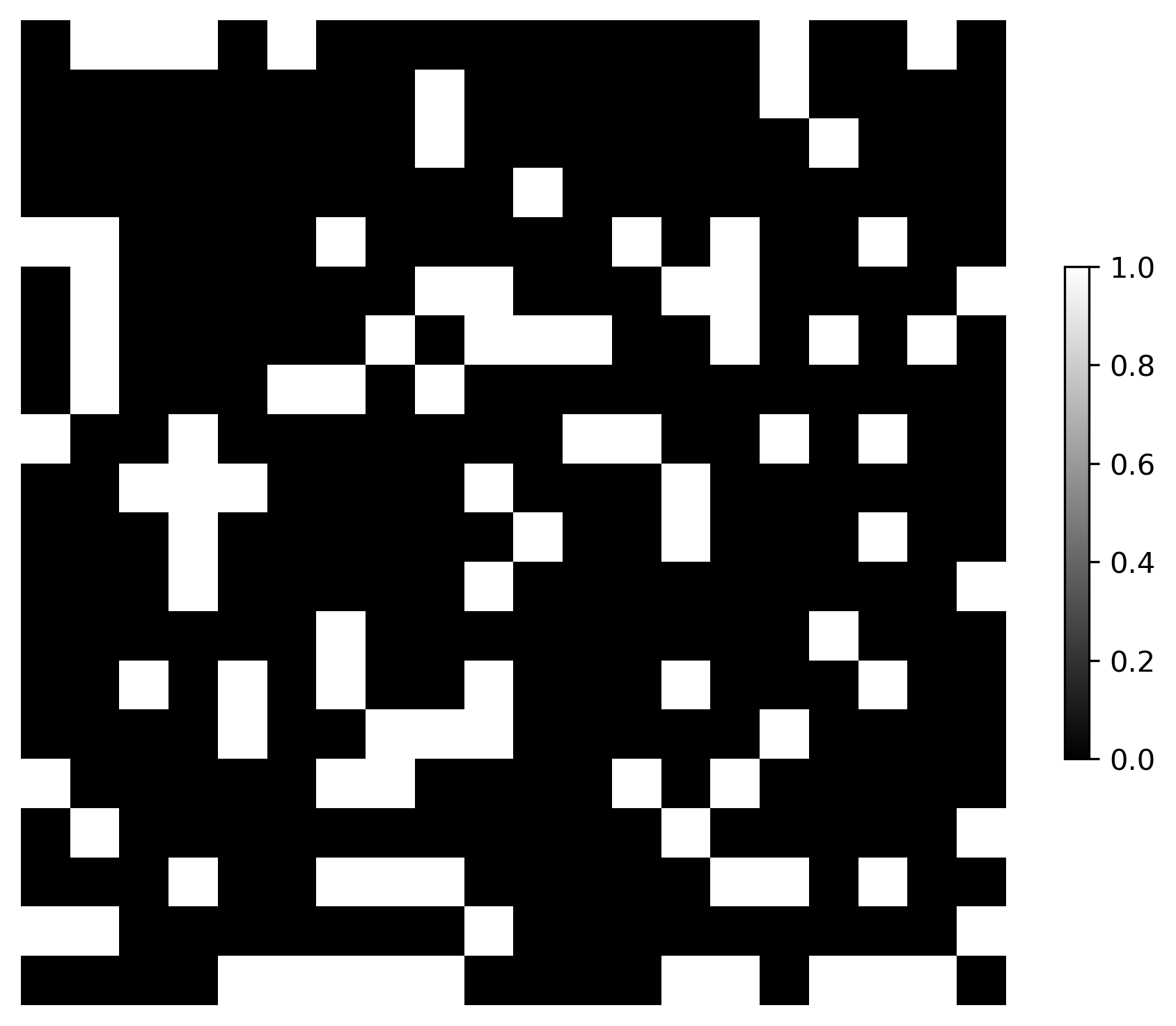}} & \includegraphics[width=.16\linewidth]{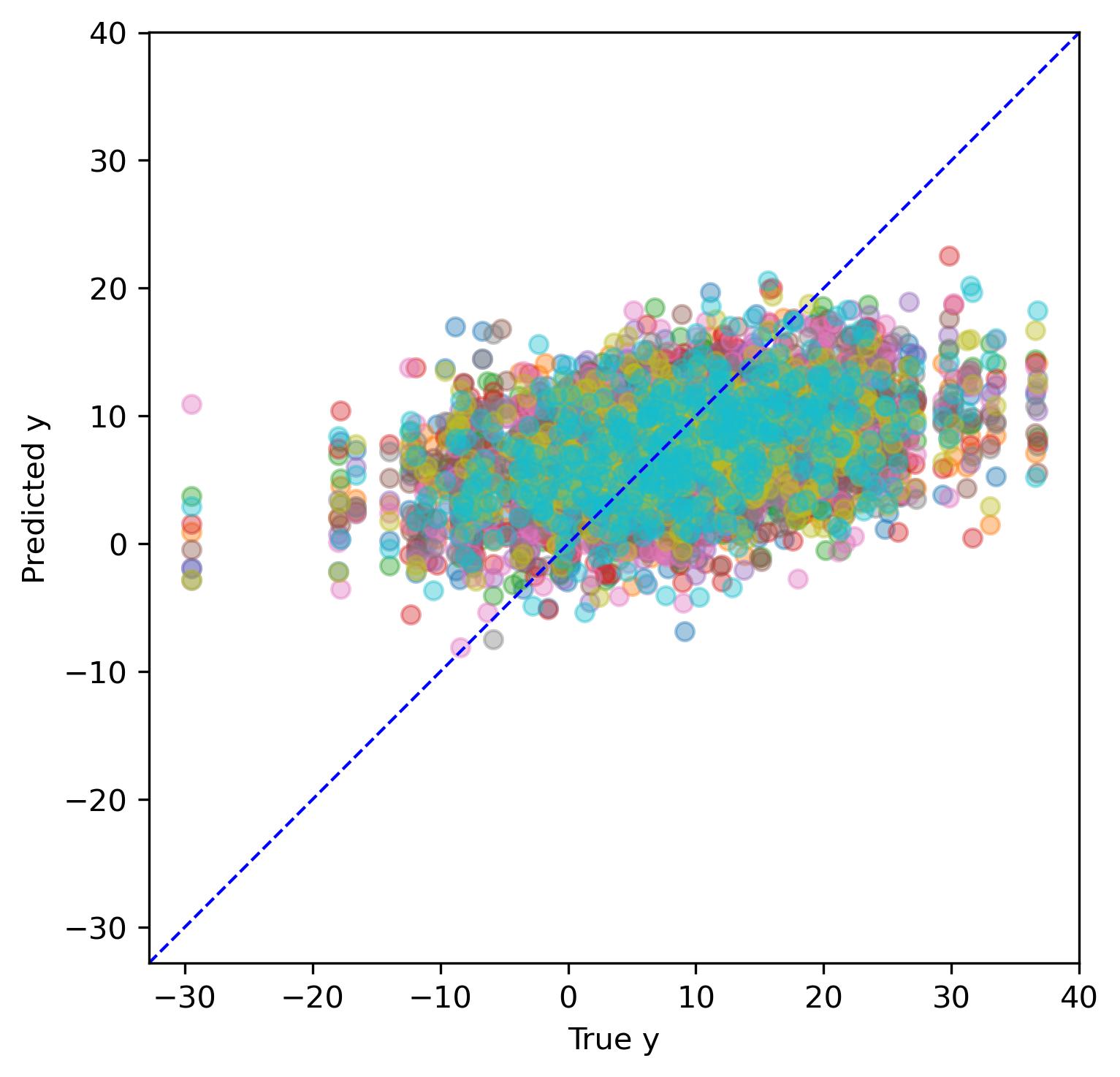} & \includegraphics[width=.16\linewidth]{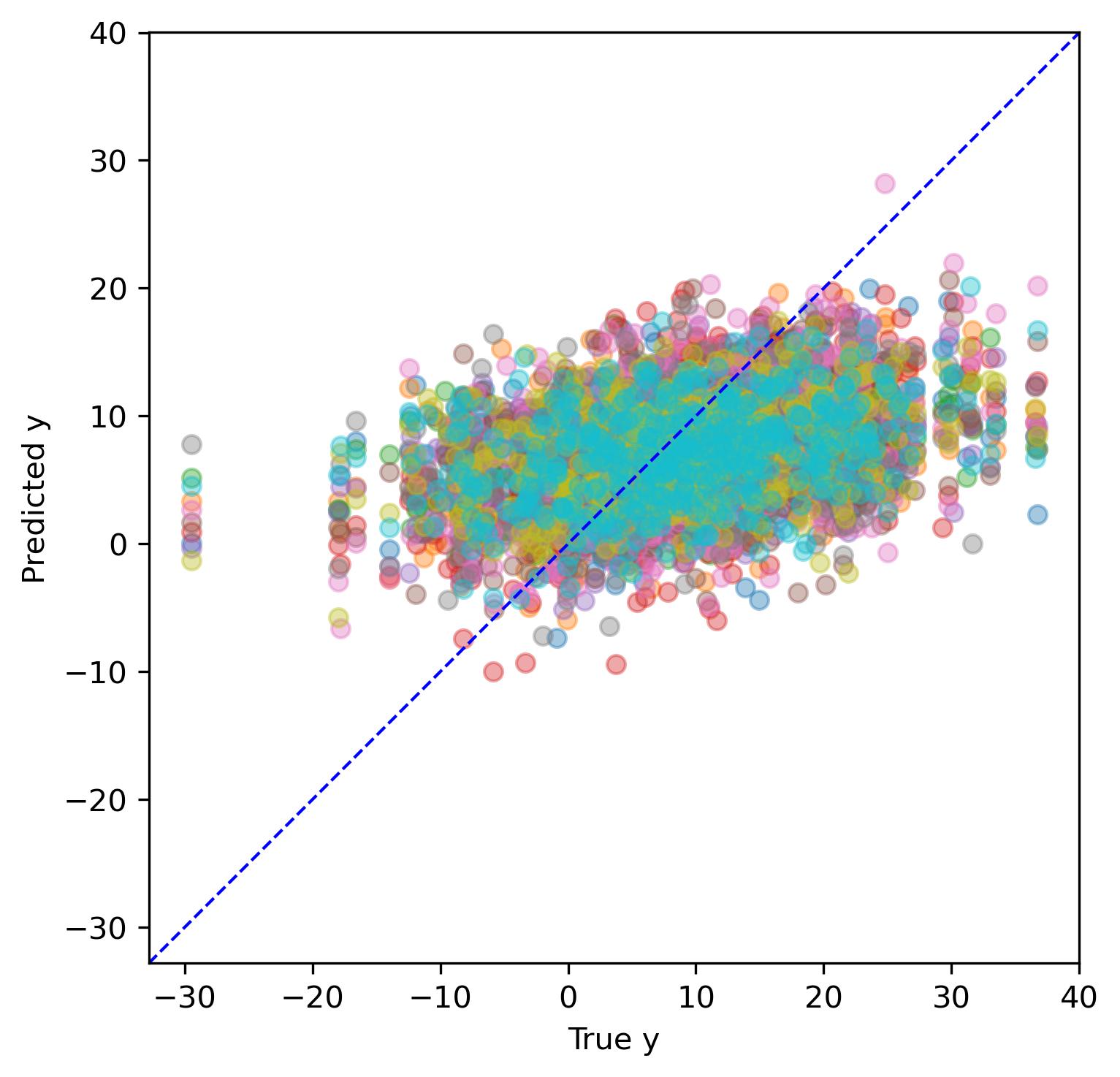} & \includegraphics[width=.16\linewidth]{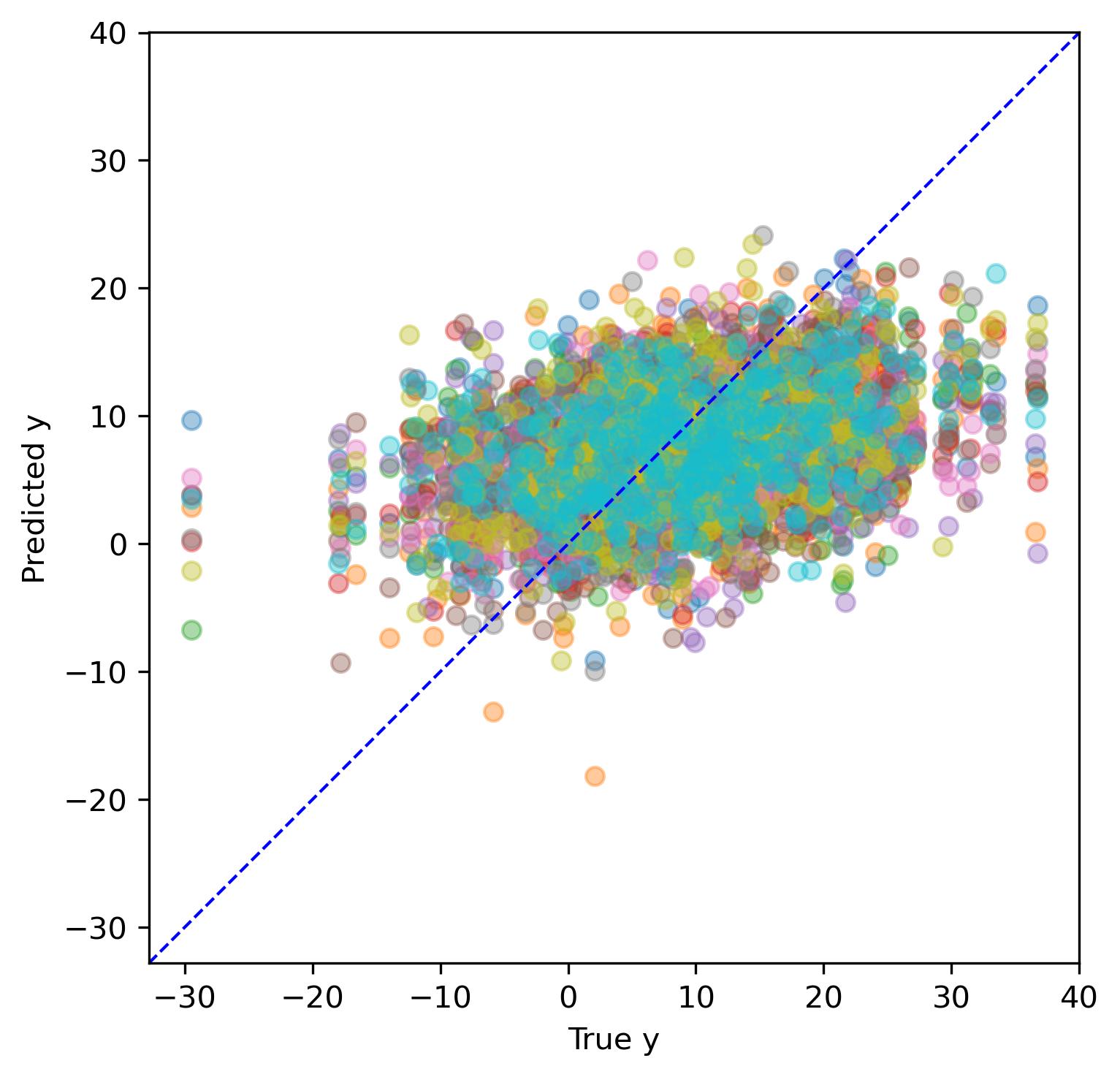} & \includegraphics[width=.16\linewidth]{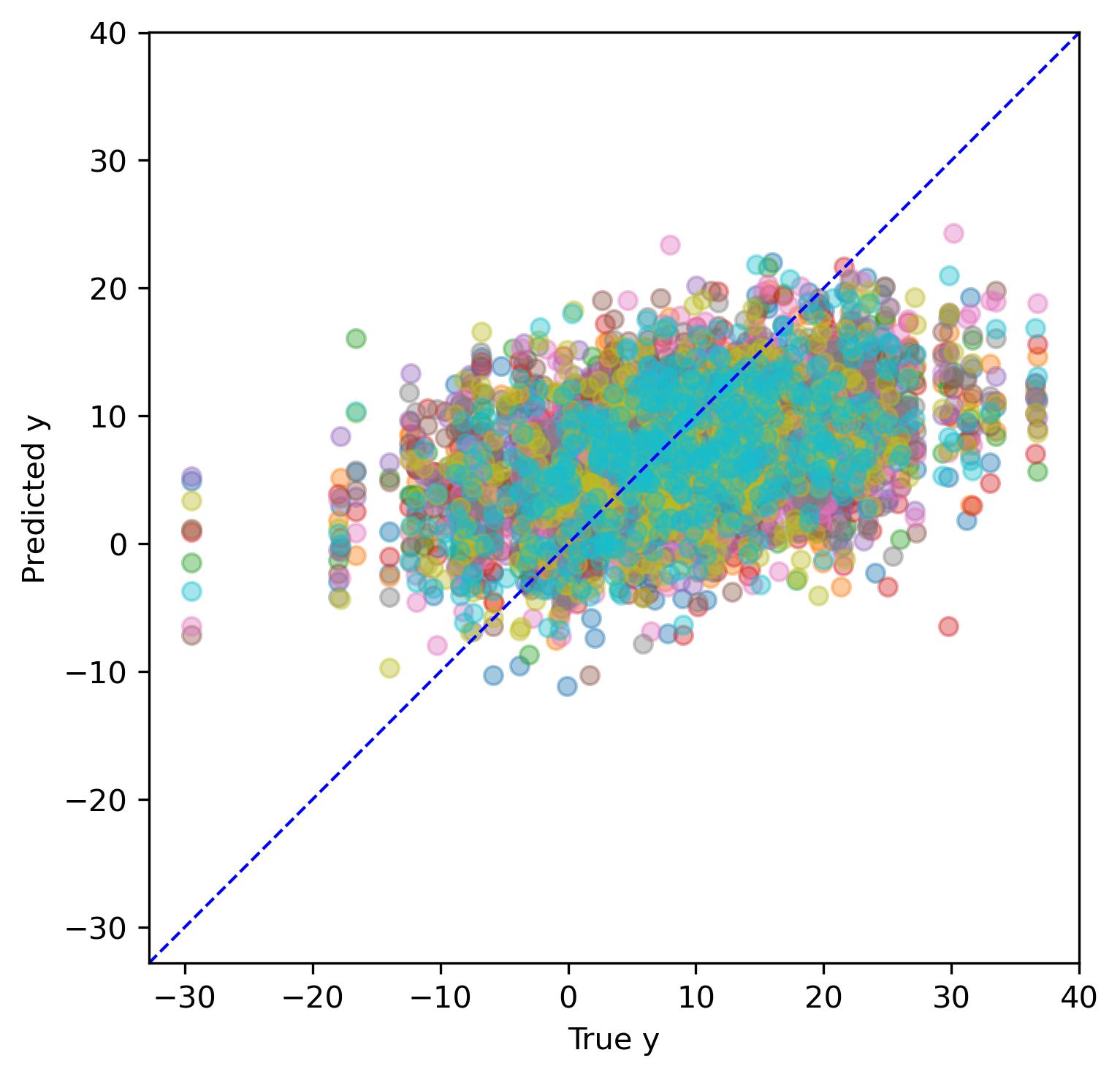}         
    
    \end{tabular}
    \caption{Scatter plots of actual data versus the predicted for three sets of coefficients with no underlying structures at sparsity levels of $25\%, 50\%$, and $75\%$. True coefficients are shown in panel (a) and forecasts are shown in panel (b). In each scatter plot: actual data (horizontal axis) against the predicted data (vertical axis) for different levels of sparsity (rows) and different types of random projections (columns), using $L=10$ independent projection matrices of the same type (colors) for each simulation. In each experiment: training sample size: $n=1000$, compression rate: $0.36$, $\psi=3$.}
    \label{fig: non-stru}
\end{figure}

\begin{figure}
    \centering
    \begin{tabular}{cccc}
   $n=500$ & $n=1000$ & $n=1500$ & $n=2000$\\
   \includegraphics[width=.23\linewidth]{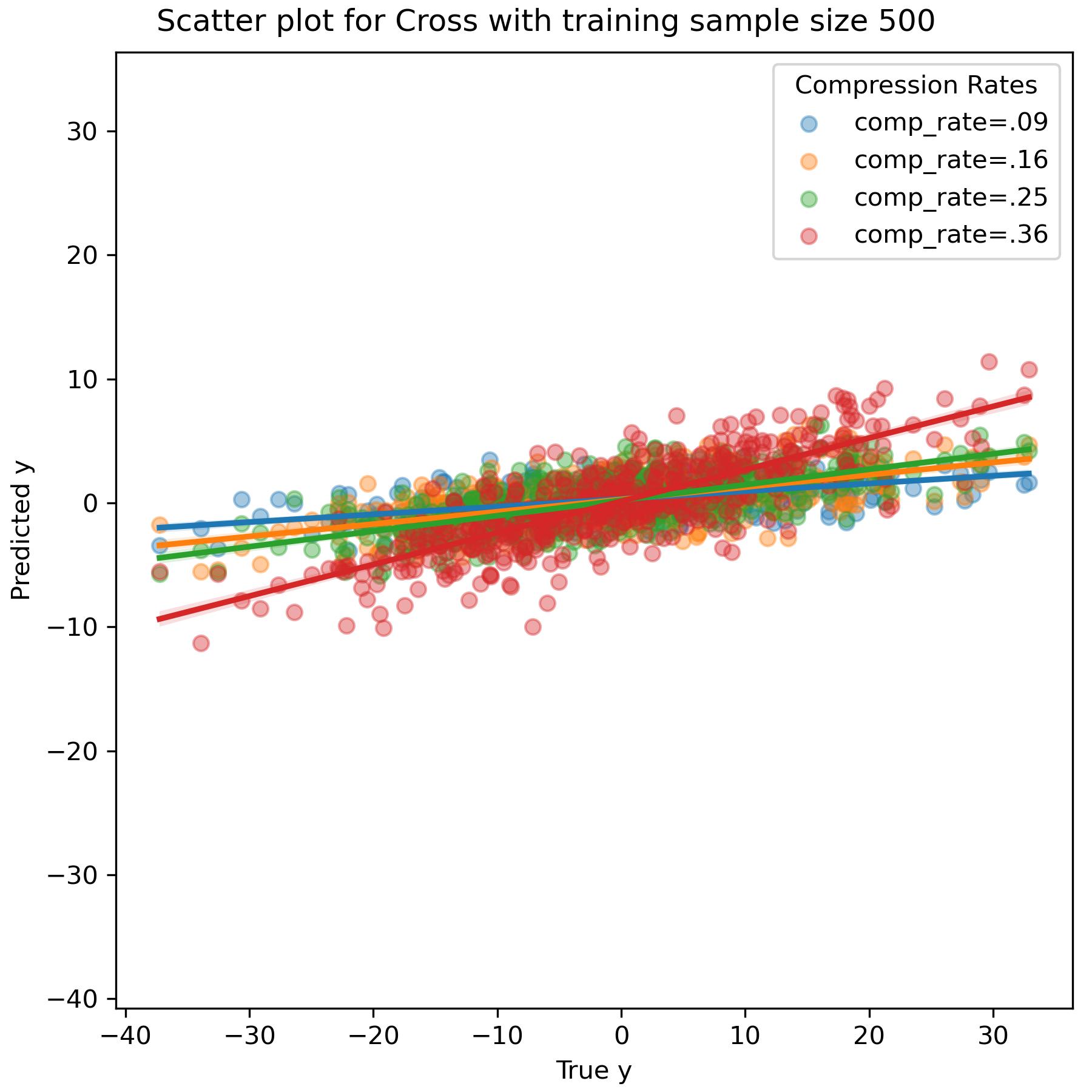} & \includegraphics[width=.23\linewidth]{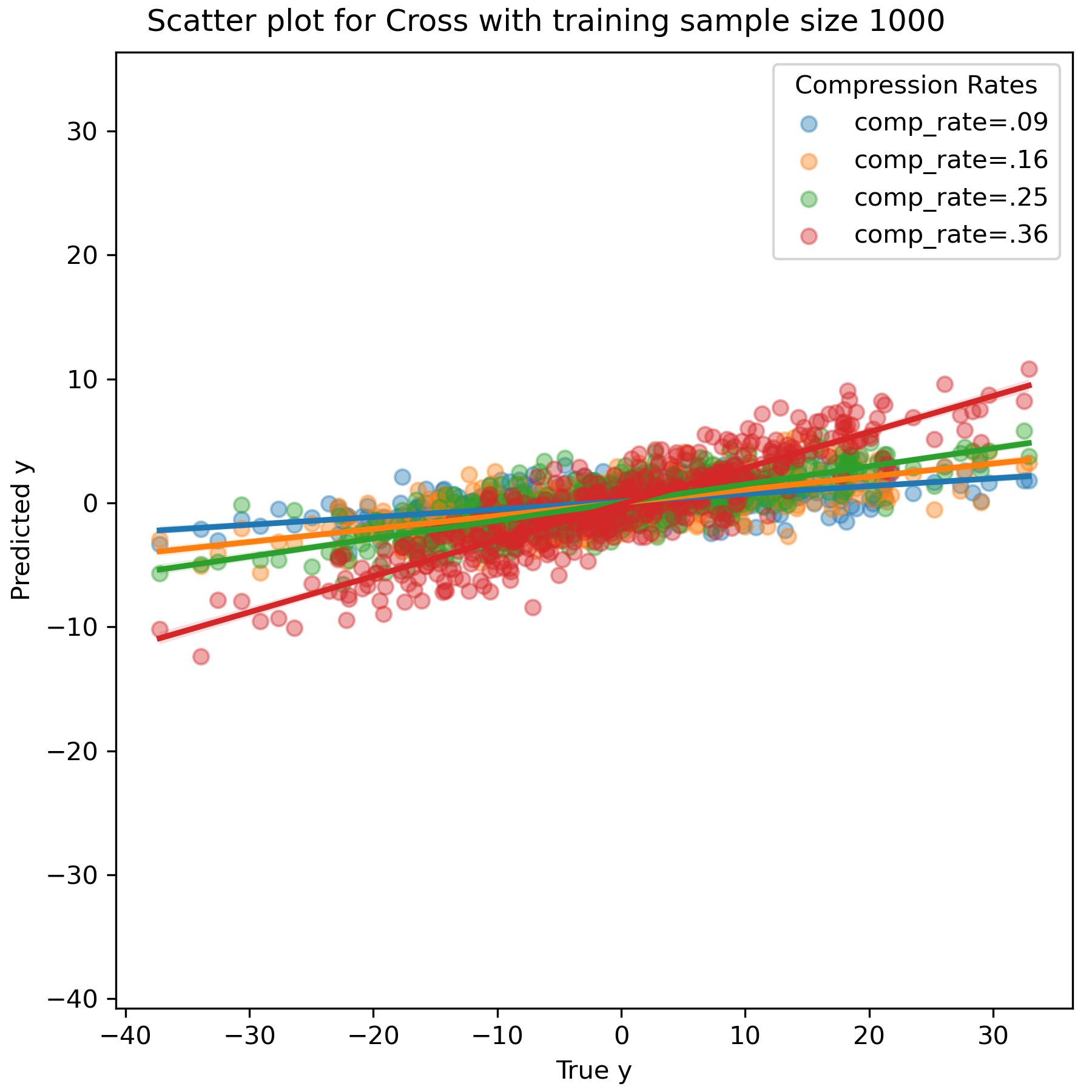} & \includegraphics[width=.23\linewidth]{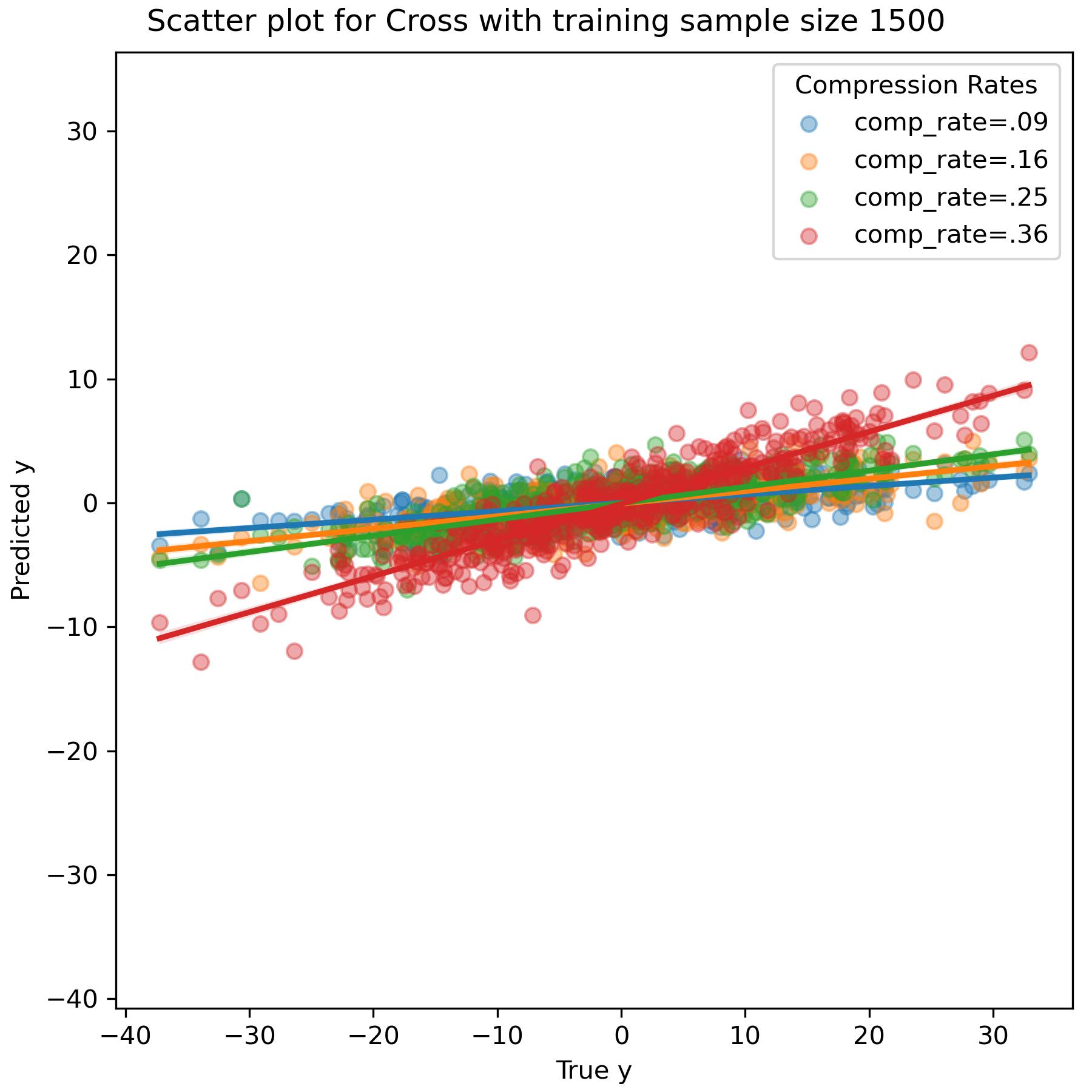} & \includegraphics[width=.23\textwidth]{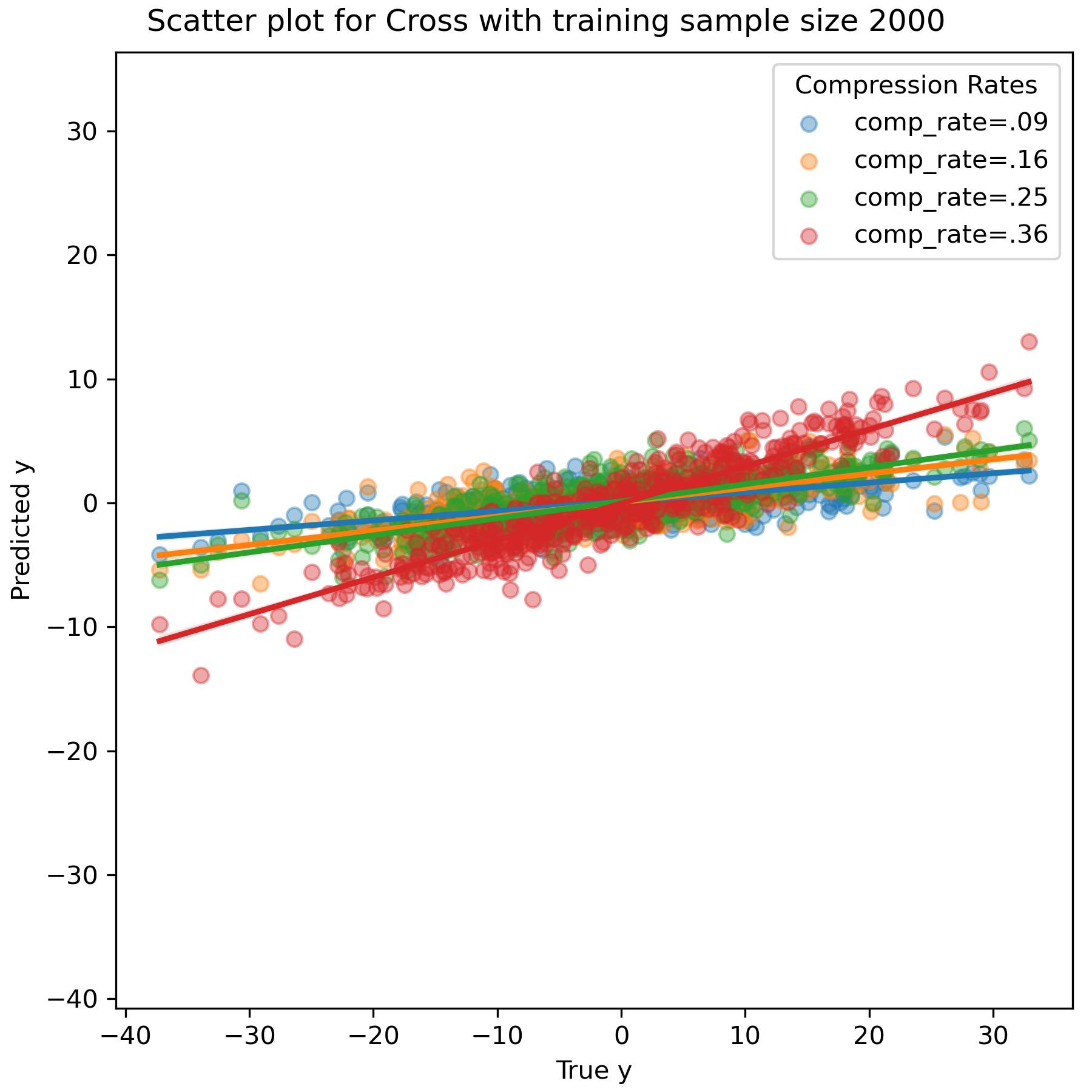} \\
   \includegraphics[width=.23\textwidth]{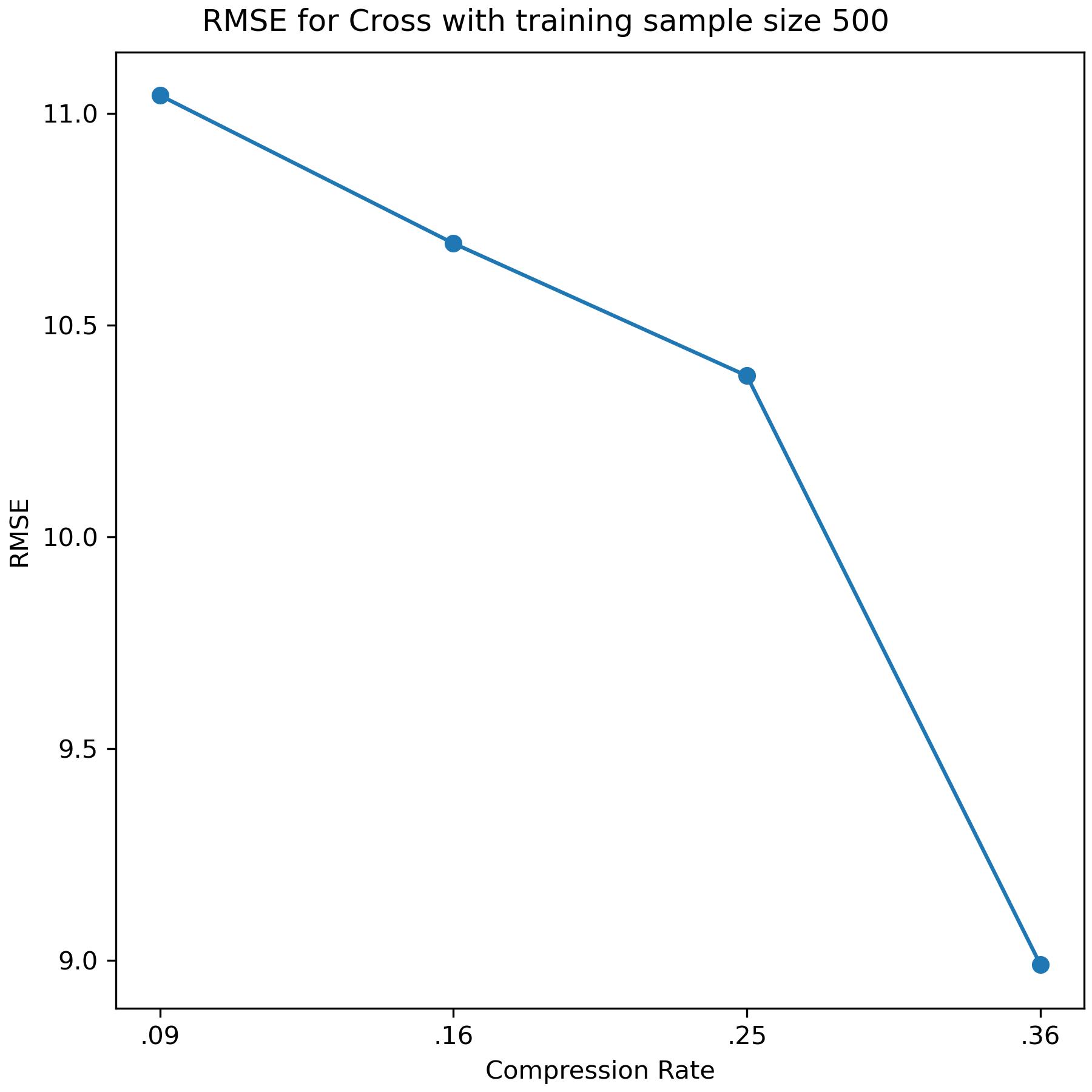} & \includegraphics[width=.23\textwidth]{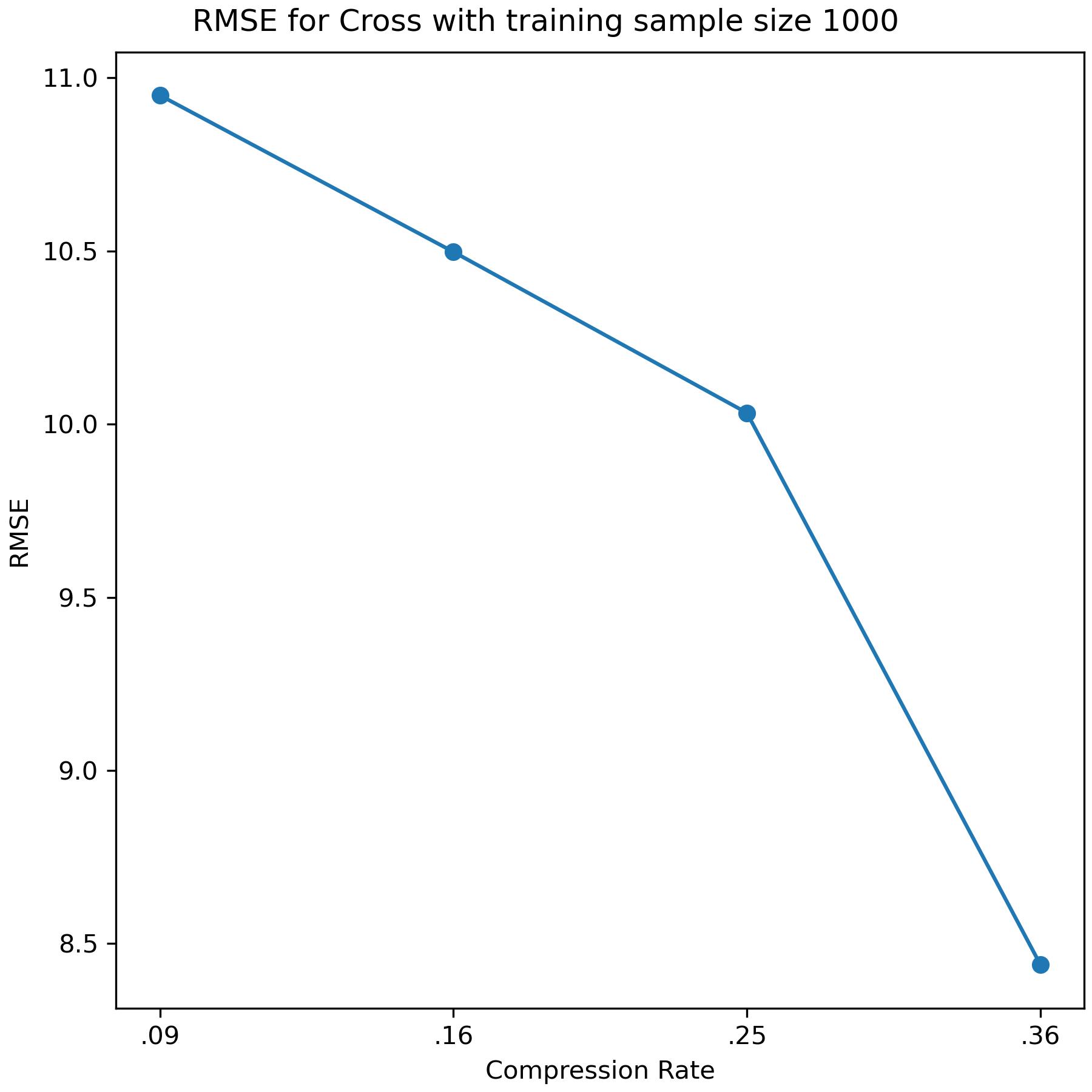} & \includegraphics[width=.23\textwidth]{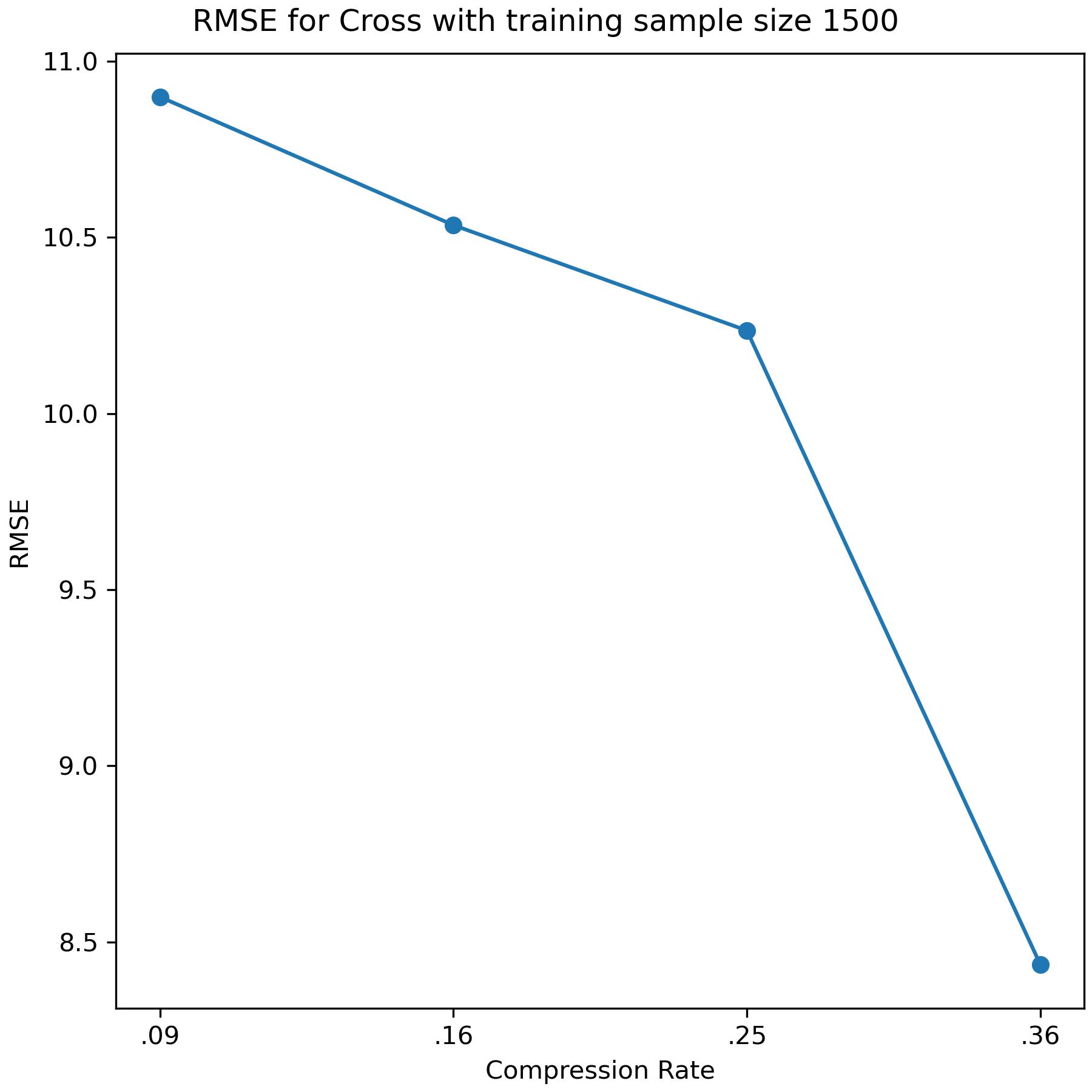} & \includegraphics[width=.23\textwidth]{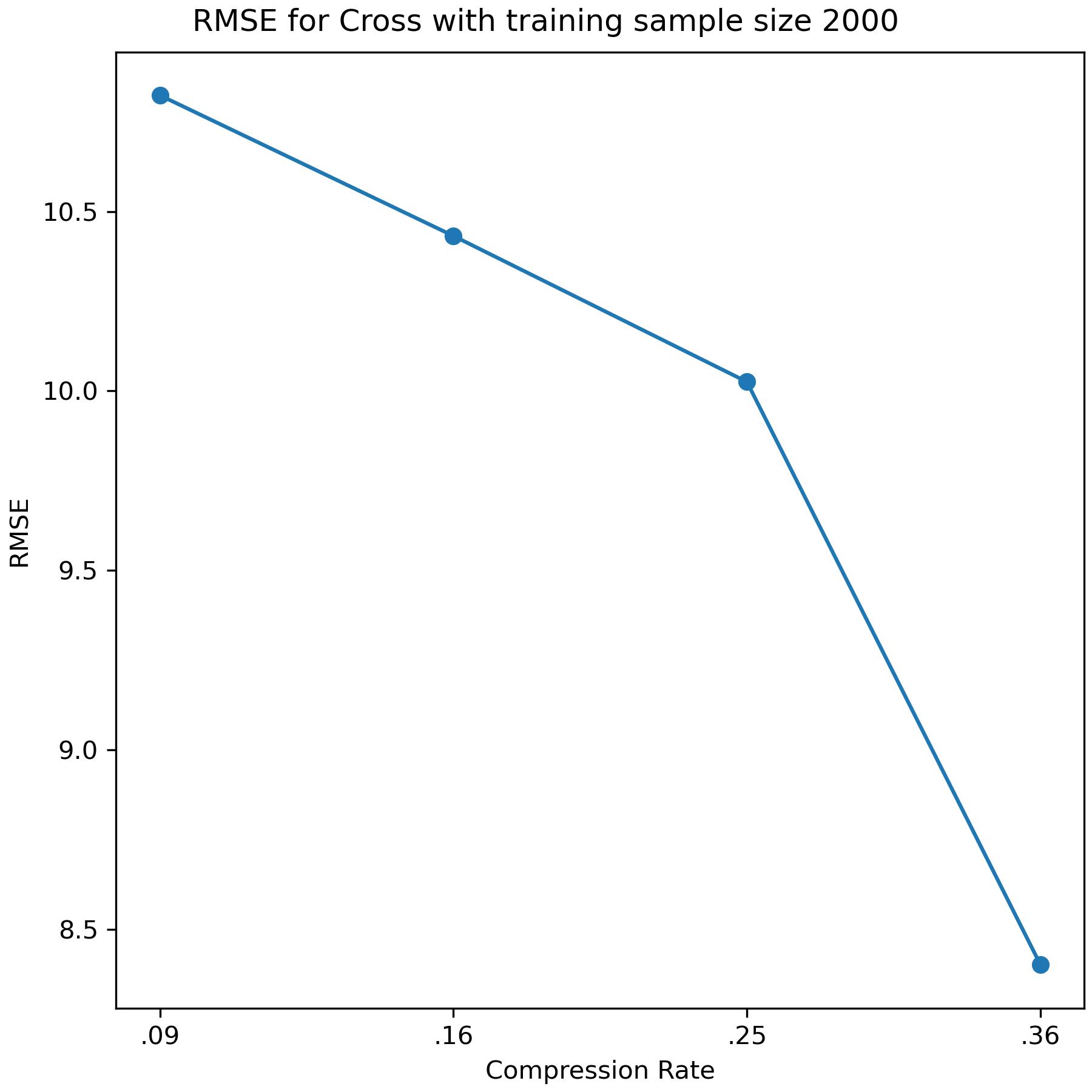}
    \end{tabular}
    \caption{Prediction performances of different compression rates $r\in\{0.09, 0.16,0.25,$ $0.36\}$ using different training sample size $n\in\{500, 1000,1500,2000\}$ (rows). First row: scatter plots of actual data (horizontal axis) versus predicted data (vertical axis) with regression lines for different compression rates in different colors ($r=0.09$: blue, $r=0.16$: orange, $r=0.25$: green, and $r=0.36$: red). Second row: prediction RMSE (vertical axis) for different compression rates (horizontal axis).}
    \label{fig: comp_rate}
\end{figure}

\renewcommand{\thesection}{D}
\renewcommand{\theequation}{D.\arabic{equation}}
\renewcommand{\thefigure}{D.\arabic{figure}}
\renewcommand{\thetable}{D.\arabic{table}}
\setcounter{table}{0}
\setcounter{figure}{0}
\setcounter{equation}{0}

\section{Further empirical results}\label{app:RealRes}

\begin{figure}[ht!]
\vspace*{2pt}
    \centering
    \makebox[\textwidth][c]{ 
    \begin{tabular}{ccc}
          \raisebox{20mm}{BTR} & \includegraphics[width=.4\linewidth]{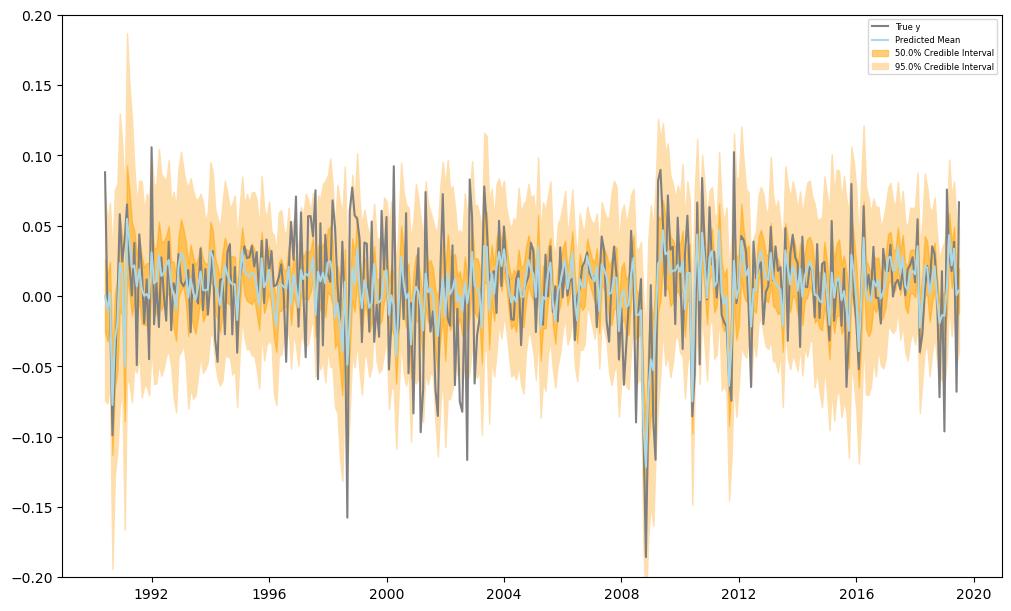} & \includegraphics[width=.4\linewidth]{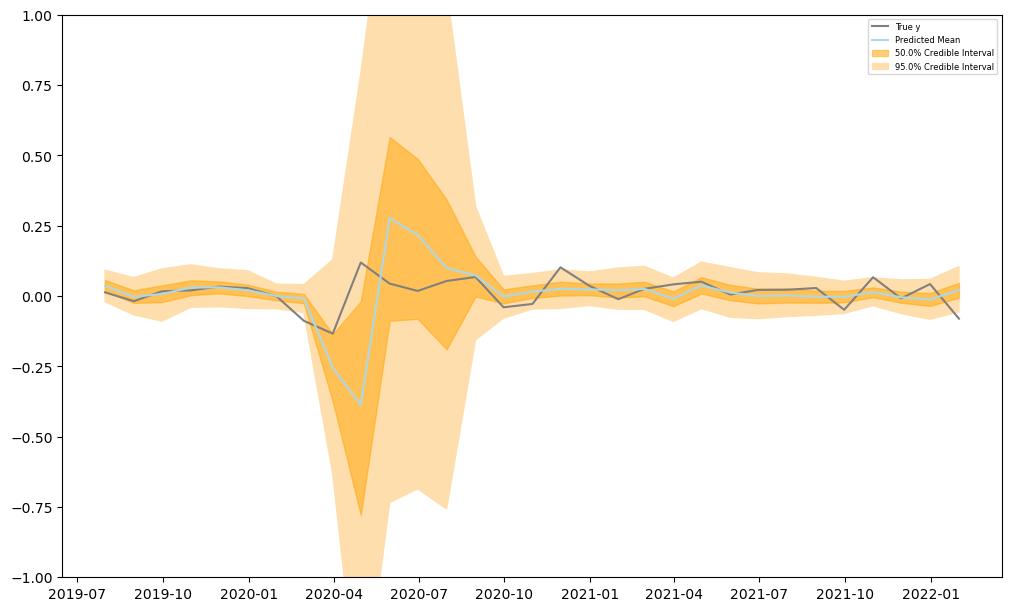} \\
          \raisebox{20mm}{TW$(0)$} & \includegraphics[width=.4\linewidth]{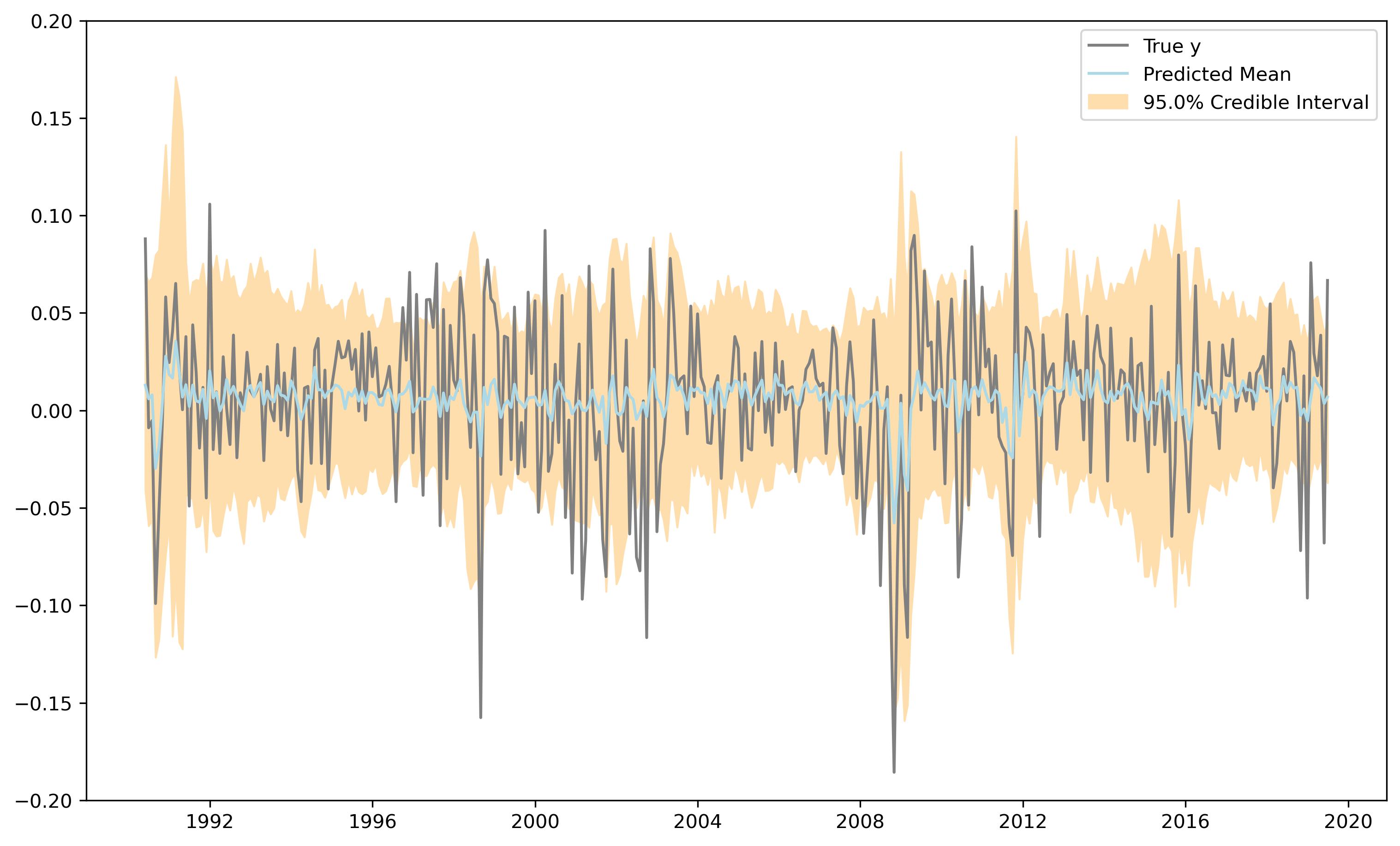}& \includegraphics[width=.4\linewidth]{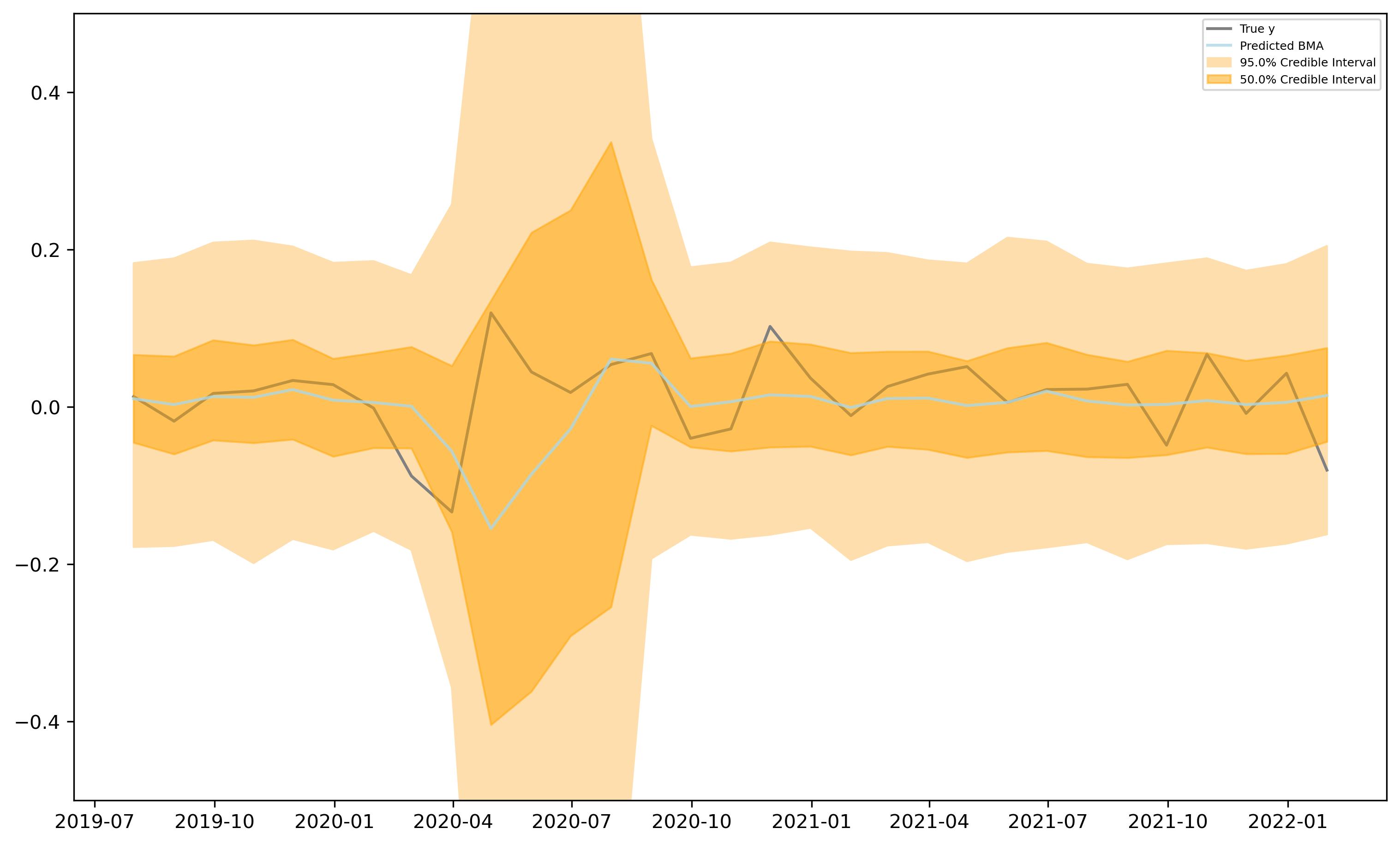}\\
          \raisebox{20mm}{MW$(0)$} & \includegraphics[width=.4\linewidth]{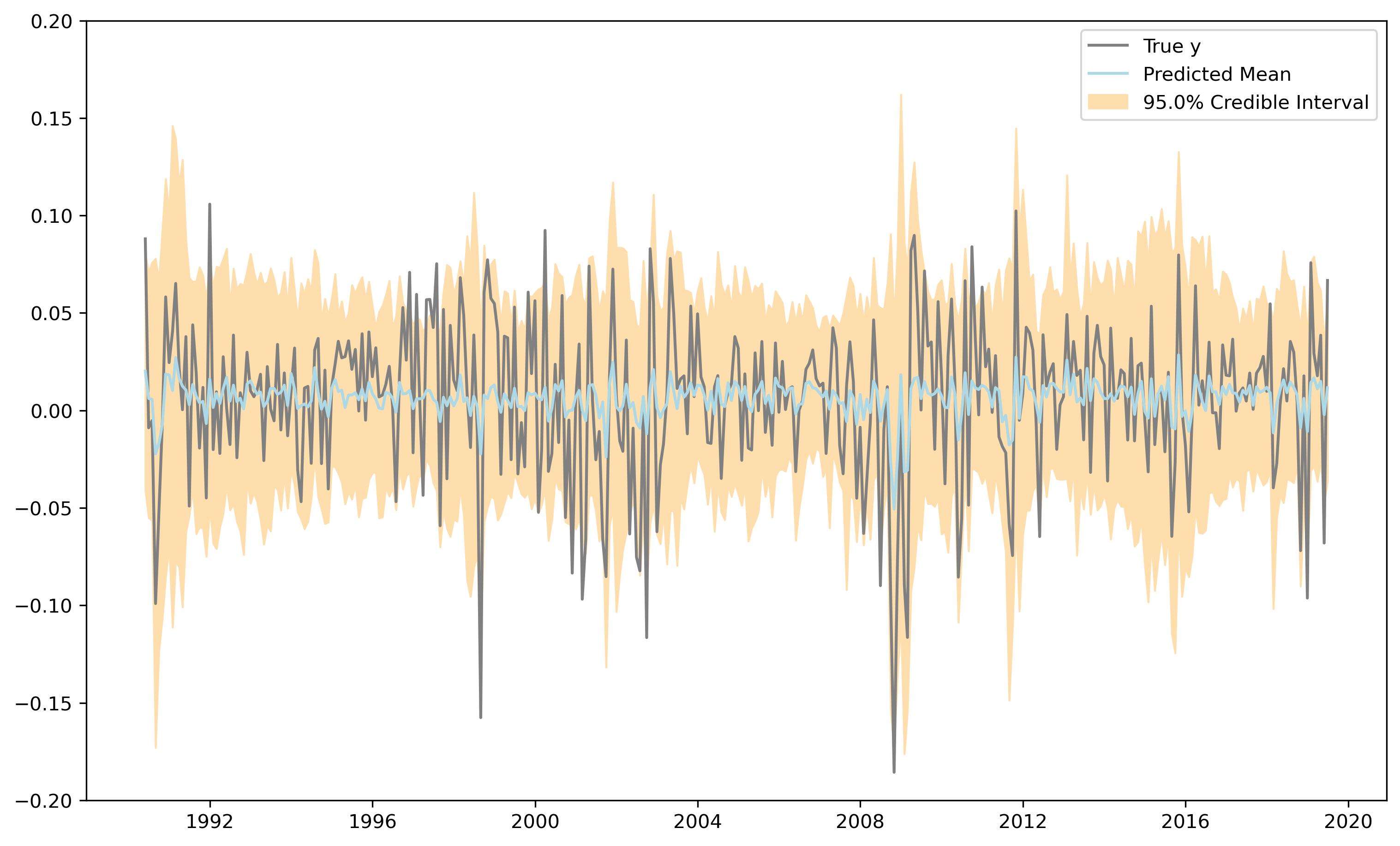} & \includegraphics[width=.4\linewidth]{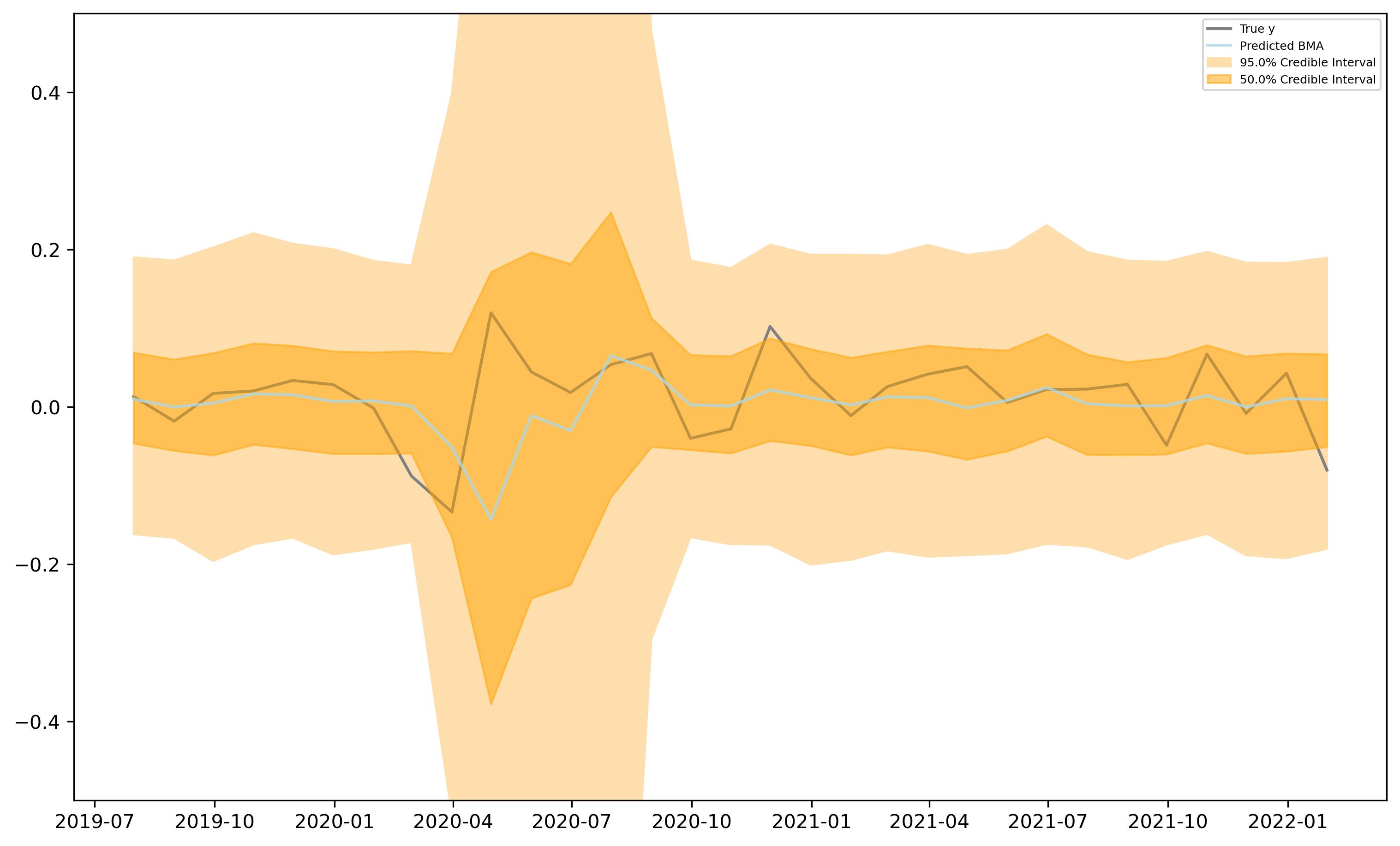}\\
          \raisebox{20mm}{MW$(1)$} & \includegraphics[width=.4\linewidth]{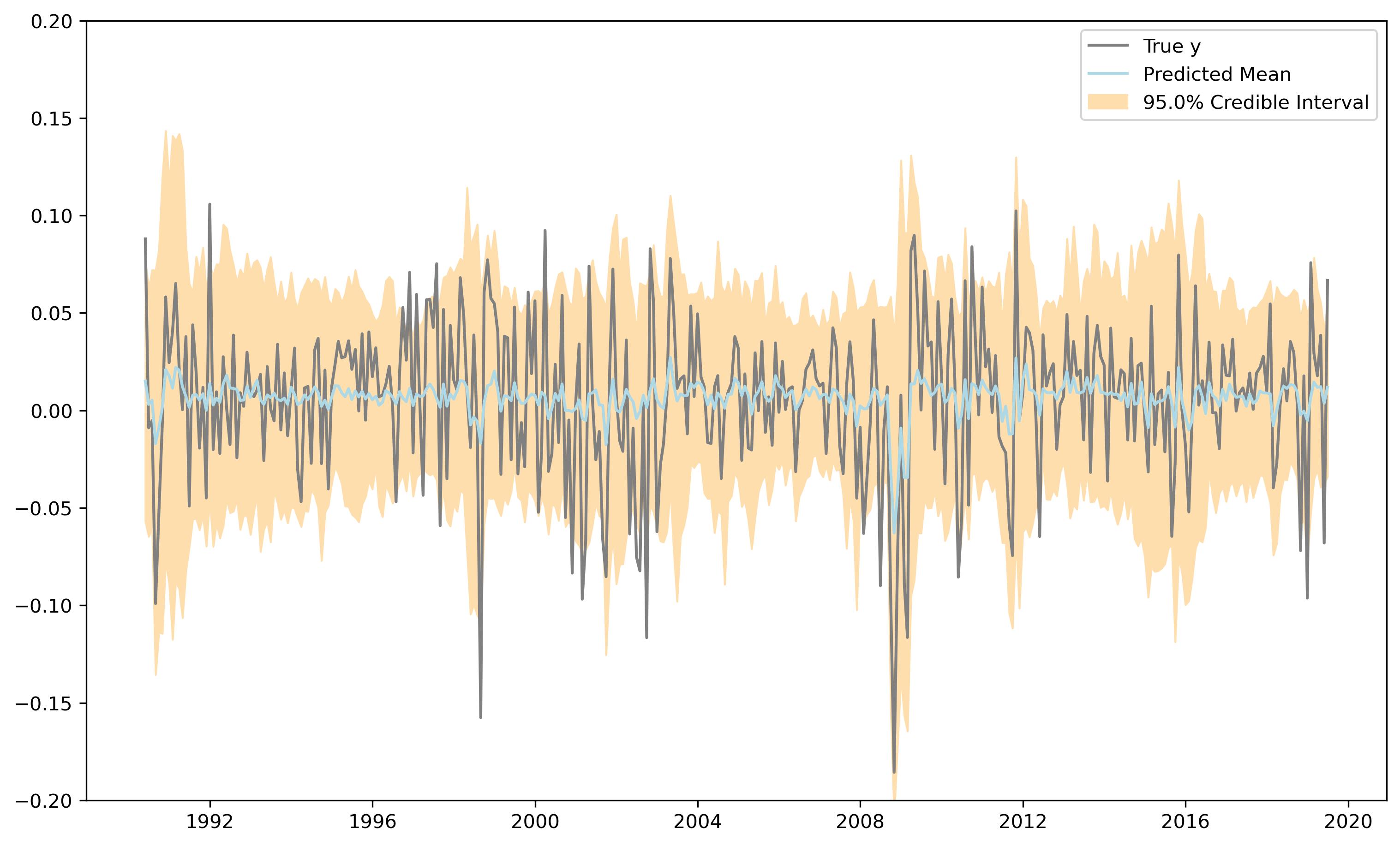} & \includegraphics[width=.4\linewidth]{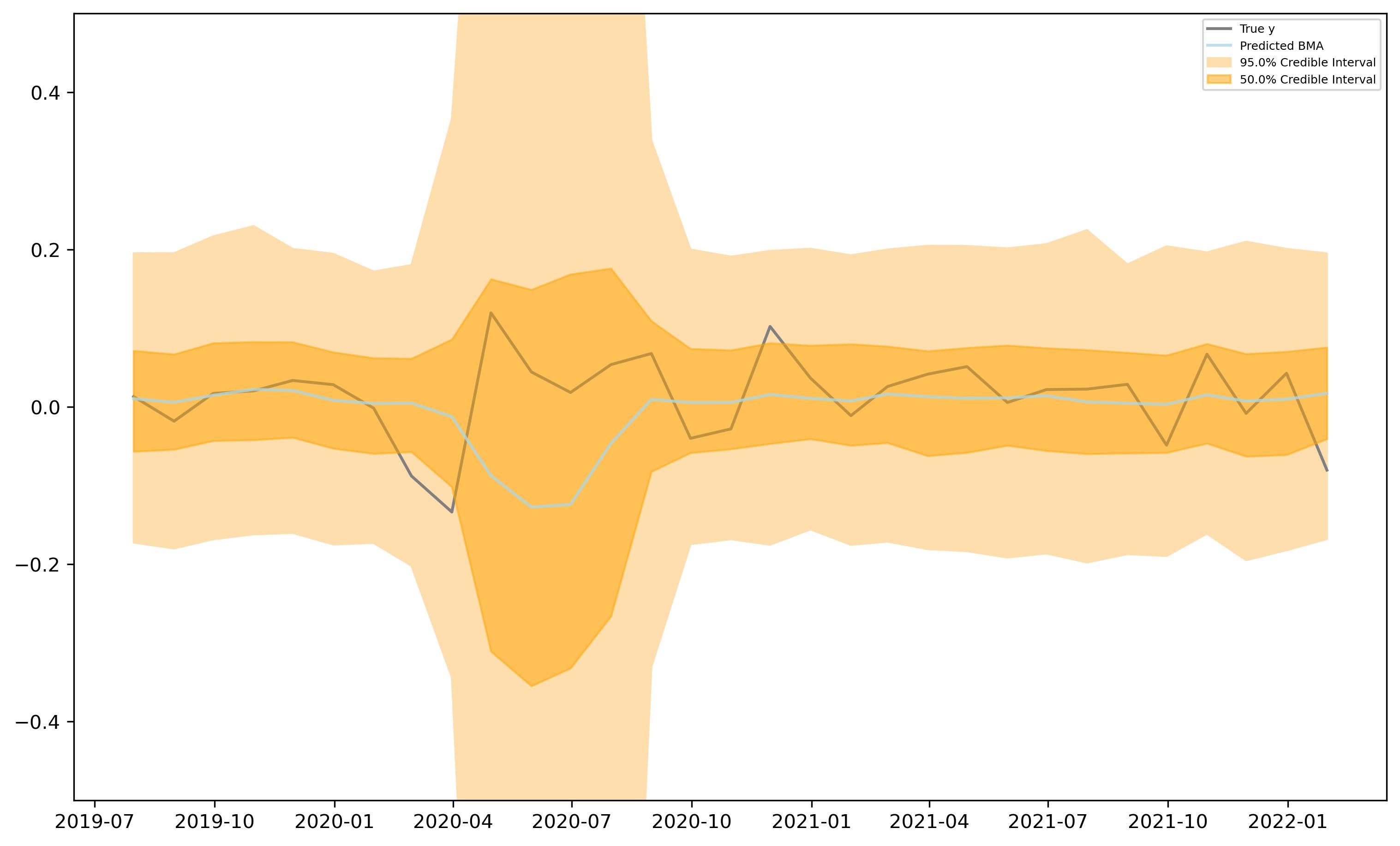} \\
           \raisebox{20mm}{MW$(1,2)$} & \includegraphics[width=.4\linewidth]{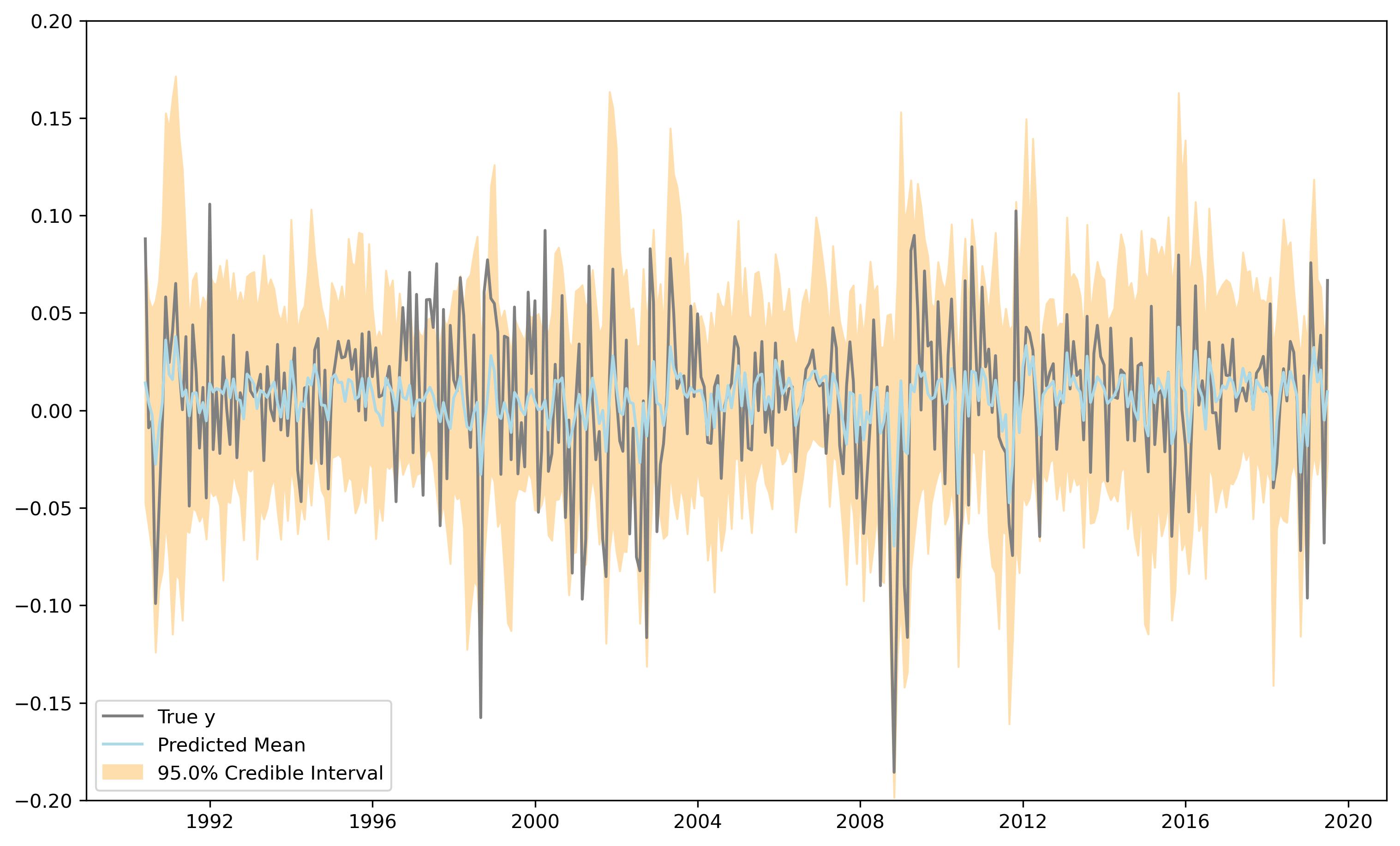} & \includegraphics[width=.4\linewidth]{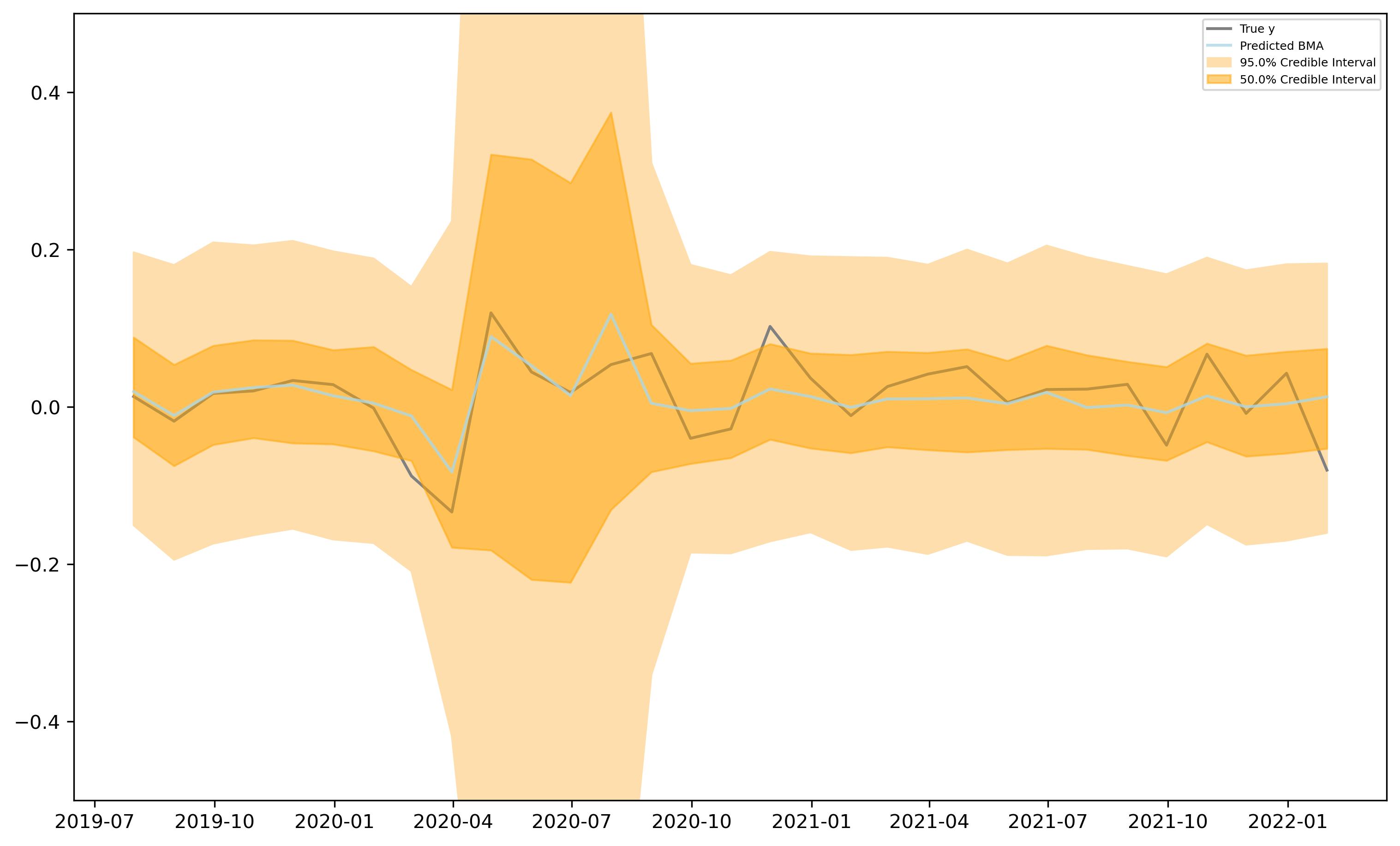}
    \end{tabular}
    }
    \caption{Fitting comparison between BTR and CBTR with different random projection methods. First column: in-sample fitting. Second column: out-of-sample prediction. Actual data are shown in gray solid line, predicted values are shown in blue solid line, light and dark orange colors represent $95\%$ and $50\%$ credible intervals, respectively.}
    \label{fig: pred_comp}
\end{figure}

The left axis of Figure \ref{fig:comp_time} presents the computational time on a logarithmic scale. As the compression rate increases, the computational time also increases because more predictors are retained after compression. Nevertheless, compared with the BTR models using LASSO and Gaussian priors, the CBTR models remain approximately two orders of magnitude faster.

To evaluate the trade-off between predictive performance and computational cost, we report the efficiency score defined as
$\text{Efficiency Score}=\frac{1}{\text{RMSE}\times \text{Cost}}$,
where ``Cost’’ denotes the computational time measured in hours. A higher efficiency score indicates better predictive performance per unit of computational cost.

The black dashed line with red markers in Figure \ref{fig:comp_time} displays the efficiency scores across the competing models. The CBTR models achieve substantially higher efficiency scores than the BTR models with LASSO and Gaussian priors, highlighting their favorable balance between predictive accuracy and computational efficiency. As expected, the efficiency scores decrease as the compression rate increases, reflecting the higher computational burden associated with retaining more predictors after compression.

\begin{figure}[h!]
\vspace*{2pt}
    \centering
    \includegraphics[width=0.6\linewidth]{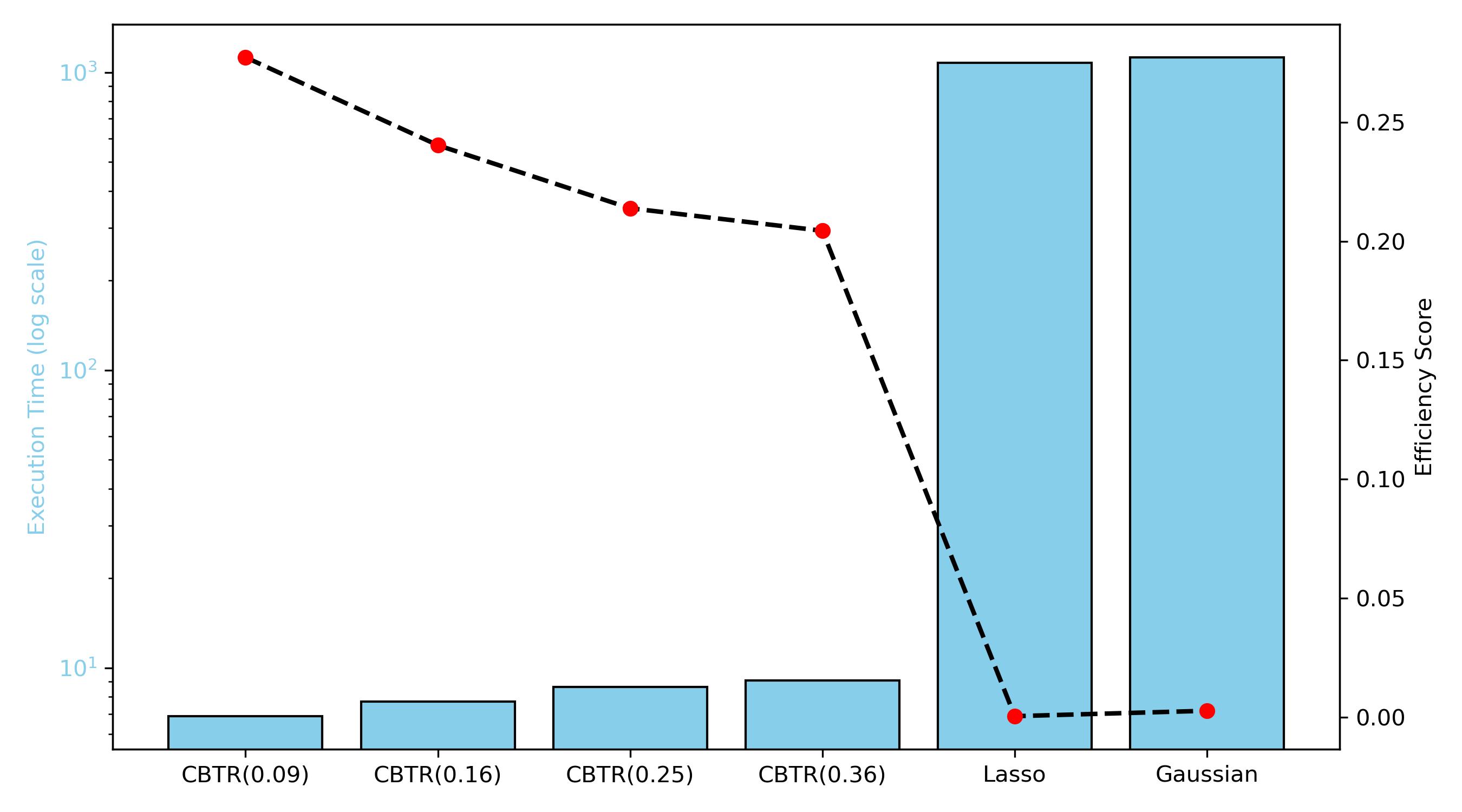}
    \caption{Total computational cost in a $\log$ scale (blue bars, left vertical axis) and efficiency scores (red dots, right vertical axis) for the Compressed Bayesian Tensor Regression with different compression rates $r\in\{0.09, 0.16,0.25,0.36\}$ (CBTR($r$)), the Bayesian Lasso regression (Lasso) and the Gaussian regression (Gaussian).}
    \label{fig:comp_time}
\end{figure}

\end{appendices}

\clearpage
\bibliography{bibliography}

\end{document}